\newif\ifSHOWPROOF \SHOWPROOFtrue
\newif\ifARXIV \ARXIVtrue
\newif\ifEXTERNALIZE \EXTERNALIZEtrue
\ifARXIV%
\documentclass[10pt,acmsmall,screen]{acmart}
\else%
\documentclass[10pt,acmsmall,screen]{acmart}
\fi

\usepackage{enumitem}
\setlist[itemize]{leftmargin=20pt}

\usepackage{xspace}
\usepackage{thm-restate}
\hypersetup{hypertexnames=false}

\usepackage[capitalize,nameinlink]{cleveref}
\crefname{relation}{relation}{relations}
\creflabelformat{relation}{#2\textup{(#1)}#3}

\usepackage{physics}
\usepackage{mathtools}
\usepackage{diagbox}

\usepackage{adjustbox}

\usepackage{tikz}
\usepackage{tikz-cd}
\usetikzlibrary{shapes.geometric,calc,positioning}
\usetikzlibrary{quantikz2}
\def\scaleratio{.8}

\usepackage[normalem]{ulem}

\tikzcdset{
arrow style=tikz,
diagrams={>={Straight Barb[scale=0.85]}}
}
\tikzset{
startip/.tip={Glyph[glyph math command=mysym]},
Rightarrow*/.style={double equal sign distance,>={Implies},->.startip,shorten >=-0.54ex},
to*/.style={->.startip,shorten >=-0.85ex},
}

\usepackage{algpseudocode}
\usepackage{algorithm}
\algrenewcommand\algorithmicrequire{\textbf{Input:}}
\algrenewcommand\algorithmicensure{\textbf{Output:}}

\usepackage{ebproof}

\usepackage{bbold}
\usepackage{stmaryrd}

\DeclarePairedDelimiter\rbra{\lparen}{\rparen}
\DeclarePairedDelimiter\sbra{\lbrack}{\rbrack}
\DeclarePairedDelimiter\cbra{\{}{\}}
\DeclarePairedDelimiter\sem{\llbracket}{\rrbracket}
\DeclarePairedDelimiter\lsem{\llparenthesis}{\rrparenthesis}

\newcommand{\ZZ}{\mathbb{Z}}
\newcommand{\NN}{\mathbb{N}}

\newcommand{\cE}{\mathcal{E}}
\newcommand{\cG}{\mathcal{G}}

\newcommand{\cT}{\mathcal{T}}

\DeclareMathOperator{\hshift}{\mathsf{shift}}
\DeclareMathOperator{\shift}{\mathsf{shift}}
\DeclareMathOperator{\hpermute}{\mathsf{permute}}
\DeclareMathOperator{\hdim}{\mathsf{dim}}

\newcommand{\defeq}{\Coloneqq}
\newcommand{\fromto}{\leftrightarrow}
\newcommand{\seqq}{\fatsemi}

\newcommand{\bbz}{\ensuremath{\mathbb{0}}}
\newcommand{\bbo}{\ensuremath{\mathbb{1}}}
\newcommand{\bb}[1]{\mathbb{#1}}
\newcommand{\bool}{\ensuremath{\mathit{bool}}}

\newcommand{\tconv}{\mathit{tconv}}
\newcommand{\pid}{\mathit{id}}

\newcommand{\swapp}{\mathit{swap}^+}
\newcommand{\swapt}{\mathit{swap}^\times}
\newcommand{\hadamard}{\mathit{hadamard}}
\newcommand{\assoclp}{\mathit{assocl}^+}
\newcommand{\assocrp}{\mathit{assocr}^+}
\newcommand{\assoclt}{\mathit{assocl}^\times}
\newcommand{\assocrt}{\mathit{assocr}^\times}
\newcommand{\unitep}{\mathit{unite}^+}
\newcommand{\unitet}{\mathit{unite}^\times}
\newcommand{\unitip}{\mathit{uniti}^+}
\newcommand{\unitit}{\mathit{uniti}^\times}
\newcommand{\dist}{\mathit{dist}}
\newcommand{\absorb}{\mathit{absorb}}
\newcommand{\factor}{\mathit{factor}}
\newcommand{\factorz}{\mathit{factorz}}
\newcommand{\pctrl}[1]{{\p{ctrl}\ #1}}

\DeclareMathOperator{\lde}{lde}
\DeclareMathOperator{\level}{level}
\DeclareMathOperator{\Real}{\mathsf{Re}}
\DeclareMathOperator{\Imag}{\mathsf{Im}}

\newcommand{\edge}[3]{\langle #1, #2, #3\rangle}

\usepackage{suffix}
\newcommand\Se[1]{\xrightarrow{#1}}
\WithSuffix\newcommand\Se*[1]{\xrightarrow{#1}{}^{\!\!*}\,}
\newcommand\Ne[1]{\xRightarrow{#1}}
\WithSuffix\newcommand\Ne*[1]{\xRightarrow{#1}{}^{\!\!*}\,}

\newcommand{\HPiLang}{\text{Hadamard-$\Pi$}\xspace}
\newcommand{\QPiLang}{\text{Q-$\Pi$}\xspace}
\newcommand{\HPiUnitary}{\ensuremath{O_n\rbra{\ZZ\sbra{\tfrac{1}{\sqrt{2}}}}}}
\newcommand{\HPiRing}{\ensuremath{\ZZ\sbra{{\tfrac{1}{\sqrt{2}}}}}}

\newcommand{\Unitary}{\ensuremath{\mathbf{Unitary}}}
\newcommand{\FinBij}{\ensuremath{\mathbf{FinBij}}}

\newcommand{\blue}[1]{{\color{blue}#1}}

\newif\ifFINAL %\FINALtrue
\ifFINAL

\newcommand{\chdeleted}[1]{}
\else

\newcommand{\chdeleted}[1]{{\color{red!30}\sout{#1}}}
\fi

\usepackage{robust-externalize}
\robExtConfigure{
    compile in parallel,
    rename backup files for arxiv,
}
\newcommandAutoForward{\qsem}[1]{\lsem{#1}_Q}
\newcommandAutoForward{\wsem}[1]{\lsem{#1}_W}
\newcommandAutoForward{\qT}[2][]{\cT_{Q}^{#1}\sbra*{#2}}
\newcommandAutoForward{\p}[1]{\mathsf{#1}}
\newcommandAutoForward{\hT}[1]{\cT_{\mathit{H}}\sbra*{#1}}
\newcommandAutoForward{\hqT}[1]{\cT_{\mathit{HQ}}\sbra*{#1}}
\newcommandAutoForward{\CH}[1]{\ensuremath{H_{[#1]}}}
\newcommandAutoForward{\CZ}[1]{\ensuremath{(-1)_{[#1]}}}
\newcommandAutoForward{\CX}[1]{\ensuremath{X_{[#1]}}}
\robExtConfigure{
    new preset={templateQuantikz}{
        latex,
        auto forward,
        add to preamble={
            \pdfminorversion=5
            \RequirePackage[T1]{fontenc}
            \RequirePackage[tt=false, type1=true]{libertine}
            \RequirePackage[varqu]{zi4}
            \RequirePackage[libertine]{newtxmath}
            \usepackage{tikz}
            \usetikzlibrary{quantikz2}
            \usepackage{bbold}
            \usepackage{stmaryrd}
            \usepackage{amsmath,mathtools}
            \DeclarePairedDelimiter\rbra{\lparen}{\rparen}
            \DeclarePairedDelimiter\sbra{\lbrack}{\rbrack}
            \DeclarePairedDelimiter\cbra{\{}{\}}
            \DeclarePairedDelimiter\sem{\llbracket}{\rrbracket}
            \DeclarePairedDelimiter\lsem{\llparenthesis}{\rrparenthesis}
            \newcommand{\ZZ}{\mathbb{Z}}
            \newcommand{\NN}{\mathbb{N}}
            
            \newcommand{\cG}{\mathcal{G}}
            
            \newcommand{\cT}{\mathcal{T}}
            \DeclareMathOperator{\hshift}{\mathsf{shift}}
            \DeclareMathOperator{\shift}{\mathsf{shift}}
            \DeclareMathOperator{\hpermute}{\mathsf{permute}}
            \DeclareMathOperator{\hdim}{\mathsf{dim}}

            \newcommand{\defeq}{\Coloneqq}
            \newcommand{\fromto}{\leftrightarrow}
            \newcommand{\seqq}{\fatsemi}
            \newcommand{\bbz}{\ensuremath{\mathbb{0}}}
            \newcommand{\bbo}{\ensuremath{\mathbb{1}}}
            \newcommand{\bool}{\ensuremath{\mathit{bool}}}
            \newcommand{\tconv}{\mathit{tconv}}
            \newcommand{\pid}{\mathit{id}}

            \newcommand{\swapp}{\mathit{swap}^+}
            \newcommand{\swapt}{\mathit{swap}^\times}
            \newcommand{\hadamard}{\mathit{hadamard}}
            \newcommand{\assoclp}{\mathit{assocl}^+}
            \newcommand{\assocrp}{\mathit{assocr}^+}
            \newcommand{\assoclt}{\mathit{assocl}^\times}
            \newcommand{\assocrt}{\mathit{assocr}^\times}
            \newcommand{\unitep}{\mathit{unite}^+}
            \newcommand{\unitet}{\mathit{unite}^\times}
            \newcommand{\unitip}{\mathit{uniti}^+}
            \newcommand{\unitit}{\mathit{uniti}^\times}
            \newcommand{\dist}{\mathit{dist}}
            \newcommand{\absorb}{\mathit{absorb}}
            \newcommand{\factor}{\mathit{factor}}
            \newcommand{\factorz}{\mathit{factorz}}
        },
    },
}
\ifEXTERNALIZE%
\cacheEnvironment{quantikz}{templateQuantikz}
\else%
\fi

\newtheorem{maincase}{Case}
\crefname{theorem}{Thm.}{Thms.}
\crefname{lemma}{Lem.}{Lems.}
\crefname{definition}{Def.}{Defs.}
\crefname{proposition}{Prop.}{Props.}

\crefname{section}{Sect.}{Sects.}
\crefname{subsection}{Sect.}{Sects.}
\crefname{appendix}{Appx.}{Appxs.}
\crefname{figure}{Fig.}{Figs.}
\crefname{maincase}{Case}{Cases}

\author{Wang Fang}
\orcid{0000-0001-7628-1185}
\affiliation{%
  \department{School of Informatics}
  \institution{University of Edinburgh}
  \city{Edinburgh}
  \country{UK}
}
\email{wang.fang@ed.ac.uk}

\author{Chris Heunen}
\orcid{0000-0001-7393-2640}
\affiliation{%
  \department{School of Informatics}
  \institution{University of Edinburgh}
  \city{Edinburgh}
  \country{UK}
}
\email{chris.heunen@ed.ac.uk}

\author{Robin Kaarsgaard}
\orcid{0000-0002-7672-799X}
\affiliation{%
  \department{Centre for Quantum Mathematics, 
  Department of Mathematics and Computer Science}
  \institution{University of Southern Denmark}
  \city{Odense}
  \country{Denmark}
}
\email{kaarsgaard@imada.sdu.dk}


\setcopyright{cc}
\setcctype{by}
\acmDOI{10.1145/3776647}
\acmYear{2026}
\acmJournal{PACMPL}
\acmVolume{10}
\acmNumber{POPL}
\acmArticle{5}
\acmMonth{1}
\received{2025-06-26}
\received[accepted]{2025-11-06}

\ifARXIV
\else
\usepackage{xr}
\fi

\begin{document}

\title{Hadamard-Pi: Equational Quantum Programming}

\begin{abstract}
Quantum computing offers advantages over classical computation, yet the precise features that set the two apart remain unclear. In the standard quantum circuit model, adding a 1-qubit basis-changing gate---commonly chosen to be the Hadamard gate---to a universal set of classical reversible gates yields computationally universal quantum computation. However, the computational behaviours enabled by this addition are not fully characterised.
We give such a characterisation by introducing a small quantum programming language  extending the universal classical reversible programming language $\Pi$ with a single primitive corresponding to the Hadamard gate. The language comes equipped with a sound and complete categorical semantics that is specified by a purely equational theory. Completeness is shown by means of a novel finite presentation, and a corresponding synthesis algorithm, for the groups of orthogonal matrices with entries in the ring $\ZZ\sbra{\tfrac{1}{\sqrt{2}}}$.
\end{abstract}


\maketitle

%!TEX root = 0_main.tex
\section{Introduction}\label{sec:intro}

Quantum computing has many advantages over classical computing~\cite{kitaevshenvyalyi}, but it remains unclear what exactly differentiates it~\cite{vedral:elusive}, and how to best equip programmers to take advantage of these advantages~\cite{hiemetal:quantumprogramming}.
In the prevalent quantum circuit model, there is one well-known answer: adding any single-qubit gate that does not preserve the computational basis suffices to upgrade a set of classical universal gates to universal quantum computation~\cite{shi2002toffolicontrollednotneedlittle}. The standard choice for such a gate is the Hadamard gate~\cite{aharonov2003simpleprooftoffolihadamard}.
This article makes this idea precise by giving a concrete low-level but universal quantum programming language that is discrete and has sound and complete semantics governed by a small number of equations.

We prove these properties of the language in three stages.
Starting with $\Pi$, a combinator language that is known to be universal for classical reversible computing~\cite{james2012information,carette2016computing,carettejamessabry:reversible} by governing controlled gates, we define the language \HPiLang by adding a primitive for the Hadamard gate.
It already follows from known results~\cite{aharonov2003simpleprooftoffolihadamard,shi2002toffolicontrollednotneedlittle} that this language is computationally universal for reversible quantum computing, but to establish the other desired properties, we will use two auxiliary languages.
First, we introduce the language \QPiLang, which additionally allows for the scalar $-1$, and has a pleasingly simple equational theory. We provide syntactic translations between both languages by emulating the scalar $-1$ with an auxiliary dimension. Second, we further provide translations of \QPiLang into a language that directly describes the operational semantics of orthogonal matrices. 

Additionally, we develop denotational semantics for all three languages and show that the translations preserve this. This is done in terms of rig categories and functors into the category of unitary matrices, which serves as the standard model of quantum circuits; see \cref{fig:translations}.

Further contributions are as follows:
\begin{itemize}
  \item We give a proof that the categorical semantics are sound and complete.
  \item To do so, we provide a presentation of the groups $O_n(\ZZ[\tfrac{1}{\sqrt{2}}])$ of orthogonal matrices with entries in the ring $\ZZ[\tfrac{1}{\sqrt{2}}]$ with finitely many generators $\mathcal{G}_n$ and relations.
  \item We provide an exact synthesis algorithm that turns an orthogonal matrix in $O_n(\ZZ\sbra{\tfrac{1}{\sqrt{2}}})$ into a sequence of words over $\mathcal{G}_n$ as a unique normal form. 
\end{itemize}

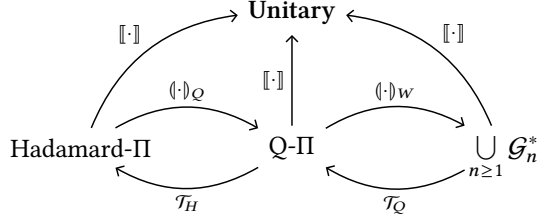
\begin{figure}
    \[\begin{tikzcd}[row sep=1.2cm,column sep=1cm,nodes={inner sep=4pt}]
        \setwiretype{n} & \Unitary{} & \\
        \setwiretype{n}\HPiLang{} \ar[ru,"\sem{\cdot}",bend left] \ar[r,bend left,"\qsem{\cdot}"] & \QPiLang{} \ar[u,"\sem{\cdot}"] \ar[l,bend left,"\cT_H"] \ar[r,bend left,"\wsem{\cdot}"] & \phantom{Words over }\tikz[overlay,baseline]{\node[xshift=-.65cm,anchor=base]{$\bigcup\limits_{n\geq 1}\cG_n^*$}}\phantom{\cG} \ar[l,bend left,"\cT_Q"] \ar[lu,"\sem{\cdot}",bend right,swap]
    \end{tikzcd}\]
	\caption{The three languages: their translations and semantics.}
	\label{fig:translations}
\end{figure}

We proceed as follows: \Cref{sec:background} gives a brief background in quantum circuits, rig categories, and the family of $\Pi$ languages. Next, \Cref{sec:language} introduces the languages \HPiLang and \QPiLang and their translations, and \Cref{sec:relations} then presents the relation to orthogonal matrices. The soundness and completeness are proven in \Cref{sec:soundness_and_completeness}. Finally, related work and concluding remarks are discussed in \Cref{sec:conclusion}. Various proofs are detailed in \ifARXIV%
\cref{app:relations_proof,app:qpi_proof,app:hpi_proof,app:misc}%
\else%
\cref{ext-app:relations_proof,ext-app:qpi_proof,ext-app:hpi_proof,ext-app:misc} of the extended version~\cite{FHK25}%
\fi.

%!TEX root = 0_main.tex
\section{Background}\label{sec:background}

We briefly describe reversible quantum computation in terms of unitary matrices, and reversible classical computing in terms of the language $\Pi$ and its categorical semantics.

\subsection{Reversible Quantum Computing}

Single qubits have as states unit vectors in $\mathbb{C}^2$, with classical bits corresponding to the computational basis states $\ket{0} = \left(\begin{smallmatrix} 1 \\ 0 \end{smallmatrix}\right)$ and $\ket{1} = \left( \begin{smallmatrix} 0 \\ 1 \end{smallmatrix}\right)$. 
Multiple qubits take states in the tensor product $\mathbb{C}^2 \otimes \cdots \otimes \mathbb{C}^2$. 
A system of $n$ qubits evolves along \emph{gates}: $2^n$-by-$2^n$ complex matrices $U$ that are unitary, that is, whose complex conjugate transpose $U^\dag$ is its inverse.
Important examples of gates are the \emph{Hadamard} and $X$ gates
\[
  H = \frac{1}{\sqrt{2}} \begin{pmatrix}
    1 & 1 \\ 
    1 & -1
  \end{pmatrix},
  \qquad \qquad
  X = \begin{pmatrix}
    0 & 1 \\
    1 & 0      
  \end{pmatrix}\text.
\]
If $U$ is an $n$-qubit gate, then the block diagonal matrix
\[
  \mathit{CU} = \begin{pmatrix}
    I &  \\
     & U      
  \end{pmatrix}
\]
is an $(n+1)$-qubit gate, called \emph{controlled}-$U$, where $I$ is the $2^n$-by-$2^n$ identity matrix.

Indeed, any $n$-qubit unitary can be approximated arbitrarily well with a \emph{quantum circuit} built up using composition and tensor product from some basic gate sets.
More precisely, a set $G$ of gates is \emph{strictly universal} if there exists a constant $n_0$ such that for any $n \geq n_0$, the subgroup generated by $G$ is dense in the group of all $2^n$-by-$2^n$ unitary matrices with determinant 1. 
The set $G$ is \emph{computationally universal} if it can be used to simulate to within $\epsilon$ error any quantum circuit on $n$ qubits with $t$ gates from a strictly universal set with only polylogarithmic overhead in $(n, t, 1/\epsilon)$.
Since any $n$-qubit unitary $U$ can be simulated by an $(n+1)$-qubit real unitary
\begin{equation}\label{eq:real_U}
    \mathcal{E}(U) = \begin{pmatrix} \Real(U) & -\Imag(U) \\ \Imag(U) & \Real(U) \end{pmatrix},
\end{equation}
where $\Real(U)$ is the real part of $U$ and $\Imag(U)$ is the imaginary part of $U$, a well-known real gate set $\{H,CCX\}$ is computationally universal~\cite{aharonov2003simpleprooftoffolihadamard,shi2002toffolicontrollednotneedlittle}.

Quantum circuits are often drawn graphically, with qubits represented by wires that combine by tensor product and go from left to right as they undergo gates~\cite{kitaevshenvyalyi,yanofskymannucci:quantumcomputing}. For example,
\[
    \scalebox{1}{\begin{quantikz}[row sep=2mm,column sep=2mm]
        &\ctrl{1} & \ctrl{1} & \targ{} && \swap{1} & \\
        &\gate{H} & \targ{} & \octrl{-1} && \targX{} &
    \end{quantikz}} 
\]
is a 2-qubit circuit that applies a controlled Hadamard gate, a controlled X gate, and then a negatively-controlled H gate and a swap gate, defined as follows.
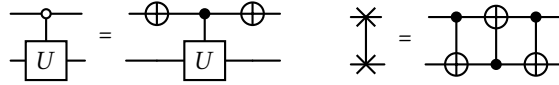
\[
    \scalebox{1}{\begin{quantikz}[row sep=2mm,column sep=2mm]
        &\octrl{1} 
        & \midstick[brackets=none,2]{=} 
        & \targ{} & \ctrl{1} & \targ{} & \\
        &\gate{U} &
        & & \gate{U} & &
    \end{quantikz}} 
    \qquad
    \scalebox{1}{\begin{quantikz}[row sep=3mm,column sep=2mm]
        &\swap{1} & \midstick[brackets=none,4]{=} 
        & \ctrl{1} & \targ{} & \ctrl{1} & \\
        &\targX{} & 
        & \targ{} & \ctrl{-1} & \targ{} & 
    \end{quantikz}}
\]

\subsection{Reversible Classical Computing}

Reversible programming languages govern permutations of finite sets. There is a universal family of languages, called $\Pi$, for reversible classical computing~\cite{james2012information}, that has elegant algebraic type constructors and combinators. It also has an equational theory~\cite{carette2016computing} that is sound and complete for permutation semantics in rig categories~\cite{choudhury2022symmetries}.

\begin{definition}
    [Syntax of $\Pi$]\label{def:pi_syntax}
    The language $\Pi$ has the following syntax for types and combinators.
    \begin{align*}
        b \defeq{}& \bbz \mid \bbo \mid b+b \mid b\times b \tag{value types} \\
        t \defeq{}& b \fromto b \tag{combinator types} \\
        \mathit{iso} \defeq{}& \pid \mid \swapp \mid \assocrp \mid \assoclp \mid \unitep \mid \unitip  \\
        & \hphantom{\pid} \mid \swapt \mid \assocrt \mid \assoclt \mid \unitet \mid \unitit \\
        & \hphantom{\pid} \mid \dist \mid \factor \mid \absorb \mid \factorz \tag{isomorphisms} \\
        c \defeq{}& \mathit{iso} \mid c \seqq c \mid c+c \mid c\times c \tag{combinators}
    \end{align*}
\end{definition}

The value types are formed from empty type ($\bbz$), unit type ($\bbo$), sum type ($+$), and the product type ($\times$).
The combinator type $b\fromto b$ indicates an isomorphism between value types.
Combinators are constructed from primitive isomorphisms $\textit{iso}$ and their compositions are subject to the typing rules in~\cref{fig:typing_rules}.

\begin{figure}[t]
\begin{gather*}
    \arraycolsep=1pt
    \begin{array}{rcl}
        \pid  :{} & b \fromto{} b & : \pid \\
        \swapp :{} & b_1+b_2 \fromto{} b_2 + b_1 & : \swapp \\
        \assocrp :{} & (b_1+b_2)+b_3 \fromto{} b_1 + (b_2+b_3) & : \assoclp \\
        \unitep :{} & \bbz + b \fromto{} b & : \unitip \\
        \swapt :{} & b_1 \times b_2 \fromto{} b_2 \times b_1 & : \swapt \\
        \assocrt :{} & (b_1 \times b_2) \times b_3 \fromto{} b_1 \times  (b_2 \times b_3) & : \assoclt \\
        \unitet :{} & \bbo \times b \fromto{} b & : \unitit \\
        \dist :{} & (b_1 + b_2) \times b_3 \fromto{} (b_1 \times b_3) + (b_2 \times b_3) & : \factor \\
        \absorb :{} & b \times \bbz \fromto{} \bbz & : \factorz
    \end{array} \\
    \begin{prooftree}
        \hypo{c_1 : b_1 \fromto b_2} \hypo{c_2 : b_2 \fromto b_3}
        \infer2{c_1 \seqq c_2 : b_1 \fromto b_3}
    \end{prooftree} \\
    \begin{array}{cc}
        \begin{prooftree}
            \hypo{c_1 : b_1 \fromto b_3} \hypo{c_2 : b_2 \fromto b_4}
            \infer2{c_1 + c_2 : b_1+b_2 \fromto b_3+b_4}
        \end{prooftree} &
        \begin{prooftree}
            \hypo{c_1 : b_1 \fromto b_3} \hypo{c_2 : b_2 \fromto b_4}
            \infer2{c_1 \times c_2 : b_1 \times b_2 \fromto b_3 \times b_4}
        \end{prooftree}
    \end{array}
\end{gather*}
\caption{Typing rules of $\Pi$.}\label{fig:typing_rules}
\end{figure}

The constructors $+$ and $\times$ function like the operations in semirings, which categorifies as follows.

\begin{definition}
    [Monoidal category]
    A \emph{monoidal category} $(\mathbf{C},\otimes,I)$ is a category equipped with:
    \begin{itemize}
        \item a bifunctor $\otimes\colon\mathbf{C}\times\mathbf{C}\to \mathbf{C}$ called tensor product; 
        \item a unit object $I\in \mathrm{Ob}(\mathbf{C})$ for $\otimes$;
        \item a natural isomorphism $\alpha_{\otimes}\colon (A\otimes B)\otimes C \to A\otimes (B\otimes C)$ called the \emph{associator} and natural isomorphisms $\lambda_{\otimes}\colon I\otimes A \to A$ and $\rho_{\otimes}\colon A\otimes I \to A$ called the \emph{unitors}, satisfying certain coherence conditions~\cite[Chapter 1]{heunen2019categories}.
    \end{itemize}
    It is a \emph{symmetric monoidal category} if additionally there is a natural isomorphism $\sigma_{\otimes} \colon  A\otimes B\to B\otimes A$, called \emph{symmetry}, that satisfies certain coherence conditions~\cite[Chapter 1]{heunen2019categories}.
\end{definition}

The monoidal structure allows us to construct objects $A\otimes B$ from objects $A$ and $B$, and morphisms $f\otimes g$ from morphisms $f$ and $g$, which provides two directions of compositions.
For example, $(f_1\otimes g_1)\circ(f_2\otimes g_2) = (f_1\circ f_2)\otimes(g_1\circ g_2)$.

\begin{definition}
    [Rig category]
    A \emph{rig category} $(\mathbf{C}, \otimes, \oplus, I, O)$ is a category with two symmetric monoidal structures, $(\mathbf{C},\otimes,I)$ and $(\mathbf{C},\oplus,O)$, along with a natural isomorphism $\delta_R \colon (A\oplus B)\otimes C \to (A\otimes C) \oplus (B\otimes C)$ called the \emph{distributor} and a natural isomorphism $\delta_0 \colon A\otimes O \to O$ called the \emph{annihilator}, satisfying certain coherence conditions~\cite{laplaza1972coherence}.
\end{definition}

Like quantum circuits, morphisms in monoidal categories are often drawn as string diagrams~\cite[Chapter~1]{heunen2019categories}. We will do so below for rig categories, but for the $\oplus$ monoidal structure rather than $\otimes$; \emph{we will put a $+$ between wires of a string diagram to avoid confusion.}

The language $\Pi$ naturally takes semantics in rig \emph{groupoids}, where all morphisms are invertible: for any morphism $f$, there is an inverse $f^{-1}$ such that $f\circ f^{-1} = \mathrm{id}$ and $f^{-1}\circ f = \mathrm{id}$.
For example, the category \FinBij{} of finite sets and bijections forms a rig groupoid under cartesian product and disjoint union of sets, and the category \Unitary{} of finite-dimensional Hilbert spaces and unitaries forms a rig groupoid under tensor product and direct sum.

\begin{definition}
    [Semantics of $\Pi$]\label{def:pi_semantics}
    The language $\Pi$ has denotational semantics in rig groupoids $(\mathbf{C}, \otimes, \oplus, I, O)$ as follows:
    \begin{itemize}
        \item Types are interpreted as objects of $\mathbf{C}$:
        \begin{align*}
            \sem{\bbz} &= O  & \sem{b_1+b_2} &= \sem{b_1}\oplus\sem{b_2} \\
            \sem{\bbo} &= I & \sem{b_1\times b_2} &= \sem{b_1}\otimes\sem{b_2}
        \end{align*}
        \item Terms are interpreted as isomorphisms of $\mathbf{C}$:
        \begin{gather*}
            \begin{aligned}
                \sem{\pid} &= \mathrm{id} &  \sem{\swapp} &= \sigma_{\oplus} & \sem{\swapt} &= \sigma_{\otimes} \\
                \sem{\dist} &= \delta_{R} & \sem{\assocrp} &= \alpha_{\oplus} & \sem{\assocrt} &= \alpha_{\otimes} \\
                \sem{\factor} &= \delta_{R}^{-1} & \sem{\assoclp} &= \alpha_{\oplus}^{-1} & \sem{\assoclt} &= \alpha_{\otimes}^{-1} \\
                \sem{\absorb} &= \delta_{0} &  \sem{\unitep} &= \lambda_{\oplus} & \sem{\unitet} &= \lambda_{\otimes} \\
                \sem{\factorz} &= \delta_{0}^{-1} & \sem{\unitip} &= \lambda_{\oplus}^{-1} & \sem{\unitit} &= \lambda_{\otimes}^{-1} \\
                \sem{c_1\seqq c_2} &= \sem{c_2}\circ\sem{c_1} & \sem{c_1+c_2} &= \sem{c_1}\oplus\sem{c_2} & \sem{c_1\times c_2} &= \sem{c_1}\otimes\sem{c_2}
            \end{aligned}
    \end{gather*}
    \end{itemize}
\end{definition}

The rig structure-induced equational theory organises a second level of $\Pi$ syntax as level-2 combinators~\cite{carette2016computing}.
Each level-2 combinator is of the form $c_1\fromto_2 c_2$ for appropriate $c_1$ and $c_2$ of the same type $b_1 \fromto b_2$, and asserts that $c_1$ and $c_2$ represent the same bijection.
For example, we have the following level-2 combinators dealing with distributivity:
\begin{align*}
    \rbra*{(c_1+c_2)\times c_3}\seqq \dist &\fromto_2{} \dist\seqq\rbra*{\rbra*{c_1\times c_3}+\rbra*{c_2\times c_3}}, \\
    \factor\seqq\rbra*{(c_1+c_2)\times c_3} &\fromto_2{} \rbra*{\rbra*{c_1\times c_3}+\rbra*{c_2\times c_3}}\seqq\factor.
\end{align*}
All level-2 combinators of $\Pi$ are listed in \ifARXIV%
\cref{app:pi_equations} for convenience%
\else%
\cref{ext-app:pi_equations} of the extended version~\cite{FHK25} for convenience%
\fi.

\begin{theorem}[${\Pi}$ full abstraction and adequacy~\cite{choudhury2022symmetries}]
    The equational theory of $\Pi$ expressed using the level-2
    combinators $\fromto_2$ is sound and complete with respect to its
    semantics in the rig groupoid $\FinBij{}$ of finite sets and
    permutations. 
\end{theorem}

\begin{figure}[th]
    \textbf{Syntax}\hfill{}
    \begin{align*}
        \mathit{iso} &\defeq{} \cdots \mid \mathit{sqrtx} \mid \mathit{w} \tag{isomorphisms}
    \end{align*}
    \textbf{Types}\hfill{}
    \begin{align*}
        \mathit{w} &: \bbo \fromto \bbo & \mathit{sqrtx} &: \bbo+\bbo \fromto \bbo+\bbo
    \end{align*}
    \textbf{Equations}\hfill{}
    \begin{gather*}
        \mathit{w}^8
                \fromto_2{} \pid_{\bbo} \\
        \mathit{sqrtx}^2
                \fromto_2{} \swapp_{\bbo+\bbo} \\
        \unitet\seqq\mathit{sqrtx}\seqq(\pid_{\bbo}+\mathit{w}^2)\seqq\mathit{sqrtx}\seqq \unitit
                \fromto_2{} \mathit{w}^2\times ((\pid_{\bbo}+\mathit{w}^2)\seqq\mathit{sqrtx}\seqq(\pid+\mathit{w}^2))     
    \end{gather*}
    \textbf{Additional Terms}\hfill{}
    \begin{align*}
        \sem{\mathit{w}} &= e^{i\frac{\pi}{4}} & \sem{\mathit{sqrtx}} &= \frac{1}{2}\begin{pmatrix*}[r]
            1+i & 1-i \\
            1-i & 1+i
        \end{pmatrix*}
    \end{align*}
    \caption{$\sqrt{\Pi}$ extension of $\Pi$ and its semantics of additional isomorphisms in the rig groupoid \Unitary{}.}\label{fig:sqrt_pi}
\end{figure}

Our work also follows the recent quantum extensions of $\Pi$~\cite{carette2024few, heunen2022quantum, carette2024bake}, among which $\sqrt{\Pi}$~\cite{carette2024few} (see \cref{fig:sqrt_pi}) is a computationally universal quantum programming language equipped with an equational theory that is sound and complete with respect to several gate sets.
However, the completeness of $\sqrt{\Pi}$ is limited to certain subsets of unitaries that it can express, for example, those with entries in the ring $\mathbb{Z}[\frac{1}{2},i]$.
The completeness for larger sets of unitaries remains open, such as those with entries in the ring $\mathbb{Z}[\frac{1}{\sqrt{2}}, i]$, which $\sqrt{\Pi}$ can express.
There remains an open question whether $c_1 \fromto_2 c_2$ iff $\sem{c_1} = \sem{c_2}$ for two terms $c_1$ and $c_2$ of $\sqrt{\Pi}$.

%!TEX root = 0_main.tex
\section{The Language: \texorpdfstring{\HPiLang}{Hadamard-Pi}}\label{sec:language}

We have mentioned that Toffoli and Hadamard gates are computationally universal for quantum computing and so can efficiently simulate any unitary~\cite{aharonov2003simpleprooftoffolihadamard,shi2002toffolicontrollednotneedlittle}.
Consequently, incorporating an isomorphism $\hadamard:\bbo+\bbo \fromto \bbo+\bbo$ into $\Pi$ suffices to achieve a computationally universal quantum programming language.
However, the story does not end here.
We will present a sound and complete equational theory for the extension of $\Pi$ with only one additional generator $\hadamard$.
We first extend $\Pi$ with two isomorphisms, $\hadamard$ and $(-1) \colon \bbo \fromto \bbo$, which facilitate simple interaction with its accurately expressible unitaries.
Then we hide the isomorphism $(-1)$ by introducing additional equations.

\subsection{Adding Minus One and Hadamard}

Our first extension of $\Pi$, named Q-$\Pi$, is presented in \cref{fig:qpi}, and comes accompanied by an equational theory.

\begin{figure}[t]
    \textbf{Syntax}\hfill{}
    \begin{align*}
        \mathit{iso} &\defeq{} \cdots \mid (-1) \mid \hadamard \tag{isomorphisms}
    \end{align*}
    \textbf{Types}\hfill{}
    \begin{align*}
        (-1) &: \bbo \fromto \bbo & \hadamard &: \bbo+\bbo \fromto \bbo+\bbo
    \end{align*}
    \textbf{Equations}\hfill{}
    \begin{gather*}
        (-1)^2 \fromto_2 \pid_{\bbo} \tag{E1}\label{eq:e1} \\
        \hadamard^2 \fromto_2 \pid_{\bbo+\bbo} \tag{E2}\label{eq:e2} \\
        \hadamard\seqq \swapp \seqq \hadamard \fromto_2 \pid_{\bbo}+(-1) \tag{E3}\label{eq:e3} \\
        \begin{prooftree}
            \hypo{c_1: b_1 \fromto b_1}
            \hypo{c_2: b_2 \fromto b_2}
            \hypo{c_1+c_2 \fromto_2 \pid_{b_1}+\pid_{b_2}}
            \infer3{c_1 \fromto_2 \pid_{b_1}}
        \end{prooftree} \tag{E4}\label{eq:e4}
    \end{gather*}
    \textbf{Additional Terms}\hfill{}
    \begin{align*}
        \sem{(-1)} &= -1 & \sem{\hadamard} &= \frac{1}{\sqrt{2}}\begin{pmatrix*}[r]
            1 & 1 \\
            1 & -1
        \end{pmatrix*}
    \end{align*}
    \caption{\QPiLang{} extension of $\Pi$ and its semantics of additional isomorphisms in the rig groupoid \Unitary{}.}\label{fig:qpi}
\end{figure}

\begin{definition}
    [Equational theory for \QPiLang{}]\label{def:qpi_eq}
    The equational theory for \QPiLang{} contains all level-2 combinators of $\Pi$ and \cref{eq:e1,eq:e2,eq:e3,eq:e4}.
\end{definition}

\cref{eq:e1,eq:e2,eq:e3} describe standard interactions of Pauli-$X$, -$Z$ and Hadamard gates.
\cref{eq:e4} is the cancellative condition for $+$.
This condition is sound for any rig category that embeds into one where the interpreted $\oplus$ is a biproduct, which covers \FinBij{} and \Unitary{}.
Although \cref{eq:e4} may appear a little burden here, it is necessary because the proof of \QPiLang{} occasionally requires additional space for intermediate equations. This extra space corresponds to an extra auxiliary qubit used in Toffoli and Hadamard circuits for efficiently simulating any unitary~\cite{aharonov2003simpleprooftoffolihadamard}. The cancellative condition then allows us to eliminate this auxiliary space when it is no longer needed.%
\footnote{The cancellative condition is only sound for $+$ but not for $\times$. For instance, in \Unitary{}, we have that $A\oplus B = I$ if and only if $A = I$ and $B = I$, but $A\otimes B = I$ does not require $A = I$, e.g., $(-1)\otimes(-1) = 1$. Notably, \cref{eq:e4} allows extracting $c \fromto_2{} \pid$ from $\pid_{\bbo+\bbo}\times c \fromto_2{} \pid$.}

The \emph{controlled}-construct for any term $c:b\fromto b$ is derived in $\Pi$ as
\[\pctrl{c} :(\bbo+\bbo)\times b \fromto (\bbo+\bbo)\times b =  \dist\seqq(\pid_{\bbo\times b}+(\pid_{\bbo}\times c))\seqq \factor.\]
Letting $\bb{2}=\bbo+\bbo$ denote the type of booleans (or qubits), we can construct the following gates in \QPiLang{}:
\begin{align*}
    \p{x}:\bb{2}\fromto\bb{2} &= \swapp, \\
    \p{h}:\bb{2}\fromto\bb{2} &= \hadamard, \\
    \p{cx}:\bb{2}\times\bb{2}\fromto\bb{2}\times\bb{2} &= \pctrl{\p{x}}, \\
    \p{ch}:\bb{2}\times\bb{2}\fromto\bb{2}\times\bb{2} &= \pctrl{\p{h}}, \\
    \p{ccx}:\bb{2}\times\bb{2}\times\bb{2}\fromto\bb{2}\times\bb{2}\times\bb{2} &= \pctrl{\p{cx}}.
\end{align*}
Thus, \QPiLang{} can express any quantum circuit over the gate set $\{\, X, \mathit{CX}, \mathit{CCX}, H, \mathit{CH}\,\}$ and is computationally universal for quantum computing~\cite{aharonov2003simpleprooftoffolihadamard,shi2002toffolicontrollednotneedlittle}.
The unitaries that it can represent exactly correspond to unitaries over the ring \HPiRing{}~\cite{Amy2020numbertheoretic}.

\begin{theorem}
    [\QPiLang{} expressivity]
    \QPiLang{} is computationally universal for quantum circuits.
\end{theorem}

When interpreting \QPiLang{} into \Unitary{}, our first main theorem is as follows.

\begin{theorem}[Soundness and completeness of \QPiLang{}]\label{thm:qpi}
    If $c_1,c_2$ are two terms of \QPiLang{}, then $c_1\fromto_2 c_2$ is a level-2 combinator of \QPiLang{} if and only if $\sem{c_1}=\sem{c_2}$.
\end{theorem}
\begin{proof}
    The proofs are deferred to \cref{sec:qpi_soundness_and_completeness}.
\end{proof}

The following lemma will be useful in \cref{sec:soundness_and_completeness}.
\begin{lemma}\label{lem:q_hhcxhh}
    In \QPiLang{},
    $(\p{h}\times \p{h})\seqq\p{cx}\seqq(\p{h}\times\p{h}) \fromto_2{} \swapt\seqq\p{cx}\seqq\swapt$.
\end{lemma}
\begin{proof}
    Omitting the type conversion isomorphisms $\assocrp$, $\assoclp$, $\unitep$, $\unitip$, $\assocrt$, $\assoclt$, $\unitet$, $\unitit$, $\dist$, $\factor$, $\absorb$, and $\factorz$:
    \begingroup
    \allowdisplaybreaks
    \begin{align*}
        &(\p{h}\times \p{h})\seqq\p{cx}\seqq(\p{h}\times\p{h}) \\
        \fromto_2{}& (\p{h}\times \p{h}) \seqq (\pid_{\bb{2}}+\p{x})\seqq (\p{h}\times \p{h}) \\
        \fromto_2{}& (\p{h}\times \pid_{\bb{2}})\seqq (\pid_{\bb{2}}\times \p{h}) \seqq (\pid_{\bb{2}}+\p{x})\seqq (\pid_{\bb{2}}\times \p{h})\seqq (\p{h}\times \pid_{\bb{2}}) \tag{by bifunctoriality of $\times$} \\
        \fromto_2{}& (\p{h}\times \pid_{\bb{2}})\seqq (\p{h} + \p{h}) \seqq (\pid_{\bb{2}}+\p{x})\seqq (\p{h} + \p{h})\seqq  (\p{h}\times \pid_{\bb{2}}) \tag{by distributivity} \\
        \fromto_2{}& (\p{h}\times \pid_{\bb{2}})\seqq (\pid_{\bb{2}}+\p{h}\seqq\p{x}\seqq\p{h})\seqq  (\p{h}\times \pid_{\bb{2}}) \tag{by bifunctoriality of $+$ and \cref{eq:e2}} \\
        \fromto_2{}& (\p{h}\times \pid_{\bb{2}})\seqq (\pid_{\bbo}+\pid_{\bbo}+\pid_{\bbo}+(-1))\seqq  (\p{h}\times \pid_{\bb{2}}) \tag{by \cref{eq:e3}} \\
        \fromto_2{}& (\p{h}\times \pid_{\bb{2}})\seqq\swapt\seqq (\pid_{\bbo}+\pid_{\bbo}+\pid_{\bbo}+(-1))\seqq\swapt\seqq(\p{h}\times \pid_{\bb{2}}) \tag{by naturality of $\swapt$} \\
        \fromto_2{}&\swapt\seqq(\pid_{\bb{2}}\times \p{h})\seqq (\pid_{\bb{2}}+\pid_{\bbo}+(-1))\seqq(\pid_{\bb{2}}\times \p{h})\seqq \swapt \tag{by naturality of $\swapt$} \\
        \fromto_2{}&\swapt\seqq(\p{h}+ \p{h})\seqq (\pid_{\bb{2}}+\pid_{\bbo}+(-1))\seqq(\p{h}+ \p{h})\seqq \swapt \tag{by distributivity} \\
        \fromto_2{}& \swapt\seqq (\pid_{\bb{2}}+\p{h}\seqq(\pid_{\bbo}+(-1))\seqq\p{h})\seqq \swapt \tag{by bifunctoriality of $+$ and \cref{eq:e2}} \\
        \fromto_2{}& \swapt\seqq (\pid_{\bb{2}}+\p{x})\seqq \swapt \tag{by \cref{eq:e3}} \\
        \fromto_2{}& \swapt\seqq\p{cx}\seqq\swapt. \tag*{\qedhere} 
    \end{align*}
    \endgroup
\end{proof}

\subsection{Hiding Minus One}

When incorporating $(-1)$ into ``higher-dimensional terms'',
e.g., using \cref{eq:e3} to simulate $(-1)$ in type $\bbo+\bbo$, we do not need $(-1)$ as a primitive isomorphism.
Consequently, we extend $\Pi$ to \HPiLang{} as shown in \cref{fig:hadamard_pi} by only introducing $\hadamard$.

\begin{figure}[t]
    \textbf{Syntax}\hfill{}
    \begin{align*}
        \mathit{iso} &\defeq{} \cdots \mid \hadamard \tag{isomorphisms}
    \end{align*}
    \textbf{Types}\hfill{}
    \begin{align*}
        \hadamard &: \bbo+\bbo \fromto \bbo+\bbo
    \end{align*}
    \textbf{Equations}\hfill{}
    \begin{gather*}
        \hadamard^2 \fromto_2 \pid_{\bbo+\bbo} \tag{H1}\label{eq:h1} \\
        \begin{aligned}
            & (\swapp+\pid_{\bbo}) \seqq \assocrp \seqq (\pid_{\bbo}+(\hadamard\seqq\swapp\seqq\hadamard))\seqq \assoclp \\
            &\fromto_2{} \assoclp\seqq (\pid_{\bbo}+(\hadamard\seqq\swapp\seqq\hadamard))\seqq \assocrp\seqq(\swapp+\pid_{\bbo})
        \end{aligned} \tag{H2}\label{eq:h2} \\
        \begin{prooftree}
            \hypo{{c_1: b_1 \fromto b_1}}
            \hypo{{c_2: b_2 \fromto b_2}}
            \hypo{c_1+c_2 \fromto_2 \pid_{b_1}+\pid_{b_2}}
            \infer3{c_1 \fromto_2 \pid_{b_1}}
        \end{prooftree} \tag{H3}\label{eq:h3}
    \end{gather*}
    \textbf{Additional Terms}\hfill{}
    \begin{align*}
        \sem{\hadamard} &= \frac{1}{\sqrt{2}}\begin{pmatrix*}[r]
            1 & 1 \\
            1 & -1
        \end{pmatrix*}
    \end{align*}
    \caption{\HPiLang{} extension of $\Pi$ and its semantics of additional isomorphisms in the rig groupoid \Unitary{}.}\label{fig:hadamard_pi}
\end{figure}

With different lines representing sum types, \cref{eq:h1,eq:h2,eq:h3} have the following graphical illustrations, where we make a shorthand:
$\mathit{hxh} = \hadamard\seqq\swapp\seqq\hadamard$.
\begingroup
\begin{equation*}
    \begin{quantikz}[row sep=1mm]
        \lstick[2,brackets=none]{$+$\hspace{-.65cm}}\push{\bbo} & \gate[2]{\hadamard} & \gate[2]{\hadamard} & \midstick[2,brackets=none]{$\fromto_2$} & \qw\\ 
        \push{\bbo}& & & & \qw 
    \end{quantikz} \tag{H1}
\end{equation*}
\begin{equation*}
    \begin{quantikz}[row sep=2mm,column sep=4mm]
        \lstick[2,brackets=none]{$+$\hspace{-.65cm}}\push{\bbo} & \permute{2,1} & & \midstick[3,brackets=none]{$\fromto_2$} & & \permute{2,1} & \\ 
        \lstick[2,brackets=none]{$+$\hspace{-.65cm}}\push{\bbo} & & \gate[2]{\mathit{hxh}} & & \gate[2]{\mathit{hxh}} & &\\ 
        \push{\bbo} & & & & & & 
    \end{quantikz} \tag{H2}
\end{equation*}
\begin{equation*}
    \begin{quantikz}[column sep=4mm]
        \lstick[2,brackets=none]{$+$\hspace{-.65cm}}\push{b_1} & \gate{c_1} & \midstick[2,brackets=none]{$\fromto_2$} & & \\
        \push{b_2} & \gate{c_2} & & &
    \end{quantikz} \Rightarrow
    \begin{quantikz}[column sep=4mm]
        & \gate{c_1} & \midstick[1,brackets=none]{$\fromto_{2}$} & &
    \end{quantikz} \tag{H3}
\end{equation*}
\endgroup
\cref{eq:h2} also has an equivalent form:
\[ \begin{quantikz}[row sep=2mm,column sep=4mm]
    \lstick[2,brackets=none]{$+$\hspace{-.65cm}}\push{\bbo}& \permute{2,1} & & \permute{2,1} & \midstick[3,brackets=none]{$\fromto_2$} & & \\ 
    \lstick[2,brackets=none]{$+$\hspace{-.65cm}}\push{\bbo}& & \gate[2]{\mathit{hxh}} & & & \gate[2]{\mathit{hxh}} & \\ 
    \push{\bbo}& & & & & &
\end{quantikz}
\]
It asserts that $\mathit{hxh}$ acts as an identity on its first line.
Compared with $(-1):\bbo\fromto \bbo$, the simulation $\mathit{hxh}:\bbo+\bbo\fromto\bbo+\bbo$ needs extra auxiliary space, which may bring some redundancies for the provable equations in \HPiLang{}.
For example, 
\[\pid_{\bb{2}}+(\hadamard\seqq\swapp)^8 \fromto_2 \pid_{\bb{2}}+\pid_{\bb{2}}\]
is provable in \HPiLang{} (see \cref{prop:hx8}).
To eliminate the cumbersome identity terms, we still need the cancellative condition \cref{eq:h3} for $+$.

\begin{definition}
    [Equational theory for \HPiLang{}]\label{def:hpi_eq}
    The equational theory for \HPiLang{} contains all level-2 combinators of $\Pi$ and \cref{eq:h1,eq:h2,eq:h3}.
\end{definition}

\begin{proposition}\label{prop:hx8}
    In \HPiLang{},
    \[(\hadamard\seqq\swapp)^8 \fromto_2 \pid_{\bbo+\bbo}.\]
\end{proposition}
\begin{proof}
    \def\scaleratio{.9}
    The formal proof can be obtained from the following graphical equations. Consider
    \begin{align}
        \scalebox{\scaleratio}{\begin{quantikz}[row sep=2mm,column sep=1mm]
            \lstick[2,brackets=none]{$+$\hspace{-.5cm}}& & & &\\
            \lstick[2,brackets=none]{$+$\hspace{-.5cm}}& \permute{2,1} & \gate[2]{\mathit{hxh}} & \permute{2,1} &\\
            & & & &
        \end{quantikz}}
        \overset{\text{\cref{eq:h2}}}{\fromto_2{}}&
        \scalebox{\scaleratio}{\begin{quantikz}[row sep=2mm,column sep=1mm]
            & & \permute{2,1} & & \permute{2,1}& &\\
            & \permute{2,1} & &\gate[2]{\mathit{hxh}} & & \permute{2,1} &\\
            & & & & & &
        \end{quantikz}}
        \overset{\text{completeness of $\Pi$}}{\fromto_2{}}
        \scalebox{\scaleratio}{\begin{quantikz}[row sep=2mm,column sep=2mm]
            & \gate[2]{\mathit{hxh}} & \\
            & & \\
            & &
        \end{quantikz}} \label{eq:hxh1}
    \end{align}
    Then,
    \begingroup
    \allowdisplaybreaks
    \begin{gather*}
        \scalebox{\scaleratio}{\begin{quantikz}[row sep=3mm,column sep=4mm]
            \lstick[2,brackets=none]{$+$\hspace{-.5cm}}& &\\
            \lstick[2,brackets=none]{$+$\hspace{-.5cm}}& &\\
            \lstick[2,brackets=none]{$+$\hspace{-.5cm}}& \gate[2]{(\hadamard\seqq\swapp)^4} &\\
            & &
        \end{quantikz}}
        \fromto_2{}
        \scalebox{\scaleratio}{\begin{quantikz}[row sep=3mm,column sep=4mm]
            \lstick[2,brackets=none]{$+$\hspace{-.5cm}} & & & & &\\
            \lstick[2,brackets=none]{$+$\hspace{-.5cm}} & & & & &\\
            \lstick[2,brackets=none]{$+$\hspace{-.5cm}}& \gate[2]{\mathit{hxh}} & \permute{2,1} & \gate[2]{\mathit{hxh}} & \permute{2,1} &\\
            & & & & &
        \end{quantikz}} 
        \overset{\text{\cref{eq:hxh1}}}{\fromto_2{}}
        \scalebox{\scaleratio}{\begin{quantikz}[row sep=2mm,column sep=4mm]
            & & & \\
            & & \gate[2]{\mathit{hxh}} & \\
            & \gate[2]{\mathit{hxh}} & & \\
            & & & 
        \end{quantikz}} \\
        \overset{\text{\cref{eq:h2}}}{\fromto_2{}}
        \scalebox{\scaleratio}{\begin{quantikz}[row sep=2mm,column sep=2mm]
            & & & & & \\
            & \permute{2,1} & & \permute{2,1} & \gate[2]{\mathit{hxh}} & \\
            & & \gate[2]{\mathit{hxh}} & & & \\
            & & & & &
        \end{quantikz}}
        \overset{\text{\cref{eq:hxh1}}}{\fromto_2{}}
        \scalebox{\scaleratio}{\begin{quantikz}[row sep=2mm,column sep=2mm]
            & & & \gate[2]{\mathit{hxh}} & & \\
            & \permute{2,1} & & & \permute{2,1} & \\
            & & \gate[2]{\mathit{hxh}} & & & \\
            & & & & &
        \end{quantikz}}
    \end{gather*}
    \endgroup
    It follows that
    \[ \pid_{\bbo+\bbo}+(\hadamard\seqq\swapp)^8 \fromto_2 \pid_{\bbo+\bbo}+\pid_{\bbo+\bbo}. \]
    Finally, \cref{eq:h3} gives $(\hadamard\seqq\swapp)^8 \fromto_2 \pid_{\bbo+\bbo}$.
\end{proof}

Expressivity, soundness, and completeness extend from \QPiLang{} to \HPiLang.

\begin{theorem}
    [\HPiLang{} expressivity]
    \HPiLang{} is computationally universal for quantum circuits.
\end{theorem}

\begin{theorem}[Soundness and completeness of \HPiLang{}]
    \label{thm:hpi}
    If $c_1,c_2$ are two terms of \HPiLang{}, then $c_1\fromto_2 c_2$ is a level-2 combinator of \HPiLang{} if and only if $\sem{c_1}=\sem{c_2}$.
\end{theorem}
\begin{proof}
    The proofs are deferred to \cref{sec:hpi_soundness_and_completeness}.
\end{proof}

Consequently, we can define any categorical model for \HPiLang as a rig groupoid with a distinguished morphism $\hadamard$ satisfying \cref{eq:h1,eq:h2,eq:h3}. The semantics of \HPiLang are sound in the sense that \Unitary{} is a model. The semantics are complete in the sense that two syntactic terms are equal in any model (according to the equations) if and only if their interpretations are equal as morphisms in \Unitary{}.

%!TEX root = 0_main.tex
\section{Presenting the Groups \texorpdfstring{\HPiUnitary{}}{On(Z[1/sqrt(2)])}}\label{sec:relations}

As discussed in \cref{sec:language}, the unitaries that can be accurately represented by \QPiLang{} and \HPiLang{} are those with entries in the ring \HPiRing{}.
In this section, we will provide a finite presentation of the group \HPiUnitary{} of $n$-dimensional orthogonal (unitary) matrices with entries in \HPiRing{}.
This presentation is fundamental to the completeness proof for our equational theories for \QPiLang{} and \HPiLang{}.
The techniques used here are largely inspired by~\cite{greylyn2014generators,Amy2020numbertheoretic,li2021generators,bian2021generators}.

\subsection{Generators and Exact Synthesis for \texorpdfstring{\HPiUnitary{}}{On(Z[1/sqrt(2)])}}

First, recall some useful definitions and lemmas from~\cite{Amy2020numbertheoretic}.

\begin{definition}
    The rings $\ZZ[\sqrt{2}]$ and $\ZZ[\frac{1}{\sqrt{2}}]$ are defined as
    \begin{align*}
        \ZZ[\sqrt{2}] 
        & = \{\,a+b{\sqrt{2}} \mid a,b\in\ZZ\,\}, \\
        \HPiRing{} 
        & = \{\, (a+b\sqrt{2})/2^k \mid a,b\in\ZZ,k\in\NN\,\}
        = \{\,v/\sqrt{2}^k\mid v\in \ZZ[\sqrt{2}], k \in\NN\,\}.
    \end{align*}
\end{definition}

\begin{definition}\label{def:lde}
    The \emph{least denominator exponent} of an element $v$ of $\HPiRing$ is defined as
    \[\lde(v) = \min\cbra*{\,k\in\NN \mid \sqrt{2}^kv \in \ZZ[\sqrt{2}]\,}.\]
\end{definition}

We also extend \cref{def:lde} to vectors or matrices over \HPiRing{}: if $M$ is a vector or matrix over \HPiRing{}, then $\lde(M)$ is the least natural number of $k$ such that $\sqrt{2}^k M$ is a vector or matrix over $\ZZ\sbra{\sqrt{2}}$, respectively.

\begin{definition}
    The group of $n$-dimensional orthogonal matrices with entries in \HPiRing{} is denoted \HPiUnitary{}.
\end{definition}

\begin{definition}
    Let $M$ be a $m\times m$ matrix, let $m<n$, and let $1\leq a_1,a_2,\ldots,a_m\leq n$. An $m$-level matrix of type $M$ is the $n\times n$ matrix $M_{[a_1,\ldots,a_m]}$ defined by 
    \[(M_{[a_1,\ldots,a_m]})_{i,j} = \begin{cases}
        M_{i',j'} & \text{ if $i=a_{i'}$ and $j = a_{j'}$,} \\
        \delta_{i,j} & \text{ otherwise.}
    \end{cases}\]
\end{definition}

\begin{definition}
    The set $\cG_n$ of \emph{$n$-dimensional generators} is the subset of \HPiUnitary{} defined as
    \begin{equation*}
        \cG_n = \left\{\,\CZ{a}, \CX{b,c}, \CH{b,c} \mid 1\leq a \leq n,1 \leq b < c \leq n \right\}.
    \end{equation*}
\end{definition}

\begin{lemma}
    [{\cite[Lemma~5.13]{Amy2020numbertheoretic}}]\label{lem:amy1}
    Let $v$ be an $n$-dimensional unit vector over $\HPiRing$ with $\lde(v) > 0$. There exist generators $G_1,G_2,\ldots,G_{\ell} \in \cG_n$ satisfying $\lde(G_1\cdots G_\ell v) < \lde(v)$.
\end{lemma}

\begin{lemma}
    [{\cite[Lemma~5.14]{Amy2020numbertheoretic}}]\label{lem:amy2}
    Let $v$ be an $n$-dimensional unit vector over $\HPiRing$. There exist generators $G_1,G_2,\ldots,G_{\ell} \in \cG_n$ satisfying $G_1\cdots G_\ell v = e_j$, where $e_j$ is the $j$-th standard basis vector.
\end{lemma}

\begin{theorem}
    [{\cite[Theorem~5.15]{Amy2020numbertheoretic}}]
    An $n$-by-$n$ complex matrix is in $\HPiUnitary$ if and only if it is a product of elements of $\cG_n$.
\end{theorem}

\begin{algorithm}[ht]
	\caption{Exact synthesis for \HPiUnitary{}.}
	\label{alg:exact_synthesis}
	\begin{algorithmic}[1]
    \Require{An element $M$ of \HPiUnitary{}}
    \Ensure{A sequence $\mathbf{W}_1, \ldots,\mathbf{W}_{\ell}$ of words over $\cG_n$ such that $\mathbf{W}_{\ell}\cdots \mathbf{W}_1M = I$}
    \State{$N \gets M$}
    \For{$j \gets n,n-1,\ldots,1$} 
        \State{$k \gets \lde(Ne_j)$}
        \While{$k > 0$} \Comment{Decrease the value of $\lde(Ne_j)$}
            \State{Let $i_1$ be the least integer such that $\sqrt{2}^kN[i_1,j]\equiv 1 \pmod{2}$ or $\equiv 1+\sqrt{2} \pmod{2}$}
            \State{Let $i_2$ be the least integer $> i_1$ such that $\sqrt{2}^kN[i_2,j]\equiv \sqrt{2}^kN[i_1,j] \pmod{2}$} \Comment{See~\cite{Amy2020numbertheoretic} for the existence of $i_1$ and $i_2$}
            \State \textbf{if}\ {$i_1 = 1$}\ \textbf{then}\ $\mathbf{W} \gets \CH{1,i_2}$
            \State \textbf{if}\ {$i_1 > 1$}\ \textbf{then}\ $\mathbf{W} \gets \CH{1,i_2}\CX{1,i_1}$
            \State{Output $\mathbf{W}$}
            \State{$N \gets \mathbf{W}N, k \gets \lde(Ne_j)$}
        \EndWhile
        \State{$v \gets N e_j$}
        \If{$v \neq e_j$} \Comment{Signed transposition}
            \State There are $a\leq j$ and $\tau\in\{0,1\}$ such that $v = (-1)^\tau e_a$.
            \State \textbf{if}\ {$a = j$} \textbf{then}\ $\mathbf{W} \gets \CZ{a}^\tau$ \Comment{Note that $\tau = 1$ in this case}
            \State \textbf{if}\ {$a < j$} \textbf{then}\  $\mathbf{W} \gets \CX{a,j}\CZ{a}^\tau $
            \State{Output $\mathbf{W}$}
            \State{$N \gets \mathbf{W}N$}
        \EndIf
    \EndFor
	\end{algorithmic}
\end{algorithm}

\Cref{alg:exact_synthesis} constructs the concrete decomposition of a given element $M$ of \HPiUnitary{} in terms of generators $\cG_n$.\footnote{Unlike synthesis algorithms designed for efficiency (e.g.,~\cite{amy2023improved}), ours provides a unique normal form for each element in $O_n(\ZZ\sbra{\tfrac{1}{\sqrt{2}}})$, which is essential for proving later \cref{thm:relation_completeness} below.}
Intuitively, for $j$ from $n$ to $1$, \cref{alg:exact_synthesis} outputs generators to decrease the value of $\lde(Me_j)$ (\cref{lem:amy1}) and outputs additional generators to get a new $M$ such that $Me_j = e_j$ (\cref{lem:amy2}).
Finally, the output generators transform the $j$-th column of $M$ into $e_j$, i.e., making $M$ become the identity matrix.
We can precisely characterise this intuition using a notion of \emph{level} similar to the one proposed in~\cite{greylyn2014generators,bian2021generators,li2021generators}.

\begin{definition}
    The \emph{level} of a unitary $M \in \HPiUnitary{}$, denoted $\level(M)$, is a triple $(j,k,l)$ defined as follows:
    \begin{itemize}
        \item $j = \max\rbra*{\{\,0\,\}\cup\{\,i \in \NN \mid 1\leq i \leq n, Me_i \neq e_i\,\}}$;
        \item $k = \lde(Me_j)$, or $k=0$ if $j=0$;  and
        \item $l$ is the number of entries in $\sqrt{2}^kMe_j$ that are $\equiv 1 \pmod{2}$ or $\equiv 1+\sqrt{2}\pmod{2}$, or $l=0$ if $k=0$. 
    \end{itemize}
    Levels are ordered lexicographically: $(j_1,k_1,l_1)<(j_2,k_2,l_2)$ iff ($j_1<j_2$) or ($j_1=j_2$ and $k_1<k_2$) or ($j_1=j_2$ and $k_1=k_2$ and $l_1< l_2$).
\end{definition}

\begin{theorem}
    [Correctness of {\cref{alg:exact_synthesis}}]
    Each output of \cref{alg:exact_synthesis} strictly decreases the level of $N$.
\end{theorem}
\begin{proof}
    Follows from the proofs of \cref{lem:amy1,lem:amy2} in~\cite{Amy2020numbertheoretic}.
\end{proof}

\subsection{Finite Presentation of \texorpdfstring{\HPiUnitary{}}{On(Z[1/sqrt(2)])}}

We already know that $\cG_n$ generates \HPiUnitary{}.
However, the complete relations for \HPiUnitary{} remain unknown in the literature.
Before diving into the complete relations we found for \HPiUnitary{}, we recall some useful terminology from~\cite{greylyn2014generators,bian2021generators,li2021generators}.

\begin{definition}[Words over $\cG_n$]
    A \emph{word} $\mathbf{W}$ over $\cG_n$ is a finite sequence $G_1G_2\cdots G_\ell$  of elements of $\cG_n$.
    The \emph{empty} word (where $\ell=0$) is denoted by $\varepsilon$.
    The collection of words over $\cG_n$ is denoted by $\cG_n^*$.
\end{definition}

\begin{definition}
    [Semantics of words]
    A word $\mathbf{W} = G_1G_2\cdots G_\ell$ over $\cG_n$ is interpreted in \HPiUnitary{}, which embeds into \Unitary{}, as
    \[\sem{\mathbf{W}} = G_1\cdot G_2\cdot\ldots\cdot G_\ell. \] 
\end{definition}

\begin{figure*}[ht]
    \begin{minipage}{0.42\linewidth}
        \begin{align}
            \CZ{a}^2 &\approx \varepsilon \label[relation]{rel:a1} \tag{a1}\\
            \CX{a,b}^2 &\approx \varepsilon \label[relation]{rel:a2} \tag{a2}\\
            \notag\\
            \CZ{a}\CZ{b} &\approx \CZ{b}\CZ{a} \label[relation]{rel:b1} \tag{b1}\\
            \CZ{a}\CX{b,c} &\approx \CX{b,c}\CZ{a} \label[relation]{rel:b2} \tag{b2}\\
            \CX{a,b}\CX{c,d} &\approx \CX{c,d}\CX{a,b} \label[relation]{rel:b3} \tag{b3}\\
            \notag\\
            \CZ{a}\CX{a,b} &\approx \CX{a,b}\CZ{b} \label[relation]{rel:c1} \tag{c1}\\
            \CX{b,c}\CX{a,b} &\approx \CX{a,b}\CX{a,c} \label[relation]{rel:c2} \tag{c2}\\
            \CX{a,c}\CX{b,c} &\approx \CX{b,c}\CX{a,b} \label[relation]{rel:c3} \tag{c3}
        \end{align}
    \end{minipage}
    \begin{minipage}{0.57\linewidth}
        \begin{align}
            \CH{a,b}^2 &\approx \varepsilon \label[relation]{rel:a3} \tag{a3}\\
            \notag\\
            \CZ{a}\CH{b,c} &\approx \CH{b,c}\CZ{a} \label[relation]{rel:b4} \tag{b4}\\
            \CX{a,b}\CH{c,d} &\approx \CH{c,d}\CX{a,b} \label[relation]{rel:b5} \tag{b5}\\
            \CH{a,b}\CH{c,d} &\approx \CH{c,d}\CH{a,b} \label[relation]{rel:b6} \tag{b6} \\
            \notag\\
            \CH{b,c}\CX{a,b} &\approx \CX{a,b}\CH{a,c} \label[relation]{rel:c4} \tag{c4}\\
            \CH{a,c}\CX{b,c} &\approx \CX{b,c}\CH{a,b} \label[relation]{rel:c5} \tag{c5}\\
            \notag\\
            \CZ{a}\CZ{b}\CH{a,b} &\approx \CH{a,b}\CZ{a}\CZ{b} \label[relation]{rel:d1} \tag{d1}\\
            \CZ{b}\CH{a,b} &\approx \CH{a,b}\CX{a,b} \label[relation]{rel:d2} \tag{d2}
        \end{align}
    \end{minipage}
    \begin{align}
        (\CH{c,d}\CH{a,c}\CH{b,d})^4 &\approx \CH{a,b}\CH{c,d}
    %    %\CX{b,c}\CH{a,b}\CX{a,c}\CH{a,d}\CH{a,b}\CX{a,c}\CH{a,d} \approx \CH{a,b}\CX{a,c}\CH{a,d}\CH{a,b}\CX{a,c}\CH{a,d}\CX{c,d}
        \label[relation]{rel:d3} \tag{d3} \\
        (\CH{a,c}\CH{b,d}\CH{a,b}\CH{a,c}\CH{b,d}\CX{c,e}\CX{d,f})^3 &\approx \CH{c,e}\CH{d,f}\CH{e,f}\CH{c,e}\CH{d,f}\CX{c,e}\CX{d,f} \label[relation]{rel:d4} \tag{d4}
    \end{align}
    \caption{Relations for \HPiUnitary{}. The indices are distinct.}\label{fig:relations}
\end{figure*}

Now an equivalence relation $\approx$ for \HPiUnitary{} can be presented in \cref{fig:relations}, along with its soundness (\cref{thm:relation_soundness}) and completeness (\cref{thm:relation_completeness}). 
It lets us reason about the equivalence of two words syntactically, rather than interpreting them as matrices and determining equivalence through matrix multiplication.

\begin{theorem}
    [Soundness of $\approx$]\label{thm:relation_soundness}
    Let $\mathbf{G}$ and $\mathbf{H}$ be words over $\cG_n$. If $\mathbf{G} \approx \mathbf{H}$, then $\sem{\mathbf{G}} = \sem{\mathbf{H}}$.
\end{theorem}
\begin{proof}
    By induction on the steps of the proof for $\mathbf{G}\approx \mathbf{H}$, it suffices to show that the relations in \cref{fig:relations} are sound. This can be verified by direct computation.
\end{proof}

\begin{restatable}[Completeness of $\approx$]{theorem}{relationcompleteness}
    \label{thm:relation_completeness}
    Let $\mathbf{G}$ and $\mathbf{H}$ be words over $\cG_n$. If $\sem{\mathbf{G}} = \sem{\mathbf{H}}$, then $\mathbf{G} \approx \mathbf{H}$.
\end{restatable}
\begin{proof}
    Once we have the relation $\approx$, the proof follows directly from the approach of~\cite{greylyn2014generators,bian2021generators,li2021generators}.
    \cref{alg:exact_synthesis} provides a unique \emph{normal form} for each element in \HPiUnitary{} in a deterministic way, similar to~\cite{li2021generators}.
    The idea of the proof is showing that each word over $\cG_n$ can be rewritten into its unique normal form, thereby establishing completeness.
    The full proof needs a very long case distinction, and is relegated to \ifARXIV%
    \cref{app:relations_proof}%
    \else%
    \cref{ext-app:relations_proof} of the extended version~\cite{FHK25}%
    \fi.
\end{proof}

\iffalse
\begin{proposition}
    \begin{align*}
        (\CH{1,2}\CH{3,4}\CH{1,3}\CH{2,4})^2 &\approx \varepsilon, \\
        (\CH{1,2}\CH{3,4}\CH{1,4}\CH{3,2})^2 &\approx \CX{1,3}\CX{2,4}.
    \end{align*}
\end{proposition}
\fi

With the generators and relations for \HPiUnitary{}, we obtain a finite presentation of \HPiUnitary{}. This presentation will bridge \Unitary{}, \QPiLang{}, and \HPiLang{}, contributing to the completeness of our equational theories.

%!TEX root = 0_main.tex
\section{Soundness and Completeness}\label{sec:soundness_and_completeness}

In this section, we prove the soundness and completeness of \HPiLang{} by using a hierarchical approach with translations as illustrated in \cref{fig:translations}.
The overall proof is divided into three major steps.
\begin{itemize}
    \item We have shown that $\cG_n^*$ with the relation $\approx$ is sound (\cref{thm:relation_soundness}) and complete (\cref{thm:relation_completeness}) for \Unitary{}.
    \item Using translations $\wsem{\cdot}$ and $\cT_Q$, we will prove that \QPiLang{} with its equational theory (\cref{def:qpi_eq}) is complete for $\cG_n^*$ with relation $\approx$ (\cref{thm:qpi_relation}).
    \item Using translations $\qsem{\cdot}$ and $\cT_H$, we will prove that \HPiLang{} with its equational theory (\cref{def:hpi_eq}) is complete for \QPiLang{} (\cref{thm:hpi_qpi}).
\end{itemize}

\subsection{Soundness and Completeness of \texorpdfstring{\QPiLang}{Q-Pi}}\label{sec:qpi_soundness_and_completeness}
Let us begin by showing how to translate a term of \QPiLang{} to a word over $\cG_n$ using $\wsem{\cdot}$.
To define $\wsem{\cdot}$, we first need to introduce some auxiliary functions.
\begin{itemize}
    \item $\hdim$: Given a term $c:b_1\fromto b_2$ of \QPiLang{}, we need to determine a specific $n$ to translate into $\cG_n^*$. Thus we define $\hdim(c) = \hdim(b_1)$ as the {dimension} of $c$, where the {dimension} of a type $b$ is defined inductively:
    \begin{align*}
        \hdim(\bbz) &= 0, &  \hdim(b_1+b_2) &= \hdim(b_1)+\hdim(b_2), \\
        \hdim(\bbo)&= 1, & \hdim(b_1\times b_2) &= \hdim(b_1)\times \hdim(b_2).
    \end{align*}
    \item $\hpermute$: Given a permutation $\pi:[n]\fromto [n]$, where $[n]=\{\,1,\ldots,n\}$ for any $n \geq 1$, $\hpermute(\pi)$ outputs $\mathbf{W}_\ell\cdots\mathbf{W}_2\mathbf{W}_1$ if the output sequence of \cref{alg:exact_synthesis} is $\mathbf{W}_1,\ldots,\mathbf{W}_\ell$ on the input $X_\pi$, the matrix representation of $\pi$.
    \item $\shift$: Shifting the subscripts of generators is defined inductively:
    \begin{align*}
        \shift\rbra*{\CZ{a},n} &= \CZ{a+n}, &
        \shift\rbra*{\CX{b,c}, n} &= \CX{b+n,c+n}, \\
        \shift\rbra*{\CH{b,c}, n} &= \CH{b+n,c+n}, &
        \shift\rbra*{\mathbf{G}_1\mathbf{G}_2,n} &= \shift\rbra*{\mathbf{G}_1,n}\shift\rbra*{\mathbf{G}_2,n}.
    \end{align*}
\end{itemize}

\begin{definition}
    [Translation $\wsem{\cdot}$]
    The translation $\wsem{\cdot}$ that maps a term $c$ of \QPiLang{} to $\cG_n^*$ with $n= \hdim(c)$ is defined inductively on the structure of terms as shown in \cref{fig:def_wsem}.
    \begin{figure}[t]
        \begin{align*}
            \wsem{(-1)} &=  \CZ{1} & \wsem{\hadamard} &= \CH{1,2} & \wsem{\pid} &= \varepsilon
        \end{align*}
        \begin{align*}
            \wsem{\assoclp} &= \varepsilon & \wsem{\assocrp} &= \varepsilon & \wsem{\unitep} &=\varepsilon & \wsem{\unitip} &= \varepsilon \\
            \wsem{\assoclt} &= \varepsilon & \wsem{\assocrt} &= \varepsilon & \wsem{\unitet} &=\varepsilon & \wsem{\unitit} &= \varepsilon \\
            \wsem{\dist} &= \varepsilon & \wsem{\factor} &= \varepsilon & \wsem{\absorb} &= \varepsilon & \wsem{\factorz} &= \varepsilon
        \end{align*}
        \begin{gather*}
            \begin{prooftree}
                \hypo{\swapp : b_1+b_2 \fromto b_2+b_1} \hypo{n_1 = \hdim(b_1)} \hypo{n_2 = \hdim(b_2)}
                \infer3{\wsem{\swapp} = \hpermute(\pi:j\in[n_1+n_2] \mapsto \mathsf{if}\, j \leq n_1 \,\mathsf{then}\, j+n_2 \,\mathsf{else}\, j-n_1)}
            \end{prooftree} \\
            \begin{prooftree}
                \hypo{\swapt : b_1\times b_2 \fromto b_2\times b_1} \hypo{n_1 = \hdim(b_1)} \hypo{n_2 = \hdim(b_2)}
                \infer3{\wsem{\swapt} = \hpermute(\pi:j\in[n_1\times n_2] \mapsto \mathsf{let}\, (k,l) = (j\ \mathrm{div}\ n_2, j\bmod n_2) \,\mathsf{in}\, l\times n_1+k+1)}
            \end{prooftree}
        \end{gather*}
        \begin{align*}
            \wsem{c_1\seqq c_2} &= \wsem{c_2}\wsem{c_1} &
            \wsem{c_1+c_2} &= \wsem{c_1}\shift\rbra*{\wsem{c_2}, \hdim(c_1)}
        \end{align*}
        \begin{equation*}
            \begin{prooftree}
                \hypo{c_1 : b_1 \fromto b_3} \hypo{c_2 : b_2 \fromto b_4}
                \infer2{\wsem{c_1\times c_2} = \wsem{\swapt}\wsem{\pid_{b_4}\times c_1}\wsem{\swapt}\wsem{\pid_{b_1}\times c_2}}
            \end{prooftree}
        \end{equation*}
        
        \begin{flushleft}
            For terms in the form of $\pid_{b}\times c$, we define:
        \end{flushleft}
        \begin{equation*}
            \wsem{\pid_{b}\times c} = \wsem{c}\shift\rbra[\big]{\wsem{c},\hdim(c)}\shift\rbra[\big]{\wsem{c},2\hdim(c)}\cdots\shift\rbra[\big]{\wsem{c},(\hdim(b)-1)\times \hdim(c)}
        \end{equation*}
        \caption{Definition of translation $\wsem{\cdot}$. \ifARXIV \cref{fig:visual_permutation} in the appendix demonstrates the permutations in the definitions of $\wsem{\swapp}$, $\wsem{\swapt}$.\fi}\label{fig:def_wsem}
    \end{figure}
\end{definition}

The following lemma follows readily.

\begin{restatable}{lemma}{lemwfaith}\label{lem:w_faith}
    If $c$ is a term of Q-$\Pi$, then $\sem{\wsem{c}} = \sem{c}$.
\end{restatable}
\begin{proof}
    By induction on the structure of $c$ and the correctness of \cref{alg:exact_synthesis}. See details in \ifARXIV%
    \cref{app:qpi_proof}%
    \else%
    \cref{ext-app:qpi_proof} of the extended version~\cite{FHK25}%
    \fi.
\end{proof}

Next, we define $\cT_Q$, which describes how to use \QPiLang{} to express the words over $\cG_n$.

\begin{definition}
    [Translation $\cT_Q$]\label{def:qt}
    The translation $\cT_Q$ that maps a word $\mathbf{G}$ over $\cG_n$ to a term $\qT[n]{\mathbf{G}}$ of Q-$\Pi$ is defined inductively as shown in \cref{fig:def_qt}, where $k\bbo$ is a shorthand for $\underbrace{\bbo+\bbo+\cdots+\bbo}_{k}$.
    \begin{figure}[t]
        \begin{align*}
            \qT[n]{\varepsilon} &= \pid_{n\bbo} \\
            \qT[n]{\CZ{a}} &= \swapp(a,n)\seqq (\pid_{(n-1)\bbo}+(-1))\seqq\swapp(a,n) \\
            \qT[n]{\CX{b,c}} &= \swapp(b,c) \\
            \qT[n]{\CH{b,c}} &= \swapp(b,n-1)\seqq\swapp(c,n)\seqq \assocrp\seqq(\pid_{(n-2)\bbo}+\hadamard)\seqq \\
            & \quad \assoclp\seqq\swapp(c,n)\seqq\swapp(b,n-1) \\
            \qT[n]{\mathbf{G}_1\mathbf{G}_2} &= \qT[n]{\mathbf{G}_2}\seqq\qT[n]{\mathbf{G}_1}
        \end{align*}
        \caption{Definition of translation $\cT_Q$. Here $\swapp(n-1,n) \triangleq \assocrp\seqq(\pid_{(n-2)\bbo}+\swapp)\seqq\assoclp$, $\swapp(k,n) \triangleq (\assocrp + \pid_{(n-k-1)\bbo}) \seqq((\pid_{(k-1)\bbo}+\swapp)+\pid_{(n-k)\bbo})\seqq(\assoclp + \pid_{(n-k-1)\bbo}) \seqq \swapp(k+1,n) \seqq(\assocrp + \pid_{(n-k-1)\bbo}) \seqq ((\pid_{(k-1)\bbo}+\swapp)+\pid_{(n-k)\bbo}) \seqq(\assoclp + \pid_{(n-k-1)\bbo})$, $\swapp(j,k) \triangleq \swapp(k,n)\seqq\swapp(j,n)\seqq\swapp(k,n)$ if $j\neq k$, and $\swapp(j,j) \triangleq \pid_{n\bbo}$.}\label{fig:def_qt}
    \end{figure} 
\end{definition}
In \cref{fig:def_qt}, $\swapp(j,k)$ represents a permutation that swaps the $j$-th and $k$-th terms of a sum of $n$ elements.
This permutation can be synthesised into a term of \QPiLang{} through the full abstraction of $\Pi$~\cite{choudhury2022symmetries}. 
For $\cT_Q$, we also use a superscript $n$ in $\qT[n]{\mathbf{G}}$ to indicate that $\mathbf{G}$ is unambiguously in $\cG_{n}^*$.
This is necessary because words can sometimes be used for different values of $n$.
For example, $\CH{1,2}$ can be regarded as a word over $\cG_n$ for any $n\geq 2$.

\begin{restatable}{lemma}{lemweq}\label{lem:w_eq}
    Let $\mathbf{G}$ and $\mathbf{H}$ be words over $\cG_n$. If $\mathbf{G}\approx \mathbf{H}$, then $\qT[n]{\mathbf{G}} \fromto_2 \qT[n]{\mathbf{H}}$ is a level-2 combinator of Q-$\Pi$. 
\end{restatable}
\begin{proof}
    By induction on the relations used for deriving $\mathbf{G}\approx \mathbf{H}$, we only need to verify that all the relations $\mathbf{G}\approx \mathbf{H}$ in \cref{fig:relations} satisfy this lemma.
    Since the level-2 combinators of $\Pi$ are sound and complete with respect to permutations~\cite{choudhury2022symmetries}, we only need to deal with the relations of \cref{fig:relations} with fixed indices.
    We provide graphical proofs for the fixed index versions of \cref{rel:d3} with $n=4$ (see \cref{lem:d3_fixed}) and \cref{rel:d4} with $n=6$ (see \cref{lem:d4_fixed}),
    while the other relations in \cref{fig:relations} are easily derived and relegated to \ifARXIV%
    \cref{app:qpi_proof}%
    \else%
    \cref{ext-app:qpi_proof} of the extended version~\cite{FHK25}%
    \fi.
\end{proof}

\begingroup
\allowdisplaybreaks
Let us list some properties and definitions, illustrated via quantum circuits, that will be required for \cref{lem:d3_fixed,lem:d4_fixed}.
\begin{gather}
    \scalebox{1}{\begin{quantikz}[row sep=2mm,column sep=2mm]
        &&\ctrl{1}& \ctrl{1} & \midstick[2,brackets=none]{$\fromto_2{}$} && \ctrl{1} &\\
        &\qwbundle{} &\gate{c_1}& \gate{c_2} & & \qwbundle{} & \gate{c_1\seqq c_2} &
    \end{quantikz}} \label{eq:qt0} \\
    \mathllap{\Leftrightarrow{}} (\pid+c_1)\seqq(\pid+c_2) \fromto_2{} (\pid+(c_1\seqq c_2)) \tag{by bifunctoriality of $+$} \\
        \scalebox{1}{\begin{quantikz}[row sep=2mm,column sep=2mm]
            &&\ctrl{1}& \octrl{1} & \midstick[2,brackets=none]{$\fromto_2{}$} && \octrl{1}& \ctrl{1} &\\
            &\qwbundle{} &\gate{c_1}& \gate{c_2} & & \qwbundle{} & \gate{c_2}& \gate{c_1} &
        \end{quantikz}} \label{eq:qt1} \\
        \mathllap{\Leftrightarrow{}} (\pid+c_1)\seqq(c_2+\pid) \fromto_2{} (c_2+\pid)\seqq(\pid+c_1) \tag{by bifunctoriality of $+$} \\
        \scalebox{1}{\begin{quantikz}[row sep=2mm,column sep=2mm]
            &&\ctrl{1}& \octrl{1} & \midstick[2,brackets=none]{$\fromto_2{}$} && &\\
            &\qwbundle{} &\gate{c}& \gate{c} & &\qwbundle{} & \gate{c} &
        \end{quantikz}} \label{eq:qt2} \\
        \mathllap{\Leftrightarrow{}} (\pid+c)\seqq(c+\pid) \fromto_2{} (c+c)\fromto_2{} \pid_{\bb{2}}\times c \tag{by bifunctoriality of $+$ and distributivity} \\
        \scalebox{1}{\begin{quantikz}[row sep=2mm,column sep=2mm]
            &\swap{1}& \gate{c_1} & \midstick[2,brackets=none]{$\fromto_2{}$} & \gate{c_2} & \swap{1} &\\
            &\targX{}& \gate{c_2} & & \gate{c_1} & \targX{} &
        \end{quantikz}} \label{eq:qt3} \\
        \mathllap{\Leftrightarrow{}} \swapt\seqq(c_1\times c_2)\fromto_2{} (c_2\times c_1)\seqq \swapt \tag{by naturality of $\swapt$} \\
        \scalebox{1}{\begin{quantikz}[row sep=2mm,column sep=2mm]
            &\gate{\p{\p{h}}} & \ctrl{1} & \gate{\p{\p{h}}} & \midstick[2,brackets=none]{$\fromto_2{}$} & \swap{1} & \ctrl{1} & \swap{1} &\\
            &\gate{\p{\p{h}}} & \targ{} & \gate{\p{\p{h}}} & & \targX{} & \targ{} & \targX{} &
        \end{quantikz}} \enspace \Leftrightarrow{} \enspace \text{\cref{lem:q_hhcxhh}} \label{eq:qt4} \\
        \scalebox{1}{\begin{quantikz}[row sep=2mm,column sep=2mm]
            &\gate{c} & \midstick[2,brackets=none]{$\triangleq$} & \swap{1} & \ctrl{1} & \swap{1} &\\
            &\ctrl{-1} & & \targX{} & \gate{c} & \targX{} &
        \end{quantikz}} \hspace{2cm}
        \scalebox{1}{\begin{quantikz}[row sep=2mm,column sep=2mm]
            &\gate{c} & \midstick[2,brackets=none]{$\triangleq$} & \swap{1} & \octrl{1} & \swap{1} &\\
            &\octrl{-1} & & \targX{} & \gate{c} & \targX{} &
        \end{quantikz}} \label{eq:qt5} \\
        \scalebox{1}{\begin{quantikz}[row sep=2mm,column sep=2mm]
            &\gate{\p{z}} & \midstick[2,brackets=none]{$\triangleq$} & \gate{\pid_\bbo+(-1)} &
        \end{quantikz}} \hspace{2cm}
        \scalebox{1}{\begin{quantikz}[row sep=2mm,column sep=2mm]
            &\ctrl{1} & \midstick[2,brackets=none]{$\triangleq$} & \ctrl{1} &\\
            &\phase{} & & \gate{\p{z}} &
        \end{quantikz}} \label{eq:qt6}
\end{gather}

The following simple lemmas are also needed.
\begin{restatable}{lemma}{lemqceqs}~\\
    \begin{minipage}{0.23\linewidth}
        \begin{equation}
            \scalebox{\scaleratio}{\begin{quantikz}[row sep=2mm,column sep=2mm]
                &\ctrl{1} & \midstick[2,brackets=none]{$\fromto_2{}$} & \gate{\p{z}} &\\
                &\gate{\p{z}} & & \ctrl{-1} &
            \end{quantikz}} \label{eq:zz}
        \end{equation}
    \end{minipage}
    \hfill
    \begin{minipage}{0.335\linewidth}
        \begin{equation}
            \scalebox{\scaleratio}{\begin{quantikz}[row sep=2mm,column sep=2mm]
                &&\gate{\p{x}}&\ctrl{1} & \gate{\p{x}} &\midstick[2,brackets=none]{$\fromto_2{}$} && \octrl{1} & \\
                &\qwbundle{}&&\gate{c} & & &\qwbundle{}& \gate{c}&
            \end{quantikz}} \label{eq:octrl}
        \end{equation}
    \end{minipage}
    \hfill
    \begin{minipage}{0.42\linewidth}
        \begin{equation}
            \scalebox{\scaleratio}{\begin{quantikz}[row sep=2mm,column sep=2mm]
                &&\ctrl{1} & \midstick[2,brackets=none]{$\fromto_2{}$} && & \ctrl{1} & & \\
                &\qwbundle{}&\gate{d\seqq c\seqq d^{-1}} & &\qwbundle{}& \gate{d} &\gate{c\vphantom{d^{-1}}} & \gate{d^{-1}} &
            \end{quantikz}} \label{eq:conjugation}
        \end{equation}
    \end{minipage}\\
    Here, $d^{-1}$ denotes the inverse of $d$. Since our language is reversible, $d^{-1}$ can be straightforwardly constructed and satisfies $d\seqq d^{-1} \fromto_2{} d^{-1}\seqq d \fromto_2 \pid$.
\end{restatable}
\begin{proof}
    This lemma is easily derived and relegated to \ifARXIV%
    \cref{app:qpi_proof}%
    \else%
    \cref{ext-app:qpi_proof} of the extended version~\cite{FHK25}%
    \fi.
\end{proof}

For $\p{h}$ gates, we derive the following decomposition lemma, inspired by~\cite{bian2022generators}.\footnote{In particular, \cref{eq:h_decomposition} represents a lifted version of the gate decomposition $H = (T\cdot H\cdot S)^{\dagger}\cdot X \cdot (T\cdot H\cdot S)$ via the mapping $\mathcal{E}$ defined in \cref{eq:real_U}. We can verify that $\cE(H) = I\otimes H, \cE(X) = I\otimes X, \cE(T) = \mathit{SWAP} \cdot (\mathit{CH}\cdot\mathit{CX}) \cdot \mathit{SWAP}$ and $\cE(S) = \mathit{SWAP} \cdot (\mathit{CZ}\cdot\mathit{CX}) \cdot \mathit{SWAP}$. For readers unfamiliar with common quantum gates, the $\pi/8$ gate is $T = 1\oplus e^{i\pi/4}$, the phase gate is $S = 1 \oplus i$, and the swap gate is $\mathit{SWAP} = 1\oplus X \oplus 1$.}
\begin{lemma}
    \begin{equation}\label{eq:cnot_h}
        \scalebox{\scaleratio}{\begin{quantikz}[row sep=2mm,column sep=2mm]
            & \ctrl{1} & \gate{\p{h}} & \ctrl{1} & \gate{\p{h}} & \ctrl{1} & \gate{\p{h}} & \ctrl{1} &\\
            & \targ{} & \ghost{} & \targ{} & & \targ{} & & \targ{} &
        \end{quantikz}}
        \fromto_2{}
        \scalebox{\scaleratio}{\begin{quantikz}[row sep=2mm,column sep=2mm]
            & \gate{\p{h}} & \\
            & \gate{\p{x}} &
        \end{quantikz}}
    \end{equation}
    \begin{equation}
        \label{eq:h_decomposition}
        \scalebox{\scaleratio}{\begin{quantikz}[row sep=2mm,column sep=2mm]
            & \gate[2][.8cm]{\p{A}} & & \gate[2][.8cm]{\p{A}^{-1}} & \midstick[2,brackets=none]{\scalebox{1.3}{$\fromto_2{}$}} & &\\
            & & \gate{\p{x}} & & & \gate{\p{h}} &
        \end{quantikz}} \text{ with }
        \scalebox{\scaleratio}{\begin{quantikz}[row sep=2mm,column sep=2mm]
            & \gate[2][.8cm]{\p{A}} & \midstick[2,brackets=none]{\scalebox{1.3}{$\triangleq$}} & \targ{} & \ctrl{1} & & \targ{} & \gate{\p{h}} &\\
            & & & \ctrl{-1} & \phase{} & \gate{\p{h}} & \ctrl{-1} & \ctrl{-1} &
        \end{quantikz}}
    \end{equation}
\end{lemma}
\begin{proof}
    \allowdisplaybreaks
    \begin{gather*}
        \scalebox{\scaleratio}{\begin{quantikz}[row sep=2mm,column sep=2mm]
            & \ctrl{1} & \gate{\p{h}} & \ctrl{1} & \gate{\p{h}} & \ctrl{1} & \gate{\p{h}} & \ctrl{1} &\\
            & \targ{} & & \targ{} & & \targ{} & & \targ{} &
        \end{quantikz}}
        \overset{\text{\cref{eq:e2,,eq:qt4,eq:qt5}}}{\fromto_2{}} \scalebox{\scaleratio}{\begin{quantikz}[row sep=2mm,column sep=2mm]
            & \ctrl{1} & & \targ{} &  & \ctrl{1} & \gate{\p{h}} & \ctrl{1} &\\
            & \targ{} & \gate{\p{h}} & \ctrl{-1} & \gate{\p{h}} & \targ{} & & \targ{} &
        \end{quantikz}} \\
        \overset{\text{\cref{eq:e3,,eq:zz,eq:conjugation}}}{\fromto_2{}}
        \scalebox{\scaleratio}{\begin{quantikz}[row sep=2mm,column sep=2mm]
            & & \ctrl{1} & \targ{} & \ctrl{1} & \gate{\p{h}} & \ctrl{1} &\\
            & \gate{\p{h}} & \phase{} & \ctrl{-1} & \phase{} & \gate{\p{h}} & \targ{} &
        \end{quantikz}}
        \overset{\text{\cref{eq:e2,eq:qt4,eq:qt5}}}{\fromto_2{}}
        \scalebox{\scaleratio}{\begin{quantikz}[row sep=2mm,column sep=2mm]
            & & \ctrl{1} & \targ{} & \ctrl{1} & \targ{} & \gate{\p{h}} &\\
            & \gate{\p{h}} & \phase{} & \ctrl{-1} & \phase{} & \ctrl{-1} & \gate{\p{h}} &
        \end{quantikz}} \\
        \overset{\text{\cref{eq:qt0,eq:qt5,eq:conjugation}}}{\fromto_2{}}
        \scalebox{\scaleratio}{\begin{quantikz}[row sep=2mm,column sep=2mm]
            & & \ctrl{1} & \gate{\p{x}} & \ctrl{1} & \gate{\p{x}} & \gate{\p{h}} & \\
            & \gate{\p{h}} & \phase{} & & \phase{} & & \gate{\p{h}} &
        \end{quantikz}}
        \overset{\text{\cref{eq:qt6}}}{\fromto_2{}}
        \scalebox{\scaleratio}{\begin{quantikz}[row sep=2mm,column sep=2mm]
            & & \ctrl{1} & \gate{\p{x}} & \ctrl{1} & \gate{\p{x}} & \gate{\p{h}} & \\
            & \gate{\p{h}} & \gate{\p{z}} & & \gate{\p{z}} & & \gate{\p{h}} &
        \end{quantikz}} \\
        \overset{\text{\cref{eq:qt2,eq:octrl}}}{\fromto_2{}}
        \scalebox{\scaleratio}{\begin{quantikz}[row sep=2mm,column sep=2mm]
            & & & \gate{\p{h}} & \\
            & \gate{\p{h}} & \gate{\p{z}} & \gate{\p{h}} &
        \end{quantikz}} 
        \overset{\text{\cref{eq:e2,eq:e3}}}{\fromto_2{}}
        \scalebox{\scaleratio}{\begin{quantikz}[row sep=2mm,column sep=2mm]
            & \gate{\p{h}} & \\
            & \gate{\p{x}} &
        \end{quantikz}}
    \end{gather*}
    \begin{gather*}
        \scalebox{\scaleratio}{\begin{quantikz}[row sep=2mm,column sep=2mm]
            & \gate[2][.8cm]{\p{A}} & & \gate[2][.8cm]{\p{A}^{-1}} &\\
            & & \gate{\p{x}} & &
        \end{quantikz}}
        \overset{}{\fromto_2{}}
        \scalebox{\scaleratio}{\begin{quantikz}[row sep=2mm,column sep=2mm]
            & \targ{} & \ctrl{1} & & \targ{} & \gate{\p{h}} & & \gate{\p{h}} & \targ{} & & \ctrl{1} & \targ{} &\\
            & \ctrl{-1} & \phase{} & \gate{\p{h}} & \ctrl{-1} & \ctrl{-1} & \gate{\p{x}} & \ctrl{-1} & \ctrl{-1} & \gate{\p{h}} & \phase{} & \ctrl{-1} &
        \end{quantikz}} \\
        \overset{\text{completeness of $\Pi$, \cref{eq:octrl}}}{\fromto_2{}}
        \scalebox{\scaleratio}{\begin{quantikz}[row sep=2mm,column sep=2mm]
            & \targ{} & \ctrl{1} & & \targ{} & \gate{\p{h}} & \gate{\p{h}} & & \targ{} & & \ctrl{1} & \targ{} &\\
            & \ctrl{-1} & \phase{} & \gate{\p{h}} & \ctrl{-1} & \ctrl{-1} & \octrl{-1} & \gate{\p{x}} & \ctrl{-1} & \gate{\p{h}} & \phase{} & \ctrl{-1} &
        \end{quantikz}}
        \overset{\text{\cref{eq:qt2}}}{\fromto_2{}}
        \scalebox{\scaleratio}{\begin{quantikz}[row sep=2mm,column sep=2mm]
            & \targ{} & \ctrl{1} & & \targ{} & \gate{\p{h}} & \targ{} & & \ctrl{1} & \targ{} &\\
            & \ctrl{-1} & \phase{} & \gate{\p{h}} & \ctrl{-1} & \gate{\p{x}} & \ctrl{-1} & \gate{\p{h}} & \phase{} & \ctrl{-1} &
        \end{quantikz}} \\
        \overset{\text{\cref{eq:e3,eq:conjugation}}}{\fromto_2{}}
        \scalebox{\scaleratio}{\begin{quantikz}[row sep=2mm,column sep=2mm]
            & \targ{} & & \ctrl{1} & \targ{} & \gate{\p{h}} & \targ{}& \ctrl{1} & & \targ{} &\\
            & \ctrl{-1} & \gate{\p{h}} & \targ{} & \ctrl{-1} & \gate{\p{x}} & \ctrl{-1} & \targ{} & \gate{\p{h}} & \ctrl{-1} &
        \end{quantikz}}
        \overset{\text{completeness of $\Pi$}}{\fromto_2{}}
        \scalebox{\scaleratio}{\begin{quantikz}[row sep=2mm,column sep=2mm]
            & \targ{} & & \targ{} &\swap{1} & \gate{\p{h}} & \swap{1} & \targ{} & & \targ{} &\\
            & \ctrl{-1} & \gate{\p{h}} & \ctrl{-1} & \targX{} & \gate{\p{x}} & \targX{} & \ctrl{-1} & \gate{\p{h}} & \ctrl{-1} &
        \end{quantikz}} \\
        \overset{\text{naturality of $\swapt$}}{\fromto_2{}}
        \scalebox{\scaleratio}{\begin{quantikz}[row sep=2mm,column sep=2mm]
            & \targ{} & & \targ{} & \gate{\p{x}} & \targ{} & & \targ{} &\\
            & \ctrl{-1} & \gate{\p{h}} & \ctrl{-1} & \gate{\p{h}} & \ctrl{-1} & \gate{\p{h}} & \ctrl{-1} &
        \end{quantikz}}
        \overset{\text{completeness of $\Pi$ and \cref{eq:conjugation}}}{\fromto_2{}}
        \scalebox{\scaleratio}{\begin{quantikz}[row sep=2mm,column sep=2mm]
            & \targ{} & & \targ{} & & \targ{} & & \targ{} & \gate{\p{x}} &\\
            & \ctrl{-1} & \gate{\p{h}} & \ctrl{-1} & \gate{\p{h}} & \ctrl{-1} & \gate{\p{h}} & \ctrl{-1} & &
        \end{quantikz}} \\
        \overset{\text{naturality of $\swapt$, \cref{eq:qt5,eq:cnot_h}}}{\fromto_2{}}
        \scalebox{\scaleratio}{\begin{quantikz}[row sep=2mm,column sep=2mm]
            & \gate{\p{x}} & \gate{\p{x}} &\\
            & \gate{\p{h}} & &
        \end{quantikz}}
        \overset{\text{completeness of $\Pi$}}{\fromto_2{}}
        \scalebox{\scaleratio}{\begin{quantikz}[row sep=2mm,column sep=2mm]
            & \ghost{} &\\
            & \gate{\p{h}} &
        \end{quantikz}} \qedhere
    \end{gather*}
\end{proof}

\begin{lemma}
    \label{lem:d3_fixed}
    In \QPiLang{},
    \begin{equation*}
        \cT_Q^{4}\left[(\CH{3,4}\CH{1,3}\CH{2,4})^4\right] \fromto_2{} \cT_Q^{4}\left[\CH{1,2}\CH{3,4}\right].
    \end{equation*}
\end{lemma}
\begin{proof}
    \allowdisplaybreaks
    We abbreviate $\cT_Q^4$ to $\cT_Q$. Then, by completeness of $\Pi$, we have
    \begin{align*}
        \qT{\CH{1,2}} &\fromto_2{} \scalebox{\scaleratio}
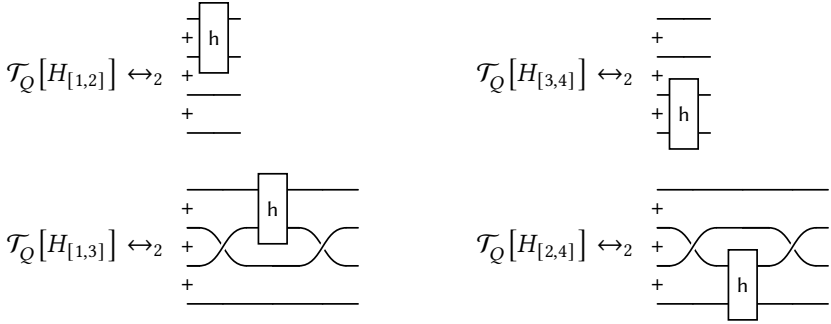
{\begin{quantikz}[row sep=1mm,column sep=2mm]
            \lstick[2,brackets=none]{$+$\hspace{-.5cm}}& \gate[2]{\p{h}} &\\
            \lstick[2,brackets=none]{$+$\hspace{-.5cm}}& &\\
            \lstick[2,brackets=none]{$+$\hspace{-.5cm}}& \ghost{\p{h}} &\\
            & \ghost{\p{h}} &
        \end{quantikz}} &
        \qT{\CH{3,4}} &\fromto_2{} \scalebox{\scaleratio}{\begin{quantikz}[row sep=1mm,column sep=2mm]
            \lstick[2,brackets=none]{$+$\hspace{-.5cm}}& \ghost{\p{h}} &\\
            \lstick[2,brackets=none]{$+$\hspace{-.5cm}}& \ghost{\p{h}} &\\
            \lstick[2,brackets=none]{$+$\hspace{-.5cm}}& \gate[2]{\p{h}} &\\
            & &
        \end{quantikz}} \\
        \qT{\CH{1,3}} &\fromto_2{} \scalebox{\scaleratio}{\begin{quantikz}[row sep=1mm,column sep=2mm]
            \lstick[2,brackets=none]{$+$\hspace{-.5cm}}& &\gate[2]{\p{h}}& &\\
            \lstick[2,brackets=none]{$+$\hspace{-.5cm}}& \permute{2,1} && \permute{2,1} &\\
            \lstick[2,brackets=none]{$+$\hspace{-.5cm}}& &\ghost{\p{h}}& &\\
            & &\ghost{\p{h}}& &
        \end{quantikz}} &
        \qT{\CH{2,4}} &\fromto_2{} \scalebox{\scaleratio}{\begin{quantikz}[row sep=1mm,column sep=2mm]
            \lstick[2,brackets=none]{$+$\hspace{-.5cm}}& &\ghost{\p{h}}& &\\
            \lstick[2,brackets=none]{$+$\hspace{-.5cm}}& \permute{2,1} &\ghost{\p{h}}& \permute{2,1} &\\
            \lstick[2,brackets=none]{$+$\hspace{-.5cm}}& &\gate[2]{\p{h}}& &\\
            & && &
        \end{quantikz}}
    \end{align*}
    By aligning $\ket{00},\ket{01},\ket{10},\ket{11}$ with $1,2,3,4$ respectively, the above terms have equivalent representations as quantum circuits:
    \begin{align*}
        \qT{\CH{1,2}} &\fromto_2{} \scalebox{\scaleratio}
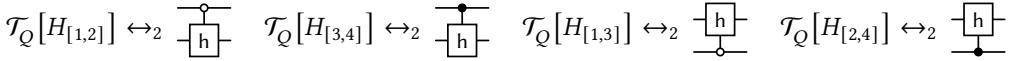
{\begin{quantikz}[row sep=2mm,column sep=2mm]
            &\octrl{1}& \\
            &\gate{\p{h}}& 
        \end{quantikz}} &
        \qT{\CH{3,4}} &\fromto_2{} \scalebox{\scaleratio}{\begin{quantikz}[row sep=2mm,column sep=2mm]
            &\ctrl{1}& \\
            &\gate{\p{h}}& 
        \end{quantikz}} &
        \qT{\CH{1,3}} &\fromto_2{} \scalebox{\scaleratio}{\begin{quantikz}[row sep=2mm,column sep=2mm]
            &\gate{\p{h}}& \\
            &\octrl{-1}&
        \end{quantikz}} &
        \qT{\CH{2,4}} &\fromto_2{} \scalebox{\scaleratio}{\begin{quantikz}[row sep=2mm,column sep=2mm]
            &\gate{\p{h}}& \\
            &\ctrl{-1}&
        \end{quantikz}}
    \end{align*}
    Then, our goal is equivalently represented as a quantum circuit as follows.
    \begin{equation*}
        \scalebox{\scaleratio}{\begin{quantikz}[row sep=2mm,column sep=2mm]
            & \gate{\p{h}} & \gate{\p{h}} &\ctrl{1} & \gate{\p{h}} & \gate{\p{h}} & \ctrl{1} & \gate{\p{h}} & \gate{\p{h}} & \ctrl{1} & \gate{\p{h}} & \gate{\p{h}} &\ctrl{1} & \\
            & \ctrl{-1} & \octrl{-1} & \gate{\p{h}} & \ctrl{-1} & \octrl{-1} & \gate{\p{h}} & \ctrl{-1} & \octrl{-1} & \gate{\p{h}} & \ctrl{-1} & \octrl{-1} &\gate{\p{h}} &
        \end{quantikz}}
        \fromto_2{}
        \scalebox{\scaleratio}{\begin{quantikz}[row sep=2mm,column sep=2mm]
            & \ctrl{1} & \octrl{1} &\\
            & \gate{\p{h}} & \gate{\p{h}} &
        \end{quantikz}}
    \end{equation*}
    By \cref{eq:qt2}, it becomes
    \begin{equation}\label{eq:d3_fixed_qc}
        \scalebox{\scaleratio}{\begin{quantikz}[row sep=2mm,column sep=2mm]
            & \gate{\p{h}} &\ctrl{1} & \gate{\p{h}} & \ctrl{1} & \gate{\p{h}} & \ctrl{1} & \gate{\p{h}} &\ctrl{1} & \\
            & & \gate{\p{h}} & & \gate{\p{h}} & & \gate{\p{h}} & &\gate{\p{h}} &
        \end{quantikz}}
        \fromto_2{}
        \scalebox{\scaleratio}{\begin{quantikz}[row sep=2mm,column sep=2mm]
            & \ghost{\p{h}} &\\
            & \gate{\p{h}} &
        \end{quantikz}}
    \end{equation}
    Using \cref{eq:h_decomposition,eq:conjugation} to decompose $\p{ch}$ with one ancilla qubit, \cref{eq:d3_fixed_qc} rewrites to
    \begin{gather*}
        \scalebox{\scaleratio}
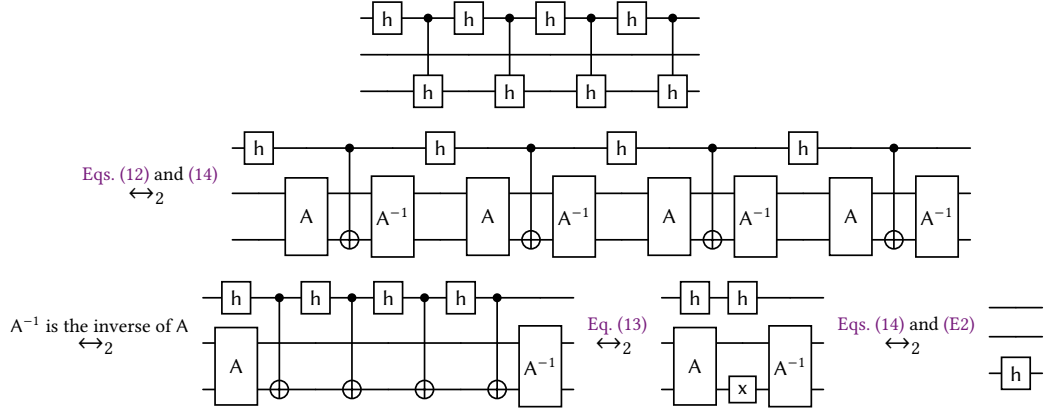
{\begin{quantikz}[row sep=2mm,column sep=2mm]
            & \gate{\p{h}} &\ctrl{2} & \gate{\p{h}} & \ctrl{2} & \gate{\p{h}} & \ctrl{2} & \gate{\p{h}} &\ctrl{2} & \\
            & \ghost{} & & & & & & & & \\
            & & \gate{\p{h}} & & \gate{\p{h}} & & \gate{\p{h}} & &\gate{\p{h}} &
        \end{quantikz}} \\
        \overset{\text{\cref{eq:h_decomposition,eq:conjugation}}}{\fromto_2{}}
        \scalebox{\scaleratio}{\begin{quantikz}[row sep=2mm,column sep=2mm]
            & \gate{\p{h}} & &\ctrl{2} & & \gate{\p{h}} & & \ctrl{2} & & \gate{\p{h}} & & \ctrl{2} & & \gate{\p{h}} & & \ctrl{2} & & \\
            & & \gate[2][.7cm]{\p{A}} & & \gate[2][.7cm]{\mathclap{\p{A}^{-1}}}& & \gate[2][.7cm]{\p{A}} & & \gate[2][.7cm]{\mathclap{\p{A}^{-1}}} & & \gate[2][.7cm]{\p{A}} & & \gate[2][.7cm]{\mathclap{\p{A}^{-1}}} & & \gate[2][.7cm]{\p{A}} & & \gate[2][.7cm]{\mathclap{\p{A}^{-1}}} & \\
            & & & \targ{} & & & & \targ{} & & & & \targ{} & & & & \targ{} & &
        \end{quantikz}} \\
        \overset{\text{$\p{A}^{-1}$ is the inverse of $\p{A}$}}{\fromto_2{}}
        \scalebox{\scaleratio}{\begin{quantikz}[row sep=2mm,column sep=2mm]
            & \gate{\p{h}} &\ctrl{2} & \gate{\p{h}} & \ctrl{2} & \gate{\p{h}} & \ctrl{2} & \gate{\p{h}} & \ctrl{2} & & \\
            & \gate[2][.7cm]{\p{A}} & & & & & & & & \gate[2][.7cm]{\mathclap{\p{A}^{-1}}} & \\
            & & \targ{} & & \targ{} & & \targ{} & & \targ{} & &
        \end{quantikz}}
        \overset{\text{\cref{eq:cnot_h}}}{\fromto_2{}}
        \scalebox{\scaleratio}{\begin{quantikz}[row sep=2mm,column sep=2mm]
            & \gate{\p{h}} & \gate{\p{h}} & & \\
            & \gate[2][.7cm]{\p{A}} & & \gate[2][.7cm]{\mathclap{\p{A}^{-1}}} & \\
            & & \gate{\p{x}} & &
        \end{quantikz}} 
        \overset{\text{\cref{eq:h_decomposition,eq:e2}}}{\fromto_2{}}
        \scalebox{\scaleratio}{\begin{quantikz}[row sep=2mm,column sep=2mm]
            & \ghost{} & \\
            & \ghost{} & \\
            & \gate{\p{h}} &
        \end{quantikz}}
    \end{gather*}
    Then, by invoking the cancellative condition \cref{eq:e4} to remove the ancilla qubit,\footnote{An ancilla qubit can be eliminated via the following derivation: $\pid_{\bbo+\bbo}\times c \fromto_2 \pid \Rightarrow c + c \fromto_2 \pid \Rightarrow c \fromto_2 \pid$.} along with the naturality of $\swapt$ and distributivity, we derive \cref{eq:d3_fixed_qc}.
\end{proof}

\begin{lemma}
    \label{lem:d4_fixed}
    In \QPiLang{},
    \[ \cT_Q^{6}\left[(\CH{1,3}\CH{2,4}\CH{1,2}\CH{1,3}\CH{2,4}\CX{3,5}\CX{4,6})^3\right] \fromto_2{} \cT_Q^{6}\left[\CH{3,5}\CH{4,6}\CH{5,6}\CH{3,5}\CH{4,6}\CX{3,5}\CX{4,6}\right].\]
\end{lemma}
\begin{proof}
    By the cancellative condition \cref{eq:e4}, it suffices to prove this lemma for $\cT_Q^{8}$.
    Similar to the proof of \cref{lem:d3_fixed}, i.e., by aligning $\ket{000},\ket{001},\ldots,\ket{111}$ with $1,2,\ldots,8$ respectively, we have the following equivalent representations via $3$-qubit circuits.
    \begin{equation*}
    \adjustbox{max width=\linewidth}{$
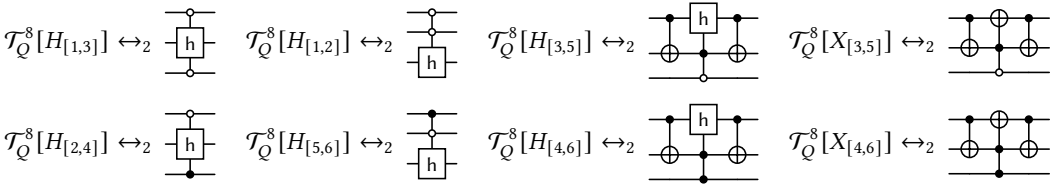
    \begin{aligned}
        \cT_Q^{8}[\CH{1,3}] &\fromto_2{} \scalebox{\scaleratio}{\begin{quantikz}[row sep=2mm,column sep=2mm]
            &\octrl{1}& \\
            &\gate{\p{h}}& \\
            &\octrl{-1}&
        \end{quantikz}} &
        \cT_Q^{8}[\CH{1,2}] &\fromto_2{} \scalebox{\scaleratio}{\begin{quantikz}[row sep=2mm,column sep=2mm]
            &\octrl{1}& \\
            &\octrl{1}& \\
            &\gate{\p{h}}&
        \end{quantikz}} &
        \cT_Q^{8}[\CH{3,5}] &\fromto_2{} \scalebox{\scaleratio}{\begin{quantikz}[row sep=2mm,column sep=2mm]
            &\ctrl{1} &\gate{\p{h}}& \ctrl{1} & \\
            &\targ{} &\ctrl{-1}& \targ{} &\\
            & &\octrl{-1} & &
        \end{quantikz}} &
        \cT_Q^{8}[\CX{3,5}] &\fromto_2{} \scalebox{\scaleratio}{\begin{quantikz}[row sep=2mm,column sep=2mm]
            &\ctrl{1} &\targ{}& \ctrl{1} & \\
            &\targ{} &\ctrl{-1}& \targ{} &\\
            & &\octrl{-1} & &
        \end{quantikz}} \\
        \cT_Q^{8}[\CH{2,4}] &\fromto_2{} \scalebox{\scaleratio}{\begin{quantikz}[row sep=2mm,column sep=2mm]
            &\octrl{1}& \\
            &\gate{\p{h}}& \\
            &\ctrl{-1}&
        \end{quantikz}} &
        \cT_Q^{8}[\CH{5,6}] &\fromto_2{} \scalebox{\scaleratio}{\begin{quantikz}[row sep=2mm,column sep=2mm]
            & \ctrl{1} & \\
            & \octrl{1} & \\
            & \gate{\p{h}} &
        \end{quantikz}} &
        \cT_Q^{8}[\CH{4,6}] &\fromto_2{} \scalebox{\scaleratio}{\begin{quantikz}[row sep=2mm,column sep=2mm]
            &\ctrl{1} &\gate{\p{h}}& \ctrl{1} & \\
            &\targ{} &\ctrl{-1}& \targ{} &\\
            & &\ctrl{-1} & &
        \end{quantikz}} &
        \cT_Q^{8}[\CX{4,6}] &\fromto_2{} \scalebox{\scaleratio}{\begin{quantikz}[row sep=2mm,column sep=2mm]
            &\ctrl{1} &\targ{}& \ctrl{1} &\\
            &\targ{} &\ctrl{-1}& \targ{} &\\
            & &\ctrl{-1} & &
        \end{quantikz}}
    \end{aligned}$}
    \end{equation*}
    Then, by completeness of $\Pi$ and \cref{eq:qt2}, we have
    \begin{gather*}
        \cT_Q^{8}[\CH{1,3}\CH{2,4}] \fromto_2{} \scalebox{\scaleratio}
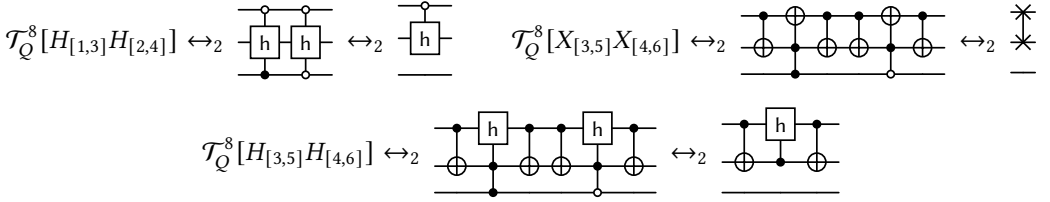
{\begin{quantikz}[row sep=2mm,column sep=2mm]
            & \octrl{1} & \octrl{1} & \\
            & \gate{\p{h}} & \gate{\p{h}} & \\
            & \ctrl{-1} &\octrl{-1}&
        \end{quantikz}}
        \fromto_2{} \scalebox{\scaleratio}{\begin{quantikz}[row sep=2mm,column sep=2mm]
            & \octrl{1} & \\
            & \gate{\p{h}} & \\
            & \ghost{} &
        \end{quantikz}} \qquad
        \cT_Q^{8}[\CX{3,5}\CX{4,6}] \fromto_2{} \scalebox{\scaleratio}{\begin{quantikz}[row sep=2mm,column sep=2mm]
            &\ctrl{1} &\targ{}& \ctrl{1} &\ctrl{1} &\targ{}& \ctrl{1} & \\
            &\targ{} &\ctrl{-1}& \targ{} &\targ{} &\ctrl{-1}& \targ{} &\\
            & &\ctrl{-1} & & &\octrl{-1} & &
        \end{quantikz}} \fromto_2{} \scalebox{\scaleratio}{\begin{quantikz}[row sep=5mm,column sep=2mm]
            & \swap{1} & \\
            & \targX{} & \\
            &  &
        \end{quantikz}}
        \\
        \cT_Q^{8}[\CH{3,5}\CH{4,6}] \fromto_2{} \scalebox{\scaleratio}{\begin{quantikz}[row sep=2mm,column sep=2mm]
            &\ctrl{1} &\gate{\p{h}}& \ctrl{1} &\ctrl{1} &\gate{\p{h}}& \ctrl{1} & \\
            &\targ{} &\ctrl{-1}& \targ{} &\targ{} &\ctrl{-1}& \targ{} &\\
            & &\ctrl{-1} & & &\octrl{-1} & &
        \end{quantikz}} \fromto_2{} \scalebox{\scaleratio}{\begin{quantikz}[row sep=2mm,column sep=2mm]
            &\ctrl{1} &\gate{\p{h}}& \ctrl{1} & \\
            &\targ{} &\ctrl{-1}& \targ{} &\\
            & & \ghost{} & &
        \end{quantikz}}
    \end{gather*}
    Thus, our goal for $\cT_Q^8$ is equivalent to
    \begin{align*}
        & \scalebox{\scaleratio}
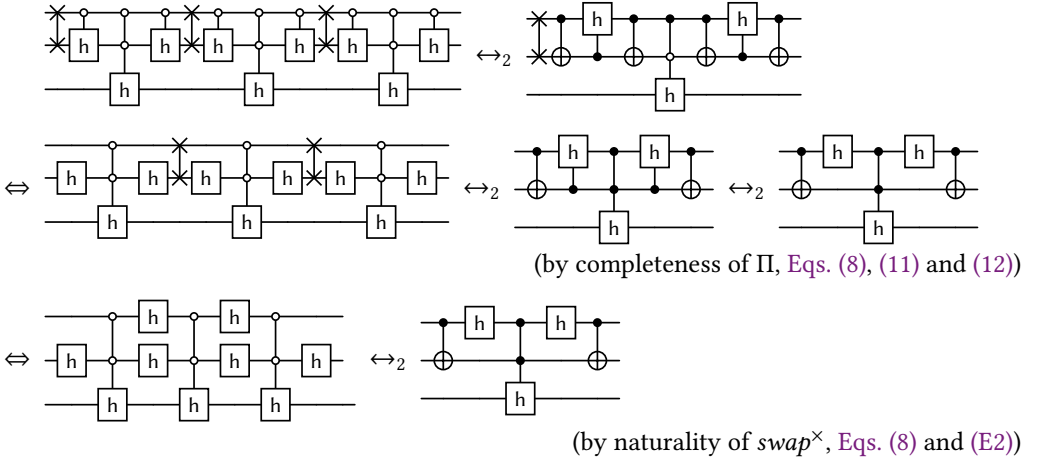
{\begin{quantikz}[row sep=2mm,column sep=2mm]
            & \swap{1} & \octrl{1} & \octrl{1} & \octrl{1} & \swap{1} & \octrl{1} & \octrl{1} & \octrl{1} & \swap{1} & \octrl{1} & \octrl{1} & \octrl{1} &\\
            & \targX{} & \gate{\p{h}} & \octrl{1} & \gate{\p{h}} & \targX{} & \gate{\p{h}} & \octrl{1} & \gate{\p{h}} & \targX{} & \gate{\p{h}} & \octrl{1} & \gate{\p{h}} &\\
            & & & \gate{\p{h}} & & & & \gate{\p{h}} & & & & \gate{\p{h}} & &
        \end{quantikz}}
        \fromto_2{}
        \scalebox{\scaleratio}{\begin{quantikz}[row sep=2mm,column sep=2mm]
            & \swap{1} & \ctrl{1} & \gate{\p{h}} & \ctrl{1} & \ctrl{1} & \ctrl{1} & \gate{\p{h}} & \ctrl{1} & \\
            & \targX{} & \targ{} & \ctrl{-1} & \targ{} & \octrl{1} & \targ{} & \ctrl{-1} & \targ{} & \\
            & & & & & \gate{\p{h}} & & & &
        \end{quantikz}} \\
        \Leftrightarrow{}&
        \scalebox{\scaleratio}{\begin{quantikz}[row sep=2mm,column sep=2mm]
            & & \octrl{1} & & \swap{1} & & \octrl{1} & & \swap{1} & & \octrl{1} & &\\
            & \gate{\p{h}} & \octrl{1} & \gate{\p{h}} & \targX{} & \gate{\p{h}} & \octrl{1} & \gate{\p{h}} & \targX{} & \gate{\p{h}} & \octrl{1} & \gate{\p{h}} &\\
            & & \gate{\p{h}} & & & & \gate{\p{h}} & & & & \gate{\p{h}} & &
        \end{quantikz}}
        \fromto_2{}
        \scalebox{\scaleratio}{\begin{quantikz}[row sep=2mm,column sep=2mm]
            & \ctrl{1} & \gate{\p{h}} & \ctrl{1} & \gate{\p{h}} & \ctrl{1} & \\
            & \targ{} & \ctrl{-1} & \ctrl{1} & \ctrl{-1} & \targ{} & \\
            & & & \gate{\p{h}} & & &
        \end{quantikz}} \fromto_2{}
        \scalebox{\scaleratio}{\begin{quantikz}[row sep=2mm,column sep=2mm]
            & \ctrl{1} & \gate{\p{h}} & \ctrl{1} & \gate{\p{h}} & \ctrl{1} & \\
            & \targ{} & & \ctrl{1} & & \targ{} & \\
            & & & \gate{\p{h}} & & &
        \end{quantikz}}
        \tag{by completeness of $\Pi$, \cref{eq:octrl,eq:qt5,,eq:conjugation}} \\
        \Leftrightarrow{}&
        \scalebox{\scaleratio}{\begin{quantikz}[row sep=2mm,column sep=2mm]
            & & \octrl{1} & \gate{\p{h}} & \octrl{1} & \gate{\p{h}} & \octrl{1} & &\\
            & \gate{\p{h}} & \octrl{1} & \gate{\p{h}} & \octrl{1} & \gate{\p{h}} & \octrl{1} & \gate{\p{h}} &\\
            & & \gate{\p{h}} & & \gate{\p{h}} & & \gate{\p{h}} & & &
        \end{quantikz}}
        \fromto_2{}
        \scalebox{\scaleratio}{\begin{quantikz}[row sep=2mm,column sep=2mm]
            & \ctrl{1} & \gate{\p{h}} & \ctrl{1} & \gate{\p{h}} & \ctrl{1} & \\
            & \targ{} & & \ctrl{1} & & \targ{} & \\
            & & & \gate{\p{h}} & & &
        \end{quantikz}} \tag{by naturality of $\swapt$, \cref{eq:qt5,,eq:e2}}
    \end{align*}
    With the cancellative condition \cref{eq:e4}, naturality of $\swapt$ and distributivity, it suffices to prove
    \begin{equation*}
        \adjustbox{max width=\linewidth}{$
        \scalebox{\scaleratio}
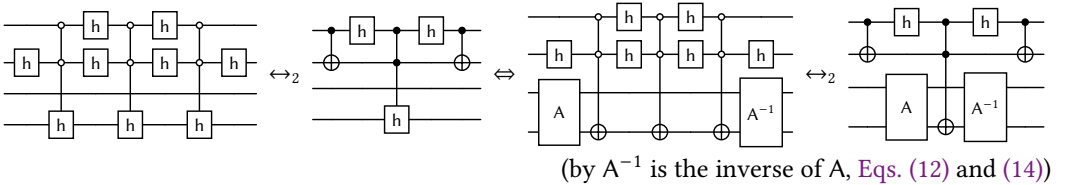
{\begin{quantikz}[row sep=2mm,column sep=2mm]
            & & \octrl{1} & \gate{\p{h}} & \octrl{1} & \gate{\p{h}} & \octrl{1} & &\\
            & \gate{\p{h}} & \octrl{2} & \gate{\p{h}} & \octrl{2} & \gate{\p{h}} & \octrl{2} & \gate{\p{h}} &\\
            & \ghost{} & & & & & & & \\
            & & \gate{\p{h}} & & \gate{\p{h}} & & \gate{\p{h}} & &
        \end{quantikz}}
        \fromto_2{}
        \scalebox{\scaleratio}{\begin{quantikz}[row sep=2mm,column sep=2mm]
            & \ctrl{1} & \gate{\p{h}} & \ctrl{1} & \gate{\p{h}} & \ctrl{1} & \\
            & \targ{} & & \ctrl{2} & & \targ{} & \\
            & \ghost{} & & & & & \\
            & & & \gate{\p{h}} & & &
        \end{quantikz}} 
        \Leftrightarrow{}
        \scalebox{\scaleratio}{\begin{quantikz}[row sep=2mm,column sep=2mm]
            & & \octrl{1} & \gate{\p{h}} & \octrl{1} & \gate{\p{h}} & \octrl{1} & &\\
            & \gate{\p{h}} & \octrl{2} & \gate{\p{h}} & \octrl{2} & \gate{\p{h}} & \octrl{2} & \gate{\p{h}} &\\
            & \gate[2][.8cm]{\p{A}} & & & & & & \gate[2][.8cm]{\p{A}^{-1}} &\\
            & & \targ{} & & \targ{} & & \targ{} & &
        \end{quantikz}}
        \fromto_2{}
        \scalebox{\scaleratio}{\begin{quantikz}[row sep=2mm,column sep=2mm]
            & \ctrl{1} & \gate{\p{h}} & \ctrl{1} & \gate{\p{h}} & \ctrl{1} & \\
            & \targ{} & & \ctrl{2} & & \targ{} & \\
            & & \gate[2][.8cm]{\p{A}} & & \gate[2][.8cm]{\p{A}^{-1}} & &\\
            & & & \targ{} & & &
        \end{quantikz}}$}
        \tag{by $\p{A}^{-1}$ is the inverse of $\p{A}$, \cref{eq:h_decomposition,eq:conjugation}}
    \end{equation*}
    where we use the same decomposition trick as in the proof of \cref{eq:d3_fixed_qc}.
    Similarly, it suffices to prove
    \begin{align*}
        & \scalebox{\scaleratio}{\begin{quantikz}[row sep=2mm,column sep=2mm]
            & & \octrl{1} & \gate{\p{h}} & \octrl{1} & \gate{\p{h}} & \octrl{1} & &\\
            & \gate{\p{h}} & \octrl{1} & \gate{\p{h}} & \octrl{1} & \gate{\p{h}} & \octrl{1} & \gate{\p{h}} &\\
            & & \targ{} & & \targ{} & & \targ{} & &
        \end{quantikz}}
        \fromto_2{}
        \scalebox{\scaleratio}{\begin{quantikz}[row sep=2mm,column sep=2mm]
            & \ctrl{1} & \gate{\p{h}} & \ctrl{1} & \gate{\p{h}} & \ctrl{1} & \\
            & \targ{} & & \ctrl{1} & & \targ{} & \\
            & & & \targ{} & & &
        \end{quantikz}} \\
        \Leftrightarrow{}&
        \scalebox{\scaleratio}{\begin{quantikz}[row sep=2mm,column sep=2mm]
            & & & \octrl{1} & \gate{\p{h}} & \gate{\p{x}} & \ctrl{1} & \gate{\p{x}} & \gate{\p{h}} & \octrl{1} & & &\\
            & \gate{\p{h}} & \gate{\p{x}} & \ctrl{1} & \gate{\p{x}} & \gate{\p{h}} & \octrl{1} & \gate{\p{h}} & \gate{\p{x}} & \ctrl{1} & \gate{\p{x}} & \gate{\p{h}} &\\
            & & & \targ{} & & & \targ{} & & & \targ{} & & &
        \end{quantikz}}
        \fromto_2{}
        \scalebox{\scaleratio}{\begin{quantikz}[row sep=2mm,column sep=2mm]
            & \ctrl{1} & \gate{\p{h}} & \ctrl{1} & \gate{\p{h}} & \ctrl{1} & \\
            & \targ{} & & \ctrl{1} & & \targ{} & \\
            & & & \targ{} & & &
        \end{quantikz}} \tag{by \cref{eq:conjugation,eq:octrl}} \\
        \Leftrightarrow{}&
        \scalebox{\scaleratio}{\begin{quantikz}[row sep=2mm,column sep=2mm]
            & & & \octrl{1} & \gate{\p{z}} & \gate{\p{h}} & \ctrl{1} & \gate{\p{h}} & \gate{\p{z}} & \octrl{1} & & &\\
            & \gate{\p{z}} & \gate{\p{h}} & \ctrl{1} & \gate{\p{h}} & \gate{\p{z}} & \octrl{1} & \gate{\p{z}} & \gate{\p{h}} & \ctrl{1} & \gate{\p{h}} & \gate{\p{z}} &\\
            & & & \targ{} & & & \targ{} & & & \targ{} & & &
        \end{quantikz}}
        \fromto_2{}
        \scalebox{\scaleratio}{\begin{quantikz}[row sep=2mm,column sep=2mm]
            & \ctrl{1} & \gate{\p{h}} & \ctrl{1} & \gate{\p{h}} & \ctrl{1} & \\
            & \targ{} & & \ctrl{1} & & \targ{} & \\
            & & & \targ{} & & &
        \end{quantikz}} \tag{by \cref{eq:e2,eq:e3}} \\
        \Leftrightarrow{}&
        \scalebox{\scaleratio}{\begin{quantikz}[row sep=2mm,column sep=2mm]
            & & \octrl{1} & \gate{\p{z}}&  & \targ{} & & \gate{\p{z}} & \octrl{1} & & \\
            & \gate{\p{z}} & \targ{} & & \gate{\p{z}} & \octrl{-1} & \gate{\p{z}} & & \targ{} & \gate{\p{z}} &\\
            & \gate{\p{h}} & \ctrl{-1} & \gate{\p{h}} & \gate{\p{h}} & \ctrl{-1} & \gate{\p{h}} & \gate{\p{h}} & \ctrl{-1} & \gate{\p{h}} &
        \end{quantikz}}
        \fromto_2{}
        \scalebox{\scaleratio}{\begin{quantikz}[row sep=2mm,column sep=2mm]
            & \ctrl{1} & \targ{} & \ctrl{1} & \\
            & \targ{} & \ctrl{-1} & \targ{} & \\
            & \gate{\p{h}} & \ctrl{-1} & \gate{\p{h}} &
        \end{quantikz}} \tag{by \cref{eq:qt4,eq:qt5,eq:conjugation,,eq:e2}}  \\
        \Leftrightarrow{}&
        \scalebox{\scaleratio}{\begin{quantikz}[row sep=2mm,column sep=2mm]
            & & \octrl{1} & \gate{\p{z}} & \targ{} & \gate{\p{z}} & \octrl{1} & & \\
            & \gate{\p{z}} & \targ{} & \gate{\p{z}} & \octrl{-1} & \gate{\p{z}} & \targ{} & \gate{\p{z}} &\\
            & & \ctrl{-1} &  & \ctrl{-1} &  & \ctrl{-1} &  &
        \end{quantikz}}
        \fromto_2{}
        \scalebox{\scaleratio}{\begin{quantikz}[row sep=2mm,column sep=2mm]
            & \ctrl{1} & \targ{} & \ctrl{1} & \\
            & \targ{} & \ctrl{-1} & \targ{} & \\
            &  & \ctrl{-1} &  &
        \end{quantikz}} \tag{by \cref{eq:e2}}
    \end{align*}
    With the naturality of $\swapt$, \cref{eq:conjugation,eq:qt5}, this reduces further to:
    \begin{equation}\label{eq:d4_final}
        \scalebox{\scaleratio}{\begin{quantikz}[row sep=2mm,column sep=2mm]
            & & \octrl{1} & \gate{\p{z}} & \targ{} & \gate{\p{z}} & \octrl{1} & & \\
            & \gate{\p{z}} & \targ{} & \gate{\p{z}} & \octrl{-1} & \gate{\p{z}} & \targ{} & \gate{\p{z}} &
        \end{quantikz}}
        \fromto_2{}
        \scalebox{\scaleratio}{\begin{quantikz}[row sep=2mm,column sep=2mm]
            & \ctrl{1} & \targ{} & \ctrl{1} & \\
            & \targ{} & \ctrl{-1} & \targ{} &
        \end{quantikz}}
    \end{equation}
    Noticing that
        \begin{gather}
            \scalebox{\scaleratio}{\begin{quantikz}[row sep=2mm,column sep=2mm]
            & \gate{\p{z}} & \octrl{1} & \gate{\p{z}} & \\
            & & \gate{c} & &
            \end{quantikz}}
            \overset{\text{\cref{eq:qt5,eq:qt2}}}{\fromto_2{}}
            \scalebox{\scaleratio}{\begin{quantikz}[row sep=2mm,column sep=2mm]
            & \gate{\p{z}} & \gate{\p{z}} & \octrl{1} & \gate{\p{z}} & \gate{\p{z}} & \\
            & \octrl{-1} & \ctrl{-1} & \gate{c} & \ctrl{-1} & \octrl{-1} &
            \end{quantikz}}
            \overset{\text{\cref{eq:zz}}}{\fromto_2{}}
            \scalebox{\scaleratio}{\begin{quantikz}[row sep=2mm,column sep=2mm]
            & \gate{\p{z}} & \ctrl{1} & \octrl{1} & \ctrl{1} & \gate{\p{z}} & \\
            & \octrl{-1} & \gate{\p{z}} & \gate{c} & \gate{\p{z}} & \octrl{-1} &
            \end{quantikz}} \notag \\
            \overset{\text{\cref{eq:qt0,,eq:qt1,,eq:qt6,eq:e1}}}{\fromto_2{}}
            \scalebox{\scaleratio}{\begin{quantikz}[row sep=2mm,column sep=2mm]
            & \gate{\p{z}} & \octrl{1} & \gate{\p{z}} & \\
            & \octrl{-1} & \gate{c} & \octrl{-1} &
            \end{quantikz}}
            \overset{\text{\cref{eq:octrl,eq:qt5}}}{\fromto_2{}}
            \scalebox{\scaleratio}{\begin{quantikz}[row sep=2mm,column sep=2mm]
            & & \gate{\p{z}} & & \octrl{1} & & \gate{\p{z}} & & \\
            & \gate{\p{x}} & \ctrl{-1} & \gate{\p{x}} & \gate{c} & \gate{\p{x}} & \ctrl{-1} & \gate{\p{x}} &
            \end{quantikz}} \notag \\
            \overset{\text{\cref{eq:zz}}}{\fromto_2{}}
            \scalebox{\scaleratio}{\begin{quantikz}[row sep=2mm,column sep=2mm]
            & & \ctrl{1} & & \octrl{1} & & \ctrl{1} & & \\
            & \gate{\p{x}} & \gate{\p{z}} & \gate{\p{x}} & \gate{c} & \gate{\p{x}} & \gate{\p{z}} & \gate{\p{x}} &
            \end{quantikz}}
            \overset{\text{\cref{eq:conjugation}}}{\fromto_2{}}
            \scalebox{\scaleratio}{\begin{quantikz}[row sep=2mm,column sep=2mm]
            & \ctrl{1} & \octrl{1} & \ctrl{1} & \\
            & \gate{\p{x}\seqq\p{z}\seqq\p{x}} & \gate{c} & \gate{\p{x}\seqq\p{z}\seqq\p{x}} &
            \end{quantikz}}
            \overset{\text{completeness of $\Pi$}}{\underset{\text{\cref{eq:qt0,eq:qt1,eq:e1}}}{\fromto_2{}}}
            \scalebox{\scaleratio}{\begin{quantikz}[row sep=2mm,column sep=2mm]
            & \octrl{1} & \\
            & \gate{c} &
            \end{quantikz}}
            \label{eq:z_octrl_z}
        \end{gather}
    \cref{eq:d4_final} rewrites to
        \begin{gather*}
            \scalebox{\scaleratio}{\begin{quantikz}[row sep=2mm,column sep=2mm]
                & & \octrl{1} & \gate{\p{z}} & \targ{} & \gate{\p{z}} & \octrl{1} & & \\
                & \gate{\p{z}} & \targ{} & \gate{\p{z}} & \octrl{-1} & \gate{\p{z}} & \targ{} & \gate{\p{z}} &
            \end{quantikz}}
            \overset{\text{\cref{eq:z_octrl_z,eq:qt5,eq:e1}}}{\fromto_2{}}
            \scalebox{\scaleratio}{\begin{quantikz}[row sep=2mm,column sep=2mm]
                & \gate{\p{z}} & \octrl{1} & & \targ{} & & \octrl{1} & \gate{\p{z}} & \\
                & \gate{\p{z}} & \targ{} & & \octrl{-1} & & \targ{} & \gate{\p{z}} &
            \end{quantikz}} \\
            \overset{\text{completeness of $\Pi$}}{\fromto_2{}}
            \scalebox{\scaleratio}{\begin{quantikz}[row sep=2mm,column sep=2mm]
                & \gate{\p{z}} & \swap{1} & \gate{\p{z}} & \\
                & \gate{\p{z}} & \targX{} & \gate{\p{z}} &
            \end{quantikz}}
            \overset{\text{naturality of $\swapt$ and \cref{eq:e1}}}{\fromto_2{}}
            \scalebox{\scaleratio}{\begin{quantikz}[row sep=2mm,column sep=2mm]
                & \ghost{} & \swap{1} & &\\
                & \ghost{} & \targX{} & &
            \end{quantikz}}
            \overset{\text{completeness of $\Pi$}}{\fromto_2{}}
            \scalebox{\scaleratio}{\begin{quantikz}[row sep=2mm,column sep=2mm]
            & \ctrl{1} & \targ{} & \ctrl{1} & \\
            & \targ{} & \ctrl{-1} & \targ{} &
        \end{quantikz}}
        \end{gather*}
    This completes the proof.
\end{proof}
\endgroup

Next, we show that the composition of $\wsem{\cdot}$ and $\cT_Q$ gives equivalent terms.

\begin{restatable}{lemma}{lemwqinv}\label{lem:wq_inv}
    If $c$ is a term of Q-$\Pi$, then $c \fromto_2 \qT[\hdim(c)]{\wsem{c}}$ is a level-2 combinator of Q-$\Pi$. 
\end{restatable}
\begin{proof}
    The proof proceeds by induction on the structure of $c$.
    Here we consider the most important case $c = c_1+c_2$. 
    Then $\hdim(c) = \hdim(c_1)+\hdim(c_2)$, so
        \begin{align*}
            \qT[\hdim(c)]{\wsem{c_1+c_2}}
            ={}& \qT[\hdim(c)]{\wsem{c_1}\shift(\wsem{c_2},\hdim(c_1))}  \\
            ={}& \qT[\hdim(c)]{\shift(\wsem{c_2},\hdim(c_1))} \seqq \qT[\hdim(c)]{\wsem{c_1}} \\
            ={}& \qT[\hdim(c_2)+\hdim(c_1)]{\shift(\wsem{c_2},\hdim(c_1))} \seqq \qT[\hdim(c)]{\wsem{c_1}} \\
            \fromto_2{}& \rbra*{\pid_{\hdim(c_1)\bbo}+\qT[\hdim(c_2)]{\wsem{c_2}}}\seqq\qT[\hdim(c)]{\wsem{c_1}} \tag{by \cref{lem:qT_shift} later} \\
            \fromto_2{}& \rbra*{\pid_{\hdim(c_1)\bbo}+c_2}\seqq\qT[\hdim(c_1)+\hdim(c_2)]{\wsem{c_1}} \tag{by inductive hypothesis} \\
            \fromto_2{}& \rbra*{\pid_{\hdim(c_1)\bbo}+c_2}\seqq\rbra*{\qT[\hdim(c_1)]{\wsem{c_1}}+\pid_{\hdim(c_2)\bbo}} \tag{by \cref{lem:qT_nm} later} \\
            \fromto_2{}& \rbra*{\pid_{\hdim(c_1)\bbo}+c_2}\seqq\rbra*{c_1+\pid_{\hdim(c_2)\bbo}} \tag{by inductive hypothesis} \\
            \fromto_2{}& c_1+c_2. \tag{bifunctoriality of $+$}
        \end{align*}
    The other cases are either basic or can be converted into $+$ using distributivity, and are detailed in \ifARXIV%
    \cref{app:qpi_proof}%
    \else%
    \cref{ext-app:qpi_proof} of the extended version~\cite{FHK25}%
    \fi.
\end{proof}

The full proof of \cref{lem:wq_inv} requires the following two lemmas, showing that when a word $\mathbf{G}$ over $\cG_n$ is overloaded as a word over $\cG_{n+k}$, the translation $\cT_Q$ will introduce additional identity terms.

\begin{lemma}\label{lem:qT_nm}
    When omitting the isomorphisms $\assocrp$, $\assoclp$, $\assocrt$, $\assoclt$, $\dist$, $\factor$ used for type conversion, for a word $\mathbf{G}$ over $\cG_n$, we have
    \[\qT[n]{\mathbf{G}} + \pid_{m\bbo}\fromto_2 \qT[n+m]{\mathbf{G}}.\]
\end{lemma}
\begin{proof}
    \begingroup
    \allowdisplaybreaks
    By induction on $\mathbf{G}$:
    \begin{itemize}
        \item If $\mathbf{G} = \CZ{a}$, then \begin{align*}
            \qT[n+m]{\CZ{a}}
            ={}& \swapp(a,n+m)\seqq(\pid_{(n+m-1)\bbo}+(-1))\seqq \swapp(a,n+m) \\
            \fromto_2{} & \swapp(a,n+m)\seqq\swapp(n,n+m)\seqq (\pid_{(n-1)\bbo}+(-1)+\pid_{m\bbo})\seqq \\
            & \quad \swapp(n,n+m)\seqq \swapp(a,n+m) \tag{by naturality of $\swapp(n,n+m)$} \\
            \fromto_2{}& \swapp(a,n)\seqq(\pid_{(n-1)\bbo}+(-1)+\pid_{m\bbo})\seqq \swapp(a,n) \tag{by completeness of $\Pi$} \\
            \fromto_2{}& \rbra*{\swapp(a,n)\seqq(\pid_{(n-1)\bbo}+(-1))\seqq \swapp(a,n)} + \pid_{m\bbo} \tag{by bifunctoriality of $+$} \\
            ={} & \qT[n]{\CZ{a}} + \pid_{m\bbo}\text.
        \end{align*}
        \item The cases $\mathbf{G} = \CX{b,c}$ and $\mathbf{G} = \CH{b,c}$ are similar.
        \item If $\mathbf{G}=\mathbf{G}_1\mathbf{G}_2$, then
        \begin{align*}
            \qT[n]{\mathbf{G}_1\mathbf{G}_2} + \pid_{m\bbo}
            ={}& \rbra*{\qT[n]{\mathbf{G}_2}\seqq\qT[n]{\mathbf{G}_1}}+\pid_{m\bbo} \\
            \fromto_2{}& (\qT[n]{\mathbf{G}_2}+\pid_{m\bbo})\seqq(\qT[n]{\mathbf{G}_1}+\pid_{m\bbo}) \tag{by bifunctoriality of $+$} \\
            \fromto_2{}& \qT[n+m]{\mathbf{G}_2}\seqq\qT[n+m]{\mathbf{G}_1} \tag{by inductive hypothesis} \\
            ={}& \qT[n+m]{\mathbf{G}_1\mathbf{G}_2}\text.
        \end{align*} 
    \end{itemize}
    \endgroup
    This completes the proof.
\end{proof}

\begin{restatable}{lemma}{lemqtshift}\label{lem:qT_shift}
    When omitting the isomorphisms $\assocrp$, $\assoclp$, $\assocrt$, $\assoclt$, $\dist$, $\factor$ used for type conversion, for a word $\mathbf{G}$ over $\cG_n$, we have
    \[\qT[n+k]{\hshift(\mathbf{G}, k)} \fromto_2 \pid_{k\bbo}+\qT[n]{\mathbf{G}}.\]
\end{restatable}
\begin{proof}
    By induction on the structure of $\mathbf{G}$, using a similar proof of \cref{lem:qT_nm}. See details in \ifARXIV%
    \cref{app:qpi_proof}%
    \else%
    \cref{ext-app:qpi_proof} of the extended version~\cite{FHK25}%
    \fi.
\end{proof}

\cref{lem:w_eq,lem:wq_inv} now give the completeness of \QPiLang{} with respect to $\approx$.
\begin{theorem}
    [Completeness of \QPiLang{} w.r.t $\approx$]\label{thm:qpi_relation}
    If \QPiLang{}-terms $c_1,c_2\colon b_1\fromto b_2$ satisfy $\wsem{c_1} \approx \wsem{c_2}$, then $c_1\fromto_2 c_2$ is a level-2 combinator of \QPiLang{}.
\end{theorem}
\begin{proof}
    Use the fact that $\hdim(b_1)=\hdim(c_1)=\hdim(c_2)$:
    \begin{align*}
        \wsem{c_1} \approx \wsem{c_2}
        \Rightarrow{}& \qT[\hdim(b_1)]{\wsem{c_1}} \fromto_2 \qT[\hdim(b_1)]{\wsem{c_2}} \tag{by \cref{lem:w_eq}} \\ 
        \Rightarrow{}& \qT[\hdim(c_1)]{\wsem{c_1}} \fromto_2 \qT[\hdim(c_2)]{\wsem{c_2}} \\
        \Rightarrow{}& c_1 \fromto_2 \qT[\hdim(c_1)]{\wsem{c_1}} \fromto_2 \qT[\hdim(c_2)]{\wsem{c_2}} \fromto_2 c_2. \tag{by \cref{lem:wq_inv}}
    \end{align*}
    This completes the proof.
\end{proof}

\cref{lem:w_faith} can take the equation $\sem{c_1} = \sem{c_2}$ for two terms $c_1,c_2$ of \QPiLang{} to the equation $\sem{\wsem{c_1}} = \sem{\wsem{c_2}}$ for two words $\wsem{c_1}, \wsem{c_2}$.
The completeness of $\approx$ (\cref{thm:relation_completeness}) then ensures that $\wsem{c_1} \approx \wsem{c_2}$, which is exactly the condition of \cref{thm:qpi_relation}.
Therefore, we prove \cref{thm:qpi} as follows.

\begin{proof}[Proof of \cref{thm:qpi}]
    For soundness, notice that by soundness of $\Pi$, it suffices to verify that \cref{eq:e1,eq:e2,eq:e3} are sound. This can be verified by direct computation.

    For completeness, use \Cref{lem:w_faith}, completeness of $\approx$ from \Cref{thm:relation_completeness}, and completeness of \QPiLang{} w.r.t. $\approx$ from \Cref{thm:qpi_relation}:
    \begin{align*}
        \sem{c_1} = \sem{c_2}
        \Rightarrow{}& \sem{\wsem{c_1}} = \sem{c_1} = \sem{c_2} = \sem{\wsem{c_2}} \\
        \Rightarrow{}& \wsem{c_1} \approx \wsem{c_2} \\
        \Rightarrow{}& c_1 \fromto_2 c_2. 
    \end{align*}
    This completes the proof.
\end{proof}

\subsection{Soundness and Completeness of \texorpdfstring{\HPiLang}{Hadamard-Pi}}\label{sec:hpi_soundness_and_completeness}
The next part follows a similar approach as in \cref{sec:qpi_soundness_and_completeness}, using translations $\qsem{\cdot}$ and $\cT_H$ to prove the completeness of \HPiLang{} with respect to \QPiLang{}.

\begin{definition}
    [Translation $\qsem{\cdot}$]
    The translation $\qsem{\cdot}$ that maps a term $c$ of \HPiLang{} to a term of Q-$\Pi$ is defined inductively on the structures of combinators:
    \begin{equation*}
        \begin{aligned}
            \qsem{\textit{iso}} &= \textit{iso}, &
            \qsem{c_1\seqq c_2} &= \qsem{c_1}\seqq \qsem{c_2}, \\
            \qsem{c_1+c_2} &= \qsem{c_1} + \qsem{c_2}, & 
            \qsem{c_1\times c_2} &= \qsem{c_1} \times \qsem{c_2}.
        \end{aligned}
    \end{equation*}
\end{definition}
$\qsem{\cdot}$ is just a simple embedding, since the syntax of combinators in \HPiLang{} is a subset of \QPiLang{}. The following is therefore easy to see.

\begin{lemma}\label{lem:q_faith}
    If $c$ is a \HPiLang{}-term, then $\sem{\qsem{c}} = \sem{c}$.
\end{lemma}

\begin{definition}
    [Translation $\cT_H$]
    The translation $\cT_H$ that maps a term $c:b_1 \fromto b_2$ of Q-$\Pi$ to a term $\hT{c}:\bbo+b_1 \fromto \bbo+b_2$ of \HPiLang{} is defined inductively on the structure of combinators and types as shown in \cref{fig:def_ht}.
    \begin{figure*}
        \begin{equation*}
            \begin{gathered}
            \begin{aligned}
                \hT{(-1)} &= \hadamard\seqq \swapp \seqq \hadamard \\
                \hT{\textit{iso}} &= \pid_{\bbo}+\textit{iso} \quad \text{(for other isomorphisms)} \\
                \hT{c_1\seqq c_2} &= \hT{c_1}\seqq \hT{c_2}
            \end{aligned} \\[2mm]
            \begin{prooftree}
                \hypo{c_1 : b_1 \fromto b_3} \hypo{c_2 : b_2 \fromto b_4}
                \infer2{\mbox{$\begin{aligned}
                    \hT{c_1 + c_2} ={}& \assoclp \seqq (\hT{c_1}+\pid_{b_2})\seqq (\swapp +\pid_{b_2}) \seqq \\
                    &\quad \assocrp \seqq (\pid_{b_3}+\hT{c_2})\seqq \assoclp \seqq (\swapp + \pid_{b_4})\seqq \assocrp
                    \end{aligned}$}}
            \end{prooftree} \\[2mm]
            \begin{prooftree}
                \hypo{c_1 : b_1 \fromto b_3} \hypo{c_2 : b_2 \fromto b_4}
                \infer2{ \hT{c_1 \times c_2} = (\pid_{\bbo}+\swapt)\seqq \hT{\pid_{b_2}\times c_1}\seqq (\pid_{\bbo}+\swapt) \seqq \hT{\pid_{b_3}\times c_2} }
            \end{prooftree}
            %\hT{c_1 \times c_2} = (\pid_{\bbo}+\swapt\seqq \dist)\seqq\cT_H[\underbrace{c_1+\cdots+c_1}_{\hdim(c_2)}]\seqq (\pid_{\bbo}+\factor\seqq\swapt\seqq\dist) \seqq \cT_H[\underbrace{c_2+\cdots+c_2}_{\hdim(c_1)}]\seqq(\pid_{\bbo}+\factor)
            \end{gathered}
        \end{equation*}
        
        \begin{flushleft}
            For terms of the form of $\pid_b\times c$, $\cT_H$ is defined inductively on the structure of $b$:
        \end{flushleft}
        \iffalse\begin{align*}
            \hT{\pid_{b}\times c} = \cT_H[\underbrace{c+\cdots+c}_{\hdim(b)}]
        \end{align*}
        \fi
        \begin{align*}
            \hT{\pid_{\bbz}\times c} &= (\pid_{\bbo}+\swapt\seqq\absorb\seqq\pid_{\bbz}\seqq\factorz\seqq\swapt) \\
            \hT{\pid_{\bbo}\times c} &= (\pid_{\bbo}+\unitet)\seqq \hT{c} \seqq (\pid_{\bbo}+\unitit) \\
            \hT{\pid_{b_1+b_2}\times c} &= (\pid_{\bbo}+\dist)\seqq\hT{\pid_{b_1}\times c + \pid_{b_2}\times c}\seqq(\pid_{\bbo}+\factor) \\
            \hT{\pid_{\bbz \times b_2}\times c} &= (\pid_{\bbo}+\assocrt\seqq\swapt\seqq\absorb\seqq\pid_{\bbz}\seqq\factorz\seqq\swapt\seqq\assoclt) \\
            \hT{\pid_{\bbo \times b_2}\times c} &= (\pid_{\bbo}+\assocrt\seqq\unitet)\seqq \hT{\pid_{b_2}\times c}\seqq (\pid_{\bbo}+\unitit\seqq\assoclt) \\
            \hT{\pid_{(b_3+b_4) \times b_2}\times c} &= (\pid_{\bbo}+\dist \times \pid)\seqq\hT{\pid_{b_3\times b_2+b_4\times b_2}\times c}\seqq (\pid_{\bbo}+\factor \times \pid) \\
            \hT{\pid_{(b_3\times b_4) \times b_2}\times c} &= (\pid_{\bbo}+\assocrt \times \pid)\seqq\hT{\pid_{b_3\times(b_4\times b_2)}\times c} \seqq (\pid_{\bbo}+\assoclt \times \pid)
        \end{align*}
        \begin{flushleft}
            The above definition is well-founded through an inductively defined ranking function $R$ on the type $b$ of $\pid_{b}\times c$:
        \end{flushleft}
        \begin{align*}
            R(\bbz) &= 1 & R(\bbo) &= 2 & R(b_1+b_2) &= R(b_1)+R(b_2) & R(b_1\times b_2) &= (R(b_1)+1)^2\times R(b_2)
        \end{align*}
        \caption{Definition of translation $\cT_H$. In the case of $\hT{\pid_b\times c}$, $\pid_b\times c$ will be transformed to an equivalent form of $\protect\underbrace{c+\cdots+c}_{\hdim(b)}$.}\label{fig:def_ht}
    \end{figure*}
\end{definition}

The translation $\cT_H$ introduces an additional $\bbo$ to make $(-1)$ simulable by $\hadamard\seqq\swapp\seqq\hadamard$.
To align (or merge) all the additional $\bbo$ of the sub-combinators with the single additional $\bbo$ of the parent-combinator, $\cT_H$ does some ``twists'' in \HPiLang{}.
The case of $\hT{c_1+c_2}$ has the following graphical presentation.
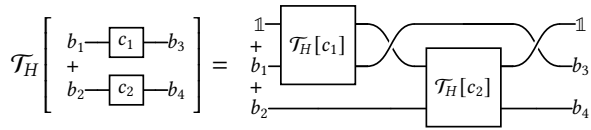
\begin{equation*}
    \hT{
        \scalebox{\scaleratio}{\begin{quantikz}[row sep=3mm,column sep=2mm]
            \lstick[2,brackets=none]{$+$\hspace{-.65cm}}\push{b_1} & & \gate{c_1} & & \push{b_3} \\
            \push{b_2}&& \gate{c_2} && \push{b_4}
        \end{quantikz}}
    }
    = 
    \scalebox{\scaleratio}{\begin{quantikz}[row sep=1mm,column sep=2mm]
        \lstick[2,brackets=none]{$+$\hspace{-.65cm}}\push{\bbo} & \gate[2]{\hT{c_1}} & \permute{2,1} & & \permute{2,1} & \push{\bbo}\\
        \lstick[2,brackets=none]{$+$\hspace{-.65cm}}\push{b_1} & & & \gate[2]{\hT{c_2}} & & \push{b_3}\\
        \push{b_2} & & & & & \push{b_4}
    \end{quantikz}}
\end{equation*}
The case of $\hT{c_1\times c_2}$ is handled by converting $c_1\times c_2$, which is equivalent to $\swapt\seqq(\pid\times c_1)\seqq\swapt\seqq(\pid\times c_2)$, into terms of $+$ by distributivity, using that $\pid\times c$ is equivalent to $c+\cdots+c$.

To establish a lemma similar to \cref{lem:w_eq}, we need to be careful with the ``twists'' above, which prevent us from directly using the rig structure.
For example, the following two terms
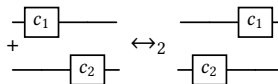
\[
    \scalebox{\scaleratio}{\begin{quantikz}[row sep=3mm,column sep=2mm]
        \lstick[2,brackets=none]{$+$\hspace{-.5cm}} & \gate{c_1} & & \midstick[2,brackets=none]{\scalebox{1.25}{$\fromto_2$}} & & \gate{c_1} &\\
        && \gate{c_2} & & \gate{c_2} & &
    \end{quantikz}}
\]
are equivalent in \QPiLang{}.
However, when translated to \HPiLang{} by $\cT_H$, they become
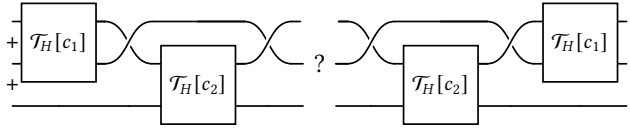
\[ 
    \scalebox{\scaleratio}{\begin{quantikz}[row sep=1mm,column sep=1.5mm]
        \lstick[2,brackets=none]{$+$\hspace{-.5cm}}& \gate[2]{\hT{c_1}} & \permute{2,1} & & \permute{2,1} & \midstick[3,brackets=none]{\scalebox{1.25}{$?$}} & \permute{2,1} & & \permute{2,1} & \gate[2]{\hT{c_1}} & \\
        \lstick[2,brackets=none]{$+$\hspace{-.5cm}}& & & \gate[2]{\hT{c_2}} & & & & \gate[2]{\hT{c_2}} & & &\\
        & & & & & & & & & &
    \end{quantikz}}
\]
The equivalence of these terms is not immediately apparent in \HPiLang{}, necessitating the following lemma.

\begin{lemma}\label{lem:ht2}
    If $c:b_1\fromto b_2$ is a term of \QPiLang{}, then \begin{align*}
        (\swapp+\pid_{b_1})\seqq\assocrp\seqq(\pid_\bbo+\hT{c})\seqq\assoclp \fromto_2{} \assoclp\seqq(\pid_\bbo+\hT{c})\seqq(\swapp+\pid_{b_2})
    \end{align*}
    is a level-2 combinator of \HPiLang{}.
\end{lemma}
\begin{proof}
\allowdisplaybreaks
This lemma shares the same graphical pattern as \cref{eq:h2} for $\mathit{hxh}$:
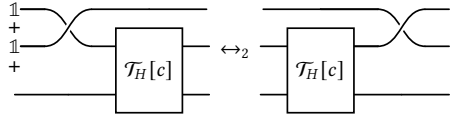
\begin{equation*}
    \scalebox{\scaleratio}{\begin{quantikz}[row sep=2mm,column sep=4mm]
        \lstick[2,brackets=none]{$+$\hspace{-.65cm}}\push{\bbo} & \permute{2,1} & & \midstick[3,brackets=none]{$\fromto_2$} & & \permute{2,1} & \\ 
        \lstick[2,brackets=none]{$+$\hspace{-.65cm}}\push{\bbo} & & \gate[2]{\hT{c}} & & \gate[2]{\hT{c}} & &\\ 
        & & & & & & 
    \end{quantikz}}
\end{equation*}
We give a graphical proof by induction on the structure of $c$.
\begin{itemize}
    \item If $c = (-1)$, then immediately invoke \cref{eq:h2}.
    \item If $c$ is another isomorphism, then
    \begin{gather*}
        \scalebox{\scaleratio}
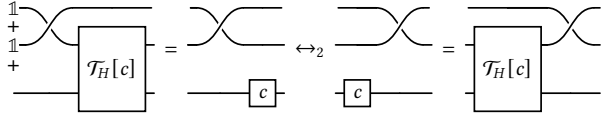
{\begin{quantikz}[row sep=2mm,column sep=1mm]
            \lstick[2,brackets=none]{$+$\hspace{-.65cm}}\push{\bbo} & \permute{2,1} & & \midstick[3,brackets=none]{$=$} & \permute{2,1} & & \midstick[3,brackets=none]{$\fromto_2$} & & \permute{2,1} & \midstick[3,brackets=none]{$=$} & &\permute{2,1}& \\ 
            \lstick[2,brackets=none]{$+$\hspace{-.65cm}}\push{\bbo} & & \gate[2]{\hT{\mathit{c}}} & & & &&&& & \gate[2]{\hT{\mathit{c}}} & &\\ 
            & & & & & \gate{\mathit{c}} & & \gate{\mathit{c}} & & & & &
        \end{quantikz}}
    \end{gather*}
    \item If $c = c_1\seqq c_2$, the proof follows quickly from the inductive hypothesis.
    \item If $c = c_1+c_2$, Then
    \begin{gather*}
        \scalebox{\scaleratio}
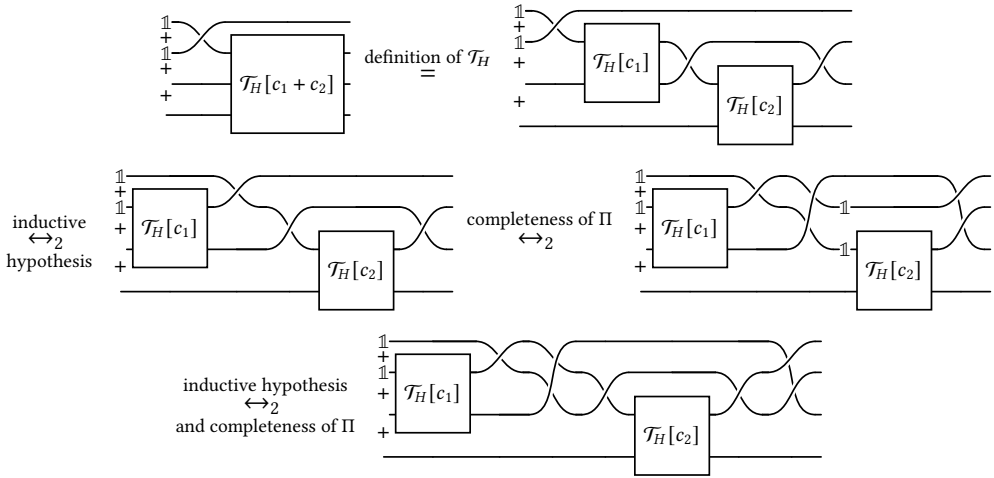
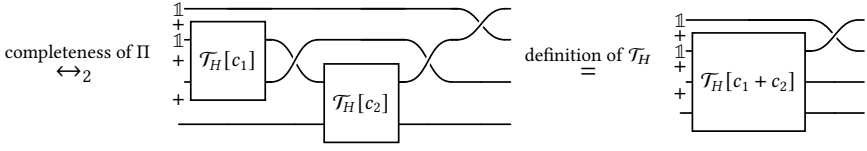
{\begin{quantikz}[row sep=1mm,column sep=1mm]
            \lstick[2,brackets=none]{$+$\hspace{-.65cm}}\push{\bbo} & \permute{2,1} & & \\ 
            \lstick[2,brackets=none]{$+$\hspace{-.65cm}}\push{\bbo} & & \gate[3]{\hT{c_1+c_2}} &\\ 
            \lstick[2,brackets=none]{$+$\hspace{-.65cm}}\push{\phantom{\bbo}} & & &\\
            & & &
        \end{quantikz}}
        \overset{\text{definition of $\cT_H$}}{=}
        \scalebox{\scaleratio}{\begin{quantikz}[row sep=1mm,column sep=1mm]
            \lstick[2,brackets=none]{$+$\hspace{-.65cm}}\push{\bbo} & \permute{2,1} & & & & &\\ 
            \lstick[2,brackets=none]{$+$\hspace{-.65cm}}\push{\bbo} & & \gate[2]{\hT{c_1}} & \permute{2,1} & & \permute{2,1} &\\ 
            \lstick[2,brackets=none]{$+$\hspace{-.65cm}}\push{\phantom{\bbo}} & & & & \gate[2]{\hT{c_2}} & &\\
            & & & & & &
        \end{quantikz}} \\
        \overset{\text{inductive}}{\underset{\text{hypothesis}}{\fromto_2{}}}
        \scalebox{\scaleratio}{\begin{quantikz}[row sep=1mm,column sep=1mm]
            \lstick[2,brackets=none]{$+$\hspace{-.65cm}}\push{\bbo} & & \permute{2,1} & & & &\\ 
            \lstick[2,brackets=none]{$+$\hspace{-.65cm}}\push{\bbo} & \gate[2]{\hT{c_1}} & & \permute{2,1} & & \permute{2,1} &\\ 
            \lstick[2,brackets=none]{$+$\hspace{-.65cm}}\push{\phantom{\bbo}} & & & & \gate[2]{\hT{c_2}} & &\\
            & & & & & &
        \end{quantikz}}
        \overset{\text{completeness of $\Pi$}}{\fromto_2{}}
        \scalebox{\scaleratio}{\begin{quantikz}[row sep=1mm,column sep=1mm]
            \lstick[2,brackets=none]{$+$\hspace{-.65cm}}\push{\bbo} & & \permute{2,1} & \permute{2,3,1} & & & \permute{3,1,2} &\\ 
            \lstick[2,brackets=none]{$+$\hspace{-.65cm}}\push{\bbo} & \gate[2]{\hT{c_1}} & & &\push{\bbo} & & &\\ 
            \lstick[2,brackets=none]{$+$\hspace{-.65cm}}\push{\phantom{\bbo}} & & & &\push{\bbo} &\gate[2]{\hT{c_2}} & &\\
            & & & & & & &
        \end{quantikz}} \\
        \overset{\text{inductive hypothesis}}{\underset{\text{and completeness of $\Pi$}}{\fromto_2{}}}
        \scalebox{\scaleratio}{\begin{quantikz}[row sep=1mm,column sep=1mm]
            \lstick[2,brackets=none]{$+$\hspace{-.65cm}}\push{\bbo} & & \permute{2,1} & \permute{2,3,1} & & & & \permute{3,1,2} &\\ 
            \lstick[2,brackets=none]{$+$\hspace{-.65cm}}\push{\bbo} & \gate[2]{\hT{c_1}} & & &\permute{2,1} & & \permute{2,1} & &\\ 
            \lstick[2,brackets=none]{$+$\hspace{-.65cm}}\push{\phantom{\bbo}} & & & & &\gate[2]{\hT{c_2}} & & &\\
            & & & & & & & &
        \end{quantikz}} \\
        \overset{\text{completeness of $\Pi$}}{\fromto_2{}}
        \scalebox{\scaleratio}{\begin{quantikz}[row sep=1mm,column sep=1mm]
            \lstick[2,brackets=none]{$+$\hspace{-.65cm}}\push{\bbo} & & \permute{1,3,2} &  & \permute{1,3,2} & \permute{2,1} &\\ 
            \lstick[2,brackets=none]{$+$\hspace{-.65cm}}\push{\bbo} & \gate[2]{\hT{c_1}} & & & & & \\ 
            \lstick[2,brackets=none]{$+$\hspace{-.65cm}}\push{\phantom{\bbo}} & & & \gate[2]{\hT{c_2}} & & & \\
            & & & & & &
        \end{quantikz}}
        \overset{\text{definition of $\cT_H$}}{=}
        \scalebox{\scaleratio}{\begin{quantikz}[row sep=1mm,column sep=1mm]
            \lstick[2,brackets=none]{$+$\hspace{-.65cm}}\push{\bbo} & & \permute{2,1} &\\ 
            \lstick[2,brackets=none]{$+$\hspace{-.65cm}}\push{\bbo} & \gate[3]{\hT{c_1+c_2}} & & \\ 
            \lstick[2,brackets=none]{$+$\hspace{-.65cm}}\push{\phantom{\bbo}} & & & \\
            & & &
        \end{quantikz}}
    \end{gather*}
    \item If $c = c_1\times c_2$, the proof can be reduced to the earlier cases of $d_1\seqq d_2$ and $d_1+d_2$.\qedhere
    \end{itemize}
\end{proof}

With the help of \cref{lem:ht2}, we can address the bifunctoriality of $+$ and $\times$ for terms translated by $\cT_H$ as follows.

\begin{restatable}{lemma}{lemht}\label{lem:ht3}
    If $c_1:b_1\fromto b_3,c_2:b_2\fromto b_4$ are two terms of \QPiLang{}, then
    \begin{align}
        \hT{c_1+c_2}
        \fromto_2{} & \hT{c_1+\pid_{b_2}}\seqq \hT{\pid_{b_3}+c_2} \fromto_2{} \hT{\pid_{b_1}+c_2}\seqq \hT{c_1+\pid_{b_4}} \label{eq:ht_plus}\\
        \hT{c_1\times c_2}
        \fromto_2{} & \hT{c_1\times \pid_{b_2}}\seqq \hT{\pid_{b_3}\times c_2} \fromto_2{} \hT{\pid_{b_1}\times c_2}\seqq \hT{c_1\times \pid_{b_4}} \label{eq:ht_times}
    \end{align}
    are level-2 combinators of \HPiLang{}.
\end{restatable}
\begin{proof}
    The proof for \cref{eq:ht_plus} is straightforward by \cref{lem:ht2,eq:h3}.
    However, the proof for \cref{eq:ht_times} involves extracting proof from $\hT{c_1}\times \hT{c_2} \fromto_2 (\hT{c_1}\times \pid_{\bbo+b_2})\seqq (\pid_{\bbo+b_3}\times \hT{c_2}) \fromto_2 (\pid_{\bbo+b_1}\times \hT{c_2})\seqq (\hT{c_1}\times \pid_{\bbo+b_4})$ using \cref{lem:ht2,eq:h3}. We postpone the details to \ifARXIV%
    \cref{app:hpi_proof}%
    \else%
    \cref{ext-app:hpi_proof} of the extended version~\cite{FHK25}%
    \fi.
\end{proof}

\begin{restatable}{lemma}{lemqeq}\label{lem:q_eq}
    If $c_1 \fromto_2 c_2$ is a level-2 combinator of Q-$\Pi$, then $\hT{c_1} \fromto_2 \hT{c_2}$ is a level-2 combinator of \HPiLang{}. 
\end{restatable}
\begin{proof}
    With \cref{lem:ht2,lem:ht3}, we can use induction on the level-2 combinators deriving $c_1 \fromto_2 c_2$. See details in \ifARXIV%
    \cref{app:hpi_proof}%
    \else%
    \cref{ext-app:hpi_proof} of the extended version~\cite{FHK25}%
    \fi.
\end{proof}

Similar to \cref{lem:wq_inv}, the next lemma tells us that the composition of $\qsem{\cdot}$ and $\cT_H$ will yield equivalent terms, with additional terms $\pid_\bbo$.
\begin{lemma}\label{lem:qh_inv}
    If $c$ is a term of \HPiLang{}, then $\pid_{\bbo}+c \fromto_2 \hT{\qsem{c}}$ is a level-2 combinator of \HPiLang{}. 
\end{lemma}
\begin{proof}
    Since $\qsem{\cdot}$ is an embedding, and for any isomorphisms $\mathit{iso}$  of \HPiLang{}, which does not contain $(-1)$, we have $\hT{\qsem{\mathit{iso}}} = \hT{\mathit{iso}} = \pid_\bbo+\mathit{iso}$, the proof is an easy induction on the structure of $c$.
\end{proof}

Analogously to the reasoning of \cref{thm:qpi_relation}, \cref{lem:qh_inv,lem:q_eq} give the completeness of \HPiLang{} with respect to \QPiLang{}.

\begin{theorem}
    [Completeness of \HPiLang{} w.r.t Q-$\Pi$]\label{thm:hpi_qpi}
    If $c_1,c_2$ are two terms of \HPiLang{} satisfying $\qsem{c_1} \fromto_2 \qsem{c_2}$, then $c_1\fromto_2 c_2$ is a level-2 combinator of \HPiLang{}.
\end{theorem}
\begin{proof}
    Invoke \Cref{lem:q_eq,lem:qh_inv,eq:h3}:
    \begin{align*}
        \qsem{c_1} \fromto_2 \qsem{c_2}
        \Rightarrow{}& \hT{\qsem{c_1}} \fromto_2 \hT{\qsem{c_2}} \\
        \Rightarrow{}& \pid_{\bbo}+c_1 \fromto_2 \hT{\qsem{c_1}} \fromto_2 \hT{\qsem{c_2}} \fromto_2 \pid_{\bbo}+c_2 \\
        \Rightarrow{}& c_1\fromto_2 c_2. \qedhere
    \end{align*}
\end{proof}

Finally, we prove \cref{thm:hpi} as follows.
\begin{proof}[Proof of \cref{thm:hpi}]
    For soundness, observe that by soundness of $\Pi$ it suffices to verify that \cref{eq:h1,eq:h2,eq:h3} are sound. This can be verified by direct computation.

    For completeness, combine \cref{lem:q_faith} with the completeness of $\QPiLang{}$ from \cref{thm:qpi} and the completeness of \HPiLang{} w.r.t. \QPiLang{} from \cref{thm:hpi_qpi}:
    \begin{align*}
            \sem{c_1} = \sem{c_2}
            \Rightarrow{}& \sem{\qsem{c_1}} = \sem{c_1} = \sem{c_2} = \sem{\qsem{c_2}} \\
            \Rightarrow{}& \qsem{c_1} \fromto_2{} \qsem{c_2} \\
            \Rightarrow{}& c_1 \fromto_2 c_2.
    \end{align*}
    This completes the proof.
\end{proof}

\section{Related work and Concluding Remarks}\label{sec:conclusion}
Our work delivers a notable improvement over the classical reversible programming language
$\Pi$~\cite{choudhury2022symmetries}: we add only one quantum primitive and three equations to $\Pi$, yet still obtain a completeness result. Several other extensions of $\Pi$ to quantum computing exist, such as $\sqrt{\Pi}$~\cite{carette2024few} and Quantum $\Pi$~\cite{carette2024bake}.
In this section, we nuance our position relative to the existing work and discuss the implications of our result and future work.

The finite presentation of $\HPiUnitary$ in \cref{sec:relations} is the cornerstone for our completeness result.
It can be regarded as a refinement of~\cite{bian2021generators,greylyn2014generators,li2021generators} as follows.
First, \citet{bian2021generators} do not consider the unitary groups over a ring that contains $1/\sqrt{2}$.
Second, \citet{li2021generators} handle $1/\sqrt{2}$ in an ad hoc way, since the orthogonal matrices of the form $M/\sqrt{2}^k$ with integer matrix $M$ and non-negative integer $k$ only form a subset of $\HPiUnitary$.
Finally, \citet{greylyn2014generators} only considers the group $U_4(\ZZ[\tfrac{1}{\sqrt{2}},i])$ of a fixed dimension $4$.
However, we do not claim that our contribution in \cref{sec:relations} represents a major breakthrough, as it adopts the general proof framework of prior work.
The difficulty is that the introduction of $\ZZ[\tfrac{1}{\sqrt{2}}]$ makes deriving \cref{rel:d3,rel:d4} particularly challenging.

Our results enable a complete mechanisation of reversible quantum program rewriting: two programs in \HPiLang implement the same quantum computation if and only if one can be rewritten into the other by a finite string of equations chosen from a finite set of equations. This is significant because physically implementing a discrete (finite) universal quantum gate set with high accuracy is generally simpler than implementing a continuous one.
In contrast, a major limitation for ZX- and ZH-calculi~\cite{coeckeduncan:zx,jeandealperdrixvilmart:zx,backenskissinger:zh} is that they do not operate within the universe of unitaries.
For example, even if a ZX-diagram represents a unitary, it is generally hard to synthesise a proper quantum circuit from it~\cite{beaudrapkissingerwetering:circuithard}.
The LOv-calculus~\cite{clementetal:lov} and subsequent work on equational theories of quantum circuits~\cite{clement2024minimal,clement2023complete,clément2023quantumcircuitcompletenessextensions} have gates that are parametrised by continuous variables and equations involving existential quantifiers. They offer a different advantage, as it operates directly on quantum circuits, but comes with continuous gates, unlike our discrete formulation.

Our work opens up several future directions of research. An implementation of \HPiLang and \QPiLang would enable term rewriting systems in transforming and optimising quantum programs, and a logic to support reasoning about quantum programs.
Based on our completeness result, it seems promising to develop a complete equational theory for quantum circuits with a discrete universal gate set.
Finally, our proofs of completeness rely on a cancellativity property that is currently baked into the languages; it is unclear whether this is necessary.

\begin{acks}
    We owe special thanks to an anonymous reviewer for identifying a gap in the proof of \cref{thm:relation_completeness} in an earlier version of this manuscript, and to Alex Rice for pointing out the need for a cancellative condition.
    This research was funded by the Engineering and Physical Sciences Research Council (EPSRC) under project EP/X025551/1 ``Rubber DUQ: Flexible Dynamic Universal Quantum programming''.
    Robin Kaarsgaard was supported by Sapere Aude: DFF-Research Leader grant 5251-00024B.
\end{acks}

\bibliographystyle{ACM-Reference-Format}
\bibliography{ref}

\ifARXIV
\appendix

\section{Completeness of Relations \texorpdfstring{(\NoCaseChange{\cref{thm:relation_completeness}})}{}}\label{app:relations_proof}
\begingroup
\tikzset{
    commutative diagrams/every diagram/.style={
        row sep=1cm,
        column sep=1cm,
        nodes={inner sep=4pt},
    }
}
\allowdisplaybreaks

We begin by adopting some useful definitions from~\cite{li2021generators} to simplify our presentation.

\begin{definition}
    [Syllable]
    The words produced by \cref{alg:exact_synthesis} are referred to as \emph{syllables}, which are of the form
    $\CZ{a}$, $\CX{a,j}\CZ{a}^{\tau}$, $\CH{1,b}$ and $\CH{1,b}\CX{1,c}$ for appropriately chosen $a,b,c$ and $\tau$.
\end{definition}

\begin{definition}[State graph]
    The \emph{state graph} for \HPiUnitary{} is a directed graph defined as follows.
    \begin{itemize}
        \item The vertices are the elements of \HPiUnitary{} and are referred to as \emph{states}.
        \item There are two types of edges: \begin{itemize}
            \item \emph{simple edges}, which are triples $\edge{s}{G}{s'}$ and denoted by $s \Se{G} s'$ or $G:s\Se{} s'$ where $s,s' \in \HPiUnitary$, $G\in\cG_n$ and $s' = Gs$.
            \item \emph{normal edges}, which are triples $\edge{s}{N}{s'}$ and denoted by $s \Ne{N} s'$ or $N:s\Ne{} s'$ where $s,s' \in \HPiUnitary$, $N$ is the unique first syllable output by \cref{alg:exact_synthesis} on input $s$ and $s' = Ns$.
        \end{itemize}
    \end{itemize} 
\end{definition}

For convenience, we write $\mathbf{G} : s \Se*{} s'$ for $\mathbf{G} = G_m\cdots G_2G_1$ if there is a sequence of simple edges
\begin{equation*}
    s \Se{G_{1}} s_1 \Se{G_2} s_2 \cdots s_{m-1} \Se{G_m} s'; 
\end{equation*}
and we write $\mathbf{N} :s \Ne*{} s'$ for $\mathbf{N} = N_m\cdots N_2N_1$ if there is a sequence of normal edges
\begin{equation*}
    s \Ne{N_{1}} s_1 \Ne{N_2} s_2 \cdots s_{m-1} \Ne{N_m} s'.
\end{equation*}
Since the syllables output by \cref{alg:exact_synthesis} are generated by $\cG_n$, normal edges can be treated as sequences of simple edges.
Thus, a sequence of normal edges can also be treated as a sequence of simple edges.

\begin{definition}
    Let $\mathbf{G}:s \Se*{} s'$ be a sequence of simple edges with $\mathbf{G} = G_m\cdots G_2G_1$.
    The \emph{level} of $\mathbf{G}: s \Se*{} s'$, denoted $\level(\mathbf{G}:s \Se*{} s')$, is the maximum level of the intermediate states in this sequence, i.e., 
    \begin{equation}
        \level(\mathbf{G}:s \Se*{} s') = \max\cbra{\,\level(s),\level(G_1s),\level(G_2G_1s),\ldots,\level(\mathbf{G}s)\,}.
    \end{equation}
\end{definition}

\begin{definition}
    Let $\mathbf{G}:s \Se*{} t, \mathbf{G}':s \Se*{} t$ be two sequences of edges. We say that the diagram
    \begin{equation}
        \begin{tikzcd}
            \setwiretype{n} s \ar[r,to*,bend right=-40,"\mathbf{G}"] \ar[r,to*, bend right=40,"\mathbf{G}'",swap]& t 
        \end{tikzcd}
    \end{equation}
    commutes equationally if $\mathbf{G} \approx \mathbf{G}'$.
\end{definition}

We follow the approach of~\cite{greylyn2014generators} to define the basic edges.

\begin{definition}
    The subset $\cG_n' \subseteq \cG_n$ of \emph{basic} generators is defined as
    \begin{equation}
        \cG_n' = \{\,\CX{1,x}, \CZ{1}, \CH{1,2} \mid 2\leq x \leq n \,\}.
    \end{equation}
    An edge $G:s \Se{} t$ is called a \emph{basic edge} if $G$ is basic.
\end{definition}

\begin{lemma}\label{lem:simple_basic}
    Let $s$ and $r$ be states, $G:s \Se{} r$ be a simple edge. Then there exists a sequence of basic edges $\mathbf{G}'$ such that 
    \[\mathbf{G}' \approx G \text{ and } \level(\mathbf{G}':s \Se*{} r) = \level(G:s \Se{} r).\] 
\end{lemma}
\begin{proof}
    Let $(j,k,l) = \level(s)$, we prove it for all cases of $G$ and $(j,k,l)$.
    \begin{enumerate}[leftmargin=25pt]
        \item $G = \CZ{a}$. Let $\mathbf{G}' = \CX{1,a}\CZ{1}\CX{1,a}$, then the following diagram
        \[ \begin{tikzcd}
                \setwiretype{n} s \ar[r,"\CZ{a}"]\ar[d,"\CX{1,a}",swap] & r \ar[d,<-, "\CX{1,a}"] \\
                \setwiretype{n} t_1 \ar[r, swap, "\CZ{1}"]& t_2
            \end{tikzcd} 
        \]
        commutes equationally by \cref{rel:c4}. \begin{enumerate}[leftmargin=15pt]
            \item $j < a$. We have $s e_a = e_a$, then $re_a = \CZ{a}se_a = -e_a, t_1e_a = \CX{1,a}se_a = e_1, t_2e_a = \CZ{1}\CX{1,a}e_a = -e_1$. Thus, $\level(r) = \level(t_1) = \level(t_2) = (a, 0, 0) > \level(s)$. Then, $\level(\mathbf{G}':s \Se*{} r) = \level(G:s \Se{} r)$.
            \item $j = a$. \begin{enumerate}[leftmargin=15pt]
                \item If $k = 0$, which also means $l=0$,we have $se_a = -e_a$, then, $t_1e_a = \CX{1,a}se_a = -e_1$, $re_a = \CZ{a}se_a = e_a$, $t_2e_a = \CZ{1}\CX{1,a}se_a = e_1$. Thus, $\level(t_1) = \level(t_2) = (a,0,0) = \level(s)$, $\level(r) < (a,0,0)$, then $\level(\mathbf{G}':s \Se*{} r) = \level(G:s \Se{} r)$.
                \item If $k\neq 0$, which means $l \neq 0$, we will have $\level(t_1)=\level(t_2)=\level(r)=\level(s)$. Thus, $\level(\mathbf{G}':s \Se*{} r) = \level(G:s \Se{} r)$.
            \end{enumerate}
            \item $j > a$. We have $\CZ{a}, \CX{1,a}$ can not increase the $j$ of level $(j,k',l')$, then $\level(t_1) = \level(s) = \level(r)$. Similarly, $\level(t_1) = \level(s)$. Thus, $\level(\mathbf{G}':s \Se*{} r) = \level(G:s \Se{} r)$.
        \end{enumerate}
        \item $G = \CX{b,c}$. Let $\mathbf{G}' = \CX{1,b}\CX{1,c}\CX{1,b}$, then the following diagram 
        \[ \begin{tikzcd}
                \setwiretype{n} s \ar[r,"\CX{b,c}"]\ar[d,"\CX{1,b}",swap] & r \ar[d,<-, "\CX{1,b}"] \\
                \setwiretype{n} t_1 \ar[r, swap, "\CX{1,c}"]& t_2
            \end{tikzcd} 
        \]
        commutes equationally by \cref{rel:c2}.
        \begin{enumerate}[leftmargin=15pt]
            \item $j< c$. We have that $se_c = e_c$. Then $t_1e_c = \CX{1,b}se_c = e_c$, $re_c = \CX{b,c}se_c = e_b$, $t_2e_c = \CX{1,c}\CX{1,b}se_c = e_1$. Thus, $\level(t_1) < (c, 0,0), \level(t_2) = (c,0,0), \level(r) = (c,0,0)$, then  $\level(\mathbf{G}':s \Se*{} r) = (c,0,0) = \level(r) = \level(G:s \Se{} r)$.
            \item $j = c$. \begin{enumerate}[leftmargin=15pt]
                \item If $k = 0$, which means $l=0$, we have $se_c = -e_c$, then $t_1e_c = \CX{1,b}se_c = -e_c$, $re_c = \CX{b,c}se_c = -e_b$, $t_2e_c = \CX{1,c}\CX{1,b}se_c = -e_1$. Thus $\level(r) = \level(t_2) = \level(t_1) = (c,0,0)$, then $\level(\mathbf{G}':s \Se*{} r) = (c,0,0) = \level(r) = \level(G:s \Se{} r)$.
                \item If $k \neq 0$, then similar to the case 1)~b)~ii), we have $\level(\mathbf{G}':s \Se*{} r) = \level(G:s \Se{} r)$.
            \end{enumerate}
            \item $j > c$. Similar to the case 1)~c), we have $\level(\mathbf{G}':s \Se*{} r) = \level(G:s \Se{} r)$.
        \end{enumerate}
        \item $G = \CH{b,c}$. Consider the following diagram
        
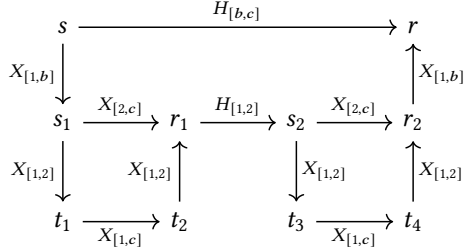
\[
            \begin{tikzcd}
                \setwiretype{n} s \ar[rrr,"\CH{b,c}"] \ar[d,"\CX{1,b}",swap] & & & r \\[-2mm]
                \setwiretype{n} s_1 \ar[r,"\CX{2,c}"] \ar[d,"\CX{1,2}",swap] & r_1\ar[r,"\CH{1,2}"] & s_2 \ar[r,"\CX{2,c}"]\ar[d,"\CX{1,2}"] & r_2 \ar[u,"\CX{1,b}",swap]  \\[-2mm]
                \setwiretype{n} t_1 \ar[r,"\CX{1,c}",swap] & t_2 \ar[u,"\CX{1,2}"] & t_3 \ar[r,"\CX{1,c}",swap] & t_4 \ar[u,"\CX{1,2}",swap]
            \end{tikzcd}
        \]
        As we have shown in case 2), the two rectangles at the bottom of the diagram commute equationally and $\level(s_1:\Se*{\CX{1,2}\CX{1,c}\CX{1,2}} r_1) = \level(s_1:\Se{\CX{2,c}}r_1),\level(s_2:\Se*{\CX{1,2}\CX{1,c}\CX{1,2}}r_2) = \level(s_2:\Se{\CX{2,c}}r_2)$.
        The rectangle at the top of the diagram commutes equationally by \cref{rel:c4,rel:c5}. Thus, we choose $\mathbf{G}' = \CX{1,b}\CX{1,2}\CX{1,c}\CX{1,2}\CH{1,2}\CX{1,2}\CX{1,c}\CX{1,2}\CX{1,b}$. With the diagram commutes equationally, we have $\mathbf{G}' \approx G$.
        \begin{enumerate}[leftmargin=15pt]
            \item $j < c$. Since $j \leq b < c$, we have that $s e_c = e_c$. Then, by direct computation , we have $s_1 e_c = e_c, r_1e_c = e_2, s_2e_c = \frac{1}{\sqrt{2}}(e_1-e_2), r_2e_c = \frac{1}{\sqrt{2}}(e_1-e_c), re_c = \frac{1}{\sqrt{2}}(e_b-e_c)$. Thus $\level(s_1) < (c,0,0), \level(r_1) < (c,0,0), \level(s_2) = \level(r_2) = \level(r) = (c,1,2)$, then $\level(\mathbf{G}':s \Se*{} r) = (c,1,2) = \level(r) = \level(G:s \Se{} r)$.
            \item $j = c$. \begin{enumerate}[leftmargin=15pt]
                \item If $k = 0$, which means $l=0$, we have $se_c = -e_c$, then $s_1 e_c = -e_c, r_1e_c = -e_2, s_2e_c = \frac{1}{\sqrt{2}}(e_2-e_1), r_2e_c = \frac{1}{\sqrt{2}}(e_c-e_1), re_c = \frac{1}{\sqrt{2}}(e_c-e_b)$. Thus, $\level(s_1) = \level(r_1) = (c,0,0), \level(s_2) = \level(r_2) = \level(r) = (c, 1,2)$, then $\level(\mathbf{G}':s \Se*{} r) = (c,1,2) = \level(r) = \level(G:s \Se{} r)$.
                \item If $k \neq 0$, which means $l \neq 0$. We have that $\CX{1,b}, \CX{2,c}$ can not increase the $j$ of level $(j,k',l')$, thus $\level(r_1) = \level(s_1) = \level(s) > (c, 0, 0)$. For $s_2,r_2,r$, we consider the following two cases: \begin{itemize}
                    \item If $\level(r) \leq (c,0,0)$, then it means $re_c = \pm e_k$ for $k \leq c$. Thus, $\level(s_2) \leq (c,0,0), \level(r_2) \leq (c,0,0)$. We have $\level(\mathbf{G}':s \Se*{} r) = \level(s) = \level(G:s \Se{} r)$.
                    \item If $\level(r) > (c,0,0)$, then, since $\CX{1,b}, \CX{2,c}$ can not increase the $j$ of level $(j,k',l')$, we have $\level(s_2) = \level(r_2) = \level(r)$. Thus, $\level(\mathbf{G}':s \Se*{} r) = \max\{\level(s), \level(r)\} = \level(G:s \Se{} r)$.
                \end{itemize}
            \end{enumerate}
            \item $j > c$. Similar to the case 1)~c), we have $\level(r_1) = \level(s_1) = \level(s), \level(s_2) = \level(r_2 = \level(r))$. Thus, $\level(\mathbf{G}':s \Se*{} r) = \max\{\level(s), \level(r)\} = \level(G:s \Se{} r)$. \qedhere
        \end{enumerate}
    \end{enumerate}
\end{proof}

Next, we can establish the Main Lemma for basic edges.

\begin{lemma}[Main Lemma]\label{lem:basic_main}
    Let $s$, $t$, and $r$ be states, $N: s \Ne{} t$ be a normal edge, and $G: s \Se{} r$ be a basic edge. Then there exist a state $q$, a sequence of normal edges $\mathbf{N}': r \Ne*{} q$, and a sequence of simple edges $\mathbf{G}':t \Se*{} q$ such that the diagram 
    \[ \begin{tikzcd}
        \setwiretype{n} s \ar[r,"G"]\ar[d,Rightarrow,"N",swap] & r \ar[d, Rightarrow*, "\mathbf{N}'"] \\
        \setwiretype{n} t \ar[r, swap, to*, "\mathbf{G}'"]& q
    \end{tikzcd} 
    \]
    commutes equationally and $\level(\mathbf{G}':t \Se*{} q) < \level(s)$.
\end{lemma}
%\begin{restatable}[Main Lemma]{lemma}{lemmain}
%    \label{lem:main}
%    Let $s$, $t$, and $r$ be states, $N: s \Ne{} t$ be a normal edge, and $G: s \Se{} r$ be a simple edge. Then there exist a state $q$, a sequence of normal edges $\mathbf{N}': r \Ne*{} q$, and a sequence of simple edges $\mathbf{G}':t \Se*{} q$ such that the diagram 
%    \[ \begin{tikzcd}[ampersand replacement=\&]
%        \setwiretype{n} s \ar[r,"G"]\ar[d,Rightarrow,"N",swap] \& r \ar[d, Rightarrow*, "\mathbf{N}'"] \\
%        \setwiretype{n} t \ar[r, swap, to*, "\mathbf{G}'"]\& q
%    \end{tikzcd} 
%    \]
%    commutes equationally and $\level(\mathbf{G}':t \Se*{} q) < \level(s)$.
%\end{restatable}
\begin{proof}
    To maintain a smooth reading flow, the proof is deferred to \cref{proof:main_lemma}.
\end{proof}

\begin{lemma}[{\cite[Lemma~32]{greylyn2014generators}}]\label{lem:normal_basic_normal}
    Let $s, r$ be two states, $G:s\Se{} r$ be a basic edge, $\mathbf{N}:s\Ne*{}I$ be the unique sequence of normal edges from $s$ to $I$ (the identity matrix) and  $\mathbf{M}:r\Ne*{}I$ be the unique sequence of normal edges from $r$ to $I$. Then $\mathbf{M}G \approx \mathbf{N}$.
\end{lemma}
\begin{proof}
    By induction on the level of $s$~\cite{greylyn2014generators}. \begin{itemize}
        \item $\level(s) = (0,0,0)$. In this case, $s = I$, then we have the following diagram 
        \[ \begin{tikzcd}
                \setwiretype{n} I \ar[rr,"G"]\ar[dr,Rightarrow*,"\mathbf{N} = \varepsilon",swap] &[-2mm]&[-2mm] r \ar[dl, Rightarrow, "\mathbf{M}=G"] \\
                \setwiretype{n} & I &
            \end{tikzcd} 
            \]
        commutes equationally by \cref{rel:a1,rel:a2,rel:a3} for all the cases of $G$.
        \item $\level(s) > (0,0,0)$. In this case, $s \neq I$. Thus, $\mathbf{N}$ must have one term $N$ such that $\mathbf{N} = \mathbf{N}'N$. By \cref{lem:basic_main}, there exist $\mathbf{G}' = G_k\cdots G_1$ and $\mathbf{M}'$ such that $\mathbf{M}'G \approx \mathbf{G}'N$ and \[\level(\mathbf{G}': t \Se*{} q) < \level(s),\] where $t = Ns$.
        Thus, we consider the following diagram
        \[ \begin{tikzcd}
                \setwiretype{n} s \ar[rrrr,"G"]\ar[d,Rightarrow,"N"{name=N},swap] &&&& r \ar[d, Rightarrow*, "\mathbf{M}"{name=M}] \\
                \setwiretype{n} t_0 = t \ar[drr,Rightarrow,"\mathbf{N}'_0 = \mathbf{N}'",swap,bend right,pos=.3]\ar[r,"G_1"{name=g1}]\ar[r,phantom,"\text{IH}"{name=g1},yshift=-7mm] & t_1\ar[r,"G_2"{name=g2}]\ar[r,phantom,"\text{IH}"{name=g1},yshift=-7mm]\ar[dr,Rightarrow,"\mathbf{N}_1'",swap,bend right,pos=.3] & t_2\ar[r,dashed,"\cdots"{name=gdots}]\ar[r,phantom,"\text{IH}"{name=g1},yshift=-7mm]\ar[d,Rightarrow,"\mathbf{N}_2'"{name=d1},swap,] & t_{k-1}\ar[r,"G_k"{name=gk}]\ar[r,phantom,"\text{IH}"{name=g1},yshift=-7mm]\ar[dl,Rightarrow,"\mathbf{N}_{k-1}'"{name=d2},bend left,pos=.3] & t_k = q \ar[dll,Rightarrow,"\mathbf{N}_k'",bend left,pos=.3] \\[5mm]
                \setwiretype{n} && I
                \arrow[phantom,from=d1,to=d2,"\cdots"]
                \arrow[phantom,from=N,to=M,"\text{\cref{lem:basic_main}}"]
            \end{tikzcd} 
            \]
            Since $\level(\mathbf{G}': t \Se*{}q) < \level(s)$, then $t_{j} <\level(s)$ for each $0\leq j \leq k-1$. Thus, by inductive hypothesis, we have $\mathbf{N}_{j+1}'G_j \approx \mathbf{N}_{j}'$ for each $0\leq j \leq k-1$, then $\mathbf{N}_k'\mathbf{G}' = \mathbf{N}_k'G_k\cdots G_1\approx \mathbf{N}_{k-1}'G_{k-1}\cdots G_1 \approx \cdots \approx \mathbf{N}_1'G_1 \approx \mathbf{N}'_0 = \mathbf{N}'$. Together with $\mathbf{M}'G \approx \mathbf{G}'N$ and the uniqueness of normal edges, we have $\mathbf{N} = \mathbf{N}'N \approx \mathbf{N}_k'\mathbf{G}' \approx \mathbf{N}_k'\mathbf{M}'G = \mathbf{M}G$. \qedhere
    \end{itemize}
\end{proof}

\begin{lemma}
    [{\cite[Lemma~33]{greylyn2014generators}}]\label{lem:basic_normal}
    Let $\mathbf{G}:s \Se*{} I$ be any sequence of basic edges with final state $I$, and let $\mathbf{N}:s \Ne*{} I$ be the unique sequence of normal edges from $s$ to $I$. Then $\mathbf{G} \approx \mathbf{N}$.
\end{lemma}
\begin{proof}
    Repeated application of \cref{lem:normal_basic_normal} gives the proof~\cite{greylyn2014generators}.
\end{proof}

\relationcompleteness*
\begin{proof}
    With \cref{lem:basic_normal,lem:simple_basic}, it is straightforward to follow the proofs in~\cite{greylyn2014generators}:
    By \cref{lem:simple_basic}, there exists two sequences $\mathbf{G}', \mathbf{H}'$ of basic edges such that $\mathbf{G}'\approx \mathbf{G}, \mathbf{H}' \approx \mathbf{H}$. By the soundness of relations (\cref{thm:relation_soundness}), we have $\sem{\mathbf{G}'} =  \sem{\mathbf{G}}, \sem{\mathbf{H}'} =  \sem{\mathbf{H}}$, then $\sem{\mathbf{G}'} = \sem{\mathbf{H}'}$.
    Let $s = \sem{\mathbf{G}'}^{-1} = \sem{\mathbf{H}'}^{-1}$, then $\mathbf{G}':s \Se*{} I$ and $\mathbf{H}':s \Se*{} I$.
    Let $\mathbf{N}:s\Ne*{} I$ be the unique sequence of normal edges from $s$ to $I$, then by \cref{lem:basic_normal}, we have $\mathbf{G}' \approx \mathbf{N} \approx \mathbf{H}'$.
    We finally get $\mathbf{G} \approx \mathbf{H}$ as $\approx$ is an equivalence relation. 
\end{proof}

\subsection{Derived Relations and Useful Lemmas}\label{proof:derived_relations}
For proof convenience, we add generators $\CX{c,b}$, $1 \leq b<c \leq n$, with relation
\begin{equation}
    \CX{c,b} \approx \CX{b,c} \tag{e1} \label[relation]{rel:e1}
\end{equation}
and generators $\CH{c,b}$, $1 \leq b<c \leq n$, with relation
\begin{equation}
    \CH{c,b} \approx \CX{b,c}\CH{b,c}\CX{b,c}\text. \tag{e2} \label[relation]{rel:e2}
\end{equation}

\begin{proposition} 
    The counterpart of \cref{rel:c4,rel:c5} with reversed subscripts of $H$ generators
    \begin{equation*}
        \CH{c,b}\CX{a,b} \approx \CX{a,b}\CH{c,a}, \quad \CH{c,a}\CX{b,c} \approx \CX{b,c} \CH{b,a}
    \end{equation*}
    is derivable from \cref{rel:e2}.
\end{proposition}
\begin{proof}
    \begin{align*}
        \CH{c,b}\CX{a,b} &\approx \CX{b,c}\CH{b,c}\CX{b,c}\CX{a,b} \tag{by \cref{rel:e2}} \\
        &\approx \CX{b,c}\CH{b,c}\blue{\CX{a,b}\CX{a,c}} \tag{by \cref{rel:c2}} \\
        &\approx \CX{b,c}\blue{\CX{a,b}\CH{a,c}}\CX{a,c} \tag{by \cref{rel:c4}} \\
        &\approx \blue{\CX{a,b}\CX{a,c}}\CH{a,c}\CX{a,c} \tag{by \cref{rel:c2}} \\
        &\approx \CX{a,b}\blue{\CH{c,a}}. \tag{by \cref{rel:e2}}\\ \\
        \CH{c,a}\CX{b,c} &\approx \CX{a,c}\CH{a,c}\CX{a,c}\CX{b,c} \tag{by \cref{rel:e2}} \\
        &\approx \blue{\CX{b,c}\CX{a,b}\CH{a,b}\CX{a,b}} \tag{by \cref{rel:c3,rel:c5}} \\
        &\approx \CX{b,c}\blue{\CH{b,a}}. \tag{by \cref{rel:e2}}
    \end{align*}
\end{proof}
Later proofs may thus apply \cref{rel:c4,rel:c5} directly for any order of the subscripts.

\begin{proposition}
    \begin{equation}
        (\CH{a,b}\CH{c,d}\CH{a,c}\CH{b,d})^2 \approx \varepsilon \tag{f1} \label[relation]{rel:f1}
    \end{equation}
\end{proposition}
\begin{proof}
    Since
    \begin{align*}
        &\CH{a,b}\CH{c,d}\CH{a,c}\CH{b,d} \\
        \approx{}& \blue{(\CH{c,d}\CH{a,c}\CH{b,d})^4}\CH{a,c}\CH{b,d} \tag{by \cref{rel:d3}} \\
        ={}& \blue{(\CH{c,d}\CH{a,c}\CH{b,d})^3\CH{c,d}\CH{a,c}\CH{b,d}}\CH{a,c}\CH{b,d} \\
        \approx{}& (\CH{c,d}\CH{a,c}\CH{b,d})^3\CH{c,d}\blue{\varepsilon} \tag{by \cref{rel:a3,rel:b6}} \\
        \approx{}& \blue{\CH{b,d}\CH{a,c}\CH{c,d}\CH{c,d}\CH{a,c}\CH{b,d}}(\CH{c,d}\CH{a,c}\CH{b,d})^3\CH{c,d} \tag{by \cref{rel:a3}} \\
        ={}& \CH{b,d}\CH{a,c}\CH{c,d}\blue{(\CH{c,d}\CH{a,c}\CH{b,d})^4}\CH{c,d} \\
        \approx{}& \CH{b,d}\CH{a,c}\CH{c,d}\blue{\CH{a,b}\CH{c,d}}\CH{c,d} \tag{by \cref{rel:d3}}\\
        \approx{}& \CH{b,d}\CH{a,c}\CH{c,d}\CH{a,b}\blue{\varepsilon}, \tag{by \cref{rel:a3}}
    \end{align*}
    we have
    \begin{align*}
        & (\CH{a,b}\CH{c,d}\CH{a,c}\CH{b,d})^2 \\
        ={}& \CH{a,b}\CH{c,d}\CH{a,c}\CH{b,d}\CH{a,b}\CH{c,d}\CH{a,c}\CH{b,d} \\
        \approx{}& \blue{\CH{b,d}\CH{a,c}\CH{c,d}\CH{a,b}}\CH{a,b}\CH{c,d}\CH{a,c}\CH{b,d} \\
        \approx{}& \blue{\varepsilon}. \tag{by \cref{rel:a3}}
    \end{align*}
\end{proof}

\begin{proposition}
    \begin{equation}
        (\CH{a,c}\CH{b,d}\CH{a,d}\CH{b,c})^2 \approx \CX{a,b}\CX{c,d} \tag{f2} \label[relation]{rel:f2}
    \end{equation}
\end{proposition}
\begin{proof}
    Since
    \begin{align*}
        &\CH{a,c}\CH{b,d}\CH{a,d}\CH{b,c} \\
        \approx{}& \CH{a,c}\CH{b,d}\CH{a,d}\blue{\CX{c,d}\CX{c,d}}\CH{b,c} \tag{by \cref{rel:a2}} \\
        \approx{}& \CH{a,c}\CH{b,d}\blue{\CX{c,d}\CH{a,c}\CH{b,d}\CX{c,d}} \tag{by \cref{rel:c5}} \\
        \approx{}& \CH{a,c}\CH{b,d}\blue{\CH{a,b}\CX{c,d}\CH{a,b}}\CH{a,c}\CH{b,d}\CX{c,d} \tag{by \cref{rel:a3,rel:b5}} \\
        \approx{}& \CH{a,c}\CH{b,d}\CH{a,b}\blue{\CH{c,d}\CZ{d}\CH{c,d}}\CH{a,b}\CH{a,c}\CH{b,d}\CX{c,d} \tag{by \cref{rel:d2}} \\
        \approx{}& \blue{\CH{a,b}\CH{c,d}\CH{a,c}\CH{b,d}}\CZ{d}\blue{\CH{a,c}\CH{b,d}\CH{c,d}\CH{a,b}}\CX{c,d} \tag{by \cref{rel:a3,,rel:b6,rel:f1}} \\
        \approx{}& \CH{a,b}\CH{c,d}\blue{\CH{b,d}\CZ{d}}\CH{b,d}\CH{c,d}\CH{a,b}\CX{c,d} \tag{by \cref{rel:a3,rel:b4}} \\
        \approx{}& \CH{a,b}\CH{c,d}\blue{\CX{b,d}}\CH{c,d}\CH{a,b}\CX{c,d} \tag{by \cref{rel:d2}} \\
        \approx{}& \CH{a,b}\CH{c,d}\CX{b,d}\blue{\CZ{d}\CH{c,d}\CH{a,b}}, \tag{by \cref{rel:b5,rel:d2}} \\
    \end{align*}
    we have
    \begin{align*}
        & (\CH{a,c}\CH{b,d}\CH{a,d}\CH{b,c})^2 \\
        \approx{}& (\blue{\CH{a,b}\CH{c,d}\CX{b,d}\CZ{d}\CH{c,d}\CH{a,b}})^2 \\
        ={}& \blue{\CH{a,b}\CH{c,d}\CX{b,d}\CZ{d}\CH{c,d}\CH{a,b}\CH{a,b}\CH{c,d}\CX{b,d}\CZ{d}\CH{c,d}\CH{a,b}} \\
        \approx{}& \CH{a,b}\CH{c,d}\CX{b,d}\CZ{d}\blue{\varepsilon}\CX{b,d}\CZ{d}\CH{c,d}\CH{a,b} \tag{by \cref{rel:a3}} \\
        \approx{}& \CH{a,b}\CH{c,d}\blue{\CZ{b}}\CZ{d}\CH{c,d}\CH{a,b} \tag{by \cref{rel:c1}} \\
        \approx{}& \blue{\CH{a,b}\CZ{b}\CH{a,b}}\blue{\CH{c,d}\CZ{d}\CH{c,d}} \tag{by \cref{rel:b4,rel:b6}} \\
        \approx{}& \blue{\CX{a,b}}\blue{\CX{c,d}}. \tag{by \cref{rel:d2}}
    \end{align*}
\end{proof}

\iffalse
\begin{proposition}
    \begin{equation}
        (\CH{a,b}\CH{a,c}\CH{b,d})^2 \approx (\CH{c,d}\CH{a,c}\CH{b,d})^2 \tag{f3} \label[relation]{rel:f3}
    \end{equation}
\end{proposition}
\begin{proof}
    \begin{align*}
        & (\CH{a,b}\CH{a,c}\CH{b,d})^2 \\
        ={}& \CH{a,b}\CH{a,c}\CH{b,d}\CH{a,b}\CH{a,c}\CH{b,d} \\
        \approx{}& \blue{\CH{c,d}\CH{c,d}}\CH{a,b}\CH{a,c}\CH{b,d}\CH{a,b}\CH{a,c}\CH{b,d} \tag{by \cref{rel:a3}} \\
        \approx{}& \CH{c,d}\blue{\CH{a,b}\CH{c,d}}\CH{a,c}\CH{b,d}\CH{a,b}\CH{a,c}\CH{b,d} \tag{by \cref{rel:b6}} \\
        \approx{}& \CH{c,d}\blue{\CH{a,c}\CH{b,d}\CH{a,b}\CH{c,d}}\CH{a,b}\CH{a,c}\CH{b,d} \tag{by \cref{rel:a3,,rel:b6,rel:f1}} \\
        \approx{}& \CH{c,d}\CH{a,c}\CH{b,d}\blue{\varepsilon\CH{c,d}}\CH{a,c}\CH{b,d} \tag{by \cref{rel:a3,rel:b6}} \\
        \approx{}& (\CH{c,d}\CH{a,c}\CH{b,d})^2.
    \end{align*}
\end{proof}
\fi

\begin{lemma}
    \label{lem:useful0}
    For any $u, v \in \HPiRing{}$ such that $\lde(u) = \lde(v) = k$ and $\sqrt{2}^ku \not\equiv \sqrt{2}^kv \pmod{2}$, let $(u',v')^T = H(u,v)^{T}$, then $\lde(u') = \lde(v') = k$ and $\sqrt{2}^ku' \not\equiv \sqrt{2}^kv' \pmod{2}$
\end{lemma}
\begin{proof}
    Without loss of generality, we assume $\sqrt{2}^ku \equiv 1 \pmod{2}$ and $\sqrt{2}^ku \equiv 1+\sqrt{2} \pmod{2}$.
    Then, there exist $a_1,b_1,a_2,b_2\in \ZZ$ such that
    \[\sqrt{2}^ku = 2a_1+1+2b_1\sqrt{2}, \enspace \sqrt{2}^kv = 2a_2+1+(2b_2+1)\sqrt{2}.\]
    By direct computation, we have
    \[\sqrt{2}^ku' = 2(b_1+b_2)+1 + (a_1+a_2+1)\sqrt{2}, \enspace \sqrt{2}^kv' = 2(b_1-b_2)-1+(a_1-a_2)\sqrt{2}.\]
    Thus, $\lde(u') = \lde(v') = k$ and
    \[\sqrt{2}^ku' \equiv 1+(a_1+a_2+1)\sqrt{2} \equiv 1+(a_1-a_2+1)\sqrt{2}\not\equiv 1+(a_1-a_2)\sqrt{2} \equiv \sqrt{2}^kv' \pmod{2}. \qedhere\]
\end{proof}

\begin{lemma}\label{lem:useful1}
    For any $v\in \HPiRing{}^4$ and $k\geq 2$ such that $\sqrt{2}^kv \in \mathbb{Z}[\sqrt{2}]^4$, $\sqrt{2}^kv^{T} \equiv (1,1,1,1) \pmod{2}$ or $ \equiv (1+\sqrt{2}, 1+\sqrt{2}, 1+\sqrt{2}, 1+\sqrt{2}) \pmod{2}$, either
    \begin{itemize}
        \item $\sqrt{2}^{k-1}\CH{1,3}\CH{2,4}\CH{1,2}\CH{3,4}v\in \ZZ[\sqrt{2}]^4$; or
        \item $\sqrt{2}^{k-1}\CH{1,4}\CH{2,3}\CH{1,2}\CH{3,4}v\in \ZZ[\sqrt{2}]^4$.
    \end{itemize}
\end{lemma}
\begin{proof}~
    \begin{itemize}
        \item If $\sqrt{2}^kv^{T} \equiv (1,1,1,1) \pmod{2}$, then there exist $a_1,b_1,a_2,b_2,a_3,b_3,a_4,b_4 \in \ZZ$ such that 
        \[\sqrt{2}^kv^T = (2a_1+1+2b_1\sqrt{2}, 2a_2+1+2b_2\sqrt{2},2a_3+1+2b_3\sqrt{2},2a_4+1+2b_4\sqrt{2}).\]
        Let $u = \CH{1,2}\CH{3,4}v$, we have 
        \begin{align*}
            \sqrt{2}^{k-1}u_1 &= (a_1+a_2+1)+(b_1+b_2)\sqrt{2}, \\ 
            \sqrt{2}^{k-1}u_2 &= (a_1-a_2 \phantom{{}+{}1})+(b_1-b_2)\sqrt{2}, \\ 
            \sqrt{2}^{k-1}u_3 &= (a_3+a_4+1)+(b_3+b_4)\sqrt{2}, \\ 
            \sqrt{2}^{k-1}u_4 &= (a_3-a_4 \phantom{{}+{}1})+(b_3-b_4)\sqrt{2}.
        \end{align*}
        Then, depending on the parity of $a_1+a_2$ and $a_3+a_4$, we have that either
        \begin{itemize}
            \item $a_1+a_2$ and $a_3+a_4$ have the same parity: $\sqrt{2}^{k-1}u_1 \equiv \sqrt{2}^{k-1}u_3 \pmod{\sqrt{2}}$, $\sqrt{2}^{k-1}u_2 \equiv \sqrt{2}^{k-1}u_4 \pmod{\sqrt{2}}$, which results in \[\sqrt{2}^{k-1}\CH{1,3}\CH{2,4}\CH{1,2}\CH{3,4}v = \sqrt{2}^{k-1}\CH{1,3}\CH{2,4}u\in \ZZ[\sqrt{2}]^4; \text{ or}\]
            \item $a_1+a_2$ and $a_3+a_4$ have different parities: $\sqrt{2}^{k-1}u_1 \equiv \sqrt{2}^{k-1}u_4 \pmod{\sqrt{2}}$, $\sqrt{2}^{k-1}u_2 \equiv \sqrt{2}^{k-1}u_3 \pmod{\sqrt{2}}$, which results in  \[\sqrt{2}^{k-1}\CH{1,4}\CH{2,3}\CH{1,2}\CH{3,4}v = \sqrt{2}^{k-1}\CH{1,4}\CH{2,3}u\in \ZZ[\sqrt{2}]^4.\]
        \end{itemize}
        \item If $\sqrt{2}^kv \equiv (1+\sqrt{2}, 1+\sqrt{2}, 1+\sqrt{2}, 1+\sqrt{2}) \pmod{2}$, we can conclude similarly to the one above. \qedhere
    \end{itemize}
\end{proof}

\begin{lemma}\label{lem:useful2}
    For any $v\in \HPiRing{}^4$ and $k\geq 2$ such that $\sqrt{2}^kv \in \mathbb{Z}[\sqrt{2}]^4$, $\sqrt{2}^kv^T \equiv (1+\sqrt{2},1,1+\sqrt{2},1) \pmod{2}$ or $ \equiv (1, 1+\sqrt{2}, 1, 1+\sqrt{2}) \pmod{2}$, one of the following four cases must hold:
    \begin{enumerate}
        \item $\sqrt{2}^{k-1}\CH{1,4}\CH{2,3}\CH{1,2}\CH{2,3}\CH{1,4}v \in \mathbb{Z}[\sqrt{2}]^4$;
        \item $\sqrt{2}^{k-1}\CH{1,4}\CH{2,3}\CH{1,3}\CH{2,3}\CH{1,4}v \in \mathbb{Z}[\sqrt{2}]^4$;
        \item $\sqrt{2}^{k-1}\CH{1,4}\CH{2,3}\CH{2,4}\CH{2,3}\CH{1,4}v \in \mathbb{Z}[\sqrt{2}]^4$;
        \item $\sqrt{2}^{k-1}\CH{1,4}\CH{2,3}\CH{3,4}\CH{2,3}\CH{1,4}v \in \mathbb{Z}[\sqrt{2}]^4$.
    \end{enumerate}
\end{lemma}
\begin{proof}~
    \begin{itemize}
        \item If $\sqrt{2}^kv^T \equiv (1+\sqrt{2},1,1+\sqrt{2},1) \pmod{2}$, there exist $a_1,b_1,a_2,b_2,a_3,b_3,a_4,b_4 \in \ZZ$ such that
        \[\sqrt{2}^kv^T = (2a_1+1+(2b_1+1)\sqrt{2}, 2a_2+1+2b_2\sqrt{2},2a_3+1+(2b_3+1)\sqrt{2},2a_4+1+2b_4\sqrt{2}).\]
        By direct computation, we have
        \begin{align*}
            \CH{1,4}\CH{2,3}\CH{1,2}\CH{2,3}\CH{1,4} &= \frac{1}{\sqrt{2}^3} \begin{psmallmatrix}
                1 + \sqrt{2} & 1 & 1 & 1 - \sqrt{2} \\
                1 & -1 + \sqrt{2} & -1 - \sqrt{2} & 1 \\
                1 & -1 - \sqrt{2} & -1 + \sqrt{2} & 1 \\
                1 - \sqrt{2} & 1 & 1 & 1 + \sqrt{2}
            \end{psmallmatrix}_{[1,2,3,4]} \\
            \CH{1,4}\CH{2,3}\CH{1,3}\CH{2,3}\CH{1,4} &= \frac{1}{\sqrt{2}^3} \begin{psmallmatrix}
                1 + \sqrt{2} & 1 & -1 & 1 - \sqrt{2} \\
                1 & -1 + \sqrt{2} & 1 + \sqrt{2} & 1 \\
                -1 & 1 + \sqrt{2} & -1 + \sqrt{2} & -1 \\
                1 - \sqrt{2} & 1 & -1 & 1 + \sqrt{2}
            \end{psmallmatrix}_{[1,2,3,4]} \\
            \CH{1,4}\CH{2,3}\CH{2,4}\CH{2,3}\CH{1,4} &= \frac{1}{\sqrt{2}^3} \begin{psmallmatrix}
                -1 + \sqrt{2} & 1 & 1 & 1 + \sqrt{2} \\
                1 & 1 + \sqrt{2} & 1 - \sqrt{2} & -1 \\
                1 & 1 - \sqrt{2} & 1 + \sqrt{2} & -1 \\
                1 + \sqrt{2} & -1 & -1 & -1 + \sqrt{2}
            \end{psmallmatrix}_{[1,2,3,4]} \\
            \CH{1,4}\CH{2,3}\CH{3,4}\CH{2,3}\CH{1,4} &= \frac{1}{\sqrt{2}^3} \begin{psmallmatrix}
                -1 + \sqrt{2} & 1 & -1 & 1 + \sqrt{2} \\
                1 & 1 + \sqrt{2} & -1 + \sqrt{2} & -1 \\
                -1 & -1 + \sqrt{2} & 1 + \sqrt{2} & 1 \\
                1 + \sqrt{2} & -1 & 1 & -1 + \sqrt{2}
            \end{psmallmatrix}_{[1,2,3,4]}
        \end{align*}
        With $u_1$, $u_2$, $u_3$, and $u_4$ denoting the first entries of the following vectors, respectively:
        \begin{gather*}
            \sqrt{2}^{k-1}\CH{1,4}\CH{2,3}\CH{1,2}\CH{2,3}\CH{1,4}v, \sqrt{2}^{k-1}\CH{1,4}\CH{2,3}\CH{1,3}\CH{2,3}\CH{1,4}v, \\
            \sqrt{2}^{k-1}\CH{1,4}\CH{2,3}\CH{2,4}\CH{2,3}\CH{1,4}v, \text{ and } \sqrt{2}^{k-1}\CH{1,4}\CH{2,3}\CH{3,4}\CH{2,3}\CH{1,4}v,
        \end{gather*}
        we have
        \begin{align*}
            u_1 &= \frac{1}{2}(a_1 + a_2 + a_3 + a_4 + 2b_1 - 2b_4+3) + \frac{1}{\sqrt{2}}(a_1 - a_4 + b_1 + b_2 + b_3 + b_4 + 1), \\
            u_2 &= \frac{1}{2}(a_1 + a_2 - a_3 + a_4 + 2b_1 - 2b_4+2) + \frac{1}{\sqrt{2}}(a_1 - a_4 + b_1 + b_2 - b_3 + b_4), \\
            u_3 &= \frac{1}{2}(-a_1 + a_2 + a_3 + a_4 + 2b_1 + 2b_4+2) + \frac{1}{\sqrt{2}}(a_1 + a_4 - b_1 + b_2 + b_3 + b_4 + 1), \\
            u_4 &= \frac{1}{2}(-a_1 + a_2 - a_3 + a_4 + 2b_1 + 2b_4+1) + \frac{1}{\sqrt{2}}(a_1 + a_4 - b_1 + b_2 - b_3 + b_4).
        \end{align*}
        Then, depending on the parity of $a_1 + a_2 + a_3 + a_4$ and $a_1 - a_4 + b_1 + b_2 + b_3 + b_4$, we have one of the following four cases:
        \begin{enumerate}
            \item $a_1 + a_2 + a_3 + a_4 \equiv 0 \pmod{2}$ and $a_1 - a_4 + b_1 + b_2 + b_3 + b_4 \equiv 0 \pmod{2}$.
            In this case, since \begin{align*}
                a_1 + a_2 - a_3 + a_4 - (a_1 + a_2 + a_3 + a_4) &= -2a_3,\\
                a_1 - a_4 + b_1 + b_2 - b_3 + b_4 - (a_1 - a_4 + b_1 + b_2 + b_3 + b_4) &= -2b_3.
            \end{align*}
            we have $a_1 + a_2 - a_3 + a_4 \equiv a_1 - a_4 + b_1 + b_2 - b_3 + b_4 \equiv 0 \pmod{2}$. Write $a_1 + a_2 - a_3 + a_4 = 2c, a_1 - a_4 + b_1 + b_2 - b_3 + b_4 = 2d$ for some $c, d \in \ZZ$, with direct computation, we have 
            \begin{align*}
                & \sqrt{2}^{k-1}\CH{1,4}\CH{2,3}\CH{1,3}\CH{2,3}\CH{1,4}v \\
                ={}& \begin{pmatrix}
                    \frac{1}{2}(a_1 + a_2 - a_3 + a_4 + 2b_1 - 2b_4+2) + \frac{1}{\sqrt{2}}(a_1 - a_4 + b_1 + b_2 - b_3 + b_4) \\
                    \frac{1}{2}(a_1 - a_2 + a_3 + a_4 + 2b_2 + 2b_3+2) + \frac{1}{\sqrt{2}}(a_2 + a_3 + b_1 - b_2 + b_3 + b_4 + 2) \\
                    \frac{1}{2}(-a_1 + a_2 - a_3 - a_4 + 2b_2 + 2b_3) + \frac{1}{\sqrt{2}}(a_2 + a_3 - b_1 + b_2 - b_3 - b_4) \\
                    \frac{1}{2}(a_1 + a_2 - a_3 + a_4 - 2b_1 + 2b_4) + \frac{1}{\sqrt{2}}(-a_1 + a_4 + b_1 + b_2 - b_3 + b_4)
                \end{pmatrix} \\
                ={}& \begin{pmatrix}
                    \frac{1}{2}(2c + 2b_1 - 2b_4+2) + \frac{1}{\sqrt{2}}(2d) \\
                    \frac{1}{2}(2c - 2a_2+2a_3 + 2b_2 + 2b_3+2) + \frac{1}{\sqrt{2}}(2d-2c+2a_2+2a_4-2b_2+2b_3 + 2) \\
                    \frac{1}{2}(2c-2a_1-2a_4 + 2b_2 + 2b_3) + \frac{1}{\sqrt{2}}(2d-2c+2a_2+2a_4-2b_1-2b_4) \\
                    \frac{1}{2}(2c- 2b_1 + 2b_4) + \frac{1}{\sqrt{2}}(2d-2a_1+2a_4)
                \end{pmatrix} \\
                ={}& \begin{pmatrix}
                    (c + b_1 - b_4+1) + \sqrt{2}(d) \\
                    (c - a_2+a_3 + b_2 + b_3+1) + {\sqrt{2}}(d-c+a_2+a_4-b_2+b_3 + 1) \\
                    (c-a_1-a_4 + b_2 + b_3) + {\sqrt{2}}(d-c+a_2+a_4-b_1-b_4) \\
                    (c- b_1 + b_4) + {\sqrt{2}}(d-a_1+a_4)
                \end{pmatrix} \in \mathbb{Z}[\sqrt{2}]^4.
            \end{align*}
            \item $a_1 + a_2 + a_3 + a_4 \equiv 0 \pmod{2}$ and $a_1 - a_4 + b_1 + b_2 + b_3 + b_4 \equiv 1 \pmod{2}$. In this case, since \begin{align*}
                -a_1 + a_2 - a_3 - a_4 - (a_1 + a_2 + a_3 + a_4) &= -2a_1-2a_3-2a_4,\\
                a_1 + a_4 - b_1 + b_2 + b_3 + b_4 - (a_1 - a_4 + b_1 + b_2 + b_3 + b_4) &= 2a_4-2b_3.
            \end{align*}
            we have $-a_1 + a_2 - a_3 - a_4 \equiv a_1 + a_4 - b_1 + b_2 + b_3 + b_4+1 \equiv 0 \pmod{2}$. Write $-a_1 + a_2 - a_3 - a_4 = 2c, a_1 + a_4 - b_1 + b_2 + b_3 + b_4+1 = 2d$ for some $c, d \in \ZZ$, with direct computation, we have 
            \begin{align*}
                & \sqrt{2}^{k-1}\CH{1,4}\CH{2,3}\CH{2,4}\CH{2,3}\CH{1,4}v \\
                ={}& \begin{pmatrix}
                    \frac{1}{2}(-a_1 + a_2 + a_3 + a_4 + 2b_1 + 2b_4+2) + \frac{1}{\sqrt{2}}(a_1 + a_4 - b_1 + b_2 + b_3 + b_4 + 1) \\
                    \frac{1}{2}(a_1 + a_2 + a_3 - a_4 + 2b_2 - 2b_3) + \frac{1}{\sqrt{2}}(a_2 - a_3 + b_1 + b_2 + b_3 - b_4 + 1) \\
                    \frac{1}{2}(a_1 + a_2 + a_3 - a_4 - 2b_2 + 2b_3+2) + \frac{1}{\sqrt{2}}(-a_2 + a_3 + b_1 + b_2 + b_3 - b_4 + 1) \\
                    \frac{1}{2}(a_1 - a_2 - a_3 - a_4 + 2b_1 + 2b_4) + \frac{1}{\sqrt{2}}(a_1 + a_4 + b_1 - b_2 - b_3 - b_4 + 1)
                \end{pmatrix} \\
                ={}& \begin{pmatrix}
                    \frac{1}{2}(2c + 2b_1 + 2b_4+2) + \frac{1}{\sqrt{2}}(2d) \\
                    \frac{1}{2}(2c+2a_1-2a_4 + 2b_2 - 2b_3) + \frac{1}{\sqrt{2}}(2d+2c+2b_1-2b_4) \\
                    \frac{1}{2}(2c+2a_1-2a_4 - 2b_2 + 2b_3+2) + \frac{1}{\sqrt{2}}(2d+2c+2a_2-2a_3+2b_1-2b_4) \\
                    \frac{1}{2}(2c+2a_1-2a_2-2a_3-2a_4 + 2b_1 + 2b_4) + \frac{1}{\sqrt{2}}(2d+2b_1 - 2b_2 - 2b_3 - 2b_4)
                \end{pmatrix} \\
                ={}& \begin{pmatrix}
                    (c + b_1 + b_4+1) + {\sqrt{2}}(d) \\
                    (c+a_1-a_4 + b_2 - b_3) + {\sqrt{2}}(d+c+b_1-b_4) \\
                    (c+a_1-a_4 - b_2 + b_3+1) + {\sqrt{2}}(d+c+a_2-a_3+b_1-b_4) \\
                    (c+a_1-a_2-a_3-a_4 + b_1 + b_4) + {\sqrt{2}}(d+b_1 - b_2 - b_3 - b_4)
                \end{pmatrix}
                \in \mathbb{Z}[\sqrt{2}]^4.
            \end{align*}
            \item $a_1 + a_2 + a_3 + a_4 \equiv 1 \pmod{2}$ and $a_1 - a_4 + b_1 + b_2 + b_3 + b_4 \equiv 0 \pmod{2}$. In this case, since \begin{align*}
                -a_1 + a_2 - a_3 + a_4 - (a_1 + a_2 + a_3 + a_4) &= -2a_1-2a_3,\\
                a_1 + a_4 - b_1 + b_2 - b_3 + b_4 - (a_1 - a_4 + b_1 + b_2 + b_3 + b_4) &= 2a_4-2b_1-2b_3.
            \end{align*}
            we have $-a_1 + a_2 - a_3 + a_4+1 \equiv a_1 + a_4 - b_1 + b_2 - b_3 + b_4 \equiv 0 \pmod{2}$. Write $-a_1 + a_2 - a_3 + a_4+1 = 2c, a_1 + a_4 - b_1 + b_2 - b_3 + b_4 = 2d$ for some $c, d \in \ZZ$, with direct computation, we have 
            \begin{align*}
                & \sqrt{2}^{k-1}\CH{1,4}\CH{2,3}\CH{3,4}\CH{2,3}\CH{1,4}v \\
                ={}& \begin{pmatrix}
                    \frac{1}{2}(-a_1 + a_2 - a_3 + a_4 + 2b_1 + 2b_4+1) + \frac{1}{\sqrt{2}}(a_1 + a_4 - b_1 + b_2 - b_3 + b_4) \\
                    \frac{1}{2}(a_1 + a_2 - a_3 - a_4 + 2b_2 + 2b_3+1) + \frac{1}{\sqrt{2}}(a_2 + a_3 + b_1 + b_2 - b_3 - b_4 + 1) \\
                    \frac{1}{2}(-a_1 - a_2 + a_3 + a_4 + 2b_2 + 2b_3+1) + \frac{1}{\sqrt{2}}(a_2 + a_3 - b_1 - b_2 + b_3 + b_4 + 1) \\
                    \frac{1}{2}(a_1 - a_2 + a_3 - a_4 + 2b_1 + 2b_4+1) + \frac{1}{\sqrt{2}}(a_1 + a_4 + b_1 - b_2 + b_3 - b_4 + 2)
                \end{pmatrix} \\
                ={}& \begin{pmatrix}
                    \frac{1}{2}(2c + 2b_1 + 2b_4) + \frac{1}{\sqrt{2}}(2d) \\
                    \frac{1}{2}(2c+2a_1 - 2a_4 + 2b_2 + 2b_3) + \frac{1}{\sqrt{2}}(2d+2c+2a_3-2a_4+2b_1 - 2b_4) \\
                    \frac{1}{2}(2c-2a_2 + 2a_3 + 2b_2 + 2b_3) + \frac{1}{\sqrt{2}}(2d+2c+2a_3-2a_4 - 2b_2 + 2b_3) \\
                    \frac{1}{2}(2c+2a_1 - 2a_2 + 2a_3 - 2a_4 + 2b_1 + 2b_4) + \frac{1}{\sqrt{2}}(2d+2b_1 - 2b_2 + 2b_3 - 2b_4 + 2)
                \end{pmatrix} \\
                ={}& \begin{pmatrix}
                    (c + b_1 + b_4) + {\sqrt{2}}(d) \\
                    (c+a_1 - a_4 + b_2 + b_3) + {\sqrt{2}}(d+c+a_3-a_4+b_1 - b_4) \\
                    (c-a_2 + a_3 + b_2 + b_3) + {\sqrt{2}}(d+c+a_3-a_4 - b_2 + b_3) \\
                    (c+a_1 - a_2 + a_3 - a_4 + b_1 + b_4) + {\sqrt{2}}(d+b_1 - b_2 + b_3 - b_4 + 1)
                \end{pmatrix}
                \in \mathbb{Z}[\sqrt{2}]^4.
            \end{align*}
            \item $a_1 + a_2 + a_3 + a_4 \equiv 1 \pmod{2}$ and $a_1 - a_4 + b_1 + b_2 + b_3 + b_4 \equiv 1 \pmod{2}$. In this case,
            we have $a_1 + a_2 + a_3 + a_4+1 \equiv a_1 - a_4 + b_1 + b_2 + b_3 + b_4+1 \equiv 0 \pmod{2}$. Write $a_1 + a_2 + a_3 + a_4+1 = 2c, a_1 - a_4 + b_1 + b_2 + b_3 + b_4+1 = 2d$ for some $c, d \in \ZZ$, with direct computation, we have 
            \begin{align*}
                & \sqrt{2}^{k-1}\CH{1,4}\CH{2,3}\CH{1,2}\CH{2,3}\CH{1,4}v \\
                ={}& \begin{pmatrix}
                    \frac{1}{2}(a_1 + a_2 + a_3 + a_4 + 2b_1 - 2b_4+3) + \frac{1}{\sqrt{2}}(a_1 - a_4 + b_1 + b_2 + b_3 + b_4 + 1) \\
                    \frac{1}{2}(a_1 - a_2 - a_3 + a_4 + 2b_2 - 2b_3-1) + \frac{1}{\sqrt{2}}(a_2 - a_3 + b_1 - b_2 - b_3 + b_4) \\
                    \frac{1}{2}(a_1 - a_2 - a_3 + a_4 - 2b_2 + 2b_3+1) + \frac{1}{\sqrt{2}}(-a_2 + a_3 + b_1 - b_2 - b_3 + b_4) \\
                    \frac{1}{2}(a_1 + a_2 + a_3 + a_4 - 2b_1 + 2b_4+1) + \frac{1}{\sqrt{2}}(-a_1 + a_4 + b_1 + b_2 + b_3 + b_4 + 1)
                \end{pmatrix} \\
                ={}& \begin{pmatrix}
                    \frac{1}{2}(2c + 2b_1 - 2b_4) + \frac{1}{\sqrt{2}}(2d) \\
                    \frac{1}{2}(2c-2a_2 - 2a_3 + 2b_2 - 2b_3-4) + \frac{1}{\sqrt{2}}(2d+2c-2a_1-2a_3- 2b_2 - 2b_3-2) \\
                    \frac{1}{2}(2c-2a_2 - 2a_3 - 2b_2 + 2b_3-2) + \frac{1}{\sqrt{2}}(2d+2c-2a_1-2a_2- 2b_2 - 2b_3-2) \\
                    \frac{1}{2}(2c - 2b_1 + 2b_4-2) + \frac{1}{\sqrt{2}}(2d-2a_1+2a_4)
                \end{pmatrix} \\
                ={}& \begin{pmatrix}
                    (c + b_1 - b_4) + {\sqrt{2}}(d) \\
                    (c-a_2 - a_3 + b_2 - b_3-2) + {\sqrt{2}}(d+c-a_1-a_3- b_2 - b_3-1) \\
                    (c-a_2 - a_3 - b_2 + b_3-1) + {\sqrt{2}}(d+c-a_1-a_2- b_2 - b_3-1) \\
                    (c - b_1 + b_4-1) + {\sqrt{2}}(d-a_1+a_4)
                \end{pmatrix}
                \in \mathbb{Z}[\sqrt{2}]^4.
            \end{align*}
        \end{enumerate}
        Therefore, we complete the proof for this case.
        \item If $\sqrt{2}^kv^T \equiv (1,1+\sqrt{2},1,1+\sqrt{2}) \pmod{2}$, there exist $a_1,b_1,a_2,b_2,a_3,b_3,a_4,b_4 \in \ZZ$ such that
        \[\sqrt{2}^kv^T = (2a_1+1+2b_1\sqrt{2}, 2a_2+1+(2b_2+1)\sqrt{2},2a_3+1+2b_3\sqrt{2},2a_4+1+2(b_4+1)\sqrt{2}).\]
        Letting $u_1$, $u_2$, $u_3$, and $u_4$ denoting the first entries of the following vectors, respectively:
        \begin{gather*}
            \sqrt{2}^{k-1}\CH{1,4}\CH{2,3}\CH{1,2}\CH{2,3}\CH{1,4}v, \sqrt{2}^{k-1}\CH{1,4}\CH{2,3}\CH{1,3}\CH{2,3}\CH{1,4}v, \\
            \sqrt{2}^{k-1}\CH{1,4}\CH{2,3}\CH{2,4}\CH{2,3}\CH{1,4}v, \text{ and } \sqrt{2}^{k-1}\CH{1,4}\CH{2,3}\CH{3,4}\CH{2,3}\CH{1,4}v,
        \end{gather*}
        by direct computation, we have
        \begin{align*}
            u_1 &= \frac{1}{2}(a_1 + a_2 + a_3 + a_4 + 2b_1 - 2b_4+1) + \frac{1}{\sqrt{2}}(a_1 - a_4 + b_1 + b_2 + b_3 + b_4 + 1), \\
            u_2 &= \frac{1}{2}(a_1 + a_2 - a_3 + a_4 + 2b_1 - 2b_4) + \frac{1}{\sqrt{2}}(a_1 - a_4 + b_1 + b_2 - b_3 + b_4 + 1), \\
            u_3 &= \frac{1}{2}(-a_1 + a_2 + a_3 + a_4 + 2b_1 + 2b_4+2) + \frac{1}{\sqrt{2}}(a_1 + a_4 - b_1 + b_2 + b_3 + b_4 + 2), \\
            u_4 &= \frac{1}{2}(-a_1 + a_2 - a_3 + a_4 + 2b_1 + 2b_4+1) + \frac{1}{\sqrt{2}}(a_1 + a_4 - b_1 + b_2 - b_3 + b_4 + 2).
        \end{align*}
        Then, depending on the parity of $a_1 + a_2 + a_3 + a_4$ and $a_1 - a_4 + b_1 + b_2 + b_3 + b_4$, we have one of the following four cases:
        \begin{enumerate}
            \item  $a_1 + a_2 + a_3 + a_4 \equiv 0 \pmod{2}$ and $a_1 - a_4 + b_1 + b_2 + b_3 + b_4 \equiv 0 \pmod{2}$. In this case, we have $u_3 \in \mathbb{Z}[\sqrt{2}]$, then we can conclude similarly to the one
            above that $\sqrt{2}^{k-1}\CH{1,4}\CH{2,3}\CH{2,4}\CH{2,3}\CH{1,4}v \in \mathbb{Z}[\sqrt{2}]^4$.
            \item  $a_1 + a_2 + a_3 + a_4 \equiv 0 \pmod{2}$ and $a_1 - a_4 + b_1 + b_2 + b_3 + b_4 \equiv 1 \pmod{2}$. In this case, we have $u_2 \in \mathbb{Z}[\sqrt{2}]$, then we can conclude similarly to the one
            above that $\sqrt{2}^{k-1}\CH{1,4}\CH{2,3}\CH{1,3}\CH{2,3}\CH{1,4}v \in \mathbb{Z}[\sqrt{2}]^4$.
            \item  $a_1 + a_2 + a_3 + a_4 \equiv 1 \pmod{2}$ and $a_1 - a_4 + b_1 + b_2 + b_3 + b_4 \equiv 0 \pmod{2}$. In this case, we have $u_4 \in \mathbb{Z}[\sqrt{2}]$, then we can conclude similarly to the one
            above that $\sqrt{2}^{k-1}\CH{1,4}\CH{2,3}\CH{3,4}\CH{2,3}\CH{1,4}v \in \mathbb{Z}[\sqrt{2}]^4$.
            \item  $a_1 + a_2 + a_3 + a_4 \equiv 1 \pmod{2}$ and $a_1 - a_4 + b_1 + b_2 + b_3 + b_4 \equiv 1 \pmod{2}$. In this case, we have $u_1 \in \mathbb{Z}[\sqrt{2}]$, then we can conclude similarly to the one
            above that $\sqrt{2}^{k-1}\CH{1,4}\CH{2,3}\CH{1,2}\CH{2,3}\CH{1,4}v \in \mathbb{Z}[\sqrt{2}]^4$. \qedhere
        \end{enumerate}
    \end{itemize}
\end{proof}

\begin{lemma}\label{lem:useful2_1}
    For any $M \in \HPiUnitary{}$ and $k\geq 2$ such that $\level(M) = (j,k,4)$ and
    \[\sqrt{2}^kv^T \equiv (1+\sqrt{2},1,1+\sqrt{2},1) \pmod{2} \text{ or } \equiv (1, 1+\sqrt{2}, 1, 1+\sqrt{2}) \pmod{2},\]
    where $v = (M[1,j], M[2,j], M[3,j], M[4,j])^T$,
    if $\sqrt{2}^{k-1}\CH{1,4}\CH{2,3}\CH{2,4}\CH{2,3}\CH{1,4}v \in \mathbb{Z}[\sqrt{2}]^4$,
    then 
    \[ \sqrt{2}^{k-1}\CH{1,3}\CH{2,4}\CH{1,2}\CH{1,3}\CH{2,4}v \in \mathbb{Z}[\sqrt{2}]^4. \]
    %for $(u_1,u_2,u_3,u_4) = \CH{1,3}\CH{2,4}\CH{4,2}\CH{2,4}\CH{1,3}\CH{2,1}v$, we have
    %\[ \lde(u_1)=\lde(u_2)=\lde(u_3)=\lde(u_4) = k, \sqrt{2}^ku_1 \equiv \sqrt{2}^ku_2 \text{ and } \sqrt{2}^ku_3 \equiv \sqrt{2}^ku_4. \]
\end{lemma}
\begin{proof}
    Without loss of generality, we assume $\sqrt{2}^kv^T \equiv (1+\sqrt{2},1,1+\sqrt{2},1) \pmod{2}$.
    Then, there exist $a_1,b_1,a_2,b_2,a_3,b_3,a_4,b_4 \in \ZZ$ such that
    \[\sqrt{2}^kv^T = (2a_1+1+(2b_1+1)\sqrt{2}, 2a_2+1+2b_2\sqrt{2},2a_3+1+(2b_3+1)\sqrt{2},2a_4+1+2b_4\sqrt{2}).\]
    Let $w_1$ be the first entry of
    $\sqrt{2}^{k-1}\CH{1,4}\CH{2,3}\CH{2,4}\CH{2,3}\CH{1,4}v \in \ZZ[\sqrt{2}]^4$, 
    \[(u_1, u_2, u_3, u_4)^T = \sqrt{2}^{k-1}\CH{1,3}\CH{2,4}\CH{1,2}\CH{1,3}\CH{2,4}v.\]
    By direct computation as in the proof of \cref{lem:useful2}, we have
    \begin{align*}
        w_1 &= \frac{1}{2}(-a_1 + a_2 + a_3 + a_4 + 2b_1 + 2b_4+2) + \frac{1}{\sqrt{2}}(a_1 + a_4 - b_1 + b_2 + b_3 + b_4 + 1), \\
        u_1 &= \frac{1}{2}(a_1+a_2+a_3+a_4+2b_1-2b_3+2) + \frac{1}{\sqrt{2}}(a_1-a_3+b_1+b_2+b_3+b_4 + 1), \\
        u_2 &= \frac{1}{2}(a_1-a_2+a_3-a_4+2b_2-2b_4) + \frac{1}{\sqrt{2}}(a_2-a_4+b_1-b_2+b_3-b_4 + 1), \\
        u_3 &= \frac{1}{2}(a_1+a_2+a_3+a_4-2b_1+2b_3+2) + \frac{1}{\sqrt{2}}(-a_1+a_3+b_1+b_2+b_3+b_4 + 1), \\
        u_4 &= \frac{1}{2}(a_1-a_2+a_3-a_4-2b_2+2b_4) + \frac{1}{\sqrt{2}}(-a_2+a_4+b_1-b_2+b_3-b_4 + 1).
    \end{align*}
    With $w_1 \in \ZZ[\sqrt{2}]$, we have $-a_1 + a_2 + a_3 + a_4 \equiv a_1 + a_4 - b_1 + b_2 + b_3 + b_4 + 1 \equiv 0 \pmod{2}$.
    We then consider the parity of $a_3+a_4$.
    \begin{itemize}
        \item $a_3+a_4 \equiv 0\pmod{2}$. In this case, 
        \[ u_1 - w_1 = \frac{1}{2}(2a_1-2b_3-2b_4)+\frac{1}{\sqrt{2}}(-a_3-a_4+2b_1) \in \ZZ[\sqrt{2}], \]
        which implies that $u_1\in \ZZ[\sqrt{2}]$. Similarly, it can also be seen that $u_2,u_3,u_4\in \ZZ[\sqrt{2}]$. Thus $\sqrt{2}^{k-1}\CH{1,3}\CH{2,4}\CH{1,2}\CH{1,3}\CH{2,4}v \in \ZZ[\sqrt{2}]^4$.
        \item $a_3+a_4 \equiv 1\pmod{2}$. In this case, we have
        \begin{align*}
            a_1+a_3+b_1+b_2+b_3+b_4 &= (a_1 + a_4 - b_1 + b_2 + b_3 + b_4) - (-a_3+a_4-2b_1) \\
            &\equiv 1 - 1 \pmod{2}\\
            &\equiv 0 \pmod{2}.
        \end{align*}
        Since $\level(M) = (j,k,4)$, we can write $\sqrt{2}^{k}M[i,j] = \sqrt{2}(a_i+b_i\sqrt{2})$ with $a_i,b_i \in \ZZ$ for $ 4< i\leq n$. Then,
        \begin{align*}
            2^k ={}& \sum_{i=1}^4(2a_i+1)^2+2\rbra*{\rbra*{2b_1+1}^2+\rbra*{2b_2}^2+\rbra*{2b_3+1}^2+\rbra*{2b_4}^2} \\
            &\quad + 2\rbra*{(2a_1+1)(2b_1+1)+(2a_3+1)(2b_3+1)}\sqrt{2} \\
            &\quad + 2\rbra*{(2a_2+1)(2b_2)+(2a_4+1)(2b_4)}\sqrt{2} \\
            &\quad + 2\rbra*{\sum_{i>4}a_i^2+2b_i^2+2a_ib_i\sqrt{2}}.
        \end{align*}
        Thus, 
        \begin{align*}
            2^k ={}& \sum_{i=1}^4(2a_i+1)^2+2\rbra*{\rbra*{2b_1+1}^2+\rbra*{2b_2}^2+\rbra*{2b_3+1}^2+\rbra*{2b_4}^2} \\
            & \quad + 2\rbra*{\sum_{i>4}a_i^2+2b_i^2}, \\
            ={}& 4\rbra*{\sum_{i=1}^4\rbra*{a_i^2+b_i^2+a_i}+b_1+b_2+2+\sum_{i>4}b_i^2}+2\sum_{i>4}a_i^2,  \\
            0 ={}&  2\rbra*{(2a_1+1)(2b_1+1)+(2a_3+1)(2b_3+1)} \\
            & \quad +2\rbra*{(2a_2+1)(2b_2)+(2a_4+1)(2b_4)}+2\rbra*{\sum_{i>4}2a_ib_i} \\
            ={}& 4\rbra*{\sum_{i=1}^4\rbra*{2a_ib_i+b_i}+a_1+a_3+\sum_{i>4}a_ib_i}.
        \end{align*}
        Since $k \geq 2$, we have
        \begin{align*}
            2\sum_{i>4}a_i^2 \equiv 0 \pmod{4} \text{ and }
            \sum_{i>4}a_ib_i \equiv  \sum_{i=1}^4b_i +a_1+a_3\pmod{2}.
        \end{align*}
        With $a_1+a_3 + b_1+b_2+b_3+b_4 \equiv 0 \pmod{2}$, we have 
        \begin{align*}
            \sum_{i>4}a_i \equiv 0 \pmod{2} \text{ and }
            \sum_{i>4}a_ib_i \equiv 0\pmod{2}.
        \end{align*}
        It follows that $a_i+\sqrt{2}b_i \equiv 1+\sqrt{2} \pmod{2}$ for evenly many $4 < i \leq n$.

        Consider the first four entries of the vector $(\CH{1,3}\CH{2,4}M)[:,j]$, we have
        \begin{align*}
            \sqrt{2}^{k-1}(\CH{1,3}\CH{2,4}M)[1,j] &= a_1+a_3+1+\sqrt{2}(b_1+b_3+1), \\
            \sqrt{2}^{k-1}(\CH{1,3}\CH{2,4}M)[2,j] &= a_2+a_4+1+\sqrt{2}(b_2+b_4), \\
            \sqrt{2}^{k-1}(\CH{1,3}\CH{2,4}M)[3,j] &= a_1-a_3+\sqrt{2}(b_1-b_3), \\
            \sqrt{2}^{k-1}(\CH{1,3}\CH{2,4}M)[4,j] &= a_2-a_4+\sqrt{2}(b_2-b_4).
        \end{align*}
        With $\sqrt{2}^{k-1}(\CH{1,3}\CH{2,4}M)[i,j] = \sqrt{2}^{k-1}M[i,j] = a_i+\sqrt{2}b_i$ for $4 < i \leq n$, it can be seen that $\lde((\CH{1,3}\CH{2,4}M)[:,j]) \leq k-1$.

        We further consider the parity of $a_1+a_3$:
        \begin{itemize}
            \item $a_1+a_3 \equiv 0 \pmod{2}$. In this case, $b_1+b_2+b_3+b_4 \equiv 0 \pmod{2}$. Then, either \begin{itemize}
                \item $\sqrt{2}^{k-1}(\CH{1,3}\CH{2,4}M)[1,j] \equiv 1+\sqrt{2}\pmod{2}$; or
                \item $\sqrt{2}^{k-1}(\CH{1,3}\CH{2,4}M)[2,j] \equiv 1+\sqrt{2}\pmod{2}$.
            \end{itemize}
            Additionally, $\sqrt{2}^{k-1}(\CH{1,3}\CH{2,4}M)[3,j] \equiv \sqrt{2}^{k-1}(\CH{1,3}\CH{2,4}M)[4,j] \equiv 0 \pmod{\sqrt{2}}$. 
            \item $a_1+a_3 \equiv 1 \pmod{2}$. In this case, $b_1+b_2+b_3+b_4 \equiv 1 \pmod{2}$. Then, either \begin{itemize}
                \item $\sqrt{2}^{k-1}(\CH{1,3}\CH{2,4}M)[3,j] \equiv 1+\sqrt{2}\pmod{2}$; or
                \item $\sqrt{2}^{k-1}(\CH{1,3}\CH{2,4}M)[4,j] \equiv 1+\sqrt{2}\pmod{2}$.
            \end{itemize}
            Additionally, $\sqrt{2}^{k-1}(\CH{1,3}\CH{2,4}M)[1,j] \equiv \sqrt{2}^{k-1}(\CH{1,3}\CH{2,4}M)[2,j] \equiv 0 \pmod{\sqrt{2}}$.
        \end{itemize}
        Therefore, there only one index $i$ chosen from $\{1,2,3,4\}$ such that $\sqrt{2}^{k-1}(\CH{1,3}\CH{2,4}M)[i,j] \equiv 1+\sqrt{2}\pmod{2}$.
        Together with $\sqrt{2}^{k-1}(\CH{1,3}\CH{2,4}M)[i,j] = a_i+\sqrt{2}b_i \equiv 1+\sqrt{2} \pmod{2}$ for evenly many $4 < i \leq n$, we have $\sqrt{2}^{k-1}(\CH{1,3}\CH{2,4}M)[i,j] \equiv 1+\sqrt{2} \pmod{2}$ for oddly many $1\leq i\leq n$, which contradicts $\lde((\CH{1,3}\CH{2,4}M)[:,j]) = k-1$.\footnote{See the proof of Lemma~5.13 in~\cite{Amy2020numbertheoretic}.} \emph{Thus, this case is impossible.}
    \end{itemize}
    Since one case leads to a valid conclusion and the other is impossible, the proof is complete.
\end{proof}

\begin{lemma}\label{lem:useful2_2}
    For any $M \in \HPiUnitary{}$ and $k\geq 2$ such that $\level(M) = (j,k,4)$ and
    \[\sqrt{2}^kv^T \equiv (1+\sqrt{2},1,1+\sqrt{2},1) \pmod{2} \text{ or } \equiv (1, 1+\sqrt{2}, 1, 1+\sqrt{2}) \pmod{2},\]
    where $v = (M[1,j], M[2,j], M[3,j], M[4,j])^T$,
    if $\sqrt{2}^{k-1}\CH{1,4}\CH{2,3}\CH{3,4}\CH{2,3}\CH{1,4}v \in \mathbb{Z}[\sqrt{2}]^4$,
    then 
    \[ \sqrt{2}^{k-1}\CH{1,3}\CH{2,4}\CH{1,4}\CH{1,3}\CH{2,4}v \in \mathbb{Z}[\sqrt{2}]^4. \]
    %for $(u_1,u_2,u_3,u_4) = \CH{1,3}\CH{2,4}\CH{4,2}\CH{2,4}\CH{1,3}\CH{2,1}v$, we have
    %\[ \lde(u_1)=\lde(u_2)=\lde(u_3)=\lde(u_4) = k, \sqrt{2}^ku_1 \equiv \sqrt{2}^ku_2 \text{ and } \sqrt{2}^ku_3 \equiv \sqrt{2}^ku_4. \]
\end{lemma}
\begin{proof}
    Without loss of generality, we assume $\sqrt{2}^kv^T \equiv (1+\sqrt{2},1,1+\sqrt{2},1) \pmod{2}$.
    Then, there exist $a_1,b_1,a_2,b_2,a_3,b_3,a_4,b_4 \in \ZZ$ such that
    \[\sqrt{2}^kv^T = (2a_1+1+(2b_1+1)\sqrt{2}, 2a_2+1+2b_2\sqrt{2},2a_3+1+(2b_3+1)\sqrt{2},2a_4+1+2b_4\sqrt{2}).\]
    Let $w_1$ be the first entry of
    $\sqrt{2}^{k-1}\CH{1,4}\CH{2,3}\CH{3,4}\CH{2,3}\CH{1,4}v \in \ZZ[\sqrt{2}]^4$, 
    \[(u_1, u_2, u_3, u_4)^T = \sqrt{2}^{k-1}\CH{1,3}\CH{2,4}\CH{1,4}\CH{1,3}\CH{2,4}v.\]
    By direct computation as in the proof of \cref{lem:useful2_2}, we have
    \begin{align*}
        w_1 &= \frac{1}{2}(-a_1+a_2-a_3+a_4+2b_1+2b_4+1) + \frac{1}{\sqrt{2}}(a_1+a_4-b_1+b_2-b_3+b_4), \\
        u_1 &= \frac{1}{2}(a_1+a_2+a_3-a_4+2b_1-2b_3+1) + \frac{1}{\sqrt{2}}(a_1-a_3+b_1+b_2+b_3+b_4+1), \\
        u_2 &= \frac{1}{2}(a_1-a_2+a_3+a_4+2b_2+2b_4+1) + \frac{1}{\sqrt{2}}(a_2+a_4+b_1-b_2+b_3+b_4 + 2), \\
        u_3 &= \frac{1}{2}(a_1+a_2+a_3-a_4+2b_1+2b_3+1) + \frac{1}{\sqrt{2}}(-a_1+a_3+b_1+b_2+b_3-b_4 + 1), \\
        u_4 &= \frac{1}{2}(-a_1+a_2-a_3-a_4+2b_2+2b_4+1) + \frac{1}{\sqrt{2}}(a_2+a_4-b_1+b_2-b_3-b_4).
    \end{align*}
    With $w_1 \in \ZZ[\sqrt{2}]$, we have $-a_1+a_2-a_3+a_4+1 \equiv a_1+a_4-b_1+b_2-b_3+b_4 \equiv 0 \pmod{2}$.
    We then consider the parity of $a_3+a_4$.
    \begin{itemize}
        \item $a_3+a_4 \equiv 1\pmod{2}$. In this case, 
        \[ u_1 - w_1 = \frac{1}{2}(2a_1+2a_3-2a_4-2b_3-2b_4)+\frac{1}{\sqrt{2}}(-a_3-a_4+2b_1+2b_3+1) \in \ZZ[\sqrt{2}], \]
        which implies that $u_1\in \ZZ[\sqrt{2}]$. Similarly, it can also be seen that $u_2,u_3,u_4\in \ZZ[\sqrt{2}]$. Thus $\sqrt{2}^{k-1}\CH{1,3}\CH{2,4}\CH{1,4}\CH{1,3}\CH{2,4}v \in \ZZ[\sqrt{2}]^4$.
        \item $a_3+a_4 \equiv 0\pmod{2}$. In this case, we have
        \begin{align*}
            a_1+a_3+b_1+b_2+b_3+b_4 &= (a_1 + a_4 - b_1 + b_2 - b_3 + b_4) - (-a_3+a_4-2b_1-2b_3) \\
            &\equiv 0 - 0 \pmod{2}\\
            &\equiv 0 \pmod{2}.
        \end{align*}
        As in the proof of \cref{lem:useful2_1}, it leads to a contradiction. \emph{Therefore, this case is also impossible.}
    \end{itemize}
    Since one case leads to a valid conclusion and the other is impossible, the proof is complete.
\end{proof}

\begin{lemma}
    \label{lem:useful3}
    Let $s$ and $r$ be states with $(j,k,l) = \level(s) = \level(r)$ and $k \geq 1$ such that
    \begin{enumerate}
        \item $\lde(s[a,j]) = \lde(s[b,j]) = k$;
        \item $\lde(s[c,j]) \leq k-1, \lde(s[d,j]) \leq k-1$;
        \item $r = \CH{a,c}\CH{b,d}\CH{a,b}\CH{b,d}\CH{a,c}s$;
        \item $\sqrt{2}^ks[a,j] \not\equiv \sqrt{2}^ks[b,j] \pmod{2}$
    \end{enumerate} 
    for distinct $1\leq a,b,c,d \leq j$, then either
    \begin{itemize}
        \item $\lde(r[a,j]) = \lde(r[b,j]) = k$, $\lde(r[c,j])\leq k-1$, $\lde(r[d,j]) \leq k-1$; or
        \item $\lde(r[c,j]) = \lde(r[d,j]) = k$, $\lde(r[a,j])\leq k-1$, $\lde(r[b,j]) \leq k-1$.
    \end{itemize}
\end{lemma}
\begin{proof}
    Without loss of generality, we assume $\sqrt{2}^ks[a,j] \equiv 1 \pmod{2}$ and $\sqrt{2}^ks[b,j] \equiv 1+\sqrt{2} \pmod{2}$.
    Then, there exist $a_1,b_1,a_2,b_2,a_3,b_3,a_4,b_4\in \ZZ$ such that
    \begin{align*}
        & \sqrt{2}^k(s[a,j],s[b,j],s[c,j],s[d,j]) \\
        ={}& (2a_1+1+2b_1\sqrt{2},2a_2+1+(2b_2+1)\sqrt{2},2a_3+b_3\sqrt{2}, 2a_4+b_4\sqrt{2}).
    \end{align*}
    By direct computation, we have
    \begin{equation*}
        \CH{a,c}\CH{b,d}\CH{a,b}\CH{b,d}\CH{a,c} = \frac{1}{\sqrt{2}^3}\begin{pmatrix}
            1 + \sqrt{2} & 1 & 1 - \sqrt{2} & 1 \\
            1 & -1 + \sqrt{2} & 1 & -\sqrt{2} - 1 \\
            1 - \sqrt{2} & 1 & 1 + \sqrt{2} & 1 \\
            1 & -1-\sqrt(2) & 1 & -1 + \sqrt{2}
        \end{pmatrix}_{[a,b,c,d]}
    \end{equation*}
    With $r = \CH{a,c}\CH{b,d}\CH{a,b}\CH{b,d}\CH{a,c}s$, we have
    \begin{align*}
        \sqrt{2}^k r[a,j] &= a_1 - a_3 + b_1 + b_2 + \frac{b_3 + b_4}{2} + 1 + \left(\frac{a_1 + a_2 + a_3 + a_4 + 2b_1 - b_3 + 2}{2}\right)\sqrt{2}, \\
        \sqrt{2}^k r[b,j] &= a_2 - a_4 + b_1 - b_2 + \frac{b_3 - b_4}{2} + 1 + \left(\frac{a_1 - a_2 + a_3 - a_4 + 2b_2 - b_4}{2}\right)\sqrt{2}, \\
        \sqrt{2}^k r[c,j] &= -a_1 + a_3 + b_1 + b_2 + \frac{b_3 + b_4}{2} + \left(\frac{a_1 + a_2 + a_3 + a_4 - 2b_1 + b_3}{2}\right)\sqrt{2}, \\
        \sqrt{2}^k r[d,j] &= -a_2 + a_4 + b_1 - b_2 + \frac{b_3 - b_4}{2} + \left(\frac{a_1 - a_2 + a_3 - a_4 - 2b_2 + b_4}{2}\right)\sqrt{2}.
    \end{align*}
    Since $\level(r) = \level(s)$, we must have $\sqrt{2}^k r[a,j],\sqrt{2}^k r[b,j],\sqrt{2}^k r[c,j],\sqrt{2}^k r[d,j] \in \ZZ[\sqrt{2}]$, then
    \[ b_3+b_4\equiv 0 \pmod{2}, \enspace a_1 + a_2 + a_3 + a_4 - b_3 \equiv 0 \pmod{2}.\]
    Write $b_3+b_4 = 2e, a_1 + a_2 + a_3 + a_4 - b_3 = 2f$, and $\sqrt{2}^k (r[a,j],r[b,j],r[c,j],r[d,j]) = (c_1+d_1\sqrt{2}, c_2+d_2\sqrt{2}, c_3+d_3\sqrt{2}, c_4+d_4\sqrt{2})$ for some $c_1,d_1,c_2,d_2,c_3,d_3,c_4,d_4,e,f \in \ZZ$, we have
    \begin{align*}
        c_1 - c_2 = -2f+2e+2a_1+2b_2 &\equiv 0 \pmod{2}, \\
        c_3 - c_4 = 2f+2e-2a_1-2a_4+2b_2&\equiv 0 \pmod{2}, \\
        c_1 - c_3 = 2a_1-2a_3+1 &\equiv 1 \pmod{2}.
    \end{align*}
    Then, depending on the parity of $c_1$, we have either
    \begin{itemize}
        \item $c_1\equiv 1 \pmod{2}$, then,
        \[\lde(r[a,j]) = \lde(r[b,j]) = k, \lde(r[c,j])\leq k-1, \lde(r[d,j]) \leq k-1 \text{; or}\]
        \item $c_1\equiv 0 \pmod{2}$, then, 
        \[\lde(r[c,j]) = \lde(r[d,j]) = k, \lde(r[a,j])\leq k-1, \lde(r[b,j]) \leq k-1. \qedhere\]
    \end{itemize}
\end{proof}

\begin{lemma}
    \label{lem:useful4}
    Let $s$ and $r$ be states with $(j,k,l) = \level(s) = \level(r)$ and $k \geq 1$ such that
    \begin{enumerate}
        \item $\lde(s[a,j]) = \lde(s[b,j]) = k$;
        \item $\lde(s[c,j]) \leq k-1, \lde(s[d,j]) \leq k-1$;
        \item $r = \CH{a,c}\CH{b,d}\CH{a,b}\CH{b,d}\CH{a,c}s$
    \end{enumerate} 
    for distinct $1\leq a,b,c,d \leq j$, then there exist a sequence of simple edges $\mathbf{G}$ such that the diagram
    
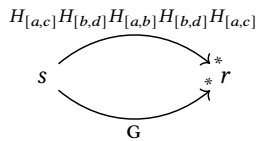
\begin{equation*}
        \begin{tikzcd}
            \setwiretype{n} s \ar[r,to*,bend right=-40,"\CH{a,c}\CH{b,d}\CH{a,b}\CH{b,d}\CH{a,c}"] \ar[r,to*, bend right=40,"\mathbf{G}",swap]&[1cm] r 
        \end{tikzcd}
    \end{equation*}
    commutes equationally and $\level(\mathbf{G}:s \Se*{} r) \leq (j,k,l+2)$.
\end{lemma}
\begin{proof}
    \let\fullwidthdisplay\relax
    We first prove that $k = 1$ is impossible. Suppose $k=1$, we write $\sqrt{2}s[i,j] = x_i+y_i\sqrt{2}$ for $1\leq i \leq n$ with $x_i,y_i \in \ZZ$. Since $s[:,j]$ is a unit vector, we have
    \[2 = \sum_{i=1}^n x_i^2+2y_i^2 +x_iy_i2\sqrt{2}.\]
    Then $\sum_{i=1}^n x_i^2+2y_i^2 = 2$, which means there exist $\abs{x_{i_1}} = \abs{x_{i_2}} = 1$ or $\abs{y_{i_3}} = 1$; and all other $x_i,y_i$ are zero.
    With $\lde(s[a,j]) = \lde(s[b,j]) = 1$, we must have
    \[\abs{s[a,j]} = \abs{s[b,j]} = \frac{1}{\sqrt{2}}\]
    and all other entries of $s[:,j]$ are zero. Then $\CH{a,b}s$ is a unit vector with a single non-zero entry, either $\abs{(\CH{a,b}s)[a,j]} = 1$ or $\abs{(\CH{a,b}s)[b,j]} = 1$.
    By direct computation,
    \begin{equation*}
        \sem{\CH{a,c}\CH{b,d}\CH{a,b}\CH{b,d}\CH{a,c}\CH{a,b}} = \frac{1}{\sqrt{2}^3}\begin{pmatrix}
            1+\sqrt{2} & 1 & 1-\sqrt{2} & 1 \\
            1 & -1+\sqrt{2} & 1 & -1-\sqrt{2} \\
            -1+\sqrt{2} & -1 & 1+\sqrt{2} & 1 \\
            -1 & 1+\sqrt{2} & 1 & -1+\sqrt{2}
        \end{pmatrix}_{[a,b,c,d]}
    \end{equation*}
    With $r = \CH{a,c}\CH{b,d}\CH{a,b}\CH{b,d}\CH{a,c}\CH{a,b}(\CH{a,b}s)$, we can verify that $\level(r) = (j,3,4)$, which contradicts $\level(r) = \level(s) = (j,1,l)$.
    \emph{Thus, $k=1$ is impossible}.

    Then, for $k \geq 2$, we prove it by induction on $l$, noting that $l \geq 2$ must be even.
    \begin{itemize}
        \item $l = 2$. In this case, we must have 
        \[ \sqrt{2}^ks[a,j] \equiv \sqrt{2}^ks[b,j] \pmod{2}.\]
        Then, we distinct subcases by \cref{lem:useful3}.
        \begin{itemize}[leftmargin=10pt]
            \item $\lde(r[a,j]) = \lde(r[b,j]) = k$, $\lde(r[c,j])\leq k-1$, $\lde(r[d,j]) \leq k-1$.
            In this case, we have
            \[ \sqrt{2}^kr[a,j] \equiv \sqrt{2}^kr[b,j] \pmod{2}. \]
            Then, consider the following diagram.
            \begin{equation*}
                \begin{tikzcd}
                    \setwiretype{n} s \ar[rrrrr,to*,"\CH{a,c}\CH{b,d}\CH{a,b}\CH{b,d}\CH{a,c}"{name=a1}]\ar[d, "\CH{a,b}", swap] &&&&& r\\[-4mm]
                    \setwiretype{n} t_1 \ar[r,"\CH{a,c}",swap] & t_2 \ar[r,"\CH{b,d}",swap] & t_3 \ar[r,"\CH{c,d}"{name=a2},swap] & t_4 \ar[r,"\CH{b,d}",swap] & t_5 \ar[r,"\CH{a,c}",swap] & t_6 \ar[u,"\CH{a,b}",swap]
                    \arrow[phantom,from=a1,to=a2,"\text{\cref{prop:h0}}"]
                \end{tikzcd}
            \end{equation*}
            The diagram commutes equationally by the postponed \cref{prop:h0}.
            It can be seen that
            \begin{gather*}
                \level(t_1), \level(t_6) < (j,k,0), \enspace
                \level(t_2), \level(t_5) \leq (j,k,2), \enspace
                \level(t_3), \level(t_4) \leq (j,k,4).
            \end{gather*}
            Let 
            \[\mathbf{G}: s \Se*{\CH{a,b}\CH{a,c}\CH{b,d}\CH{c,d}\CH{b,d}\CH{a,c}\CH{a,b}} r.\]
            We have $\mathbf{G} \approx  \CH{a,c}\CH{b,d}\CH{a,b}\CH{b,d}\CH{a,c}$ and $\level(\mathbf{G}) \leq (j,k,4) = (j,k,l+2)$.
            \item $\lde(r[c,j]) = \lde(r[d,j]) = k$, $\lde(r[a,j])\leq k-1$, $\lde(r[b,j]) \leq k-1$. In this case, we have
            \[ \sqrt{2}^kr[c,j] \equiv \sqrt{2}^kr[d,j] \pmod{2}. \]
            Then, consider the following diagram.
            \begin{equation*}
                \begin{tikzcd}
                    \setwiretype{n} s \ar[rrrrr,to*,"\CH{a,c}\CH{b,d}\CH{a,b}\CH{b,d}\CH{a,c}"{name=a1}]\ar[d, "\CH{a,b}", swap] &&&&& r\\[-4mm]
                    \setwiretype{n} t_1 \ar[r,"\CH{a,c}",swap] & t_2 \ar[r,"\CH{b,d}",swap] & t_3 \ar[r,"\CH{a,b}"{name=a2},swap] & t_4 \ar[r,"\CH{b,d}",swap] & t_5 \ar[r,"\CH{a,c}",swap] & t_6 \ar[u,"\CH{c,d}",swap]
                    \arrow[phantom,from=a1,to=a2,"\text{\cref{prop:h0'}}"]
                \end{tikzcd}
            \end{equation*}
            The diagram commutes equationally by the postponed \cref{prop:h0'}.
            It can be seen that
            \begin{gather*}
                \level(t_1), \level(t_6) < (j,k,0), \enspace
                \level(t_2), \level(t_5) \leq (j,k,2), \enspace
                \level(t_3), \level(t_4) \leq (j,k,4).
            \end{gather*}
            Let 
            \[\mathbf{G}: s \Se*{\CH{c,d}\CH{a,c}\CH{b,d}\CH{a,b}\CH{b,d}\CH{a,c}\CH{a,b}} r.\]
            We have $\mathbf{G} \approx  \CH{a,c}\CH{b,d}\CH{a,b}\CH{b,d}\CH{a,c}$ and $\level(\mathbf{G}) \leq (j,k,4) = (j,k,l+2)$.
        \end{itemize}
        \item $l = 4$. We distinguish subcases depending on $s[a,j]$ and $s[b,j]$.
        \begin{enumerate}[leftmargin=10pt]
            \item $\sqrt{2}^ks[a,j] \equiv \sqrt{2}^ks[b,j] \pmod{2}$. Since $l = 4$, there must exist distinct $e,f \not\in\{a,b,c,d\}$ with $1\leq e,f \leq j$ such that
            \[\lde(s[e,j]) = \lde(s[f,j]) = k \text{ and } \sqrt{2}^ks[e,j] \equiv \sqrt{2}^ks[f,j] \pmod{2}.\]
            Then, consider the following diagram.
            \begin{equation*}
                \begin{tikzcd}
                    \setwiretype{n} s \ar[rrrrr,to*,"\CH{a,c}\CH{b,d}\CH{a,b}\CH{b,d}\CH{a,c}"]\ar[d, "\CH{e,f}", swap] &&&&& r\\[-4mm]
                    \setwiretype{n} s'\ar[rrrrr,to*,"\CH{a,c}\CH{b,d}\CH{a,b}\CH{b,d}\CH{a,c}"{name=a1}]\ar[rrrrr,to*, bend right=30,"\mathbf{G}'"{name=a2},swap] &&&&& r' \ar[u, "\CH{e,f}",swap]
                    \arrow[phantom,from=a1,to=a2,"\text{Case $l=2$}"]
                \end{tikzcd}
            \end{equation*}
            The rectangle at the top of the diagram commute equationally by \cref{rel:b6}.
            It can be seen that $\level(s') = \level(r') = (j,k,2)$,
            \begin{gather*}
                \lde(s'[a,j]) = \lde(s[a,j]) = k, \lde(s'[b,j]) = \lde(s[b,j]) = k, \\
                \lde(s'[c,j]) = \lde(s[c,j]) \leq k-1, \lde(s'[d,j]) = \lde(s[d,j]) \leq k-1
            \end{gather*}
            and $r' = \CH{a,c}\CH{b,d}\CH{a,b}\CH{b,d}\CH{a,c}s'$. Then, by the basic case $l=2$, there exist $\mathbf{G}'$ such that
            \begin{gather*}
                \mathbf{G}' \approx \CH{a,c}\CH{b,d}\CH{a,b}\CH{b,d}\CH{a,c}, \\
                \level(\mathbf{G}':s'\Se*{}r') \leq (j,k,4).
            \end{gather*}
            Let $\mathbf{G} = \CH{e,f}\mathbf{G}'\CH{e,f}$,
            it can be seen that
            \begin{align*}
                \mathbf{G} &= \CH{e,f}\mathbf{G}'\CH{e,f} \\
                &\approx \CH{e,f}\CH{a,c}\CH{b,d}\CH{a,b}\CH{b,d}\CH{a,c}\CH{e,f} \\
                &\approx \CH{a,c}\CH{b,d}\CH{a,b}\CH{b,d}\CH{a,c}
            \end{align*}
            and
            \begin{align*}
                \level(\mathbf{G}:s \Se*{} r) &= \max\left\{\level(s), \level(r), \level(\mathbf{G}':s'\Se*{}r')\right\} \\
                &= (j,k,4) \leq (j,k,l+2).
            \end{align*}
            \item $\sqrt{2}^ks[a,j] \not\equiv \sqrt{2}^ks[b,j] \pmod{2}$.
            Since $l=4$, there must exist distinct $e,f \not\in\{a,b,c,d\}$ with $1\leq e,f \leq j$ such that
            \begin{gather*}
                \lde(s[e,j]) = \lde(s[f,j]) = k, \\
                \sqrt{2}^ks[e,j] \equiv \sqrt{2}^ks[a,j] \pmod{2}, \enspace\enspace
                \sqrt{2}^ks[f,j] \equiv \sqrt{2}^ks[b,j] \pmod{2}.
            \end{gather*}
            Then, $(s[a,j],s[b,j],s[e,j],s[f,j])^T$ and $k$ satisfy the condition of \cref{lem:useful2}. Thus, one of the following four cases must hold:
            \begin{enumerate}[leftmargin=10pt]
                \item $\sqrt{2}^{k-1}\CH{1,4}\CH{2,3}\CH{1,2}\CH{2,3}\CH{1,4}(s[a,j],s[b,j],s[e,j],s[f,j])^T \in \ZZ[\sqrt{2}]^4$.  We have
                \[\sqrt{2}^{k-1}(\CH{a,f}\CH{b,e}\CH{a,b}\CH{b,e}\CH{a,f}s)[:,j] \in \ZZ[\sqrt{2}]^4.\]
                Then, we distinct subcases by \cref{lem:useful3}.
                \begin{itemize}[leftmargin=10pt]
                    \item $\lde(r[a,j]) = \lde(r[b,j]) = k$, $\lde(r[c,j])\leq k-1$, $\lde(r[d,j]) \leq k-1$.
                    In this case, consider the following diagram.
                    \begin{equation*}
                        \hspace{-6mm}
                        \adjustbox{max width=1.05\linewidth}
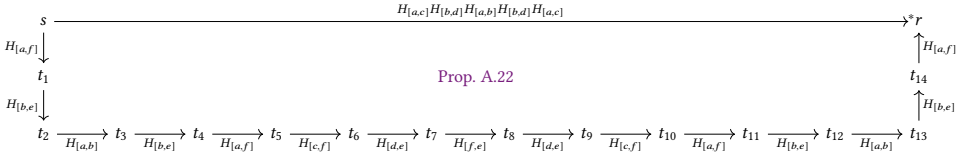
{
                        \begin{tikzcd}
                            \setwiretype{n} s \ar[rrrrrrrrrrr,to*,"\CH{a,c}\CH{b,d}\CH{a,b}\CH{b,d}\CH{a,c}"{name=a1}]\ar[d, "\CH{a,f}", swap] &&&&&&&&&&& r\\[-4mm]
                            \setwiretype{n} t_1 \ar[d, "\CH{b,e}", swap] &&&&&&&&&&& t_{14}\ar[u, "\CH{a,f}", swap] \\[-4mm]
                            \setwiretype{n} t_2 \ar[r, "\CH{a,b}", swap] & t_3 \ar[r, "\CH{b,e}", swap] & t_4\ar[r, "\CH{a,f}", swap] & t_5\ar[r,swap,"\CH{c,f}"] & t_6\ar[r, swap, "\CH{d,e}"] & t_7\ar[r, swap, "\CH{f,e}"{name=a2}] & t_8\ar[r, swap, "\CH{d,e}"] & t_9\ar[r,swap,"\CH{c,f}"] & t_{10}\ar[r, "\CH{a,f}", swap] & t_{11}\ar[r, "\CH{b,e}", swap] & t_{12}\ar[r, "\CH{a,b}", swap] & t_{13}\ar[u, "\CH{b,e}", swap]
                            \arrow[phantom,from=a1,to=a2,"\text{\cref{prop:h1}}"]
                        \end{tikzcd}}
                    \end{equation*}
                    The diagram commutes equationally by the postponed \cref{prop:h1}.
                    
                    Since $t_5 = \CH{a,f}\CH{b,e}\CH{a,b}\CH{b,e}\CH{a,f}s$, we have
                    \begin{gather*}
                        \lde(t_5[a,j]) \leq k-1, \lde(t_5[b,j]) \leq k-1, \lde(t_5[e,j]) \leq k-1, \lde(t_5[f,j]) \leq k-1
                    \end{gather*}
                    and
                    \[\lde(t_5[c,j]) = \lde(s[c,j]) \leq k-1, \enspace \lde(t_5[d,j]) = \lde(s[d,j]) \leq k-1.\]
                    Then, with \cref{lem:useful0}, we must have
                    \begingroup
                    \begin{align*}
                        \lde(t_1[a,j]) &= k, &\lde(t_1[b,j]) &= k, &\lde(t_1[e,j]) &= k, &\lde(t_1[f,j]) &= k, \\
                        \lde(t_2[a,j]) &= k, &\lde(t_2[b,j]) &= k, &\lde(t_2[e,j]) &= k, &\lde(t_2[f,j]) &= k, \\
                        \lde(t_3[a,j]) &= k, &\lde(t_3[b,j]) &= k, &\lde(t_3[e,j]) &= k, &\lde(t_3[f,j]) &= k, \\
                        \lde(t_4[a,j]) &= k, &\lde(t_4[b,j]) &\leq k-1, &\lde(t_4[e,j]) &\leq k-1, &\lde(t_4[f,j]) &= k,
                    \end{align*}
                    \begin{gather*}
                        \lde(t_4[c,j]) = \lde(t_3[c,j]) = \lde(t_2[c,j]) = \lde(t_1[c,j]) = \lde(s[c,j]) \leq k-1, \\
                        \lde(t_4[d,j]) = \lde(t_3[d,j]) = \lde(t_2[d,j]) = \lde(t_1[d,j]) = \lde(s[d,j]) \leq k-1
                    \end{gather*}
                    and
                    \begin{align*}
                        \lde(t_6[c,j]) &\leq k, &\lde(t_6[d,j]) &\leq k-1, &\lde(t_6[e,j]) &\leq k-1, &\lde(t_6[f,j]) &\leq k, \\
                        \lde(t_7[c,j]) &\leq k, &\lde(t_7[d,j]) &\leq k, &\lde(t_7[e,j]) &\leq k, &\lde(t_7[f,j]) &\leq k,
                    \end{align*}
                    \begin{gather*}
                        \lde(t_7[a,j]) = \lde(t_6[a,j]) = \lde(t_5[a,j]) \leq k-1, \\
                        \lde(t_7[b,j]) = \lde(t_6[b,j]) = \lde(t_5[b,j]) \leq k-1.
                    \end{gather*}
                    \endgroup
                    Since $\lde(t_9[d,j]) = \lde(r[d,j]) \leq k-1$ and $\lde(t_7[e,j]), \lde(t_7[f,j]) \leq k$, we have $\lde(t_8[e,j]), \lde(t_8[f,j])\leq k$; otherwise, if $\lde(t_8[e,j])$ or $\lde(t_8[f,j])$ is greater than $k$, with $t_8 = \CH{f,e}t_7$, we will have $\lde(t_8[e,j]) = \lde(t_8[f,j]) = k+1$. Then, with $\lde(t_8[d,j]) = \lde(t_7[d,j]) \leq k$ and $t_9 = \CH{d,e}t_8$, we will have $\lde(t_9[d,j]) = k+2$, which contradicts $\lde(t_9[d,j]) \leq k-1$. Thus,
                    \begin{gather*}
                        \lde(t_8[a,j]), \lde(t_8[b,j]) \leq k-1, \\
                        \lde(t_8[c,j]), \lde(t_8[d,j]), \lde(t_8[e,j]), \lde(t_8[f,j]) \leq k, \\
                        \lde(t_9[a,j]), \lde(t_9[b,j]),\lde(t_9[d,j]) \leq k-1, \\
                        \lde(t_9[c,j]), \lde(t_9[e,j]), \lde(t_9[f,j]) \leq k.
                    \end{gather*}
                    Since $\lde(t_{10}[c,j]) = \lde(r[c,j]) \leq k-1$ and $t_{10} = \CH{c,f}t_{9}$, we can also have
                    \begin{gather*}
                        \lde(t_{10}[a,j]), \lde(t_{10}[b,j]),\lde(t_{10}[c,j]), \lde(t_{10}[d,j]) \leq k-1, \\
                        \lde(t_{10}[e,j]), \lde(t_{10}[f,j]) \leq k.
                    \end{gather*}
                    With $t_{14}=\CH{a,f}r, t_{13}=\CH{b,e}\CH{a,f}r, t_{12}=\CH{a,b}\CH{b,e}\CH{a,f}r$, it can be seen that\footnote{If $\sqrt{2}^kr[a,j]\equiv \sqrt{2}^kr[e,j] \pmod{2}$, we will have $\lde(t_{13}[a,j]) = \lde(t_{13}[b,j]) = \lde(t_{13}[e,j]) = \lde( t_{13}[f,j]) = k$; if $\sqrt{2}^kr[a,j]\not\equiv \sqrt{2}^kr[e,j] \pmod{2}$, we will have $\lde(t_{13}[a,j]), \lde(t_{13}[b,j]), \lde(t_{13}[e,j]), \lde( t_{13}[f,j]) \leq k-1$.}
                    \begin{gather*}
                        \lde(t_{14}[c,j]), \lde(t_{14}[d,j]) \leq k-1, \\
                        \lde(t_{14}[a,j]), \lde(t_{14}[b,j]), \lde(t_{14}[e,j]), \lde(t_{14}[f,j]) \leq k, \\
                        \lde(t_{13}[c,j]), \lde(t_{13}[d,j]) \leq k-1, \\
                        \lde(t_{13}[a,j]), \lde(t_{13}[b,j]), \lde(t_{13}[e,j]), \lde(t_{13}[f,j]) \leq k. \\
                        \lde(t_{12}[c,j]), \lde(t_{12}[d,j]) \leq k-1, \\
                        \lde(t_{12}[a,j]), \lde(t_{12}[b,j]), \lde(t_{12}[e,j]), \lde(t_{12}[f,j]) \leq k.
                    \end{gather*}
                    Based on the values of $\lde$ of $t_{10},t_{12}$, we have $\level(t_{10}) \leq (j,k,2), \level(t_{12}) \leq (j,k,4)$. With $t_{12} = \CH{b,e} t_{11} = \CH{b,e}\CH{a,f} t_{10}$, we must have $\level(t_{11}) \leq (j, k, 4)$. Then $\level(t_{i}) \leq (j,k,4)$ for $1\leq i\leq 14$.
                    Let 
                    \[\mathbf{G}: s \Se*{\CH{a,f}\CH{b,e}\CH{a,b}\CH{b,e}\CH{a,f} \CH{c,f}\CH{d,e}\CH{f,e}\CH{d,e}\CH{c,f} \CH{a,f}\CH{b,e}\CH{a,b}\CH{b,e}\CH{a,f}} r.\]
                    We have $\mathbf{G} \approx  \CH{a,c}\CH{b,d}\CH{a,b}\CH{b,d}\CH{a,c}$ and $\level(\mathbf{G}) = (j,k,4) \leq (j,k,l+2)$.
                    \item $\lde(r[c,j]) = \lde(r[d,j]) = k$, $\lde(r[a,j])\leq k-1$, $\lde(r[b,j]) \leq k-1$. In this case, consider the following diagram.
                    \begin{equation*}
                        \hspace{-6mm}
                        \adjustbox{max width=1.05\linewidth}
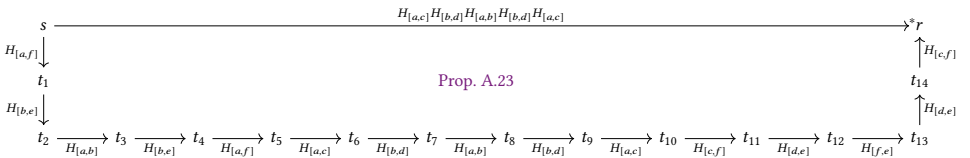
{
                        \begin{tikzcd}
                            \setwiretype{n} s \ar[rrrrrrrrrrr,to*,"\CH{a,c}\CH{b,d}\CH{a,b}\CH{b,d}\CH{a,c}"{name=a1}]\ar[d, "\CH{a,f}", swap] &&&&&&&&&&& r\\[-4mm]
                            \setwiretype{n} t_1 \ar[d, "\CH{b,e}", swap] &&&&&&&&&&& t_{14}\ar[u, "\CH{c,f}", swap] \\[-4mm]
                            \setwiretype{n} t_2 \ar[r, "\CH{a,b}", swap] & t_3 \ar[r, "\CH{b,e}", swap] & t_4\ar[r, "\CH{a,f}", swap] & t_5\ar[r,swap,"\CH{a,c}"] & t_6\ar[r, swap, "\CH{b,d}"] & t_7\ar[r, swap, "\CH{a,b}"{name=a2}] & t_8\ar[r, swap, "\CH{b,d}"] & t_9\ar[r,swap,"\CH{a,c}"] & t_{10}\ar[r, "\CH{c,f}", swap] & t_{11}\ar[r, "\CH{d,e}", swap] & t_{12}\ar[r, "\CH{f,e}", swap] & t_{13}\ar[u, "\CH{d,e}", swap]
                            \arrow[phantom,from=a1,to=a2,"\text{\cref{prop:h2}}"]
                        \end{tikzcd}}
                    \end{equation*}
                    The diagram commutes equationally by the postponed \cref{prop:h2}. 
                    Let 
                    \[\mathbf{G}: s \Se*{\CH{c,f}\CH{d,e}\CH{f,e}\CH{d,e}\CH{c,f} \CH{a,c}\CH{b,d}\CH{a,b}\CH{b,d}\CH{a,c} \CH{a,f}\CH{b,e}\CH{a,b}\CH{b,e}\CH{a,f}} r,\]
                    we can have $\mathbf{G} \approx \CH{a,c}\CH{b,d}\CH{a,b}\CH{b,d}\CH{a,c}$ and $\level(\mathbf{G}) = (j,k,4) \leq (j,k,l+2)$ by the same reasoning as the previous case.
                \end{itemize}

                \item $\sqrt{2}^{k-1}\CH{1,4}\CH{2,3}\CH{1,3}\CH{2,3}\CH{1,4}(s[a,j],s[b,j],s[e,j],s[f,j])^T \in \ZZ[\sqrt{2}]^4$. We have
                \[\sqrt{2}^{k-1}(\CH{a,f}\CH{b,e}\CH{a,e}\CH{b,e}\CH{a,f}s)[:,j] \in \ZZ[\sqrt{2}]^4.\]
                Then, we distinct subcases by \cref{lem:useful3}.
                \begin{itemize}
                    \item $\lde(r[a,j]) = \lde(r[b,j]) = k$, $\lde(r[c,j])\leq k-1$, $\lde(r[d,j]) \leq k-1$.
                    In this case, consider the following diagram.
                    \begin{equation*}
                        \hspace{-6mm}
                        \adjustbox{max width=1.05\linewidth}{
                        \begin{tikzcd}
                            \setwiretype{n} s \ar[rrrrrrrrrrr,to*,"\CH{a,c}\CH{b,d}\CH{a,b}\CH{b,d}\CH{a,c}"{name=a1}]\ar[d, "\CH{a,f}", swap] &&&&&&&&&&& r\\[-4mm]
                            \setwiretype{n} t_1 \ar[d, "\CH{b,e}", swap] &&&&&&&&&&& t_{14}\ar[u, "\CH{a,f}", swap] \\[-4mm]
                            \setwiretype{n} t_2 \ar[r, "\CH{a,e}", swap] & t_3 \ar[r, "\CH{b,e}", swap] & t_4\ar[r, "\CH{a,f}", swap] & t_5\ar[r,swap,"\CH{c,f}"] & t_6\ar[r, swap, "\CH{d,e}"] & t_7\ar[r, swap, "\CH{f,d}"{name=a2}] & t_8\ar[r, swap, "\CH{d,e}"] & t_9\ar[r,swap,"\CH{c,f}"] & t_{10}\ar[r, "\CH{a,f}", swap] & t_{11}\ar[r, "\CH{b,e}", swap] & t_{12}\ar[r, "\CH{a,e}", swap] & t_{13}\ar[u, "\CH{b,e}", swap]
                            \arrow[phantom,from=a1,to=a2,"\text{\cref{prop:h3}}"]
                        \end{tikzcd}}
                    \end{equation*}
                    The diagram commutes equationally by the postponed \cref{prop:h3}.
                    Let 
                    \[\mathbf{G}: s \Se*{\CH{a,f}\CH{b,e}\CH{a,e}\CH{b,e}\CH{a,f} \CH{c,f}\CH{d,e}\CH{f,d}\CH{d,e}\CH{c,f} \CH{a,f}\CH{b,e}\CH{a,e}\CH{b,e}\CH{a,f}} r,\]
                    we can have $\mathbf{G} \approx \CH{a,c}\CH{b,d}\CH{a,b}\CH{b,d}\CH{a,c}$ and $\level(\mathbf{G}) = (j,k,4) \leq (j,k,l+2)$ by the same reasoning as the previous case.
                    \item $\lde(r[c,j]) = \lde(r[d,j]) = k$, $\lde(r[a,j])\leq k-1$, $\lde(r[b,j]) \leq k-1$. In this case, consider the following diagram.
                    \begin{equation*}
                        \hspace{-6mm}
                        \adjustbox{max width=1.05\linewidth}{
                        \begin{tikzcd}
                            \setwiretype{n} s \ar[rrrrrrrrrrr,to*,"\CH{a,c}\CH{b,d}\CH{a,b}\CH{b,d}\CH{a,c}"{name=a1}]\ar[d, "\CH{a,f}", swap] &&&&&&&&&&& r\\[-4mm]
                            \setwiretype{n} t_1 \ar[d, "\CH{b,e}", swap] &&&&&&&&&&& t_{14}\ar[u, "\CH{c,f}", swap] \\[-4mm]
                            \setwiretype{n} t_2 \ar[r, "\CH{a,e}", swap] & t_3 \ar[r, "\CH{b,e}", swap] & t_4\ar[r, "\CH{a,f}", swap] & t_5\ar[r,swap,"\CH{a,c}"] & t_6\ar[r, swap, "\CH{b,d}"] & t_7\ar[r, swap, "\CH{a,b}"{name=a2}] & t_8\ar[r, swap, "\CH{b,d}"] & t_9\ar[r,swap,"\CH{a,c}"] & t_{10}\ar[r, "\CH{c,f}", swap] & t_{11}\ar[r, "\CH{d,e}", swap] & t_{12}\ar[r, "\CH{f,d}", swap] & t_{13}\ar[u, "\CH{d,e}", swap]
                            \arrow[phantom,from=a1,to=a2,"\text{\cref{prop:h4}}"]
                        \end{tikzcd}}
                    \end{equation*}
                    The diagram commutes equationally by the postponed \cref{prop:h4}. 
                    Let 
                    \[\mathbf{G}: s \Se*{\CH{c,f}\CH{d,e}\CH{f,d}\CH{d,e}\CH{c,f} \CH{a,c}\CH{b,d}\CH{a,b}\CH{b,d}\CH{a,c} \CH{a,f}\CH{b,e}\CH{a,e}\CH{b,e}\CH{a,f}} r,\]
                    we can have $\mathbf{G} \approx \CH{a,c}\CH{b,d}\CH{a,b}\CH{b,d}\CH{a,c}$ and $\level(\mathbf{G}) = (j,k,4) \leq (j,k,l+2)$ by the same reasoning as the previous case.
                \end{itemize}
                \item $\sqrt{2}^{k-1}\CH{1,4}\CH{2,3}\CH{2,4}\CH{2,3}\CH{1,4}(s[a,j],s[b,j],s[e,j],s[f,j])^T \in \ZZ[\sqrt{2}]^4$. In this case, we apply \cref{lem:useful2_1} to get
                \[ \sqrt{2}^{k-1}\CH{1,3}\CH{2,4}\CH{1,2}\CH{2,4}\CH{1,3}(s[a,j],s[b,j],s[e,j],s[f,j])^T \in \ZZ[\sqrt{2}]^4, \]
                which means
                \[\sqrt{2}^{k-1}(\CH{a,e}\CH{b,f}\CH{a,b}\CH{b,f}\CH{a,e}s)[:,j] \in \ZZ[\sqrt{2}]^4.\]
                Then, we distinct subcases by \cref{lem:useful3}.
                \begin{itemize}
                    \item $\lde(r[a,j]) = \lde(r[b,j]) = k$, $\lde(r[c,j])\leq k-1$, $\lde(r[d,j]) \leq k-1$.
                    In this case, consider the following diagram.
                    \begin{equation*}
                        \hspace{-6mm}
                        \adjustbox{max width=1.05\linewidth}{
                        \begin{tikzcd}
                            \setwiretype{n} s \ar[rrrrrrrrrrr,to*,"\CH{a,c}\CH{b,d}\CH{a,b}\CH{b,d}\CH{a,c}"{name=a1}]\ar[d, "\CH{a,e}", swap] &&&&&&&&&&& r\\[-4mm]
                            \setwiretype{n} t_1 \ar[d, "\CH{b,f}", swap] &&&&&&&&&&& t_{14}\ar[u, "\CH{a,e}", swap] \\[-4mm]
                            \setwiretype{n} t_2 \ar[r, "\CH{a,b}", swap] & t_3 \ar[r, "\CH{b,f}", swap] & t_4\ar[r, "\CH{a,e}", swap] & t_5\ar[r,swap,"\CH{c,e}"] & t_6\ar[r, swap, "\CH{d,f}"] & t_7\ar[r, swap, "\CH{e,f}"{name=a2}] & t_8\ar[r, swap, "\CH{d,f}"] & t_9\ar[r,swap,"\CH{c,e}"] & t_{10}\ar[r, "\CH{a,e}", swap] & t_{11}\ar[r, "\CH{b,f}", swap] & t_{12}\ar[r, "\CH{a,b}", swap] & t_{13}\ar[u, "\CH{b,f}", swap]
                            \arrow[phantom,from=a1,to=a2,"\text{\cref{prop:h5}}"]
                    \end{tikzcd}}
                    \end{equation*}
                    The diagram commutes equationally by the postponed \cref{prop:h5}.
                    Let 
                    \begin{align*}
                        \mathbf{G}&: s \Se*{\CH{a,e}\CH{b,f}\CH{a,b}\CH{b,f}\CH{a,e} \CH{c,e}\CH{d,f}\CH{e,f}\CH{d,f}\CH{c,e} \CH{a,e}\CH{b,f}\CH{a,b}\CH{b,f}\CH{a,e}} r,
                    \end{align*}
                    we can have $\mathbf{G} \approx \CH{a,c}\CH{b,d}\CH{a,b}\CH{b,d}\CH{a,c}$ and $\level(\mathbf{G}) = (j,k,4) \leq (j,k,l+2)$ by the same reasoning as the previous case.
                    \item $\lde(r[c,j]) = \lde(r[d,j]) = k$, $\lde(r[a,j])\leq k-1$, $\lde(r[b,j]) \leq k-1$. In this case, consider the following diagram.
                    \begin{equation*}
                        \hspace{-6mm}
                        \adjustbox{max width=1.05\linewidth}{
                        \begin{tikzcd}
                            \setwiretype{n} s \ar[rrrrrrrrrrr,to*,"\CH{a,c}\CH{b,d}\CH{a,b}\CH{b,d}\CH{a,c}"{name=a1}]\ar[d, "\CH{a,e}", swap] &&&&&&&&&&& r\\[-4mm]
                            \setwiretype{n} t_1 \ar[d, "\CH{b,f}", swap] &&&&&&&&&&& t_{14}\ar[u, "\CH{c,e}", swap] \\[-4mm]
                            \setwiretype{n} t_2 \ar[r, "\CH{a,b}", swap] & t_3 \ar[r, "\CH{b,f}", swap] & t_4\ar[r, "\CH{a,e}", swap] & t_5\ar[r,swap,"\CH{a,c}"] & t_6\ar[r, swap, "\CH{b,d}"] & t_7\ar[r, swap, "\CH{a,b}"{name=a2}] & t_8\ar[r, swap, "\CH{b,d}"] & t_9\ar[r,swap,"\CH{a,c}"] & t_{10}\ar[r, "\CH{c,e}", swap] & t_{11}\ar[r, "\CH{d,f}", swap] & t_{12}\ar[r, "\CH{e,f}", swap] & t_{13}\ar[u, "\CH{d,f}", swap]
                            \arrow[phantom,from=a1,to=a2,"\text{\cref{prop:h6}}"]
                        \end{tikzcd}}
                    \end{equation*}
                    The diagram commutes equationally by the postponed \cref{prop:h6}. 
                    Let 
                    \[\mathbf{G}: s \Se*{\CH{c,e}\CH{d,f}\CH{e,f}\CH{d,f}\CH{c,e} \CH{a,c}\CH{b,d}\CH{a,b}\CH{b,d}\CH{a,c} \CH{a,e}\CH{b,f}\CH{a,b}\CH{b,f}\CH{a,e}} r,\]
                    we can have $\mathbf{G} \approx \CH{a,c}\CH{b,d}\CH{a,b}\CH{b,d}\CH{a,c}$ and $\level(\mathbf{G}) = (j,k,4) \leq (j,k,l+2)$ by the same reasoning as the previous case.
                \end{itemize}
                \item $\sqrt{2}^{k-1}\CH{1,4}\CH{2,3}\CH{3,4}\CH{2,3}\CH{1,4}(s[a,j],s[b,j],s[e,j],s[f,j])^T \in \ZZ[\sqrt{2}]^4$. In this case, we apply \cref{lem:useful2_2} to get
                \[ \sqrt{2}^{k-1}\CH{1,3}\CH{2,4}\CH{1,4}\CH{2,4}\CH{1,3}(s[a,j],s[b,j],s[e,j],s[f,j])^T \in \ZZ[\sqrt{2}]^4, \]
                which means
                \[\sqrt{2}^{k-1}(\CH{a,e}\CH{b,f}\CH{a,f}\CH{b,f}\CH{a,e}s)[:,j] \in \ZZ[\sqrt{2}]^4.\]
                Then, we distinct subcases by \cref{lem:useful3}.
                \begin{itemize}
                    \item $\lde(r[a,j]) = \lde(r[b,j]) = k$, $\lde(r[c,j])\leq k-1$, $\lde(r[d,j]) \leq k-1$.
                    In this case, consider the following diagram.
                    \begin{equation*}
                        \hspace{-6mm}
                        \adjustbox{max width=1.05\linewidth}{
                        \begin{tikzcd}
                            \setwiretype{n} s \ar[rrrrrrrrrrr,to*,"\CH{a,c}\CH{b,d}\CH{a,b}\CH{b,d}\CH{a,c}"{name=a1}]\ar[d, "\CH{a,e}", swap] &&&&&&&&&&& r\\[-4mm]
                            \setwiretype{n} t_1 \ar[d, "\CH{b,f}", swap] &&&&&&&&&&& t_{14}\ar[u, "\CH{a,e}", swap] \\[-4mm]
                            \setwiretype{n} t_2 \ar[r, "\CH{a,f}", swap] & t_3 \ar[r, "\CH{b,f}", swap] & t_4\ar[r, "\CH{a,e}", swap] & t_5\ar[r,swap,"\CH{c,e}"] & t_6\ar[r, swap, "\CH{d,f}"] & t_7\ar[r, swap, "\CH{e,d}"{name=a2}] & t_8\ar[r, swap, "\CH{d,f}"] & t_9\ar[r,swap,"\CH{c,e}"] & t_{10}\ar[r, "\CH{a,e}", swap] & t_{11}\ar[r, "\CH{b,f}", swap] & t_{12}\ar[r, "\CH{a,f}", swap] & t_{13}\ar[u, "\CH{b,f}", swap]
                            \arrow[phantom,from=a1,to=a2,"\text{\cref{prop:h7}}"]
                        \end{tikzcd}}
                    \end{equation*}
                    The diagram commutes equationally by the postponed \cref{prop:h7}.
                    Let 
                    \begin{align*}
                        \mathbf{G}&: s \Se*{\CH{a,e}\CH{b,f}\CH{a,f}\CH{b,f}\CH{a,e}\CH{c,e}\CH{d,f}\CH{e,d}\CH{d,f}\CH{c,e}\CH{a,e}\CH{b,f}\CH{a,f}\CH{b,f}\CH{a,e}} r,
                    \end{align*}
                    we can have $\mathbf{G} \approx \CH{a,c}\CH{b,d}\CH{a,b}\CH{b,d}\CH{a,c}$ and $\level(\mathbf{G}) = (j,k,4) \leq (j,k,l+2)$ by the same reasoning as the previous case.
                    \item $\lde(r[c,j]) = \lde(r[d,j]) = k$, $\lde(r[a,j])\leq k-1$, $\lde(r[b,j]) \leq k-1$. In this case, consider the following diagram.
                    \begin{equation*}
                        \hspace{-6mm}
                        \adjustbox{max width=1.05\linewidth}{
                        \begin{tikzcd}
                            \setwiretype{n} s \ar[rrrrrrrrrrr,to*,"\CH{a,c}\CH{b,d}\CH{a,b}\CH{b,d}\CH{a,c}"{name=a1}]\ar[d, "\CH{a,e}", swap] &&&&&&&&&&& r\\[-4mm]
                            \setwiretype{n} t_1 \ar[d, "\CH{b,f}", swap] &&&&&&&&&&& t_{14}\ar[u, "\CH{c,e}", swap] \\[-4mm]
                            \setwiretype{n} t_2 \ar[r, "\CH{a,f}", swap] & t_3 \ar[r, "\CH{b,f}", swap] & t_4\ar[r, "\CH{a,e}", swap] & t_5\ar[r,swap,"\CH{a,c}"] & t_6\ar[r, swap, "\CH{b,d}"] & t_7\ar[r, swap, "\CH{a,b}"{name=a2}] & t_8\ar[r, swap, "\CH{b,d}"] & t_9\ar[r,swap,"\CH{a,c}"] & t_{10}\ar[r, "\CH{c,e}", swap] & t_{11}\ar[r, "\CH{d,f}", swap] & t_{12}\ar[r, "\CH{e,d}", swap] & t_{13}\ar[u, "\CH{d,f}", swap]
                            \arrow[phantom,from=a1,to=a2,"\text{\cref{prop:h8}}"]
                        \end{tikzcd}}
                    \end{equation*}
                    The diagram commutes equationally by the postponed \cref{prop:h8}. 
                    Let 
                    \[\mathbf{G}: s \Se*{\CH{c,e}\CH{d,f}\CH{e,d}\CH{d,f}\CH{c,e} \CH{a,c}\CH{b,d}\CH{a,b}\CH{b,d}\CH{a,c} \CH{a,e}\CH{b,f}\CH{a,f}\CH{b,f}\CH{a,e}} r,\]
                    we can have $\mathbf{G} \approx \CH{a,c}\CH{b,d}\CH{a,b}\CH{b,d}\CH{a,c}$ and $\level(\mathbf{G}) = (j,k,4) \leq (j,k,l+2)$ by the same reasoning as the previous case.
                \end{itemize}
            \end{enumerate}
        \end{enumerate}
        \item $l = 2(m+1)$ with $m \geq 2$. Since $l \geq 6$, there must exist distinct $e,f \not\in\{a,b,c,d\}$ with $1\leq e,f \leq j$ such that
        \[\lde(s[e,j]) = \lde(s[f,j]) = k \text{ and } \sqrt{2}^ks[e,j] \equiv \sqrt{2}^ks[f,j] \pmod{2}.\]
        Then, consider the following diagram.
        
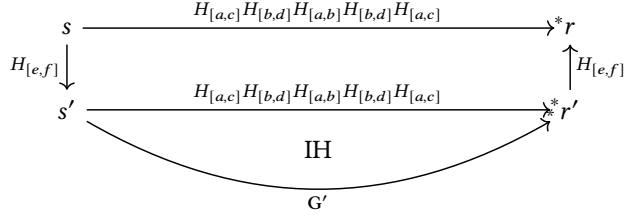
\begin{equation*}
            \begin{tikzcd}
                \setwiretype{n} s \ar[rrrrr,to*,"\CH{a,c}\CH{b,d}\CH{a,b}\CH{b,d}\CH{a,c}"]\ar[d, "\CH{e,f}", swap] &&&&& r\\[-4mm]
                \setwiretype{n} s'\ar[rrrrr,to*,"\CH{a,c}\CH{b,d}\CH{a,b}\CH{b,d}\CH{a,c}"{name=a1}]\ar[rrrrr,to*, bend right=30,"\mathbf{G}'"{name=a2},swap] &&&&& r' \ar[u, "\CH{e,f}",swap]
                \arrow[phantom,from=a1,to=a2,"\text{IH}"]
            \end{tikzcd}
        \end{equation*}
        The rectangle at the top of the diagram commute equationally by \cref{rel:b6}.
        It can be seen that $\level(s') = \level(r') = (j,k,2m)$,
        \begin{gather*}
            \lde(s'[a,j]) = \lde(s[a,j]) = k, \lde(s'[b,j]) = \lde(s[b,j]) = k, \\
            \lde(s'[c,j]) = \lde(s[c,j]) \leq k-1, \lde(s'[d,j]) = \lde(s[d,j]) \leq k-1
        \end{gather*}
        and $r' = \CH{a,c}\CH{b,d}\CH{a,b}\CH{b,d}\CH{a,c}s'$. Then, by inductive hypothesis of $m$, there exist $\mathbf{G}'$ such that
        \begin{gather*}
            \mathbf{G}' \approx \CH{a,c}\CH{b,d}\CH{a,b}\CH{b,d}\CH{a,c}, \\
            \level(\mathbf{G}':s'\Se*{}r') \leq (j,k,2m+2).
        \end{gather*}
        Let $\mathbf{G} = \CH{e,f}\mathbf{G}'\CH{e,f}$,
        it can be seen that
        \begin{align*}
            \mathbf{G} &= \CH{e,f}\mathbf{G}'\CH{e,f} \\
            &\approx \CH{e,f}\CH{a,c}\CH{b,d}\CH{a,b}\CH{b,d}\CH{a,c}\CH{e,f} \\
            &\approx \CH{a,c}\CH{b,d}\CH{a,b}\CH{b,d}\CH{a,c}
        \end{align*}
        and
        \begin{align*}
            \level(\mathbf{G}:s \Se*{} r) &= \max\left\{\level(s), \level(r), \level(\mathbf{G}':s'\Se*{}r')\right\} \\
            &= (j,k,2(m+1)) \leq (j,k,l+2).
        \end{align*} \qedhere
    \end{itemize}
\end{proof}

\begin{proposition}
    \label{prop:h0}
    \begin{equation*}
        \CH{a,b}\CH{a,c}\CH{b,d}\CH{c,d}\CH{b,d}\CH{a,c}\CH{a,b} \approx \CH{a,c}\CH{b,d}\CH{a,b}\CH{b,d}\CH{a,c}.
    \end{equation*}
\end{proposition}
\begin{proof}
    \begin{align*}
        & \CH{a,b}\CH{a,c}\CH{b,d}\CH{c,d}\CH{b,d}\CH{a,c}\CH{a,b} \approx \CH{a,c}\CH{b,d}\CH{a,b}\CH{b,d}\CH{a,c} \\
        \Leftrightarrow{}& \CH{c,d} \approx \blue{\CH{b,d}\CH{a,c}\CH{a,b}}\CH{a,c}\CH{b,d}\CH{a,b}\CH{b,d}\CH{a,c}\blue{\CH{a,b}\CH{a,c}\CH{b,d}} \tag{by \cref{rel:a3}} \\
        \Leftrightarrow{}& \CH{c,d} \approx \CH{b,d}\CH{a,c}\blue{\CH{c,d}\CH{a,b}\CH{c,d}}\CH{a,c}\CH{b,d}\CH{a,b}\CH{b,d}\CH{a,c}\CH{a,b}\CH{a,c}\CH{b,d} \tag{by \cref{rel:a3,rel:b6}} \\
        \Leftrightarrow{}& \CH{c,d} \approx \CH{b,d}\CH{a,c}\CH{c,d}\blue{\CH{a,c}\CH{b,d}\CH{a,b}\CH{c,d}}\CH{a,b}\CH{b,d}\CH{a,c}\CH{a,b}\CH{a,c}\CH{b,d} \tag{by \cref{rel:a3,rel:b6,,rel:f1}} \\
        \Leftrightarrow{}& \CH{c,d} \approx \CH{b,d}\CH{a,c}\CH{c,d}\CH{a,c}\CH{b,d}\blue{\CH{c,d}}\CH{b,d}\CH{a,c}\CH{a,b}\CH{a,c}\CH{b,d} \tag{by \cref{rel:a3,rel:b6}} \\
        \Leftrightarrow{}& \blue{\varepsilon} \approx \blue{\CH{c,d}}\CH{b,d}\CH{a,c}\CH{c,d}\CH{a,c}\CH{b,d}\CH{c,d}\CH{b,d}\CH{a,c}\CH{a,b}\CH{a,c}\CH{b,d} \tag{by \cref{rel:a3}} \\
        \Leftrightarrow{}& \varepsilon \approx \blue{\rbra*{\CH{c,d}\CH{a,c}\CH{b,d}}^3}\CH{a,b}\CH{a,c}\CH{b,d} \tag{by \cref{rel:b6}} \\
        \Leftrightarrow{}& \varepsilon \approx \blue{\rbra*{\CH{c,d}\CH{a,c}\CH{b,d}}^3}\blue{\CH{c,d}\CH{a,b}\CH{c,d}}\CH{a,c}\CH{b,d} \tag{by \cref{rel:a3,rel:b6}} \\
        \Leftrightarrow{}& \varepsilon \approx \rbra*{\CH{c,d}\CH{a,c}\CH{b,d}}^3\CH{c,d}\blue{\CH{a,c}\CH{b,d}\CH{a,b}\CH{c,d}} \tag{by \cref{rel:a3,rel:b6,,rel:f1}} \\
        \Leftrightarrow{}& \varepsilon \approx \blue{\rbra*{\CH{c,d}\CH{a,c}\CH{b,d}}^4}\CH{a,b}\CH{c,d} \\
        \Leftrightarrow{}& \varepsilon \approx \blue{\CH{a,b}\CH{c,d}}\CH{a,b}\CH{c,d} \tag{by \cref{rel:d3}}\\
        \Leftrightarrow{}& \varepsilon \approx \blue{\varepsilon}. \tag{by \cref{rel:a3,rel:b6}}
    \end{align*}
\end{proof}

\begin{proposition}
    \label{prop:h0'}
    \begin{equation*}
        \CH{c,d}\CH{a,c}\CH{b,d}\CH{a,b}\CH{b,d}\CH{a,c}\CH{a,b} \approx \CH{a,c}\CH{b,d}\CH{a,b}\CH{b,d}\CH{a,c}.
    \end{equation*}
\end{proposition}
\begin{proof}
    \begin{align*}
        & \CH{c,d}\CH{a,c}\CH{b,d}\CH{a,b}\CH{b,d}\CH{a,c}\CH{a,b} \\
        \approx{}& \blue{\CH{a,b}\CH{a,b}\CH{c,d}}\CH{a,c}\CH{b,d}\CH{a,b}\CH{b,d}\CH{a,c}\CH{a,b} \tag{by \cref{rel:a3,rel:b6}} \\
        \approx{}& \CH{a,b}\blue{\CH{a,c}\CH{b,d}\CH{a,b}\CH{c,d}}\CH{a,b}\CH{b,d}\CH{a,c}\CH{a,b} \tag{by \cref{rel:a3,rel:b6,,rel:f1}} \\
        \approx{}& \CH{a,b}\CH{a,c}\CH{b,d}\blue{\CH{c,d}}\CH{b,d}\CH{a,c}\CH{a,b} \tag{by \cref{rel:a3,rel:b6}} \\
        \approx{}& \blue{\CH{a,c}\CH{b,d}\CH{a,b}\CH{b,d}\CH{a,c}}. \tag{by \cref{prop:h0}}
    \end{align*}
\end{proof}

\begin{proposition}
    \label{prop:h1}
    \begin{gather*}
        \CH{a,f}\CH{b,e}\CH{a,b}\CH{b,e}\CH{a,f} \CH{c,f}\CH{d,e}\CH{f,e}\CH{d,e}\CH{c,f} \CH{a,f}\CH{b,e}\CH{a,b}\CH{b,e}\CH{a,f} \\
        \approx{} \CH{a,c}\CH{b,d}\CH{a,b}\CH{b,d}\CH{a,c}.
    \end{gather*}
\end{proposition}
\begin{proof}
    \begin{align*}
        & \CH{a,f}\CH{b,e}\CH{a,b}\CH{b,e}\CH{a,f} \CH{c,f}\CH{d,e}\CH{f,e}\CH{d,e}\CH{c,f} \CH{a,f}\CH{b,e}\CH{a,b}\CH{b,e}\CH{a,f} \\
        \approx{}& \blue{\CX{e,f}\CX{e,f}}\CH{a,f}\CH{b,e}\CH{a,b}\CH{b,e}\CH{a,f} \CH{c,f}\CH{d,e}\CH{f,e}\CH{d,e}\CH{c,f} \\
        & \quad \CH{a,f}\CH{b,e}\CH{a,b}\CH{b,e}\CH{a,f} \tag{by \cref{rel:a2}} \\
        \approx{}& \CX{e,f}\blue{\CH{a,e}\CH{b,f}\CH{a,b}\CH{b,f}\CH{a,e} \CH{c,e}\CH{d,f}\CH{e,f}\CH{d,f}\CH{c,e}} \\
        & \quad \blue{\CH{a,e}\CH{b,f}\CH{a,b}\CH{b,f}\CH{a,e}\CX{e,f}} \tag{by \cref{rel:c5,rel:b5,,rel:e2}} \\
        \approx{}& \CX{e,f}\blue{\CH{a,c}\CH{b,d}\CH{a,b}\CH{b,d}\CH{a,c}}\CX{e,f} \tag{by \cref{prop:h5}} \\
        \approx{}& \blue{\CH{a,c}\CH{b,d}\CH{a,b}\CH{b,d}\CH{a,c}}. \tag{by \cref{rel:a2,rel:b5}}
    \end{align*}
\end{proof}

\begin{proposition}
    \label{prop:h2}
    \begin{gather*}
        \CH{c,f}\CH{d,e}\CH{f,e}\CH{d,e}\CH{c,f} \CH{a,c}\CH{b,d}\CH{a,b}\CH{b,d}\CH{a,c} \CH{a,f}\CH{b,e}\CH{a,b}\CH{b,e}\CH{a,f} \\
        \approx{} \CH{a,c}\CH{b,d}\CH{a,b}\CH{b,d}\CH{a,c}.
    \end{gather*}
\end{proposition}
\begin{proof}
    \begin{align*}
        & \CH{c,f}\CH{d,e}\CH{f,e}\CH{d,e}\CH{c,f} \CH{a,c}\CH{b,d}\CH{a,b}\CH{b,d}\CH{a,c} \CH{a,f}\CH{b,e}\CH{a,b}\CH{b,e}\CH{a,f} \\
        \approx{}& \blue{\CX{e,f}\CX{e,f}}\CH{c,f}\CH{d,e}\CH{f,e}\CH{d,e}\CH{c,f} \CH{a,c}\CH{b,d}\CH{a,b}\CH{b,d}\CH{a,c} \\
        & \quad \CH{a,f}\CH{b,e}\CH{a,b}\CH{b,e}\CH{a,f} \tag{by \cref{rel:a2}} \\
        \approx{}& \CX{e,f}\blue{\CH{c,e}\CH{d,f}\CH{e,f}\CH{d,f}\CH{c,e} \CH{a,c}\CH{b,d}\CH{a,b}\CH{b,d}\CH{a,c}} \\
        & \quad \blue{\CH{a,e}\CH{b,f}\CH{a,b}\CH{b,f}\CH{a,e}\CX{e,f}} \tag{by \cref{rel:c5,rel:b5,,rel:e2}} \\
        \approx{}& \CX{e,f}\blue{\CH{a,c}\CH{b,d}\CH{a,b}\CH{b,d}\CH{a,c}}\CX{e,f} \tag{by \cref{prop:h6}} \\
        \approx{}& \blue{\CH{a,c}\CH{b,d}\CH{a,b}\CH{b,d}\CH{a,c}}. \tag{by \cref{rel:a2,rel:b5}}
    \end{align*}
\end{proof}

\begin{proposition}
    \label{prop:h3}
    \begin{gather*}
        \CH{a,f}\CH{b,e}\CH{a,e}\CH{b,e}\CH{a,f} \CH{c,f}\CH{d,e}\CH{f,d}\CH{d,e}\CH{c,f} \CH{a,f}\CH{b,e}\CH{a,e}\CH{b,e}\CH{a,f} \\
        \approx{} \CH{a,c}\CH{b,d}\CH{a,b}\CH{b,d}\CH{a,c}.
    \end{gather*}
\end{proposition}
\begin{proof}
    \begin{align*}
        & \CH{a,f}\CH{b,e}\CH{a,e}\CH{b,e}\CH{a,f} \CH{c,f}\CH{d,e}\CH{f,d}\CH{d,e}\CH{c,f} \CH{a,f}\CH{b,e}\CH{a,e}\CH{b,e}\CH{a,f} \\
        \approx{}& \CH{a,f}\CH{b,e}\blue{\CX{b,e}\CX{b,e}}\CH{a,e}\CH{b,e}\CH{a,f} \CH{c,f}\CH{d,e}\CH{f,d}\CH{d,e}\CH{c,f} \\
        & \quad \CH{a,f}\CH{b,e}\CH{a,e}\CH{b,e}\CH{a,f} \tag{by \cref{rel:a2}} \\
        \approx{}& \CH{a,f}\CH{b,e}\CX{b,e}\blue{\CH{a,b}\CX{b,e}}\CH{b,e}\CH{a,f} \CH{c,f}\CH{d,e}\CH{f,d}\CH{d,e}\CH{c,f} \\
        & \quad \CH{a,f}\CH{b,e}\CH{a,e}\CH{b,e}\CH{a,f} \tag{by \cref{rel:c5}} \\
        \approx{}& \CH{a,f}\CH{b,e}\CX{b,e}\CH{a,b}\blue{\CH{b,e}\CZ{e}}\CH{a,f} \CH{c,f}\CH{d,e}\CH{f,d}\CH{d,e}\CH{c,f} \\
        & \quad \CH{a,f}\CH{b,e}\CH{a,e}\CH{b,e}\CH{a,f} \tag{by \cref{rel:d2}} \\
        \approx{}& \CH{a,f}\CH{b,e}\CX{b,e}\CH{a,b}\CH{b,e}\CH{a,f} \CH{c,f}\blue{\CH{d,e}\CX{d,e}}\CH{f,d}\CH{d,e}\CH{c,f} \\
        & \quad \CH{a,f}\CH{b,e}\CH{a,e}\CH{b,e}\CH{a,f} \tag{by \cref{rel:b4,rel:d2}} \\
        \approx{}& \CH{a,f}\CH{b,e}\CX{b,e}\CH{a,b}\CH{b,e}\CH{a,f} \CH{c,f}\CH{d,e}\blue{\CH{f,e}\CX{d,e}}\CH{d,e}\CH{c,f} \\
        & \quad \CH{a,f}\CH{b,e}\CH{a,e}\CH{b,e}\CH{a,f} \tag{by \cref{rel:c5}} \\
        \approx{}& \CH{a,f}\CH{b,e}\CX{b,e}\CH{a,b}\CH{b,e}\CH{a,f} \CH{c,f}\CH{d,e}\CH{f,e}\blue{\CH{d,e}\CZ{e}}\CH{c,f} \\
        & \quad \CH{a,f}\CH{b,e}\CH{a,e}\CH{b,e}\CH{a,f} \tag{by \cref{rel:d2}} \\
        \approx{}& \CH{a,f}\CH{b,e}\CX{b,e}\CH{a,b}\CH{b,e}\CH{a,f} \CH{c,f}\CH{d,e}\CH{f,e}\CH{d,e}\blue{\CH{c,f}} \\
        & \quad \blue{\CH{a,f}\CH{b,e}\CX{b,e}}\CH{a,e}\CH{b,e}\CH{a,f} \tag{by \cref{rel:b4,rel:d2}} \\
        \approx{}& \CH{a,f}\CH{b,e}\CX{b,e}\CH{a,b}\CH{b,e}\CH{a,f} \CH{c,f}\CH{d,e}\CH{f,e}\CH{d,e}{\CH{c,f}} \\
        & \quad\CH{a,f}\CH{b,e} \blue{\CH{a,b}\CX{b,e}}\CH{b,e}\CH{a,f} \tag{by \cref{rel:c5}} \\
        \approx{}& \CH{a,f}\blue{\CZ{e}\CH{b,e}}\CH{a,b}\CH{b,e}\CH{a,f} \CH{c,f}\CH{d,e}\CH{f,e}\CH{d,e}{\CH{c,f}} \\
        & \quad\CH{a,f}\CH{b,e} \CH{a,b}\blue{\CH{b,e}\CZ{e}}\CH{a,f} \tag{by \cref{rel:d2}} \\
        \approx{}& \blue{\CZ{e}\CH{a,f}}\CH{b,e}\CH{a,b}\CH{b,e}\CH{a,f} \CH{c,f}\CH{d,e}\CH{f,e}\CH{d,e}{\CH{c,f}} \\
        & \quad\CH{a,f}\CH{b,e} \CH{a,b}\CH{b,e}\blue{\CH{a,f}\CZ{e}} \tag{by \cref{rel:b4}} \\
        \approx{}& \CZ{e}\blue{\CH{a,c}\CH{b,d}\CH{a,b}\CH{b,d}\CH{a,c}}\CZ{e} \tag{by \cref{prop:h1}} \\
        \approx{}& \blue{\CH{a,c}\CH{b,d}\CH{a,b}\CH{b,d}\CH{a,c}}. \tag{by \cref{rel:a1,rel:b4}}
    \end{align*}
\end{proof}

\begin{proposition}
    \label{prop:h4}
    \begin{gather*}
        \CH{c,f}\CH{d,e}\CH{f,d}\CH{d,e}\CH{c,f} \CH{a,c}\CH{b,d}\CH{a,b}\CH{b,d}\CH{a,c} \CH{a,f}\CH{b,e}\CH{a,e}\CH{b,e}\CH{a,f} \\
        \approx{} \CH{a,c}\CH{b,d}\CH{a,b}\CH{b,d}\CH{a,c}.
    \end{gather*}
\end{proposition}
\begin{proof}
    \begin{align*}
        & \CH{c,f}\CH{d,e}\CH{f,d}\CH{d,e}\CH{c,f} \CH{a,c}\CH{b,d}\CH{a,b}\CH{b,d}\CH{a,c} \CH{a,f}\CH{b,e}\CH{a,e}\CH{b,e}\CH{a,f} \\
        \approx{}& \CH{c,f}\CH{d,e}\blue{\CX{d,e}\CX{d,e}}\CH{f,d}\CH{d,e}\CH{c,f} \CH{a,c}\CH{b,d}\CH{a,b}\CH{b,d}\CH{a,c} \\
        & \quad \CH{a,f}\CH{b,e}\CH{a,e}\CH{b,e}\CH{a,f} \tag{by \cref{rel:a2}} \\
        \approx{}& \CH{c,f}\CH{d,e}\CX{d,e}\blue{\CH{f,e}\CX{d,e}}\CH{d,e}\CH{c,f} \CH{a,c}\CH{b,d}\CH{a,b}\CH{b,d}\CH{a,c} \\
        & \quad \CH{a,f}\CH{b,e}\CH{a,e}\CH{b,e}\CH{a,f} \tag{by \cref{rel:c5}} \\
        \approx{}& \CH{c,f}\CH{d,e}\CX{d,e}\CH{f,e}\blue{\CH{d,e}\CZ{e}}\CH{c,f} \CH{a,c}\CH{b,d}\CH{a,b}\CH{b,d}\CH{a,c} \\
        & \quad \CH{a,f}\CH{b,e}\CH{a,e}\CH{b,e}\CH{a,f} \tag{by \cref{rel:d2}} \\
        \approx{}& \CH{c,f}\CH{d,e}\CX{d,e}\CH{f,e}\CH{d,e}\blue{\CH{c,f} \CH{a,c}\CH{b,d}\CH{a,b}\CH{b,d}\CH{a,c}} \\
        & \quad \blue{\CH{a,f}\CZ{e}}\CH{b,e}\CH{a,e}\CH{b,e}\CH{a,f} \tag{by \cref{rel:b4}} \\
        \approx{}& \CH{c,f}\CH{d,e}\CX{d,e}\CH{f,e}\CH{d,e}{\CH{c,f} \CH{a,c}\CH{b,d}\CH{a,b}\CH{b,d}\CH{a,c}} \\
        & \quad\CH{a,f} \blue{\CH{b,e}\CX{b,e}}\CH{a,e}\CH{b,e}\CH{a,f} \tag{by \cref{rel:d2}} \\
        \approx{}& \CH{c,f}\CH{d,e}\CX{d,e}\CH{f,e}\CH{d,e}{\CH{c,f} \CH{a,c}\CH{b,d}\CH{a,b}\CH{b,d}\CH{a,c}} \\
        & \quad\CH{a,f} \CH{b,e}\blue{\CH{a,b}\CX{b,e}}\CH{b,e}\CH{a,f} \tag{by \cref{rel:c5}} \\
        \approx{}& \CH{c,f}\blue{\CZ{e}\CH{d,e}}\CH{f,e}\CH{d,e}{\CH{c,f} \CH{a,c}\CH{b,d}\CH{a,b}\CH{b,d}\CH{a,c}} \\
        & \quad\CH{a,f} \CH{b,e}\CH{a,b}\blue{\CH{b,e}\CZ{e}}\CH{a,f} \tag{by \cref{rel:d2}} \\
        \approx{}& \blue{\CZ{e}}\CH{c,f}\CH{d,e}\CH{f,e}\CH{d,e}{\CH{c,f} \CH{a,c}\CH{b,d}\CH{a,b}\CH{b,d}\CH{a,c}} \\
        & \quad\CH{a,f} \CH{b,e}\CH{a,b}\CH{b,e}\blue{\CH{a,f}\CZ{e}} \tag{by \cref{rel:b4}} \\
        \approx{}& \CZ{e}\blue{\CH{a,c}\CH{b,d}\CH{a,b}\CH{b,d}\CH{a,c}} \CZ{e} \tag{by \cref{prop:h2}} \\
        \approx{}& \blue{\CH{a,c}\CH{b,d}\CH{a,b}\CH{b,d}\CH{a,c}}. \tag{by \cref{rel:a1,rel:b4}}
    \end{align*}
\end{proof}

\begin{proposition}
    \label{prop:h5}
    \begin{gather*}
        \CH{a,e}\CH{b,f}\CH{a,b}\CH{b,f}\CH{a,e} \CH{c,e}\CH{d,f}\CH{e,f}\CH{d,f}\CH{c,e} \CH{a,e}\CH{b,f}\CH{a,b}\CH{b,f}\CH{a,e} \\
        \approx{} \CH{a,c}\CH{b,d}\CH{a,b}\CH{b,d}\CH{a,c}.
    \end{gather*}
\end{proposition}
\begin{proof}
    \begin{align*}
        & \CH{a,e}\CH{b,f}\CH{a,b}\CH{b,f}\CH{a,e} \CH{c,e}\CH{d,f}\CH{e,f}\CH{d,f}\CH{c,e} \CH{a,e}\CH{b,f}\CH{a,b}\CH{b,f}\CH{a,e} \\
        & \approx{} \CH{a,c}\CH{b,d}\CH{a,b}\CH{b,d}\CH{a,c} \\
        \Leftrightarrow{}& \CH{a,e}\CH{b,f}\CH{a,b}\CH{b,f}\CH{a,e} \CH{c,e}\CH{d,f}\CH{e,f}\CH{d,f}\CH{c,e}\\
        & \approx{} \CH{a,c}\CH{b,d}\CH{a,b}\CH{b,d}\CH{a,c} \blue{\CH{a,e}\CH{b,f}\CH{a,b}\CH{b,f}\CH{a,e} } \tag{by \cref{rel:a3}} \\
        \Leftrightarrow{}& \CH{a,e}\CH{b,f}\CH{a,b}\CH{b,f}\CH{a,e} \CH{c,e}\CH{d,f}\CH{e,f}\CH{d,f}\CH{c,e}\\
        & \approx{} \CH{a,c}\CH{b,d}\CH{a,b}\CH{b,d}\CH{a,c} \CH{a,e}\CH{b,f}\CH{a,b}\CH{b,f}\CH{a,e} \blue{\CX{c,e}\CX{d,f}\CX{c,e}\CX{d,f}} \tag{by \cref{rel:a2,rel:b3}} \\
        \Leftrightarrow{}& \CH{a,e}\CH{b,f}\CH{a,b}\CH{b,f}\CH{a,e} \CH{c,e}\CH{d,f}\CH{e,f}\CH{d,f}\CH{c,e} \\
        & \approx{} \CH{a,c}\CH{b,d}\CH{a,b}\CH{b,d}\CH{a,c} \blue{\CX{c,e}\CX{d,f}\CH{a,c}\CH{b,d}\CH{a,b}\CH{b,d}\CH{a,c}} \CX{c,e}\CX{d,f} \tag{by \cref{rel:c5,rel:e1}} \\
        \Leftrightarrow{}& \CH{a,e}\CH{b,f}\CH{a,b}\CH{b,f}\CH{a,e} \CH{c,e}\CH{d,f}\CH{e,f}\CH{d,f}\CH{c,e} \\
        & \approx{} \blue{(\CH{a,c}\CH{b,d}\CH{a,b}\CH{b,d}\CH{a,c} \CX{c,e}\CX{d,f})^2} \\
        \Leftrightarrow{}& \CH{a,e}\CH{b,f}\CH{a,b}\CH{b,f}\CH{a,e}\blue{\CX{c,e}\CX{d,f}\CX{c,e}\CX{d,f}} \CH{c,e}\CH{d,f}\CH{e,f}\CH{d,f}\CH{c,e} \\
        & \approx{} (\CH{a,c}\CH{b,d}\CH{a,b}\CH{b,d}\CH{a,c} \CX{c,e}\CX{d,f})^2 \tag{by \cref{rel:a2,rel:b3}}\\
        \Leftrightarrow{}& \blue{\CX{c,e}\CX{d,f}\CH{a,c}\CH{b,d}\CH{a,b}\CH{b,d}\CH{a,c}}\CX{c,e}\CX{d,f} \CH{c,e}\CH{d,f}\CH{e,f}\CH{d,f}\CH{c,e} \\
        & \approx{} (\CH{a,c}\CH{b,d}\CH{a,b}\CH{b,d}\CH{a,c} \CX{c,e}\CX{d,f})^2 \tag{by \cref{rel:c5,rel:e1}}\\
        \Leftrightarrow{}& \CX{c,e}\CX{d,f} \CH{c,e}\CH{d,f}\CH{e,f}\CH{d,f}\CH{c,e} \\
        & \approx{} \blue{\CH{a,c}\CH{b,d}\CH{a,b}\CH{b,d}\CH{a,c} \CX{c,e}\CX{d,f}}(\CH{a,c}\CH{b,d}\CH{a,b}\CH{b,d}\CH{a,c} \CX{c,e}\CX{d,f})^2 \tag{by \cref{rel:a2,rel:a3,,rel:b3}}\\
        \Leftrightarrow{}& \CX{c,e}\CX{d,f} \CH{c,e}\CH{d,f}\CH{e,f}\CH{d,f}\CH{c,e} \approx{} \blue{(\CH{a,c}\CH{b,d}\CH{a,b}\CH{b,d}\CH{a,c} \CX{c,e}\CX{d,f})^3} \\
        \Leftrightarrow{}& \blue{\CH{c,e}\CH{d,f}\CZ{e}\CZ{f}}\CH{e,f}\CH{d,f}\CH{c,e} \approx{} (\CH{a,c}\CH{b,d}\CH{a,b}\CH{b,d}\CH{a,c} \CX{c,e}\CX{d,f})^3 \tag{by \cref{rel:a3,,rel:b1,rel:b4,,rel:d2}}\\
        \Leftrightarrow{}& \CH{c,e}\CH{d,f}\blue{\CH{e,f}\CZ{e}\CZ{f}}\CH{d,f}\CH{c,e} \approx{} (\CH{a,c}\CH{b,d}\CH{a,b}\CH{b,d}\CH{a,c} \CX{c,e}\CX{d,f})^3 \tag{by \cref{rel:d1}}\\
        \Leftrightarrow{}&  \CH{c,e}\CH{d,f}\CH{e,f}\blue{\CH{d,f}\CH{c,e}\CX{c,e}\CX{d,f}} \approx{} (\CH{a,c}\CH{b,d}\CH{a,b}\CH{b,d}\CH{a,c} \CX{c,e}\CX{d,f})^3, \tag{by \cref{rel:a3,,rel:b1,rel:b4,,rel:d2}}
    \end{align*}
    which is obtained directly from \cref{rel:d4}.
\end{proof}

\begin{proposition}
    \label{prop:h6}
    \begin{gather*}
        \CH{c,e}\CH{d,f}\CH{e,f}\CH{d,f}\CH{c,e} \CH{a,c}\CH{b,d}\CH{a,b}\CH{b,d}\CH{a,c} \CH{a,e}\CH{b,f}\CH{a,b}\CH{b,f}\CH{a,e} \\
        \approx{} \CH{a,c}\CH{b,d}\CH{a,b}\CH{b,d}\CH{a,c}.
    \end{gather*}
\end{proposition}
\begin{proof}
    \begin{align*}
        & \CH{c,e}\CH{d,f}\CH{e,f}\CH{d,f}\CH{c,e} \CH{a,c}\CH{b,d}\CH{a,b}\CH{b,d}\CH{a,c} \CH{a,e}\CH{b,f}\CH{a,b}\CH{b,f}\CH{a,e} \\
        &\approx{} \CH{a,c}\CH{b,d}\CH{a,b}\CH{b,d}\CH{a,c} \\
        \Leftrightarrow{}& \blue{\varepsilon}\CH{a,c}\CH{b,d}\CH{a,b}\CH{b,d}\CH{a,c} \CH{a,e}\CH{b,f}\CH{a,b}\CH{b,f}\CH{a,e} \\
        &\approx{} \blue{\CH{c,e}\CH{d,f}\CH{e,f}\CH{d,f}\CH{c,e}} \CH{a,c}\CH{b,d}\CH{a,b}\CH{b,d}\CH{a,c} \tag{by \cref{rel:a3}}\\
        \Leftrightarrow{}& \CH{a,c}\CH{b,d}\CH{a,b}\CH{b,d}\CH{a,c} \CH{a,e}\CH{b,f}\CH{a,b}\CH{b,f}\CH{a,e}\blue{\CX{c,e}\CX{d,f}\CX{c,e}\CX{d,f}} \\
        &\approx{} \CH{c,e}\CH{d,f}\CH{e,f}\CH{d,f}\CH{c,e} \CH{a,c}\CH{b,d}\CH{a,b}\CH{b,d}\CH{a,c} \tag{by \cref{rel:a2,rel:b3}}\\
        \Leftrightarrow{}& \CH{a,c}\CH{b,d}\CH{a,b}\CH{b,d}\CH{a,c} \blue{\CX{c,e}\CX{d,f}\CH{a,c}\CH{b,d}\CH{a,b}\CH{b,d}\CH{a,c}}\CX{c,e}\CX{d,f} \\
        &\approx{} \CH{c,e}\CH{d,f}\CH{e,f}\CH{d,f}\CH{c,e} \CH{a,c}\CH{b,d}\CH{a,b}\CH{b,d}\CH{a,c} \tag{by \cref{rel:c5,rel:e1}} \\
        \Leftrightarrow{}& (\CH{a,c}\CH{b,d}\CH{a,b}\CH{b,d}\CH{a,c} \CX{c,e}\CX{d,f})^2\blue{\CH{a,c}\CH{b,d}\CH{a,b}\CH{b,d}\CH{a,c}} \\
        &\approx{} \CH{c,e}\CH{d,f}\CH{e,f}\CH{d,f}\CH{c,e}\blue{\varepsilon} \tag{by \cref{rel:a3}} \\
        \Leftrightarrow{}& (\CH{a,c}\CH{b,d}\CH{a,b}\CH{b,d}\CH{a,c} \CX{c,e}\CX{d,f})^2{\CH{a,c}\CH{b,d}\CH{a,b}\CH{b,d}\CH{a,c}}\blue{\CX{c,e}\CX{d,f}} \\
        &\approx{} \CH{c,e}\CH{d,f}\CH{e,f}\CH{d,f}\CH{c,e}\blue{\CX{c,e}\CX{d,f}} \tag{by \cref{rel:a3}} \\
        \Leftrightarrow{}& \blue{(\CH{a,c}\CH{b,d}\CH{a,b}\CH{b,d}\CH{a,c} \CX{c,e}\CX{d,f})^3} \approx{} \CH{c,e}\CH{d,f}\CH{e,f}\CH{d,f}\CH{c,e}\CX{c,e}\CX{d,f},
    \end{align*}
    which is obtained directly from \cref{rel:d4}.
\end{proof}

\begin{proposition}
    \label{prop:h7}
    \begin{gather*}
        \CH{c,e}\CH{d,f}\CH{e,d}\CH{d,f}\CH{c,e} \CH{a,c}\CH{b,d}\CH{a,b}\CH{b,d}\CH{a,c} \CH{a,e}\CH{b,f}\CH{a,f}\CH{b,f}\CH{a,e} \\
        \approx{} \CH{a,c}\CH{b,d}\CH{a,b}\CH{b,d}\CH{a,c}.
    \end{gather*}
\end{proposition}
\begin{proof}
    \begin{align*}
        & \CH{c,e}\CH{d,f}\CH{e,d}\CH{d,f}\CH{c,e} \CH{a,c}\CH{b,d}\CH{a,b}\CH{b,d}\CH{a,c} \CH{a,e}\CH{b,f}\CH{a,f}\CH{b,f}\CH{a,e} \\
        \approx{}& \CH{c,e}\CH{d,f}\blue{\CX{d,f}\CX{d,f}}\CH{e,d}\CH{d,f}\CH{c,e} \CH{a,c}\CH{b,d}\CH{a,b}\CH{b,d}\CH{a,c} \\
        & \quad \CH{a,e}\CH{b,f}\CH{a,f}\CH{b,f}\CH{a,e} \tag{by \cref{rel:a2}} \\
        \approx{}& \CH{c,e}\blue{\CZ{f}\CH{d,f}}\CX{d,f}\CH{e,d}\CH{d,f}\CH{c,e} \CH{a,c}\CH{b,d}\CH{a,b}\CH{b,d}\CH{a,c} \\
        & \quad \CH{a,e}\CH{b,f}\CH{a,f}\CH{b,f}\CH{a,e} \tag{by \cref{rel:d2}} \\
        \approx{}& \CH{c,e}\blue{\CZ{f}\CH{d,f}}\CX{d,f}\CH{e,d}\CH{d,f}\CH{c,e} \CH{a,c}\CH{b,d}\CH{a,b}\CH{b,d}\CH{a,c} \\
        & \quad \CH{a,e}\CH{b,f}\CH{a,f}\CH{b,f}\CH{a,e} \tag{by \cref{rel:d2}} \\
        \approx{}& \CH{c,e}\CZ{f}\CH{d,f}\blue{\CH{e,f}\CX{d,f}}\CH{d,f}\CH{c,e} \CH{a,c}\CH{b,d}\CH{a,b}\CH{b,d}\CH{a,c} \\
        & \quad \CH{a,e}\CH{b,f}\CH{a,f}\CH{b,f}\CH{a,e} \tag{by \cref{rel:c5}} \\
        \approx{}& \CH{c,e}\CZ{f}\CH{d,f}\CH{e,f}\blue{\CH{d,f}\CZ{f}}\CH{c,e} \CH{a,c}\CH{b,d}\CH{a,b}\CH{b,d}\CH{a,c} \\
        & \quad \CH{a,e}\CH{b,f}\CH{a,f}\CH{b,f}\CH{a,e} \tag{by \cref{rel:a3,,rel:d2}} \\
        \approx{}& \CH{c,e}\CZ{f}\CH{d,f}\CH{e,f}\CH{d,f}\blue{\CH{c,e} \CH{a,c}\CH{b,d}\CH{a,b}\CH{b,d}\CH{a,c}} \\
        & \quad \blue{\CH{a,e}\CZ{f}}\CH{b,f}\CH{a,f}\CH{b,f}\CH{a,e} \tag{by \cref{rel:b4}} \\
        \approx{}& \CH{c,e}\CZ{f}\CH{d,f}\CH{e,f}\CH{d,f}\CH{c,e} \CH{a,c}\CH{b,d}\CH{a,b}\CH{b,d}\CH{a,c} \\
        & \quad \CH{a,e}\blue{\CH{b,f}\CX{b,f}}\CH{a,f}\CH{b,f}\CH{a,e} \tag{by \cref{rel:d2}} \\
        \approx{}& \CH{c,e}\CZ{f}\CH{d,f}\CH{e,f}\CH{d,f}\CH{c,e} \CH{a,c}\CH{b,d}\CH{a,b}\CH{b,d}\CH{a,c} \\
        & \quad \CH{a,e}\CH{b,f}\blue{\CH{a,b}\CX{b,f}}\CH{b,f}\CH{a,e} \tag{by \cref{rel:c5}} \\
        \approx{}& \CH{c,e}\CZ{f}\CH{d,f}\CH{e,f}\CH{d,f}\CH{c,e} \CH{a,c}\CH{b,d}\CH{a,b}\CH{b,d}\CH{a,c} \\
        & \quad \CH{a,e}\CH{b,f}\CH{a,b}\blue{\CH{b,f}\CZ{f}}\CH{a,e} \tag{by \cref{rel:a3,,rel:d2}} \\
        \approx{}& \blue{\CZ{f}\CH{c,e}}\CH{d,f}\CH{e,f}\CH{d,f}\CH{c,e} \CH{a,c}\CH{b,d}\CH{a,b}\CH{b,d}\CH{a,c} \\
        & \quad \CH{a,e}\CH{b,f}\CH{a,b}\CH{b,f}\blue{\CH{a,e}\CZ{f}} \tag{by \cref{rel:b4}} \\
        \approx{}& \CZ{f}\blue{\CH{a,c}\CH{b,d}\CH{a,b}\CH{b,d}\CH{a,c}}\CZ{f} \tag{by \cref{prop:h6}} \\
        \approx{}& \blue{\CH{a,c}\CH{b,d}\CH{a,b}\CH{b,d}\CH{a,c}}. \tag{by \cref{rel:a1,rel:b4}}
    \end{align*}
\end{proof}

\begin{proposition}
    \label{prop:h8}
    \begin{gather*}
        \CH{a,f}\CH{b,e}\CH{a,b}\CH{b,e}\CH{a,f} \CH{c,f}\CH{d,e}\CH{f,e}\CH{d,e}\CH{c,f} \CH{a,f}\CH{b,e}\CH{a,b}\CH{b,e}\CH{a,f} \\
        \approx{} \CH{a,c}\CH{b,d}\CH{a,b}\CH{b,d}\CH{a,c}.
    \end{gather*}
\end{proposition}
\begin{proof}
    \begin{align*}
        & \CH{a,f}\CH{b,e}\CH{a,b}\CH{b,e}\CH{a,f} \CH{c,f}\CH{d,e}\CH{f,e}\CH{d,e}\CH{c,f} \CH{a,f}\CH{b,e}\CH{a,b}\CH{b,e}\CH{a,f} \\
        \approx{}& \blue{\CX{e,f}\CX{e,f}} \CH{a,f}\CH{b,e}\CH{a,b}\CH{b,e}\CH{a,f} \CH{c,f}\CH{d,e}\CH{f,e}\CH{d,e}\CH{c,f} \\
        & \quad \CH{a,f}\CH{b,e}\CH{a,b}\CH{b,e}\CH{a,f} \tag{by \cref{rel:a2}} \\
        \approx{}& \CX{e,f}\blue{\CH{a,e}\CH{b,f}\CH{a,b}\CH{b,f}\CH{a,e} \CH{c,e}\CH{d,f}\CH{e,f}\CH{d,f}\CH{c,e}} \\
        & \quad \blue{\CH{a,e}\CH{b,f}\CH{a,b}\CH{b,f}\CH{a,e}\CX{e,f}} \tag{by \cref{rel:a2,,rel:c5,rel:e2}} \\
        \approx{}& \CX{e,f}\blue{\CH{a,c}\CH{b,d}\CH{a,b}\CH{b,d}\CH{a,c}}\CX{e,f} \tag{by \cref{prop:h5}} \\
        \approx{}& \blue{\CH{a,c}\CH{b,d}\CH{a,b}\CH{b,d}\CH{a,c}}. \tag{by \cref{rel:a3,rel:b6}}
    \end{align*}
\end{proof}

\subsection{The Main Lemma}\label{proof:main_lemma}

We prove \nameref{lem:basic_main} by splitting it into lemmas for all cases of $G$ and $N$, as listed in \cref{tab:all_cases}.

\begin{table}[ht]
    \centering
    \caption{Case distinction for \nameref{lem:basic_main}.}\label{tab:all_cases}
    \begin{tabular}{|c|c|c|c|}
        \hline
        \diagbox{$N=$}{Cases}{$G=$} & $\CX{1,x}$ & $\CZ{1}$ & $\CH{1,2}$\\
        \hline 
        $\CH{1,i_2}$ & \cref{maincase:1} & \cref{maincase:5} & \cref{maincase:9} \\
        \hline 
        $\CH{1,i_2}\CX{1,i_1}$ & \cref{maincase:2} & \cref{maincase:6} & \cref{maincase:10} \\
        \hline 
        $\CZ{a}^{\tau}$ & \cref{maincase:3} & \cref{maincase:7} & \cref{maincase:11} \\
        \hline 
        $\CX{a,j}\CZ{a}^{\tau}$ & \cref{maincase:4} & \cref{maincase:8} & \cref{maincase:12} \\
        \hline
    \end{tabular}
\end{table}

\begin{maincase}\label{maincase:1}
    \nameref{lem:basic_main} holds for the case of $G=\CX{1,x}$ and $N=\CH{1,i_2}$.
\end{maincase}
\begin{proof}
    We distinguish further subcases depending on $x$ and $i_2$.
    \begin{enumerate}
        \item $x < i_2$. Then the first syllable output by \cref{alg:exact_synthesis} on input $\CX{1,x}s$ is $\CH{1,i_2}\CX{1,x}$. We can complete the diagram as follows.
        \begin{equation*}
            \begin{tikzcd}
                \setwiretype{n} s \ar[r,"\CX{1,x}"]\ar[d,Rightarrow,"\CH{1,i_2}",swap] & r \ar[d, Rightarrow, "\CH{1,i_2}\CX{1,x}"] \\
                \setwiretype{n} t \ar[r, to*, swap, "\varepsilon"]& q
            \end{tikzcd}
        \end{equation*}
        The diagram commutes equationally by \cref{rel:a2,rel:a3}, and $\level(q) = \level(t) < \level(s)$.
        \item $x = i_2$. Then the first syllable output by \cref{alg:exact_synthesis} on input $\CX{1,i_2}s$ is $\CH{1,i_2}$. We can complete the diagram as follows.
        \begin{equation*}
            \begin{tikzcd}
                \setwiretype{n} s \ar[r,"\CX{1,i_2}"]\ar[d,Rightarrow,"\CH{1,i_2}",swap] & r \ar[d, Rightarrow, "\CH{1,i_2}"] \\
                \setwiretype{n} t \ar[r, swap, "\CZ{i_2}"]& q
            \end{tikzcd} 
        \end{equation*}
        The diagram commutes equationally by \cref{rel:d2} and  $\level(t) <\level(s), \level(q) < \level(r) = \level(s)$.
        \item $x > i_2$. Let $(j,k,l) = \level(s)$.
        \begin{enumerate}[leftmargin=10pt]
            \item \label{case:1_3a} If $\sqrt{2}^k s[x,j] \equiv \sqrt{2}^k s[i_2,j] \pmod{2}$, then the first syllable output by \cref{alg:exact_synthesis} on input $\CX{1,x}s$ will be $\CH{1,i_2}$. There must be $i_3 > i_2$ (w.l.o.g we can assume $x < i_3$) such that $i_2$, $x$ and $i_3$ are the first, the second and the third smallest integers, respectively, that satisfy the Line $6$ in \cref{alg:exact_synthesis} on input $s$. We can complete the diagram in two ways.
            \begin{equation*}
                \begin{tikzcd}
                    \setwiretype{n} s \ar[rrrrr,"\CX{1,x}"]\ar[d,Rightarrow,"\CH{1,i_2}",swap] &[5mm] &[5mm] &[5mm] &[5mm] &[5mm] r \ar[d, Rightarrow, "\CH{1,i_2}"] \\[-4mm]
                    \setwiretype{n} t \ar[d,Rightarrow,"\CH{1,i_3}\CX{1,x}",swap] & & & & & r_1 \ar[d, Rightarrow, "\CH{1,i_3}\CX{1,x}"] \\[-4mm]
                    \setwiretype{n} t_1 \ar[rd,swap,"\CH{1,i_2}"] \ar[r,swap,"\CH{1,x}"] & t_2 \ar[r, swap, "\CH{i_2,i_3}"] & t_3 \ar[r,, swap, "\CX{x,i_3}"] & t_4 \ar[r, swap, "\CH{i_2,i_3}"] & t_5 \ar[r,, swap, "\CH{1,x}"] & q \\[-4mm]
                    \setwiretype{n} & t_2' \ar[r, swap, "\CH{x,i_3}"] & t_3' \ar[r, , swap, "\CX{i_2,i_3}"] & t_4' \ar[r, , swap, "\CH{x,i_3}"] & t_5' \ar[ru, , swap, "\CH{1,i_2}"]
                \end{tikzcd}
            \end{equation*}
            We can verify that the rectangle at the top of the diagram commutes equationally:
            \begin{align*}
                & \mathbf{G}'N \approx \mathbf{N}'G \\
                \Leftrightarrow{}& \CH{1,x}\CH{i_2,i_3}\CX{x,i_3}\CH{i_2,i_3}\CH{1,x}\CH{1,i_3}\CX{1,x}\CH{1,i_2} \\
                & \quad \approx \CH{1,i_3}\CX{1,x}\CH{1,i_2}\CX{1,x} \\
                \Leftrightarrow{}& \CH{1,x}\CH{i_2,i_3}\CX{x,i_3}\CH{i_2,i_3}\CH{1,x}\blue{\CX{1,x}\CH{x,i_3}}\CH{1,i_2} \\
                & \quad \approx \CH{1,i_3}\blue{\CH{x,i_2}} \tag{by \cref{rel:a3,rel:c4}}\\
                \Leftrightarrow{}& \CH{1,x}\CH{i_2,i_3}\CX{x,i_3}\CH{i_2,i_3}\blue{\CZ{x}\CH{1,x}}\CH{x,i_3}\CH{1,i_2} \\
                & \quad \approx \CH{1,i_3}\CH{x,i_2} \tag{by \cref{rel:d2}}\\
                \Leftrightarrow{}& \CH{1,x}\CH{i_2,i_3}\blue{\CZ{i_3}\CX{x,i_3}\CH{i_2,i_3}}\CH{1,x}\CH{x,i_3}\CH{1,i_2} \\
                & \quad \approx \CH{1,i_3}\CH{x,i_2} \tag{by \cref{rel:c1,rel:b4}}\\
                \Leftrightarrow{}& \blue{\CX{i_2,i_3}\CH{1,x}\CH{i_2,i_3}}\CX{x,i_3}\CH{i_2,i_3}\CH{1,x}\CH{x,i_3}\CH{1,i_2} \\
                & \quad \approx \CH{1,i_3}\CH{x,i_2} \tag{by \cref{rel:b5,rel:d2}}\\
                \Leftrightarrow{}& \CH{1,x}\CH{i_2,i_3}\CX{x,i_3}\CH{i_2,i_3}\CH{1,x}\CH{x,i_3}\CH{1,i_2} \\
                & \quad \approx \blue{\CX{i_2,i_3}}\CH{1,i_3}\CH{x,i_2} \tag{by \cref{rel:a2}}\\
                \Leftrightarrow{}& \CH{1,x}\CH{i_2,i_3}\CX{x,i_3}\CH{i_2,i_3}\CH{1,x}\CH{x,i_3}\CH{1,i_2} \\
                & \quad \approx \blue{\CH{1,i_2}\CH{x,i_3}\CX{i_2,i_3}} \tag{by \cref{rel:c5}}\\
                \Leftrightarrow{}& \CX{x,i_3}\CH{i_2,i_3}\CH{1,x}\CH{x,i_3}\CH{1,i_2} \\
                & \quad \approx \blue{\CH{1,x}\CH{i_2,i_3}}\CH{1,i_2}\CH{x,i_3}\CX{i_2,i_3} \tag{by \cref{rel:a3}}\\
                \Leftrightarrow{}& \CX{x,i_3}\CH{i_2,i_3}\CH{1,x}\CH{x,i_3}\CH{1,i_2} \\
                & \quad \approx \blue{\CH{1,i_2}\CH{x,i_3}\CH{1,x}\CH{i_2,i_3}}\CX{i_2,i_3} \tag{by \cref{rel:f1} with $a=1,b=i_2,c=x,d=i_3$, \cref{rel:a3,rel:b6}}\\
                \Leftrightarrow{}& \CX{x,i_3}\CH{i_2,i_3}\CH{1,x}\CH{x,i_3}\CH{1,i_2} \\
                & \quad \approx \CH{1,i_2}\CH{x,i_3}\CH{1,x}\blue{\CZ{i_3}\CH{i_2,i_3}} \tag{by \cref{rel:d2}}\\
                \Leftrightarrow{}& \CX{x,i_3}\CH{i_2,i_3}\CH{1,x}\CH{x,i_3}\CH{1,i_2} \\
                & \quad \approx \blue{\CX{x,i_3}\CH{1,i_2}\CH{x,i_3}\CH{1,x}}\CH{i_2,i_3} \tag{by \cref{rel:b4,rel:b5,,rel:d2}}\\
                \Leftrightarrow{}& \CH{i_2,i_3}\CH{1,x}\CH{x,i_3}\CH{1,i_2} \approx \CH{1,i_2}\CH{x,i_3}\CH{1,x}\CH{i_2,i_3} \tag{by \cref{rel:a2}}\\
                \Leftrightarrow{}&  \blue{\varepsilon} \approx \CH{1,i_2}\CH{x,i_3}\CH{1,x}\CH{i_2,i_3}\blue{\CH{1,i_2}\CH{x,i_3}\CH{1,x}\CH{i_2,i_3}} \tag{by \cref{rel:a3}}\\
                \Leftrightarrow{}& \varepsilon \approx \blue{\varepsilon}. \tag{by \cref{rel:f1} with $a=1,b=i_2,c=x,d=i_3$}
            \end{align*}
            The decagon at the bottom of the diagram also commutes equationally:
            \begin{align*}
                & \CH{1,x}\CH{i_2,i_3}\CX{x,i_3}\CH{i_2,i_3}\CH{1,x} \\
                \approx{}& \blue{\CH{1,i_2}\CH{x,i_3}\CH{1,i_2}\CH{x,i_3}}\CH{1,x}\CH{i_2,i_3}\CX{x,i_3}\CH{i_2,i_3}\CH{1,x} \tag{by \cref{rel:a3,rel:b6}} \\
                \approx{}& \CH{1,i_2}\CH{x,i_3}\blue{\CH{1,x}\CH{i_2,i_3}\CH{1,i_2}\CH{x,i_3}}\CX{x,i_3}\CH{i_2,i_3}\CH{1,x} \tag{by \cref{rel:f1} with $a=1,b=i_2,c=x,d=i_3$, \cref{rel:a3,rel:b6}} \\
                \approx{}& \CH{1,i_2}\CH{x,i_3}\CH{1,x}\CH{i_2,i_3}\CH{1,i_2}\blue{\CZ{i_3}\CH{x,i_3}}\CH{i_2,i_3}\CH{1,x} \tag{by \cref{rel:d2}} \\
                \approx{}& \CH{1,i_2}\CH{x,i_3}\CH{1,x}\CH{i_2,i_3}\blue{\CZ{i_3}\CH{1,i_2}}\CH{x,i_3}\CH{i_2,i_3}\CH{1,x} \tag{by \cref{rel:b4}} \\
                \approx{}& \CH{1,i_2}\CH{x,i_3}\CH{1,x}\CH{i_2,i_3}\CZ{i_3}\blue{\CH{i_2,i_3}\CH{1,x}\CH{x,i_3}\CH{1,i_2}} \tag{by \cref{rel:f1} with $a=1,b=i_2,c=x,d=i_3$, \cref{rel:a3,rel:b6}} \\
                \approx{}& \CH{1,i_2}\CH{x,i_3}\CH{1,x}\blue{\CX{i_2,i_3}}\CH{1,x}\CH{x,i_3}\CH{1,i_2} \tag{by \cref{rel:d2}} \\
                \approx{}& \CH{1,i_2}\CH{x,i_3}\blue{\varepsilon\CX{i_2,i_3}}\CH{x,i_3}\CH{1,i_2} \tag{by \cref{rel:a2,rel:b5}} \\
                \approx{}& \CH{1,i_2}\CH{x,i_3}\CX{i_2,i_3}\CH{x,i_3}\CH{1,i_2}.
            \end{align*}
            Thus, we have two sequences of simple edges for $\mathbf{G}'$:
            \begin{align*}
                \mathbf{G}_1'&: t \Ne{\CH{1,i_3}\CX{1,x}} t_1 \Se{\CH{1,x}} t_2 \Se{\CH{i_2,i_3}} t_3 \Se{\CX{x,i_3}} t_4 \Se{\CH{i_2,i_3}} t_5 \Se{\CH{1,x}} q, \\
                \mathbf{G}_2'&: t \Ne{\CH{1,i_3}\CX{1,x}} t_1 \Se{\CH{1,i_2}} t_2' \Se{\CH{x,i_3}} t_3' \Se{\CX{i_2,i_3}} t_4' \Se{\CH{x,i_3}} t_5' \Se{\CH{1,i_2}} q.
            \end{align*}
            We then prove that there exists a $\mathbf{G}'$ chosen from $\mathbf{G}_1'$ or $\mathbf{G}_2'$ such that $\level(\mathbf{G}')<\level(s)$.
            Recall that from Lines 3-9 in \cref{alg:exact_synthesis}, we know that $\sqrt{2}^ks[1,j]\equiv 1 \pmod{2}$ or $\equiv 1+\sqrt{2} \pmod{2}$, $i_2, x$ and $i_3$ are the first, the second and the third smallest integer $i > 1$ satisfying $\sqrt{2}^ks[i,j] \equiv \sqrt{2}^ks[1,j] \pmod{2}$ respectively, $k\geq 2$\footnote{Since there are $4$ indices $i$ such that $\sqrt{2}^k s[i,j] \not\equiv 0 \pmod{2}$, $k$ can not be $1$ otherwise the norm of the $j$-th row of $s$ would be greater than $1$} and $l \geq 4$. 
            Then $v= s[:,j][1,i_2,x,i_3]$ and $k$ satisfy the condition of \cref{lem:useful1}, thus we have that either 
            \begin{itemize}
                \item $\sqrt{2}^{k-1}(\CH{1,x}\CH{i_2,i_3}\CH{1,i_2}\CH{x,i_3}s)[:,j][1,i_2,x,i_3]\in \ZZ[\sqrt{2}]^4$; or
                \item $\sqrt{2}^{k-1}(\CH{1,i_3}\CH{i_2,x}\CH{1,i_2}\CH{x,i_3}s)[:,j][1,i_2,x,i_3]\in \ZZ[\sqrt{2}]^4$,
            \end{itemize}
            which are equivalent to
            \begin{itemize}
                \item $ \sqrt{2}^{k-1}(\CZ{x}\CH{i_2,i_3}\CH{1,x}\CH{1,i_3}\CX{1,x}\CH{1,i_2}s)[:,j][1,i_2,x,i_3]\in \ZZ[\sqrt{2}]^4$; or
                \item $ \sqrt{2}^{k-1}(\CX{1,x}\CX{1,i_2}\CZ{i_2}\CH{x,i_3}\CH{1,i_2}\CH{1,i_3}\CX{1,x}\CH{1,i_2}s)[:,j][1,i_2,x,i_3]\in \ZZ[\sqrt{2}]^4$,
            \end{itemize}
            respectively, i.e.,
            \begin{itemize}
                \item $ \sqrt{2}^{k-1}(\CZ{x}t_3)[:,j][1,i_2,x,i_3] \in \ZZ[\sqrt{2}]^4$; or
                \item $\sqrt{2}^{k-1}(\CX{1,x}\CX{1,i_2}\CZ{i_2}t_3')[:,j][1,i_2,x,i_3] \in \ZZ[\sqrt{2}]^4$,
            \end{itemize}
            They correspond exactly to $\mathbf{G}_1'$ or $\mathbf{G}_2'$ respectively.
            \emph{Without loss of generality}, we assume that $\sqrt{2}^{k-1}(\CZ{x}t_3)[:,j][1,i_2,x,i_3] \in \ZZ[\sqrt{2}]^4$.
            Since the generators involved here only act on indices $1,i_2,x,i_3$, we only need to consider the $\lde$ of indices $1,i_2,x,i_3$.
            It is clearly that $\sqrt{2}^{k}t_3[:,j][1,i_2,x,i_3] \equiv (0,0,0,0) \pmod{2}$, which means
            \[\lde(t_3[:,j][1,i_2,x,i_3]) \leq k-1.\]
            Then the $4$ elements at $(1,j), (i_2,j), (x,j), (i_3,j)$ must not contribute to the integer $l$ when the level is $(j,k,l) = \level(s)$. Thus,
            \[\level(t_3) \leq (j,k,l-4).\]
            Since $t_4 = \CX{x,i_3}t_3$, we can get a similar result that $\level(t_4) \leq (j,k,l-4)$.
            With
            $\level(t) = (j,k,l-2)$, $\level(t_1) \leq (j,k,l-4)$, $\level(q) < (j,k,l-4)$,
            and each $H$ generator is only possible to add at most $2$ to $l'$ when the $k'$ of level $(j',k',l')$ is unchanged, we have $\level(t_2) \leq (j,k,l-2), \level(t_5) \leq (j,k,l-2)$.
            Finally, we have \[\level(\mathbf{G}'_1:t \Se*{} q) \leq (j,k,l-2) < (j,k,l) = \level(s).\] 
            \item \label{case:1_3b} If $\sqrt{2}^ks[x,j] \not\equiv \sqrt{2}^ks[i_2,j] \pmod{2}$ but $\sqrt{2}^ks[x,j] \equiv 1 \pmod{2}$ or $\equiv 1+\sqrt{2} \pmod{2}$ (i.e., $x$ is an integer satisfying the Line $5$ but not the Line $6$ in \cref{alg:exact_synthesis} on input $s$), there must be another smallest $i_3 \neq x$ (w.l.og. we can assume $x < i_3$) such that $\sqrt{2}^{k}s[i_3,j]\equiv \sqrt{2}^ks[x,j] \pmod{2}$. Then, the first syllable output by \cref{alg:exact_synthesis} on input $\CX{1,x}s$ will be $\CH{1,i_3}$. We can complete the diagram as follows.
            \begin{equation*}
                \begin{tikzcd}
                    \setwiretype{n} s \ar[rrr,"\CX{1,x}"]\ar[d,Rightarrow,"\CH{1,i_2}",swap] & & & r \ar[d, Rightarrow, "\CH{1,i_3}"] \\[-4mm]
                    \setwiretype{n} t \ar[r,"\CH{x,i_3}",swap] & t_1 \ar[r, swap,"\CH{1,i_2}"] & t_2 \ar[r,"\CX{1,x}",swap] & q 
                \end{tikzcd}
            \end{equation*}
            The diagram commutes equationally by \cref{rel:a3,,rel:c4,,rel:b6}. Since $\CH{1,i_2}$ and $\CH{x,i_3}$ act on different indices, $\sqrt{2}^ks[1,j] \equiv \sqrt{2}^ks[i_2,j]\pmod{2}$ and $\sqrt{2}^ks[i_3,j] \equiv \sqrt{2}^ks[x,j]\pmod{2}$, we have
            \begin{gather*}
                \level(t) = \level(\CH{1,i_2}s) = (j,k,l-2), \\
                \level(t_1) = \level(\CH{x,i_3}\CH{1,i_2}s) \leq (j,k,l-4), \\
                \level(t_2) = \level(\CH{1,i_2}\CH{x,i_3}\CH{1,i_2}s) = \level(\CH{x,i_3}s) = (j,k,l-2), \\
                \level(q) = \level(\CX{1,x}t_2) \leq (j,k,l-2),
            \end{gather*}
            where $l\geq 4$ because  $\sqrt{2}^ks[1,j], \sqrt{2}^ks[i_2,j], \sqrt{2}^ks[i_3,j]$ and $\sqrt{2}^ks[x,j]$ are $4$ elements that contribute to $l$. Therefore,
            \[\level(\mathbf{G}': t \Se{\CH{x,i_3}} t_1 \Se{\CH{1,i_2}} t_2 \Se{\CX{1,x}} q) \leq (j,k,l-2) < (j,k,l) = \level(s).\]
            \item \label{case:1_3c} If $\sqrt{2}^ks[x,j] \equiv 0 \pmod{2}$, then the first syllable output by \cref{alg:exact_synthesis} on input $\CX{1,x}s$ will be $\CH{1,x}\CX{1,i_2}$. We can complete the diagram as follows.
            \[ \begin{tikzcd}
                \setwiretype{n} s \ar[r,"\CX{1,x}"]\ar[d,Rightarrow,"\CH{1,i_2}",swap] &[1cm] r \ar[d, Rightarrow, "\CH{1,x}\CX{1,i_2}"] \\[-4mm]
                \setwiretype{n} t \ar[r, to*,swap, "\CZ{x}\CX{i_2,x}"]& q
            \end{tikzcd} 
            \]
            The diagram commutes equationally by \cref{rel:a2,rel:d2}.
            We have $l \geq 2$. Then $\level(r) = (j,k,l)$, $\level(t) \leq (j,k,l-2)$, and $\level(q) \leq (j,k,l-2)$. Thus, we can easily verify that $\level(\mathbf{G}': t\Se*{\CZ{x}\CX{i_2,x}} q) \leq (j,k,l-2) < \level(s)$. \qedhere
        \end{enumerate} 
    \end{enumerate}
\end{proof}

\begin{maincase}\label{maincase:2}
    \nameref{lem:basic_main} holds for the case of $G=\CX{1,x}$ and $N=\CH{1,i_2}\CX{1,i_1}$.
\end{maincase}
\begin{proof}
    We distinguish further subcases depending on $x$, $i_1$ and $i_2$.
    \begin{enumerate}
        \item $x=i_1$. Then the first syllable output by \cref{alg:exact_synthesis} on input $\CX{1,x}s$ will be $\CH{1,i_2}$. We can complete the diagram as follows.
        \[ \begin{tikzcd}
            \setwiretype{n} s \ar[r,"\CX{1,i_1}"]\ar[d,Rightarrow,"\CH{1,i_2}\CX{1,i_1}",swap] &[1.4cm] r \ar[d, Rightarrow, "\CH{1,i_2}"] \\[-4mm]
            \setwiretype{n} t \ar[r, to*,swap, "\varepsilon"]& q
        \end{tikzcd} 
        \]
        The diagram commutes equationally and $\level(q) = \level(t)<\level(s)$.
        \item $x=i_2$. Then the first syllable output by \cref{alg:exact_synthesis} on input $\CX{1,x}s$ will be $\CH{1,i_1}$. We can complete the diagram as follows.
        \[ \begin{tikzcd}
            \setwiretype{n} s \ar[r,"\CX{1,i_2}"]\ar[d,Rightarrow,"\CH{1,i_2}\CX{1,i_1}",swap] &[1.4cm] r \ar[d, Rightarrow, "\CH{1,i_1}"] \\[-4mm]
            \setwiretype{n} t \ar[r, to*, swap, "\CZ{i_1}\CX{i_1,i_2}"]& q
        \end{tikzcd} 
        \]
        The diagram commutes equationally:
        \begin{align*}
            &\phantom{={}} \CZ{i_1}\CX{i_1,i_2}\CH{1,i_2}\CX{1,i_1} \\
            &\approx \blue{\CX{i_1,i_2}\CZ{i_2}}\CH{1,i_2}\CX{1,i_1} \tag{by \cref{rel:c1}} \\
            &\approx \CX{i_1,i_2}\blue{\CH{1,i_2}\CX{1,i_2}}\CX{1,i_1} \tag{by \cref{rel:d2}} \\
            &\approx \blue{\CH{1,i_1}\CX{i_1,i_2}}\CX{1,i_2}\CX{1,i_1} \tag{by \cref{rel:c5}} \\
            &\approx \CH{1,i_1}\CX{i_1,i_2}\blue{\CX{1,i_1}\CX{i_1,i_2}} \tag{by \cref{rel:c2}} \\
            &\approx \CH{1,i_1}\blue{\CX{1,i_2}} \tag{by \cref{rel:c3}}
        \end{align*}
        Let $(j,k,l) = \level(s)$, we have $l \geq 2$. Then $\level(r) = (j,k,l)$, $\level(t) \leq (j,k,l-2)$, and $\level(q) \leq (j,k,l-2)$. Thus, we can easily verify that \[\level(\mathbf{G}': t\Se*{\CZ{i_1}\CX{i_1,i_2}} q) \leq (j,k,l-2) < \level(s).\]
        \item $x\neq i_1,i_2$. Let $(j,k,l) = \level(s)$.
        \begin{enumerate}
            \item If $\sqrt{2}^ks[x,j] \equiv \sqrt{2}^ks[i_1,j]$ (i.e., $x$ satisfies Lines $5$ and $6$ of \cref{alg:exact_synthesis} on input $s$), then, similar to \cref{case:1_3a} of \cref{maincase:1}, we can complete the diagram in two ways.\footnote{Here, we assume $1<i_1<i_2<x<i_3$. For the other cases of $1<i_1< i_2,x,i_3$, we can apply a permutation that doesn't affect $1, i_1$, which will not change the first syllable output by \cref{alg:exact_synthesis}.}
            \begin{equation*}
                \hspace{4mm}\begin{tikzcd}
                \setwiretype{n} s \ar[rrrrrr,"\CX{1,x}"]\ar[d,Rightarrow,"\CH{1,i_2}\CX{1,i_1}"] &[2mm] &[2mm] &[2mm] &[2mm] &[2mm] &[18mm] r \ar[d, swap,Rightarrow, "\CH{1,i_1}"] \\[-4mm]
                \setwiretype{n} t \ar[d,Rightarrow,"\CH{1,i_3}\CX{1,x}"] & & & & & & r_1 \ar[d, swap,Rightarrow, "\CH{1,i_3}\CX{1,i_2}"] \\[-4mm]
                \setwiretype{n} t_1 \ar[rd,swap,"\CH{1,i_2}"] \ar[r,,swap,"\CH{1,x}"] & t_2 \ar[r, , swap, "\CH{i_2,i_3}"] & t_3 \ar[r, , swap, "\CX{i_2,x}"] & t_4 \ar[r, , swap, "\CH{i_2,i_3}"] & t_5 \ar[r, , swap, "\CH{1,x}"] & t_6 \ar[r, , swap, "\scalebox{.7}{$\CX{1,x}\CZ{i_3}\CX{i_2,i_3}\CX{i_1,i_3}\CX{1,i_1}$}"] & q \\[-3mm]
                \setwiretype{n} & t_2' \ar[r, , swap, "\CH{x,i_3}"] & t_3' \ar[r, , swap, "\CX{i_2,x}"] & t_4' \ar[r, , swap, "\CH{x,i_3}"] & t_5' \ar[ru, , swap, "\CH{1,i_2}"]
                \end{tikzcd}
            \end{equation*}
            We can verify that the rectangle at the top of the diagram commutes equationally:
            \begin{align*}
                & \mathbf{G}'N \approx \mathbf{N}'G \\
                \Leftrightarrow{}& \CX{1,x}\CZ{i_3}\CX{i_2,i_3}\CX{i_1,i_3}\CX{1,i_1}\CH{1,x}\CH{i_2,i_3}\CX{i_2,x}\CH{i_2,i_3}\CH{1,x}\CH{1,i_3}\CX{1,x}\CH{1,i_2}\CX{1,i_1} \\
                & \approx \CH{1,i_3}\CX{1,i_2}\CH{1,i_1}\CX{1,x} \\
                \Leftrightarrow{}& \CX{1,x}\CZ{i_3}\CX{i_2,i_3}\CX{i_1,i_3}\blue{\CH{i_1,x}\CH{i_2,i_3}\CX{i_2,x}\CH{i_2,i_3}\CH{i_1,x}\CH{i_1,i_3}\CX{i_1,x}\CH{i_1,i_2}\CX{1,i_1}}\CX{1,i_1} \\
                & \approx \blue{\CX{1,x}\CH{x,i_3}\CX{x,i_2}\CH{x,i_1}} \tag{by \cref{rel:c2,rel:c4}} \\
                \Leftrightarrow{}& \CZ{i_3}\CX{i_2,i_3}\CX{i_1,i_3}\CH{i_1,x}\CH{i_2,i_3}\CX{i_2,x}\CH{i_2,i_3}\CH{i_1,x}\CH{i_1,i_3}\CX{i_1,x}\CH{i_1,i_2}\blue{\varepsilon} \\
                & \approx \CH{x,i_3}\CX{x,i_2}\CH{x,i_1} \tag{by \cref{rel:a2}} \\
                \Leftrightarrow{}& \CZ{i_3}\CX{i_2,i_3}\CX{i_1,i_3}\blue{\CX{i_2,x}\CH{i_1,i_2}\CH{x,i_3}}\CH{i_2,i_3}\CH{i_1,x}\CH{i_1,i_3}\CX{i_1,x}\CH{i_1,i_2} \\
                & \approx \CH{x,i_3}\CX{x,i_2}\CH{x,i_1} \tag{by \cref{rel:c4,rel:c5}} \\
                \Leftrightarrow{}& \CZ{i_3}\CX{i_2,i_3}\CX{i_1,i_3}\CX{i_2,x}\CH{i_1,i_2}\CH{x,i_3}\blue{\CH{i_1,x}\CH{i_2,i_3}}\CH{i_1,i_3}\CX{i_1,x}\CH{i_1,i_2} \\
                & \approx \CH{x,i_3}\CX{x,i_2}\CH{x,i_1} \tag{by \cref{rel:b6}} \\
                \Leftrightarrow{}& \CZ{i_3}\CX{i_2,i_3}\CX{i_1,i_3}\CX{i_2,x}\blue{\CH{i_1,x}\CH{i_2,i_3}\CH{i_1,i_2}\CH{x,i_3}}\CH{i_1,i_3}\CX{i_1,x}\CH{i_1,i_2} \\
                & \approx \CH{x,i_3}\CX{x,i_2}\CH{x,i_1} \tag{by \cref{rel:f1} with $a=i_1,b=i_2,c=x,d=i_3$} \\
                \Leftrightarrow{}& \CZ{i_3}\CX{i_2,i_3}\CX{i_1,i_3}\CX{i_2,x}\CH{i_1,x}\blue{\CX{i_1,x}\CH{i_2,i_3}\CH{x,i_2}\CH{i_1,i_3}\CH{x,i_3}}\CH{i_1,i_2} \\
                & \approx \CH{x,i_3}\CX{x,i_2}\CH{x,i_1} \tag{by \cref{rel:b5,rel:c4}} \\
                \Leftrightarrow{}& \CZ{i_3}\CX{i_2,i_3}\CX{i_1,i_3}\CX{i_2,x}\blue{\CZ{x}\CH{i_1,x}}\CH{i_2,i_3}\CH{x,i_2}\CH{i_1,i_3}\CH{x,i_3}\CH{i_1,i_2} \\
                & \approx \CH{x,i_3}\CX{x,i_2}\CH{x,i_1} \tag{by \cref{rel:d2}} \\
                \Leftrightarrow{}& \CZ{i_3}\blue{\CZ{i_3}\CX{i_2,i_3}\CX{i_1,i_3}\CX{i_2,x}}\CH{i_1,x}\CH{i_2,i_3}\CH{x,i_2}\CH{i_1,i_3}\CH{x,i_3}\CH{i_1,i_2} \\
                & \approx \CH{x,i_3}\CX{x,i_2}\CH{x,i_1} \tag{by \cref{rel:b2,rel:c1}} \\
                \Leftrightarrow{}& \blue{\varepsilon}\CX{i_2,i_3}\CX{i_1,i_3}\CX{i_2,x}\CH{i_1,x}\CH{i_2,i_3}\CH{x,i_2}\CH{i_1,i_3}\CH{x,i_3}\CH{i_1,i_2} \\
                & \approx \CH{x,i_3}\CX{x,i_2}\CH{x,i_1} \tag{by \cref{rel:a1}} \\
                \Leftrightarrow{}& \blue{\CX{i_2,x}\CX{x,i_3}\CX{i_1,i_3}}\CH{i_1,x}\CH{i_2,i_3}\CH{x,i_2}\CH{i_1,i_3}\CH{x,i_3}\CH{i_1,i_2} \\
                & \approx \blue{\CX{x,i_2}\CH{i_2,i_3}}\CH{x,i_1} \tag{by \cref{rel:a1}} \\
                \Leftrightarrow{}& \CX{x,i_3}\CX{i_1,i_3}\CH{i_1,x}\CH{i_2,i_3}\CH{x,i_2}\CH{i_1,i_3}\CH{x,i_3}\CH{i_1,i_2} \\
                & \approx \CH{i_2,i_3}\CH{x,i_1} \tag{by \cref{rel:e1}} \\
                \Leftrightarrow{}& \CX{x,i_3}\CX{i_1,i_3}\CH{i_1,x}\CH{i_2,i_3}\blue{\CX{i_2,x}\CH{i_2,x}\CX{i_2,x}}\CH{i_1,i_3}\CH{x,i_3}\CH{i_1,i_2} \\
                & \approx \CH{i_2,i_3}\blue{\CX{i_1,x}\CH{i_1,x}\CX{i_1,x}} \tag{by \cref{rel:e2}} \\
                \Leftrightarrow{}& \CX{x,i_3}\CX{i_1,i_3}\blue{\CX{i_2,x}\CH{i_1,i_2}\CH{x,i_3}}\CH{i_2,x}\blue{\CH{i_1,i_3}\CH{i_2,i_3}\CH{i_1,x}\CX{i_2,x}} \\
                & \approx \blue{\CX{i_1,x}\CH{i_2,i_3}}\CH{i_1,x}\CX{i_1,x} \tag{by \cref{rel:c4,rel:c5,,rel:b5}} \\
                \Leftrightarrow{}& \blue{\CX{i_1,x}}\CX{x,i_3}\CX{i_1,i_3}\CX{i_2,x}\CH{i_1,i_2}\CH{x,i_3}\CH{i_2,x}\CH{i_1,i_3}\CH{i_2,i_3}\CH{i_1,x}\CX{i_2,x} \\
                & \approx \blue{\varepsilon}\CH{i_2,i_3}\CH{i_1,x}\CX{i_1,x} \tag{by \cref{rel:a2}} \\
                \Leftrightarrow{}& \blue{\CX{x,i_3}}\CX{i_2,x}\CH{i_1,i_2}\CH{x,i_3}\CH{i_2,x}\CH{i_1,i_3}\CH{i_2,i_3}\CH{i_1,x}\CX{i_2,x} \approx \CH{i_2,i_3}\CH{i_1,x}\CX{i_1,x} \tag{by \cref{rel:c3}} \\
                \Leftrightarrow{}& \CX{x,i_3}\CX{i_2,x}\CH{i_1,i_2}\CH{x,i_3}\CH{i_2,x}\CH{i_1,i_3}\CH{i_2,i_3}\CH{i_1,x}\CX{i_2,x}\blue{\CX{i_1,x}} \approx \CH{i_2,i_3}\CH{i_1,x}\blue{\varepsilon} \tag{by \cref{rel:a2}} \\
                \Leftrightarrow{}& \CX{x,i_3}\CX{i_2,x}\CH{i_1,i_2}\CH{x,i_3}\CH{i_2,x}\CH{i_1,i_3}\CH{i_2,i_3}\CH{i_1,x}\blue{\CX{i_1,i_2}\CX{i_2,x}} \approx \CH{i_2,i_3}\CH{i_1,x} \tag{by \cref{rel:c3,rel:e1}} \\
                \Leftrightarrow{}& \CX{x,i_3}\CX{i_2,x}\CH{i_1,i_2}\blue{\CX{i_1,i_2}\CH{x,i_3}\CH{i_1,x}\CH{i_2,i_3}\CH{i_1,i_3}\CH{i_2,x}}\CX{i_2,x} \approx \CH{i_2,i_3}\CH{i_1,x} \tag{by \cref{rel:c4}} \\
                \Leftrightarrow{}& \CX{x,i_3}\CX{i_2,x}\CH{i_1,i_2}\CX{i_1,i_2}\CH{x,i_3}\CH{i_1,x}\CH{i_2,i_3}\CH{i_1,i_3}\blue{\CZ{x}\CH{i_2,x}} \approx \CH{i_2,i_3}\CH{i_1,x} \tag{by \cref{rel:d2}} \\
                \Leftrightarrow{}& \CX{x,i_3}\CX{i_2,x}\CH{i_1,i_2}\CX{i_1,i_2}\CH{x,i_3}\blue{\CX{i_1,x}\CH{i_1,x}\CH{i_2,i_3}\CH{i_1,i_3}}\CH{i_2,x} \approx \CH{i_2,i_3}\CH{i_1,x} \tag{by \cref{rel:b2,rel:b4,,rel:d2}} \\
                \Leftrightarrow{}& \CX{x,i_3}\CX{i_2,x}\CH{i_1,i_2}\CX{i_1,i_2}\CH{x,i_3}\CX{i_1,x}\blue{\CX{i_1,i_2}\CX{x,i_3}\CH{i_1,i_3}\CH{i_2,x}\CH{i_2,i_3}\CH{i_1,x}} \approx \CH{i_2,i_3}\CH{i_1,x} \tag{by \cref{rel:f2} with $a=i_1,b=i_2,c=x,d=i_3$ and \cref{rel:a3,rel:b6}} \\
                \Leftrightarrow{}& \CX{x,i_3}\CX{i_2,x}\CH{i_1,i_2}\CX{i_1,i_2}\CH{x,i_3}\CX{i_1,x}\CX{i_1,i_2}\CX{x,i_3}\CH{i_1,i_3}\CH{i_2,x} \approx \varepsilon \tag{by \cref{rel:a3}} \\
                \Leftrightarrow{}& \CX{x,i_3}\CX{i_2,x}\CH{i_1,i_2}\CX{i_1,i_2}\blue{\CX{i_1,i_2}\CH{x,i_3}\CX{i_2,x}}\CX{x,i_3}\CH{i_1,i_3}\CH{i_2,x} \approx \varepsilon \tag{by \cref{rel:b6,rel:c2}} \\
                \Leftrightarrow{}& \CX{x,i_3}\CX{i_2,x}\CH{i_1,i_2}\blue{\varepsilon}\CH{x,i_3}\CX{i_2,x}\CX{x,i_3}\CH{i_1,i_3}\CH{i_2,x} \approx \varepsilon \tag{by \cref{rel:a2}} \\
                \Leftrightarrow{}& \CX{x,i_3}\CX{i_2,x}\blue{\CX{i_2,x}\CH{i_1,x}\CH{i_2,i_3}}\CX{x,i_3}\CH{i_1,i_3}\CH{i_2,x} \approx \varepsilon \tag{by \cref{rel:c4,rel:c4}} \\
                \Leftrightarrow{}& \CX{x,i_3}\blue{\varepsilon}\CH{i_1,x}\CH{i_2,i_3}\CX{x,i_3}\CH{i_1,i_3}\CH{i_2,x} \approx \varepsilon \tag{by \cref{rel:a2}} \\
                \Leftrightarrow{}& \CX{x,i_3}\blue{\CX{x,i_3}\CH{i_1,i_3}\CH{i_2,x}}\CH{i_1,i_3}\CH{i_2,x} \approx \varepsilon \tag{by \cref{rel:c5}} \\
                \Leftrightarrow{}& \blue{\varepsilon} \approx \varepsilon \tag{by \cref{rel:a2,rel:a3,,rel:b6}}
            \end{align*}
            and the decagon at the bottom of the diagram also commutes equationally:
                \begin{align*}
                    &\CH{1,x}\CH{i_2,i_3}\CX{i_2,x}\CH{i_2,i_3}\CH{1,x} \approx \CH{1,i_2}\CH{x,i_3}\CX{i_2,x}\CH{x,i_3}\CH{1,i_2} \\
                    \Leftrightarrow{}& \blue{\CX{i_2,x}\CH{1,i_2}\CH{x,i_3}}\CH{i_2,i_3}\CH{1,x} \approx \blue{\CX{i_2,x}\CH{1,x}\CH{i_2,i_3}}\CX{i_2,x}\CH{x,i_3}\CH{1,i_2} \tag{by \cref{rel:c4,rel:c5}} \\
                    \Leftrightarrow{}& \blue{\varepsilon}\CH{1,i_2}\CH{x,i_3}\CH{i_2,i_3}\CH{1,x} \approx \blue{\varepsilon}\CH{1,x}\CH{i_2,i_3}\CX{i_2,x}\CH{x,i_3}\CH{1,i_2} \tag{by \cref{rel:a2}} \\
                    \Leftrightarrow{}& \varepsilon \approx \varepsilon \tag{by \cref{rel:f1} with $a=1,b=i_2,c=x,d=i_3$ and \cref{rel:a3}}
                \end{align*}
                Thus, we have two sequences of simple edges for $\mathbf{G}'$:
                \begin{align*}
                    \mathbf{G}_1' &:t \Ne{\CH{1,i_3}\CX{1,x}} t_1 \Se{\CH{1,x}} t_2 \Se{\CH{i_2,i_3}} t_3 \Se{\CX{i_2,x}} t_4 \Se{\CH{i_2,i_3}} t_5 \Se{\CH{1,x}} t_6 \Se*{\scalebox{.7}{$\CX{1,x}\CZ{i_3}\CX{i_2,i_3}\CX{i_1,i_3}\CX{1,i_1}$}} q, \\
                    \mathbf{G}_2' &:t \Ne{\CH{1,i_3}\CX{1,x}} t_1 \Se{\CH{1,i_2}} t_2 \Se{\CH{x,i_3}} t_3 \Se{\CX{i_2,x}} t_4 \Se{\CH{x,i_3}} t_5 \Se{\CH{1,i_2}} t_6 \Se*{\scalebox{.7}{$\CX{1,x}\CZ{i_3}\CX{i_2,i_3}\CX{i_1,i_3}\CX{1,i_1}$}} q.
                \end{align*}
                Similar to \cref{case:1_3a} of \cref{maincase:1}, we can choose $\mathbf{G}'$ to be one of $\mathbf{G}_1'$ or $\mathbf{G}_2'$ such that $\level(\mathbf{G}') < \level(s)$, then this case is proved.
            \item If $\sqrt{2}^ks[x,j] \not\equiv \sqrt{2}^ks[i_1,j] \pmod{2}$ but $\sqrt{2}^ks[x,j] \equiv 1 \pmod{2}$ or $\equiv 1+\sqrt{2} \pmod{2}$ (i.e., $x$ is an integer satisfying the Line $5$ but not the Line $6$ in \cref{alg:exact_synthesis} on input $s$). Then, similar to \cref{case:1_3b} of \cref{maincase:1}, we can complete the diagram as follows.
            \[ \begin{tikzcd}
                \setwiretype{n} s \ar[rrr,"\CX{1,x}"]\ar[d,Rightarrow,"\CH{1,i_2}\CX{1,i_1}",swap] & & &[1cm] r \ar[d, Rightarrow, "\CH{1,i_3}"] \\
                \setwiretype{n} t \ar[r,"\CH{x,i_3}",swap] & t_1 \ar[r, swap,"\CH{1,i_2}"] & t_2 \ar[r,to*,"\CX{1,x}\CX{1,i_1}",swap] & q 
            \end{tikzcd} 
            \]
            The diagram commutes equationally by \cref{rel:a2,rel:a3,,rel:c4,,rel:b5,rel:b6}. Also similar to \cref{case:1_3b} of \cref{maincase:1}, we can prove that \[\level(\mathbf{G}':t \Se{\CH{x,i_3}} t_1\Se{\CH{1,i_2}}t_2\Se*{\CX{1,x}\CX{1,i_2}}q)<\level(s).\]
            \item If $\sqrt{2}^ks[x,j] \equiv 0 \pmod{2}$, then the first syllable output by \cref{alg:exact_synthesis} on input $\CX{1,x}s$ is still $\CH{1,i_2}\CX{1,i_1}$.  We can complete the diagram as follows.
            \[ \begin{tikzcd}
                \setwiretype{n} s \ar[r,"\CX{1,x}"]\ar[d,Rightarrow,"\CH{1,i_2}\CX{1,i_1}",swap] &[1.4cm] r \ar[d, Rightarrow, "\CH{1,i_2}\CX{1,i_1}"] \\
                \setwiretype{n} t \ar[r, to*,swap, "\CX{1,i_1}\CX{1,x}\CX{1,i_1}"]& q
            \end{tikzcd} 
            \]
            The diagram commutes equationally by \cref{rel:a3,rel:c4,rel:b5}. We have $l\geq 2$. Then, $\level(r) = (j,k,l)$, $\level(t) \leq (j,k,l-2)$, and $\level(q) \leq (j,k,l-2)$. Thus, we can easily verify that $\level(\mathbf{G}': t\Se*{\CX{1,i_1}\CX{1,x}\CX{1,i_1}} q) \leq (j,k,l-2) < \level(s)$.
            \qedhere
        \end{enumerate}
    \end{enumerate}
\end{proof}

\begin{maincase}\label{maincase:3}
    \nameref{lem:basic_main} holds for the case of $G=\CX{1,x}$ and $N=\CZ{a}^{\tau}$.
\end{maincase}
\begin{proof}
    In this case, $\tau = 1$ (see Line 15 in \cref{alg:exact_synthesis}). We distinguish further subcases depending on $x$ and $a$.
    \begin{enumerate}
        \item $x< a$. Then the first syllable output by \cref{alg:exact_synthesis} on input $\CX{1,x}s$ is still $\CZ{a}$. We can complete the diagram as follows. 
        \[ \begin{tikzcd}
            \setwiretype{n} s \ar[r,"\CX{1,x}"]\ar[d,Rightarrow,"\CZ{a}",swap] & r \ar[d, Rightarrow, "\CZ{a}"] \\
            \setwiretype{n} t \ar[r, swap, "\CX{1,x}"]& q
        \end{tikzcd} 
        \]
        The diagram commutes equationally by \cref{rel:b2} and we have $\level(t) < \level(s)$, $\level(q) < \level(r) = \level(s)$.
        \item $x = a$. Then the first syllable output by \cref{alg:exact_synthesis} on input $\CX{1,x}s$ is $\CX{1,a}\CZ{1}$. We can complete the diagram as follows. 
        \[ \begin{tikzcd}
            \setwiretype{n} s \ar[r,"\CX{1,a}"]\ar[d,Rightarrow,"\CZ{a}",swap] & r \ar[d, Rightarrow, "\CX{1,a}\CZ{1}"] \\
            \setwiretype{n} t \ar[r, to*,swap, "\varepsilon"]& q
        \end{tikzcd} 
        \]
        The diagram commutes equationally by \cref{rel:c1} and we have $\level(q) = \level(t) < \level(s)$.
        \item $x > a$. Then the first and second syllables output by \cref{alg:exact_synthesis} on input $\CX{1,x}s$ would be $\CX{1,x}$ and $\CZ{a}$. We can complete the diagram as follows. 
        \[ \begin{tikzcd}
            \setwiretype{n} s \ar[r,"\CX{1,x}"]\ar[d,Rightarrow,"\CZ{a}",swap] & r \ar[d, Rightarrow*, "\CZ{a}\CX{1,x}"] \\
            \setwiretype{n} t \ar[r, to*,swap, "\varepsilon"]& q
        \end{tikzcd} 
        \]
        The diagram commutes equationally by \cref{rel:a1,rel:a2} and we have $\level(q) = \level(t) < \level(s)$. \qedhere
    \end{enumerate}
\end{proof}

\begin{maincase}\label{maincase:4}
    \nameref{lem:basic_main} holds for the case of $G=\CX{1,x}$ and $N=\CX{a,j}\CZ{a}^{\tau}$.
\end{maincase}
\begin{proof}
    We distinguish further subcases depending on $x$, $j$ and $a$.
    \begin{enumerate}
        \item $x < j$ and $a = 1$. Then the first syllable output by \cref{alg:exact_synthesis} on input $\CX{1,x}s$ is $\CX{x,j}\CZ{x}^{\tau}$. We can complete the diagram as follows. 
        \[ \begin{tikzcd}
            \setwiretype{n} s \ar[r,"\CX{1,x}"]\ar[d,Rightarrow,"\CX{1,j}\CZ{1}^{\tau}",swap] & r \ar[d, Rightarrow, "\CX{x,j}\CZ{x}^{\tau}"] \\
            \setwiretype{n} t \ar[r,swap, "\CX{1,x}"]& q
        \end{tikzcd} 
        \]
        The diagram commutes equationally by \cref{rel:c1,rel:c2} and we have $\level(t) < \level(s), \level(q) < \level(r) = \level(s)$.
        \item $x = j$ and $a = 1$. Then the first syllable output by \cref{alg:exact_synthesis} on input $\CX{1,x}s$ is $\CZ{j}^{\tau}$ (or $\varepsilon$ if $\tau = 0$). We can complete the diagram as follows. 
        \[ \begin{tikzcd}
            \setwiretype{n} s \ar[r,"\CX{1,j}"]\ar[d,Rightarrow,"\CX{1,j}\CZ{1}^{\tau}",swap] & r \ar[d, Rightarrow*, "\CZ{j}^{\tau} \text{ or $\varepsilon$ if $\tau = 0$}"] \\
            \setwiretype{n} t \ar[r,to*,swap, "\varepsilon"]& q
        \end{tikzcd} 
        \]
        The diagram commutes equationally by \cref{rel:c1} and we have $\level(q) = \level(t) < \level(s)$.
        \item $x \leq j$, $x \neq a$ and $a\neq 1$. Then the first syllable output by \cref{alg:exact_synthesis} on input $\CX{1,x}s$ is still $\CX{a,j}\CZ{a}^{\tau}$. We can complete the diagram as follows. 
        \[ \begin{tikzcd}
            \setwiretype{n} s \ar[r,"\CX{1,x}"]\ar[d,Rightarrow,"\CX{a,j}\CZ{a}^{\tau}",swap] & r \ar[d, Rightarrow, "\CX{a,j}\CZ{a}^{\tau}"] \\
            \setwiretype{n} t \ar[r,swap, "\CX{1,x}"]& q
        \end{tikzcd} 
        \]
        The diagram commutes equationally by \cref{rel:b2,rel:b3}. And we have $\level(t) < \level(s), \level(q) < \level(r) = \level(s)$.
        \item $x = a$. Then the first syllable output by \cref{alg:exact_synthesis} on input $\CX{1,x}s$ is $\CX{1,j}\CZ{1}^{\tau}$. We can complete the diagram as follows. 
        \[ \begin{tikzcd}
            \setwiretype{n} s \ar[r,"\CX{1,a}"]\ar[d,Rightarrow,"\CX{a,j}\CZ{a}^{\tau}",swap] & r \ar[d, Rightarrow, "\CX{1,j}\CZ{1}^{\tau}"] \\
            \setwiretype{n} t \ar[r,swap, "\CX{1,a}"]& q
        \end{tikzcd} 
        \]
        The diagram commutes equationally by \cref{rel:a1,rel:a2,,rel:c1,rel:c2,,}. And we have $\level(t) < \level(s), \level(q) < \level(r) = \level(s)$.
        \item $x>j$. Then the first and second syllables output by \cref{alg:exact_synthesis} on input $\CX{1,x}s$ is $\CX{1,x}$ and $\CX{a,j}\CZ{a}^\tau$. We can complete the diagram as follows. 
        \[ \begin{tikzcd}
            \setwiretype{n} s \ar[r,"\CX{1,x}"]\ar[d,Rightarrow,"\CX{a,j}\CZ{a}^{\tau}",swap] & r \ar[d, Rightarrow*, "\CX{a,j}\CZ{a}^{\tau}\CX{1,x}"] \\
            \setwiretype{n} t \ar[r,to*,swap, "\varepsilon"]& q
        \end{tikzcd} 
        \]
        The diagram commutes equationally by \cref{rel:a2} and we have $\level(q) = \level(t) < \level(s)$. \qedhere
    \end{enumerate}
\end{proof}

\begin{maincase}\label{maincase:5}
    \nameref{lem:basic_main} holds for the case of $G=\CZ{1}$ and $N=\CH{1,i_2}$.
\end{maincase}
\begin{proof}
    In this case, the first syllable output by \cref{alg:exact_synthesis} on input $\CZ{1}s$ is still $\CH{1,i_2}$. We can complete the diagram as follows. 
    \[ \begin{tikzcd}
        \setwiretype{n} s \ar[r,"\CZ{1}"]\ar[d,Rightarrow,"\CH{1,i_2}",swap] &[1.8cm] r \ar[d, Rightarrow, "\CH{1,i_2}"] \\
        \setwiretype{n} t \ar[r,to*,swap, "\CZ{1}\CZ{i_2}\CX{1,i_2}"]& q
    \end{tikzcd} 
    \]
    The diagram commutes equationally by \cref{rel:a1,,rel:d1,rel:d2}. Let $(j,k,l) = \level(s)$, we have $l \geq 2$. Then $\level(r) = (j,k,l)$, $\level(t) \leq (j,k,l-2), \level(q) \leq (j,k,l-2)$. Thus, we can easily verify that \[\level(\mathbf{G}': t\Se*{\CZ{1}\CZ{i_2}\CX{1,i_2}} q) \leq (j,k,l-2) < \level(s).\]
\end{proof}

\begin{maincase}\label{maincase:6}
    \nameref{lem:basic_main} holds for the case of $G=\CZ{1}$ and $N=\CH{1,i_2}\CX{1,i_1}$.
\end{maincase}
\begin{proof}
    In this case, the first syllable output by \cref{alg:exact_synthesis} on input $\CZ{1}s$ is still $\CH{1,i_2}\CX{1,i_1}$. We can complete the diagram as follows. 
    \[ \begin{tikzcd}
        \setwiretype{n} s \ar[r,"\CZ{1}"]\ar[d,Rightarrow,"\CH{1,i_2}\CX{1,i_1}",swap] & r \ar[d, Rightarrow, "\CH{1,i_2}\CX{1,i_1}"] \\
        \setwiretype{n} t \ar[r,swap, "\CZ{i_1}"]& q
    \end{tikzcd} 
    \]
    The diagram commutes equationally by \cref{rel:b4,rel:c1}. And we have $\level(t) < \level(s)$, $\level(q) < \level(r) = \level(s)$.
\end{proof}

\begin{maincase}\label{maincase:7}
    \nameref{lem:basic_main} holds for the case of $G=\CZ{1}$ and $N=\CZ{a}^{\tau}$.
\end{maincase}
\begin{proof}
    In this case, $\tau=1$ (see Line 15 in \cref{alg:exact_synthesis}).
    We distinguish further subcases depending on $a$.
    \begin{enumerate}
        \item $a = 1$. Then we can complete the diagram as follows. 
        \[ \begin{tikzcd}
            \setwiretype{n} s \ar[r,"\CZ{1}"]\ar[d,Rightarrow,"\CZ{1}",swap] & r \ar[d, Rightarrow*, "\varepsilon"] \\
            \setwiretype{n} t \ar[r,to*,swap, "\varepsilon"]& q
            \end{tikzcd} 
        \]
        The diagram commutes equationally and we have $\level(q) = \level(t) < \level(s)$.
        \item $a > 1$. Then the first syllable output by \cref{alg:exact_synthesis} on input $\CZ{1}s$ is still $\CZ{a}^\tau$. We can complete the diagram as follows. 
        \[ \begin{tikzcd}
            \setwiretype{n} s \ar[r,"\CZ{1}"]\ar[d,Rightarrow,"\CZ{a}^\tau",swap] & r \ar[d, Rightarrow, "\CZ{a}^\tau"] \\
            \setwiretype{n} t \ar[r,swap, "\CZ{1}"]& q
            \end{tikzcd} 
        \]
        The diagram commutes equationally by \cref{rel:b1} and we have $\level(t) < \level(s)$, $\level(q) < \level(r) = \level(s)$. 
        \qedhere
    \end{enumerate}
\end{proof}

\begin{maincase}\label{maincase:8}
    \nameref{lem:basic_main} holds for the case of $G=\CZ{1}$ and $N=\CX{a,j}\CZ{a}^{\tau}$.
\end{maincase}
\begin{proof}
    We distinguish further subcases depending on $a$.
    \begin{enumerate}
        \item $a = 1$. Then the first syllable output by \cref{alg:exact_synthesis} on input $\CZ{1}s$ would be $\CX{1,j}\CZ{1}^{\tau+1}$. We can complete the diagram as follows. 
        \[ \begin{tikzcd}
            \setwiretype{n} s \ar[r,"\CZ{1}"]\ar[d,Rightarrow,"\CX{1,j}\CZ{1}^{\tau}",swap] & r \ar[d, Rightarrow, "\CX{1,j}\CZ{1}^{\tau+1}"] \\
            \setwiretype{n} t \ar[r,to*,swap, "\varepsilon"]& q
        \end{tikzcd} 
        \]
        The diagram commutes equationally by \cref{rel:a1} and we have $\level(q) = \level(t) < \level(s)$.
        \item $a > 1$. Then the first syllable output by \cref{alg:exact_synthesis} on input $\CZ{1}s$ is still $\CX{a,j}\CZ{a}^{\tau}$. We can complete the diagram as follows. 
        \[ \begin{tikzcd}
            \setwiretype{n} s \ar[r,"\CZ{1}"]\ar[d,Rightarrow,"\CX{a,j}\CZ{a}^{\tau}",swap] & r \ar[d, Rightarrow, "\CX{a,j}\CZ{a}^{\tau}"] \\
            \setwiretype{n} t \ar[r,swap, "\CZ{1}"]& q
        \end{tikzcd} 
        \]
        The diagram commutes equationally by \cref{rel:b1,rel:b2} and we have $\level(t) < \level(s)$, $\level(q) < \level(r) = \level(s)$. \qedhere
    \end{enumerate}
\end{proof}

\begin{maincase}\label{maincase:9}
    \nameref{lem:basic_main} holds for the case of $G=\CH{1,2}$ and $N=\CH{1,i_2}$.
\end{maincase}
\begin{proof}
    \let\fullwidthdisplay\relax
    We distinguish further subcases depending on $i_2$.
    \begin{enumerate}[leftmargin=25pt]
        \item $i_2 = 2$. Then we can complete the diagram as follows. 
        \[ \begin{tikzcd}
            \setwiretype{n} s \ar[r,"\CH{1,2}"]\ar[d,Rightarrow,"\CH{1,2}",swap] & r \ar[d, Rightarrow*, "\varepsilon"] \\
            \setwiretype{n} t \ar[r,to*,swap, "\varepsilon"]& q
        \end{tikzcd} 
        \]
        The diagram commutes equationally and we have $\level(q) = \level(t) < \level(s)$.
        \item $i_2 > 2$.
        Let $(j,k,l) = \level(s)$.

        Since the first syllable output by \cref{alg:exact_synthesis} on input $s$ is $\CH{1,i_2}$, we have $\sqrt{2}^ks[1,j] \equiv 1 \pmod{2}$ or $\equiv 1+\sqrt{2} \pmod{2}$. Consider the following two cases:
        \begin{enumerate}[leftmargin=5pt]
            \item $\sqrt{2}^ks[2,j] \equiv 0 \pmod{2}$ or $\equiv \sqrt{2} \pmod{2}$.
            In this case,
            \begin{equation*}
                \begin{gathered}
                    \lde((\CH{1,2}s)[1,j]) = \lde((\CH{1,2}s)[2,j]) = k+1, \\
                    \sqrt{2}^{k+1}(\CH{1,2}s)[1,j] \equiv \sqrt{2}^{k+1}(\CH{1,2}s)[2,j] \pmod{2}.
                \end{gathered}
            \end{equation*}
            Then, the first syllable output by \cref{alg:exact_synthesis} on input $\CH{1,2}s$ will be $\CH{1,2}$.
            The diagram can be completed as follows.
            \begin{equation*}
                \begin{tikzcd}
                    \setwiretype{n} s \ar[r,"\CH{1,2}"]\ar[d,Rightarrow,"\CH{1,i_2}",swap] & r \ar[d, Rightarrow*, "\CH{1,i_2}\CH{1,2}"] \\
                    \setwiretype{n} t \ar[r,to*,swap, "\varepsilon"]& q
                \end{tikzcd}
            \end{equation*}
            The diagram commutes equationally by \cref{rel:a3} and $\level(q) = \level(t) < \level(s)$.
            \item $\sqrt{2}^ks[2,j] \equiv 1 \pmod{2}$ or $\equiv 1+\sqrt{2} \pmod{2}$.
            In this case, since the first syllable output by \cref{alg:exact_synthesis} on input $s$ is $\CH{1,i_2}$, we have
            \begin{equation*}
                \begin{gathered}
                    \sqrt{2}^k s[1,j] \equiv \sqrt{2}^k s[i_2,j] \not\equiv \sqrt{2}^ks[2,j] \pmod{2}, \\
                    \lde(s[1,j]) = \lde(s[2,j]) = k,
                \end{gathered}
            \end{equation*}
            from which, we  can prove that
            \begin{equation*}
                \begin{gathered}
                    \lde((\CH{1,2}s)[1,j]) = \lde((\CH{1,2}s)[2,j]) = k, \\
                    \sqrt{2}^k (\CH{1,2}s)[1,j] \not\equiv \sqrt{2}^k(\CH{1,2}s)[2,j] \pmod{2}.
                \end{gathered}
            \end{equation*}
            \begin{enumerate}[leftmargin=5pt]
                \item If $\sqrt{2}^k (\CH{1,2}s)[1,j] \equiv \sqrt{2}^k s[1,j] \pmod{2}$, then
                \begin{equation*}
                    \sqrt{2}^k (\CH{1,2}s)[i_2,j] = \sqrt{2}^k s[i_2,j] \equiv \sqrt{2}^k s[1,j] \equiv \sqrt{2}^k (\CH{1,2}s)[1,j] \pmod{2}.
                \end{equation*}
                Thus, the first syllable output by \cref{alg:exact_synthesis} on input $\CH{1,2}s$ will be $\CH{1,i_2}$.

                \begin{itemize}[leftmargin=10pt]
                    \item If $l \geq 6$, there must exist distinct $i_4,i_5 \not\in\{1,2,i_2\}$ such that
                    \[\lde(s[i_4,j]) = \lde(s[i_5,j]) = k, \enspace \sqrt{2}^ks[i_4,j] \equiv \sqrt{2}^ks[i_5,j] \pmod{2}.\]
                    Then, the diagram can be completed as follows.
                    
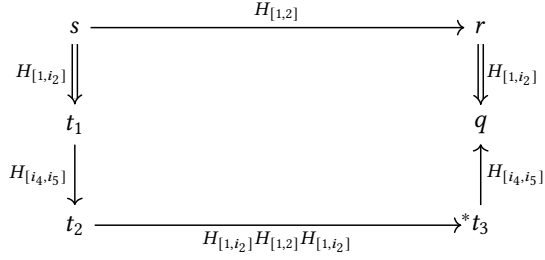
\begin{equation*}
                        \begin{tikzcd}
                            \setwiretype{n} s \ar[d,Rightarrow,"\CH{1,i_2}",swap] \ar[rrrr,"\CH{1,2}"] &&&& r \ar[d,Rightarrow,"\CH{1,i_2}"] \\[-2mm]
                            \setwiretype{n} t_1 \ar[d,"\CH{i_4,i_5}",swap] &&&& q \\[-2mm]
                            \setwiretype{n} t_2 \ar[rrrr,to*,swap,"\CH{1,i_2}\CH{1,2}\CH{1,i_2}"{name=a1}]  &&&& t_3 \ar[u,"\CH{i_4,i_5}",swap]
                        \end{tikzcd}
                    \end{equation*}
                    The diagram commutes equationally by \cref{rel:a3,rel:b6}.
                    Let
                    \[\mathbf{G}': t_1 \Se*{\CH{i_4,i_5}\CH{1,i_2}\CH{1,2}\CH{1,i_2}\CH{i_4,i_5}} q.\]
                    It can be seen that $\level(t_2), \level(t_3) \leq (j,k,l-4)$, thus $\level(\mathbf{G}') \leq (j,k,l-2) < \level(s)$.
                    \item If $l < 6$, which implies $l = 4$. There is also $i_3$ such that
                    \begin{equation*}
                        \sqrt{2}^ks[2,j] \equiv \sqrt{2}^ks[i_3,j] \equiv \sqrt{2}^k(\CH{1,2}s)[2,j] \equiv \sqrt{2}^k(\CH{1,2}s)[i_3,j] \pmod{2}.
                    \end{equation*}
                    Let
                    \begin{align*}
                        v_1 &= (\CH{2,i_3}\CH{1,i_2}s)[1,j], &
                        v_2 &= (\CH{2,i_3}\CH{1,i_2}s)[2,j], \\
                        v_3 &= (\CH{2,i_3}\CH{1,i_2}s)[i_2,j], &
                        v_4 &= (\CH{2,i_3}\CH{1,i_2}s)[i_3,j].
                    \end{align*}
                    It can be seen that either\footnote{It is easy to see that $\lde(v_1), \lde(v_3) \leq k-1$, and they cannot both be equal to $k-1$ nor both strictly less than $k-1$.}
                    \begin{itemize}[leftmargin=10pt]
                        \item $\lde(v_1) = k-1$, $\lde(v_3) \leq k-2$; or
                        \item $\lde(v_1) \leq k-2$, $\lde(v_3) = k-1$.
                    \end{itemize}
                    Similarly, we have that either
                    \begin{itemize}[leftmargin=10pt]
                        \item $\lde(v_2) = k-1$, $\lde(v_4) \leq k-2$; or
                        \item $\lde(v_2) \leq k-2$, $\lde(v_4) = k-1$.
                    \end{itemize}
                    Then, we consider subcases depending on $v_1,v_2,v_3$ and $v_4$.
                    \begin{itemize}[leftmargin=10pt]
                        \item $\lde(v_1) = k-1$, $\lde(v_3) \leq k-2$, $\lde(v_2) = k-1$, $\lde(v_4) \leq k-2$.
                        In this case, the diagram can be completed as follows.
                        
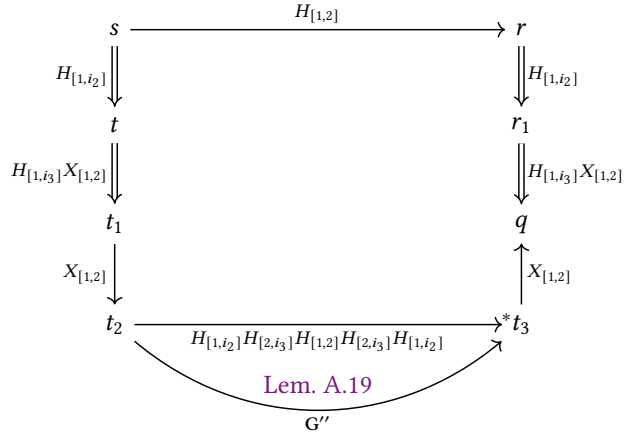
\begin{equation*}
                            \begin{tikzcd}
                                \setwiretype{n} s \ar[d,Rightarrow,"\CH{1,i_2}",swap] \ar[rrrr,"\CH{1,2}"] &&&& r \ar[d,Rightarrow,"\CH{1,i_2}"] \\[-2mm]
                                \setwiretype{n} t \ar[d,Rightarrow,"\CH{1,i_3}\CX{1,2}",swap] &&&& r_1 \ar[d,Rightarrow,"\CH{1,i_3}\CX{1,2}"] \\[-2mm]
                                \setwiretype{n} t_1 \ar[d,"\CX{1,2}",swap] &&&& q \\[-2mm]
                                \setwiretype{n} t_2 \ar[rrrr,to*,swap,"\CH{1,i_2}\CH{2,i_3}\CH{1,2}\CH{2,i_3}\CH{1,i_2}"{name=a1}]\ar[rrrr,bend right=40,"\mathbf{G}''"{name=a2},swap]  &&&& t_3 \ar[u,"\CX{1,2}",swap]
                                \arrow[phantom,from=a1,to=a2,"\text{\cref{lem:useful4}}"]
                            \end{tikzcd}
                        \end{equation*}
                        The rectangle at the top of the diagram commutes equationally by \cref{rel:a3,rel:c4,,rel:b6}.

                        Since $\level(s) = (j,k,4)$, we have that there exist $l'\geq 2$ such that 
                        \[\level(t_1) = \level(q)  = (j,k-1,l'),\]
                        thus
                        \[\level(t_2) = \level(t_3)  = (j,k-1,l').\]
                        Given that $t_2 = \CX{1,2}\CH{1,i_3}\CX{1,2}\CH{1,i_2}s = \CH{2,i_3}\CH{1,i_2}s$, we obtain
                        \begin{align*}
                            \lde(t_2[1,j]) &= \lde(v_1) = k-1, & \lde(t_2[2,j]) &= \lde(v_2) = k-1, \\
                            \lde(t_2[i_2,j]) &= \lde(v_3) \leq k-2, & \lde(t_2[i_3,j]) &= \lde(v_4) \leq  k-2.
                        \end{align*}
                        Together with $t_3 = \CH{1,i_2}\CH{2,i_3}\CH{1,2}\CH{2,i_3}\CH{1,i_2}t_2$, we can apply \cref{lem:useful4} to get a sequence of simple edges $\mathbf{G}''$ such that
                        \[ \mathbf{G}'' \approx \CH{1,i_2}\CH{2,i_3}\CH{1,2}\CH{2,i_3}\CH{1,i_2} \text{ and } \level(\mathbf{G}'':t_2\Se*{}t_3) \leq (j,k-1,l'+2). \]
                        Let
                        \begin{align*}
                            \mathbf{G}'&: t \Ne{\CH{1,i_3}\CX{1,2}} t_1 \Se{\CX{1,2}} t_2 \Se*{\mathbf{G}''} t_3 \Se{\CX{1,2}} q, \\
                            \mathbf{N}'&: r \Ne*{\CH{1,i_3}\CX{1,2}\CH{1,i_2}} q.
                        \end{align*}
                        We have
                        \begin{align*}
                            \mathbf{G}'N ={}& \CX{1,2}\mathbf{G}''\CX{1,2}\CH{1,i_3}\CX{1,2}\CH{1,i_2} \\
                            \approx{}& \CX{1,2}\blue{\CH{1,i_2}\CH{2,i_3}\CH{1,2}\CH{2,i_3}\CH{1,i_2}}\CX{1,2}\CH{1,i_3}\CX{1,2}\CH{1,i_2} \\
                            \approx{}& \CX{1,2}\CH{1,i_2}\CH{2,i_3}\CH{1,2}\CH{2,i_3}\CH{1,i_2}\blue{\CH{2,i_3}}\CH{1,i_2} \tag{by \cref{rel:a2,rel:c4}} \\
                            \approx{}& \CX{1,2}\CH{1,i_2}\CH{2,i_3}\CH{1,2}\blue{\varepsilon} \tag{by \cref{rel:a3,rel:b6}} \\
                            \approx{}& \blue{\CH{1,i_3}\CX{1,2}}\CH{1,i_2}\CH{1,2} \tag{by \cref{rel:b6,rel:c4}} \\
                            ={}& \mathbf{N}'G,
                        \end{align*}
                        and $\level(\mathbf{G}') = (j,k,2) < (j,k,4) = \level(s)$.
                        \item $\lde(v_1) = k-1$, $\lde(v_3) \leq k-2$, $\lde(v_2) \leq k-2$, $\lde(v_4) = k-1$. In this case, the diagram can be completed as follows.
                        
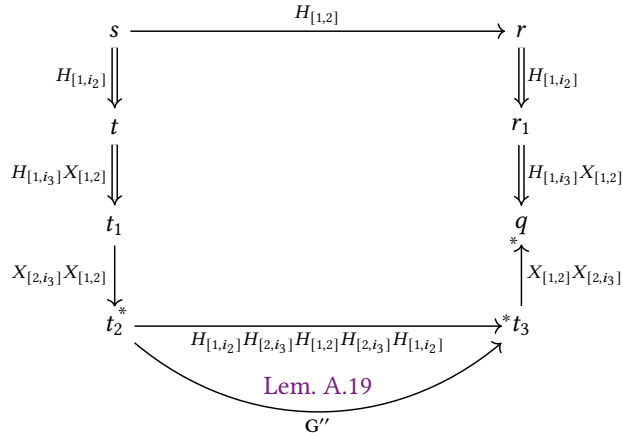
\begin{equation*}
                            \begin{tikzcd}
                                \setwiretype{n} s \ar[d,Rightarrow,"\CH{1,i_2}",swap] \ar[rrrr,"\CH{1,2}"] &&&& r \ar[d,Rightarrow,"\CH{1,i_2}"] \\[-2mm]
                                \setwiretype{n} t \ar[d,Rightarrow,"\CH{1,i_3}\CX{1,2}",swap] &&&& r_1 \ar[d,Rightarrow,"\CH{1,i_3}\CX{1,2}"] \\[-2mm]
                                \setwiretype{n} t_1 \ar[d,to*,"\CX{2,i_3}\CX{1,2}",swap] &&&& q \\[-2mm]
                                \setwiretype{n} t_2 \ar[rrrr,to*,swap,"\CH{1,i_2}\CH{2,i_3}\CH{1,2}\CH{2,i_3}\CH{1,i_2}"{name=a1}]\ar[rrrr,bend right=40,"\mathbf{G}''"{name=a2},swap]  &&&& t_3 \ar[u,to*,"\CX{1,2}\CX{2,i_3}",swap]
                                \arrow[phantom,from=a1,to=a2,"\text{\cref{lem:useful4}}"]
                            \end{tikzcd}
                        \end{equation*}
                        The rectangle at the top of the diagram commutes equationally:
                        \begin{align*}
                            & \CX{1,2}\CX{2,i_3}\CH{1,i_2}\CH{2,i_3}\CH{1,2}\CH{2,i_3}\CH{1,i_2}\CX{2,i_3}\CX{1,2}\CH{1,i_3}\CX{1,2}\CH{1,i_2} \\
                            \approx{}& \CX{1,2}\CX{2,i_3}\CH{1,i_2}\CH{2,i_3}\CH{1,2}\CH{2,i_3}\CH{1,i_2}\CX{2,i_3}\blue{\CH{2,i_3}}\CH{1,i_2} \tag{by \cref{rel:a2,rel:c4}} \\
                            \approx{}& \CX{1,2}\CX{2,i_3}\CH{1,i_2}\CH{2,i_3}\CH{1,2}\blue{\CZ{i_3}\CH{2,i_3}\CH{1,i_2}}\CH{2,i_3}\CH{1,i_2} \tag{by \cref{rel:b5,rel:d2}} \\
                            \approx{}& \CX{1,2}\blue{\CH{1,i_2}\CH{2,i_3}\CZ{i_3}}\CH{1,2}{\CZ{i_3}\CH{2,i_3}\CH{1,i_2}}\CH{2,i_3}\CH{1,i_2} \tag{by \cref{rel:b5,rel:d2}} \\
                            \approx{}& \CX{1,2}\CH{1,i_2}\CH{2,i_3}\blue{\CH{1,2}} \tag{by \cref{rel:a1,rel:a3,,rel:b4,rel:b6}} \\
                            \approx{}& \CX{1,2}\blue{\CH{2,i_3}\CH{1,i_2}}\CH{1,2} \tag{by \cref{rel:b6}} \\
                            \approx{}& \blue{\CH{1,i_3}\CX{1,2}}\CH{1,i_2}\CH{1,2}. \tag{by \cref{rel:c4}}
                        \end{align*}
                        Similar to the previous case, we can also apply \cref{lem:useful4} to $t_2$ and $t_3$ to get a sequence of simple edges $\mathbf{G}''$ such that
                        \[ \mathbf{G}'' \approx \CH{1,i_2}\CH{2,i_3}\CH{1,2}\CH{2,i_3}\CH{1,i_2} \text{ and } \level(\mathbf{G}'':t_2\Se*{}t_3) \leq (j,k-1,l'+2), \]
                        where $\level(t_2) = (j,k-1,l')$.
                        Let
                        \begin{align*}
                            \mathbf{G}'&: t \Ne{\CH{1,i_3}\CX{1,2}} t_1 \Se*{\CX{2,i_3}\CX{1,2}} t_2 \Se*{\mathbf{G}''} t_3 \Se*{\CX{1,2}\CX{2,i_3}} q, \\
                            \mathbf{N}'&: r \Ne*{\CH{1,i_3}\CX{1,2}\CH{1,i_2}} q.
                        \end{align*}
                        We have
                        \begin{align*}
                            \mathbf{G}'N ={}& \CX{1,2}\CX{2,i_3}\mathbf{G}''\CX{2,i_3}\CX{1,2}\CH{1,i_3}\CX{1,2}\CH{1,i_2} \\
                            \approx{}& \CX{1,2}\CX{2,i_3}\blue{\CH{1,i_2}\CH{2,i_3}\CH{1,2}\CH{2,i_3}\CH{1,i_2}}\CX{2,i_3}\CX{1,2}\CH{1,i_3}\CX{1,2}\CH{1,i_2} \\
                            \approx{}& \blue{\CH{1,i_3}\CX{1,2}\CH{1,i_2}\CH{1,2}} \\
                            ={}& \mathbf{N}'G,
                        \end{align*}
                        and $\level(\mathbf{G}') = (j,k,2) < (j,k,4) = \level(s)$.
                        \item $\lde(v_1) \leq k-2$, $\lde(v_3) = k-1$, $\lde(v_2) = k-1$, $\lde(v_4) \leq k-2$. In this case, the diagram can be completed as follows.
                        
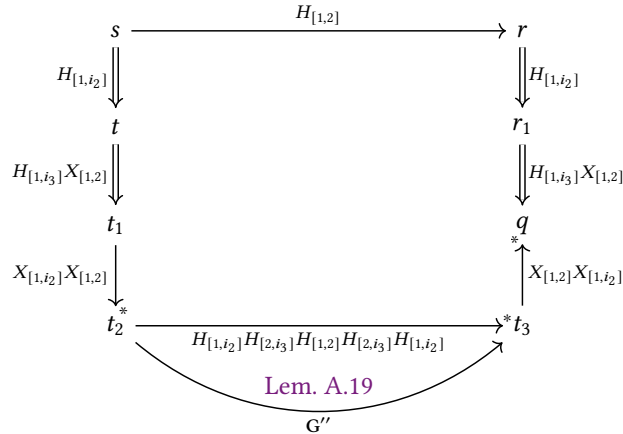
\begin{equation*}
                            \begin{tikzcd}
                                \setwiretype{n} s \ar[d,Rightarrow,"\CH{1,i_2}",swap] \ar[rrrr,"\CH{1,2}"] &&&& r \ar[d,Rightarrow,"\CH{1,i_2}"] \\[-2mm]
                                \setwiretype{n} t \ar[d,Rightarrow,"\CH{1,i_3}\CX{1,2}",swap] &&&& r_1 \ar[d,Rightarrow,"\CH{1,i_3}\CX{1,2}"] \\[-2mm]
                                \setwiretype{n} t_1 \ar[d,to*,"\CX{1,i_2}\CX{1,2}",swap] &&&& q \\[-2mm]
                                \setwiretype{n} t_2 \ar[rrrr,to*,swap,"\CH{1,i_2}\CH{2,i_3}\CH{1,2}\CH{2,i_3}\CH{1,i_2}"{name=a1}]\ar[rrrr,bend right=40,"\mathbf{G}''"{name=a2},swap]  &&&& t_3 \ar[u,to*,"\CX{1,2}\CX{1,i_2}",swap]
                                \arrow[phantom,from=a1,to=a2,"\text{\cref{lem:useful4}}"]
                            \end{tikzcd}
                        \end{equation*}
                        The rectangle at the top of the diagram commutes equationally:
                        \begin{align*}
                            & \CX{1,2}\CX{1,i_2}\CH{1,i_2}\CH{2,i_3}\CH{1,2}\CH{2,i_3}\CH{1,i_2}\CX{1,i_2}\CX{1,2}\CH{1,i_3}\CX{1,2}\CH{1,i_2} \\
                            \approx{}& \CX{1,2}\CX{1,i_2}\CH{1,i_2}\CH{2,i_3}\CH{1,2}\CH{2,i_3}\CH{1,i_2}\CX{1,i_2}\blue{\CH{2,i_3}}\CH{1,i_2} \tag{by \cref{rel:a2,rel:c4}} \\
                            \approx{}& \CX{1,2}\blue{\CH{1,i_2}\CZ{i_2}}\CH{2,i_3}\CH{1,2}\CH{2,i_3}\blue{\CZ{i_2}\CH{1,i_2}}{\CH{2,i_3}}\CH{1,i_2} \tag{by \cref{rel:d2}} \\
                            \approx{}& \CX{1,2}\CH{1,i_2}\blue{\CH{2,i_3}\CH{1,2}\CH{2,i_3}}\CH{1,i_2}{\CH{2,i_3}}\CH{1,i_2} \tag{by \cref{rel:b4,rel:a1}} \\
                            \approx{}& \CX{1,2}\CH{1,i_2}\CH{2,i_3}\CH{1,2}\blue{\varepsilon} \tag{by \cref{rel:a3,rel:b6}} \\
                            \approx{}& \CX{1,2}\blue{\CH{2,i_3}\CH{1,i_2}}\CH{1,2} \tag{by \cref{rel:b6}} \\
                            \approx{}& \blue{\CH{1,i_3}\CX{1,2}}\CH{1,i_2}\CH{1,2}. \tag{by \cref{rel:c4}}
                        \end{align*}
                        Similar to the previous case, we can also apply \cref{lem:useful4} to $t_2$ and $t_3$ to get a sequence of simple edges $\mathbf{G}''$ such that
                        \[ \mathbf{G}'' \approx \CH{1,i_2}\CH{2,i_3}\CH{1,2}\CH{2,i_3}\CH{1,i_2} \text{ and } \level(\mathbf{G}'':t_2\Se*{}t_3) \leq (j,k-1,l'+2), \]
                        where $\level(t_2) = (j,k-1,l')$.
                        Let
                        \begin{align*}
                            \mathbf{G}'&: t \Ne{\CH{1,i_3}\CX{1,2}} t_1 \Se*{\CX{1,i_2}\CX{1,2}} t_2 \Se*{\mathbf{G}''} t_3 \Se*{\CX{1,2}\CX{1,i_2}} q, \\
                            \mathbf{N}'&: r \Ne*{\CH{1,i_3}\CX{1,2}\CH{1,i_2}} q.
                        \end{align*}
                        We have
                        \begin{align*}
                            \mathbf{G}'N ={}& \CX{1,2}\CX{1,i_2}\mathbf{G}''\CX{1,i_2}\CX{1,2}\CH{1,i_3}\CX{1,2}\CH{1,i_2} \\
                            \approx{}& \CX{1,2}\CX{1,i_2}\blue{\CH{1,i_2}\CH{2,i_3}\CH{1,2}\CH{2,i_3}\CH{1,i_2}}\CX{1,i_2}\CX{1,2}\CH{1,i_3}\CX{1,2}\CH{1,i_2} \\
                            \approx{}& \blue{\CH{1,i_3}\CX{1,2}\CH{1,i_2}\CH{1,2}} \\
                            ={}& \mathbf{N}'G,
                        \end{align*}
                        and $\level(\mathbf{G}') = (j,k,2) < (j,k,4) = \level(s)$.
                        \item $\lde(v_1) \leq k-2$, $\lde(v_3) = k-1$, $\lde(v_2) \leq k-2$, $\lde(v_4) = k-1$. In this case, the diagram can be completed as follows.
                        
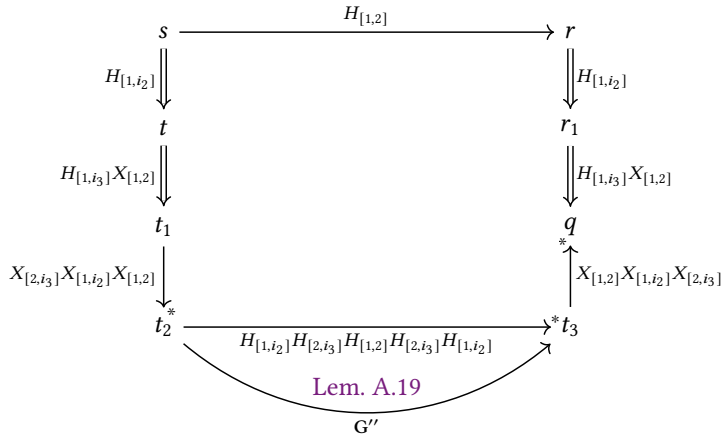
\begin{equation*}
                            \begin{tikzcd}
                                \setwiretype{n} s \ar[d,Rightarrow,"\CH{1,i_2}",swap] \ar[rrrr,"\CH{1,2}"] &&&& r \ar[d,Rightarrow,"\CH{1,i_2}"] \\[-2mm]
                                \setwiretype{n} t \ar[d,Rightarrow,"\CH{1,i_3}\CX{1,2}",swap] &&&& r_1 \ar[d,Rightarrow,"\CH{1,i_3}\CX{1,2}"] \\[-2mm]
                                \setwiretype{n} t_1 \ar[d,to*,"\CX{2,i_3}\CX{1,i_2}\CX{1,2}",swap] &&&& q \\[-2mm]
                                \setwiretype{n} t_2 \ar[rrrr,to*,swap,"\CH{1,i_2}\CH{2,i_3}\CH{1,2}\CH{2,i_3}\CH{1,i_2}"{name=a1}]\ar[rrrr,bend right=40,"\mathbf{G}''"{name=a2},swap]  &&&& t_3 \ar[u,to*,"\CX{1,2}\CX{1,i_2}\CX{2,i_3}",swap]
                                \arrow[phantom,from=a1,to=a2,"\text{\cref{lem:useful4}}"]
                            \end{tikzcd}
                        \end{equation*}
                        The rectangle at the top of the diagram commutes equationally:
                        \begin{align*}
                            & \CX{1,2}\CX{1,i_2}\CX{2,i_3}\CH{1,i_2}\CH{2,i_3}\CH{1,2}\CH{2,i_3}\CH{1,i_2}\CX{2,i_3}\CX{1,i_2}\CX{1,2}\CH{1,i_3}\CX{1,2}\CH{1,i_2} \\
                            \approx{}& \CX{1,2}\CX{1,i_2}\CX{2,i_3}\CH{1,i_2}\CH{2,i_3}\CH{1,2}\CH{2,i_3}\CX{2,i_3}\CH{1,i_2}\CX{1,i_2}\blue{\CH{2,i_3}}\CH{1,i_2} \tag{by \cref{rel:a2,rel:c4}} \\
                            \approx{}& \CX{1,2}\blue{\CH{1,i_2}\CH{2,i_3}\CZ{i_2}\CZ{i_3}}\CH{1,2}\blue{\CZ{i_3}\CZ{i_2}\CH{2,i_3}\CH{1,i_2}}{\CH{2,i_3}}\CH{1,i_2} \tag{by \cref{rel:b4,rel:b5,,rel:d2}} \\
                            \approx{}& \CX{1,2}\CH{1,i_2}\CH{2,i_3}\blue{\CH{1,2}}\CH{2,i_3}\CH{1,i_2}{\CH{2,i_3}}\CH{1,i_2} \tag{by \cref{rel:b4,rel:a1}} \\
                            \approx{}& \CX{1,2}\CH{1,i_2}\CH{2,i_3}\CH{1,2}\blue{\varepsilon} \tag{by \cref{rel:a3,rel:b6}} \\
                            \approx{}& \CX{1,2}\blue{\CH{2,i_3}\CH{1,i_2}}\CH{1,2} \tag{by \cref{rel:b6}} \\
                            \approx{}& \blue{\CH{1,i_3}\CX{1,2}}\CH{1,i_2}\CH{1,2}. \tag{by \cref{rel:c4}}
                        \end{align*}
                        Similar to the previous case, we can also apply \cref{lem:useful4} to $t_2$ and $t_3$ to get a sequence of simple edges $\mathbf{G}''$ such that
                        \[ \mathbf{G}'' \approx \CH{1,i_2}\CH{2,i_3}\CH{1,2}\CH{2,i_3}\CH{1,i_2} \text{ and } \level(\mathbf{G}'':t_2\Se*{}t_3) \leq (j,k-1,l'+2), \]
                        where $\level(t_2) = (j,k-1,l')$.
                        Let
                        \begin{align*}
                            \mathbf{G}'&: t \Ne{\CH{1,i_3}\CX{1,2}} t_1 \Se*{\CX{2,i_3}\CX{1,i_2}\CX{1,2}} t_2 \Se*{\mathbf{G}''} t_3 \Se*{\CX{1,2}\CX{1,i_2}\CX{2,i_3}} q, \\
                            \mathbf{N}'&: r \Ne*{\CH{1,i_3}\CX{1,2}\CH{1,i_2}} q.
                        \end{align*}
                        We have
                        \begin{align*}
                            & \mathbf{G}'N \\
                            ={}& \CX{1,2}\CX{1,i_2}\CX{2,i_3}\mathbf{G}''\CX{2,i_3}\CX{1,i_2}\CX{1,2}\CH{1,i_3}\CX{1,2}\CH{1,i_2} \\
                            \approx{}& \CX{1,2}\CX{1,i_2}\CX{2,i_3}\blue{\CH{1,i_2}\CH{2,i_3}\CH{1,2}\CH{2,i_3}\CH{1,i_2}}\CX{2,i_3}\CX{1,i_2}\CX{1,2}\CH{1,i_3}\CX{1,2}\CH{1,i_2} \\
                            \approx{}& \blue{\CH{1,i_3}\CX{1,2}\CH{1,i_2}\CH{1,2}} \\
                            ={}& \mathbf{N}'G,
                        \end{align*}
                        and $\level(\mathbf{G}') = (j,k,2) < (j,k,4) = \level(s)$.
                    \end{itemize}
                \end{itemize}
                \item If $\sqrt{2}^k (\CH{1,2}s)[1,j] \not\equiv \sqrt{2}^k s[1,j] \pmod{2}$, then the first syllable output by \cref{alg:exact_synthesis} on input $\CH{1,2}s$ will be $\CH{1,i_3}$ for $i_3\not\in \{1,2,i_2\}$.
                \begin{itemize}[leftmargin=10pt]
                    \item If $l \geq 6$, there must exist distinct $i_4,i_5 \not\in\{1,2,i_2,i_3\}$ such that
                    \[\lde(s[i_4,j]) = \lde(s[i_5,j]) = k, \enspace \sqrt{2}^ks[i_4,j] \equiv \sqrt{2}^ks[i_5,j] \pmod{2}.\]
                    Then, the diagram can be completed as follows.
                    \begin{equation*}
                        \begin{tikzcd}
                            \setwiretype{n} s \ar[d,Rightarrow,"\CH{1,i_2}",swap] \ar[rrrr,"\CH{1,2}"] &&&& r \ar[d,Rightarrow,"\CH{1,i_3}"] \\[-2mm]
                            \setwiretype{n} t_1 \ar[d,"\CH{i_4,i_5}",swap] &&&& q \\[-2mm]
                            \setwiretype{n} t_2 \ar[rrrr,to*,swap,"\CH{1,i_3}\CH{1,2}\CH{1,i_2}"{name=a1}]  &&&& t_3 \ar[u,"\CH{i_4,i_5}",swap]
                        \end{tikzcd}
                    \end{equation*}
                    The diagram commutes equationally by \cref{rel:a3,rel:b6}.
                    Let
                    \[\mathbf{G}': t_1 \Se*{\CH{i_4,i_5}\CH{1,i_3}\CH{1,2}\CH{1,i_2}\CH{i_4,i_5}} q.\]
                    It can be seen that $\level(t_2), \level(t_3) \leq (j,k,l-4)$, thus $\level(\mathbf{G}') \leq (j,k,l-2) < \level(s)$.
                    \item If $l < 6$, which means $l = 4$. In this case, we first consider the following diagram.
                    
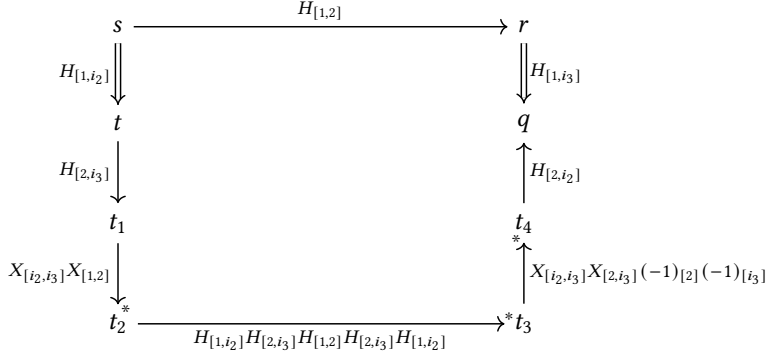
\begin{equation*}
                        \begin{tikzcd}
                            \setwiretype{n} s \ar[d,Rightarrow,"\CH{1,i_2}",swap] \ar[rrrr,"\CH{1,2}"] &&&& r \ar[d,Rightarrow,"\CH{1,i_3}"] \\[-2mm]
                            \setwiretype{n} t \ar[d,"\CH{2,i_3}",swap] &&&& q  \\[-2mm]
                            \setwiretype{n} t_1 \ar[d,to*,"\CX{i_2,i_3}\CX{1,2}",swap] &&&& t_4 \ar[u,"\CH{2,i_2}",swap]\\[-2mm]
                            \setwiretype{n} t_2 \ar[rrrr,to*,swap,"\CH{1,i_2}\CH{2,i_3}\CH{1,2}\CH{2,i_3}\CH{1,i_2}"{name=a1}]  &&&& t_3 \ar[u,to*,"\CX{i_2,i_3}\CX{2,i_3}\CZ{2}\CZ{i_3}",swap]
                        \end{tikzcd}
                    \end{equation*}
                    The diagram commutes equationally as follows.
                    \begin{align*}
                        &\CH{2,i_2}\CX{i_2,i_3}\CX{2,i_3}\CZ{2}\CZ{i_3}\CH{1,i_2}\CH{2,i_3}\CH{1,2}\CH{2,i_3}\CH{1,i_2}\CX{i_2,i_3}\CX{1,2}\CH{2,i_3}\CH{1,i_2} \\
                        \hspace{-10mm}\approx{}& \CH{2,i_2}\CX{i_2,i_3}\CX{2,i_3}\CZ{2}\CZ{i_3}\CH{1,i_2}\CH{2,i_3}\CH{1,2}\blue{\CX{i_2,i_3}\CH{2,i_2}\CH{1,i_3}}\CX{1,2}\CH{2,i_3}\CH{1,i_2} \tag{by \cref{rel:c5}} \\
                        \hspace{-10mm}\approx{}& \CH{2,i_2}\CX{i_2,i_3}\CX{2,i_3}\CZ{2}\CZ{i_3}\CH{1,i_2}\CH{2,i_3}\CH{1,2}\CX{i_2,i_3}\blue{\CX{1,2}\CH{1,i_2}\CH{2,i_3}}\CH{2,i_3}\CH{1,i_2} \tag{by \cref{rel:c4}} \\
                        \hspace{-10mm}\approx{}& \CH{2,i_2}\CX{i_2,i_3}\CX{2,i_3}\CZ{2}\CZ{i_3}\CH{1,i_2}\CH{2,i_3}\CH{1,2}\CX{i_2,i_3}\CX{1,2}\blue{\varepsilon} \tag{by \cref{rel:a3}} \\
                        \hspace{-10mm}\approx{}& \CH{2,i_2}\CX{i_2,i_3}\CX{2,i_3}\blue{\CH{1,i_2}\CH{2,i_3}\CZ{2}\CZ{i_3}}\CH{1,2}\CX{i_2,i_3}\CX{1,2} \tag{by \cref{rel:b1,rel:b4,,rel:d1}} \\
                        \hspace{-10mm}\approx{}& \CH{2,i_2}\CX{i_2,i_3}\CX{2,i_3}\CH{1,i_2}\CH{2,i_3}\CZ{2}\CZ{i_3}\blue{\CZ{2}\CH{1,2}}\CX{i_2,i_3} \tag{by \cref{rel:b3,,rel:d2}} \\
                        \hspace{-10mm}\approx{}& \CH{2,i_2}\CX{i_2,i_3}\CX{2,i_3}\CH{1,i_2}\CH{2,i_3}\blue{\CZ{i_3}}\CH{1,2}\CX{i_2,i_3} \tag{by \cref{rel:b1,rel:a1}} \\
                        \hspace{-10mm}\approx{}& \CH{2,i_2}\CX{i_2,i_3}\CX{2,i_3}\CH{1,i_2}\blue{\CX{2,i_3}\CH{2,i_3}}\CH{1,2}\CX{i_2,i_3} \tag{by \cref{rel:d2}} \\
                        \hspace{-10mm}\approx{}& \CH{2,i_2}\CX{i_2,i_3}\blue{\CH{1,i_2}}\CH{2,i_3}\CH{1,2}\CX{i_2,i_3} \tag{by \cref{rel:a2,rel:b5}} \\
                        \hspace{-10mm}\approx{}& \CH{2,i_2}\CX{i_2,i_3}\CH{1,i_2}\blue{\CX{i_2,i_3}\CH{2,i_2}\CH{1,2}} \tag{by \cref{rel:b5,rel:c5}} \\
                        \hspace{-10mm}\approx{}& \CH{2,i_2}\blue{\CH{1,i_3}}\CH{2,i_2}\CH{1,2} \tag{by \cref{rel:a2,rel:c5}} \\
                        \hspace{-10mm}\approx{}& \blue{\CH{1,i_3}}\CH{1,2}. \tag{by \cref{rel:a3,rel:b6}}
                    \end{align*}
                    Then, similar to the case of $\sqrt{2}^k (\CH{1,2}s)[1,j] \equiv \sqrt{2}^k s[1,j] \pmod{2}$ and $l<6$, we can choose $\mathbf{G}_0, \mathbf{G}_1$ according to $\lde(t_2[1,j]), \lde(t_2[2,j]), \lde(t_2[i_2,j])$ and $\lde(t_2[i_3,j])$:
                    \begin{itemize}[leftmargin=10pt]
                        \item $\lde(t_2[1,j]) = k-1, \lde(t_2[2,j]) = k-1, \lde(t_2[i_2,j]) \leq k-2$ and $\lde(t_2[i_3,j]) \leq k-2$. Choose
                        \[\mathbf{G}_0 = \varepsilon, \enspace \mathbf{G}_1 = \varepsilon. \]
                        \item $\lde(t_2[1,j]) = k-1, \lde(t_2[2,j]) \leq k-2, \lde(t_2[i_2,j]) \leq k-2$ and $\lde(t_2[i_3,j]) = k-1$. Choose
                        \[\mathbf{G}_0 = \CX{2,i_3}, \enspace \mathbf{G}_1 = \CX{2,i_3}. \]
                        \item $\lde(t_2[1,j]) \leq k-2, \lde(t_2[2,j]) = k-1, \lde(t_2[i_2,j]) = k-1$ and $\lde(t_2[i_3,j]) \leq k-2$. Choose
                        \[\mathbf{G}_0 = \CX{1,i_2}, \enspace \mathbf{G}_1 = \CX{1,i_2}. \]
                        \item $\lde(t_2[1,j]) \leq k-2, \lde(t_2[2,j]) \leq k-2, \lde(t_2[i_2,j]) = k-1$ and $\lde(t_2[i_3,j]) = k-1$. Choose
                        \[\mathbf{G}_0 = \CX{2,i_3}\CX{1,i_2}, \enspace \mathbf{G}_1 = \CX{1,i_2}\CX{2,i_3}. \]
                    \end{itemize}
                    such that the following diagram commutes equationally with $t_2'$ and $t_3'$ satisfy the conditions of \cref{lem:useful4}.
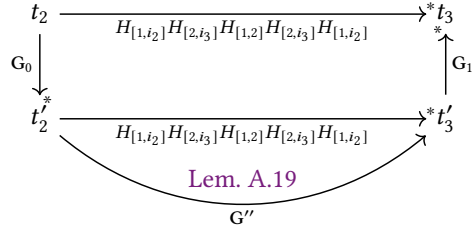
                    \begin{equation*}
                        \begin{tikzcd}
                            \setwiretype{n} t_2 \ar[rrrr,to*,swap,"\CH{1,i_2}\CH{2,i_3}\CH{1,2}\CH{2,i_3}\CH{1,i_2}"]\ar[d,to*,swap,"\mathbf{G}_0"]  &&&& t_3 \\[-2mm]
                            \setwiretype{n} t_2' \ar[rrrr,to*,swap,"\CH{1,i_2}\CH{2,i_3}\CH{1,2}\CH{2,i_3}\CH{1,i_2}"{name=a1}] \ar[rrrr,bend right=40,"\mathbf{G}''"{name=a2},swap]  &&&& t_3' \ar[u,to*,"\mathbf{G}_1",swap]
                            \arrow[phantom,from=a1,to=a2,"\text{\cref{lem:useful4}}"]
                        \end{tikzcd}
                    \end{equation*}
                    Thus, we can obtain a sequence of simple edges $\mathbf{G}''$ such that
                    \[ \mathbf{G}'' \approx \CH{1,i_2}\CH{2,i_3}\CH{1,2}\CH{2,i_3}\CH{1,i_2} \text{ and } \level(\mathbf{G}'':t_2'\Se*{}t_3') \leq (j,k-1,l'+2), \]
                    where $\level(t_2') = \level(t_2) = \level(t_1) = (j,k-1,l')$.
                    Let
                    \begin{align*}
                        \mathbf{G}'&: t \Se{\CH{2,i_3}} t_1 \Se*{\CX{i_2,i_3}\CX{1,2}} t_2 \Se*{\mathbf{G}_0} t_2' \Se*{\mathbf{G}''} t_3' \Se*{\mathbf{G}_1} t_2 \Se*{\CH{2,i_2}\CX{i_2,i_3}\CX{2,i_3}\CZ{2}\CZ{i_3}} q, \\
                        \mathbf{N}'&: r \Ne{\CH{1,i_3}} q.
                    \end{align*}
                    We can check that $\mathbf{G}'N
                        \approx{} \mathbf{N}'G$
                    and $\level(\mathbf{G}') = (j,k,2) < (j,k,4) = \level(s)$. \qedhere
                \end{itemize}
            \end{enumerate}
        \end{enumerate}
    \end{enumerate}
\end{proof}

\begin{maincase}\label{maincase:10}
    \nameref{lem:basic_main} holds for the case of $G=\CH{1,2}$ and $N=\CH{1,i_2}\CX{1,i_1}$.
\end{maincase}
\begin{proof}
    We distinguish further subcases depending on $i_1$.
    \begin{enumerate}
        \item $i_1 = 2$. Then the first syllable output by \cref{alg:exact_synthesis} on input $\CH{1,2}s$ will be $\CH{1,2}$. We can complete the diagram as follows. 
        
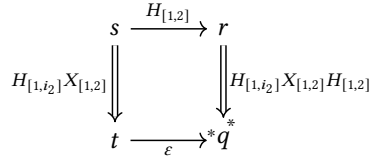
\[ \begin{tikzcd}
            \setwiretype{n} s \ar[r,"\CH{1,2}"]\ar[d,Rightarrow,"\CH{1,i_2}\CX{1,2}",swap] & r \ar[d, Rightarrow*, "\CH{1,i_2}\CX{1,2}\CH{1,2}"] \\
            \setwiretype{n} t \ar[r,to*,swap, "\varepsilon"]& q
        \end{tikzcd} 
        \]
        The diagram commutes equationally by \cref{rel:a3} and we have $\level(q) = \level(t) < \level(s)$.
        \item $i_1 > 2$. Then the first syllable output by \cref{alg:exact_synthesis} on input $\CH{1,2}s$ will be $\CH{1,2}$ or $\CH{1,i_2}\CX{1,i_1}$. 
        \begin{itemize}[leftmargin=10pt]
            \item If output $\CH{1,2}$, we can complete the diagram as follows. 
            
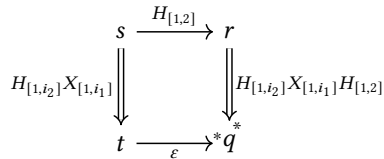
\begin{equation*}
                \begin{tikzcd}
                    \setwiretype{n} s \ar[r,"\CH{1,2}"]\ar[d,Rightarrow,"\CH{1,i_2}\CX{1,i_1}",swap] & r \ar[d, Rightarrow*, "\CH{1,i_2}\CX{1,i_1}\CH{1,2}"] \\
                    \setwiretype{n} t \ar[r,to*,swap, "\varepsilon"]& q
                \end{tikzcd} 
            \end{equation*}
            The diagram commutes equationally by \cref{rel:a3} and we have $\level(q) = \level(t) < \level(s)$.
            \item If output $\CH{1,i_2}\CX{1,i_1}$, we can complete the diagram as follows.
            \begin{equation*}
                \begin{tikzcd}
                    \setwiretype{n} s \ar[rrr,"\CH{1,2}"]\ar[d,Rightarrow,"\CH{1,i_2}\CX{1,i_1}",swap] &&& r \ar[d, Rightarrow, "\CH{1,i_2}\CX{1,i_1}"] \\
                    \setwiretype{n} t \ar[r,swap, "\CX{1,i_1}"]& t_1 \ar[r,swap,"\CH{1,2}"] & t_2 \ar[r,swap, "\CX{1,i_1}"] & q
                \end{tikzcd} 
            \end{equation*}
            The diagram commutes equationally by \cref{rel:b6,rel:c4}.
            Let $(j,k,l) = \level(s)$, where $l\geq 2$. Since the first syllable output by \cref{alg:exact_synthesis} on input $\CH{1,2}s$ is $\CH{1,i_2}\CX{1,i_1}$, we have that indices $1,2$ do not contribute to the $l'$ of $(j,k,l') = \level(r)$, then $\level(r) = (j,k,l)$.
            Thus, we have $\level(t)\leq (j,k,l-2)$ and $\level(q) \leq (j,k,l-2)$, which ensure that $\level(t_1), \level(t_2) \leq (j,k,l-2)$. Thus,
            \[\level(\mathbf{G}':t \Se{\CX{1,i_1}} t_1 \Se{\CH{1,2}} t_2 \Se{\CX{1,i_1}}) \leq (j,k,l-2) < (j,k,l) = \level(s). \qedhere \]
        \end{itemize}
    \end{enumerate}
\end{proof}

\begin{maincase}\label{maincase:11}
    \nameref{lem:basic_main} holds for the case of $G=\CH{1,2}$ and $N=\CZ{a}^{\tau}$.
\end{maincase}
\begin{proof}
    In this case, $\tau=1$ (see Line 15 in \cref{alg:exact_synthesis}).
    We distinguish further subcases depending on $a$.
    \begin{enumerate}
            \item $1 \leq a \leq 2$. Then the first syllable output by \cref{alg:exact_synthesis} on input $\CH{1,2}s$ will be $\CH{1,2}$. We can complete the diagram as follows. 
            \[ \begin{tikzcd}
                \setwiretype{n} s \ar[r,"\CH{1,2}"]\ar[d,Rightarrow,"\CZ{a}^{\tau}",swap] & r \ar[d, Rightarrow*, "\CZ{a}^{\tau}\CH{1,2}"] \\
                \setwiretype{n} t \ar[r,to*,swap, "\varepsilon"]& q
            \end{tikzcd} 
            \]
            The diagram commutes equationally by \cref{rel:a3} and we have $\level(q) = \level(t) < \level(s)$.
            \item $a > 2$. Then the first syllable output by \cref{alg:exact_synthesis} on input $\CH{1,2}s$ is still $\CZ{a}^{\tau}$. We can complete the diagram as follows. 
            \[ \begin{tikzcd}
                \setwiretype{n} s \ar[r,"\CH{1,2}"]\ar[d,Rightarrow,"\CZ{a}^{\tau}",swap] & r \ar[d, Rightarrow, "\CZ{a}^{\tau}"] \\
                \setwiretype{n} t \ar[r,swap, "\CH{1,2}"]& q
            \end{tikzcd} 
            \]
            The diagram commutes equationally by \cref{rel:b4}, and we have $\level(t) < \level(s)$, $\level(q) < \level(r) = \level(s)$. \qedhere
        \end{enumerate}
\end{proof}

\begin{maincase}\label{maincase:12}
    \nameref{lem:basic_main} holds for the case of $G=\CH{1,2}$ and $N=\CX{a,j}\CZ{a}^{\tau}$.
\end{maincase}
\begin{proof}
    We distinguish further subcases depending on $a$.
    \begin{enumerate}
        \item $1\leq a \leq 2$. Then the first and second syllables output by \cref{alg:exact_synthesis} on input $\CH{1,2}s$ will be $\CH{1,2}$ and $\CX{1,j}\CZ{1}^{\tau}$. We can complete the diagram as follows. 
        \[ \begin{tikzcd}
            \setwiretype{n} s \ar[r,"\CH{1,2}"]\ar[d,Rightarrow,"\CX{1,j}\CZ{1}^{\tau}",swap] & r \ar[d, Rightarrow*, "\CX{1,j}\CZ{1}^{\tau}\CH{1,2}"] \\
            \setwiretype{n} t \ar[r,to*,swap, "\varepsilon"]& q
        \end{tikzcd} 
        \]
        The diagram commutes equationally by \cref{rel:a3} and we have $\level(q) = \level(t) < \level(s)$.
        \item $a > 2$. Then the first syllable output by \cref{alg:exact_synthesis} on input $\CH{1,2}s$ is still $\CX{a,j}\CZ{a}^{\tau}$. We can complete the diagram as follows. 
        \[ \begin{tikzcd}
            \setwiretype{n} s \ar[r,"\CH{1,2}"]\ar[d,Rightarrow,"\CX{a,j}\CZ{a}^{\tau}",swap] & r \ar[d, Rightarrow, "\CX{a,j}\CZ{a}^{\tau}"] \\
            \setwiretype{n} t \ar[r, swap,"\CH{1,2}"] & q
        \end{tikzcd} 
        \]
        The diagram commutes equationally by \cref{rel:b4,rel:b5}. Let $(j,k,l) = \level(s)$, we have $k=0$ and $l = 0$. Then, it is easy to see that $\level(r) = (j,0,0) = \level(s)$. Thus, $\level(t) < \level(s), \level(q) < \level(r) = \level(s)$. 
        \qedhere
    \end{enumerate}
\end{proof}

\endgroup

\section{Completeness of \texorpdfstring{\QPiLang{}}{Q-Pi} \texorpdfstring{(\NoCaseChange{\cref{thm:qpi}})}{}}\label{app:qpi_proof}
\begingroup
\allowdisplaybreaks

\begin{lemma}\label{lem:n_embedded_n+m}
    Let $\mathbf{G} \in \cG_{n}^*$, which can be embedded into $\cG_{n+m}^*$, then 
    $\sem{\mathbf{G}\in  \cG_{n+m}^*} = \sem{ \mathbf{G} \in \cG_{n}^*}\oplus I_{m}$.
\end{lemma}
\begin{proof}
    By induction on structure of $\mathbf{G}$ and direct computation.
\end{proof}

\begin{lemma}\label{lem:n_shift}
    Let $\mathbf{G} \in \cG_{n}^*$, then 
    $\sem{\hshift(\mathbf{G}, m) \in \cG_{n+m}^*} = I_{m} \oplus \sem{ \mathbf{G} \in \cG_{n}^*}$.
\end{lemma}
\begin{proof}
    By induction on structure of $\mathbf{G}$ and direct computation.
\end{proof}

\begin{lemma}\label{lem:qt_faith}
    Let $\mathbf{G} \in \cG_n^*$, then
    $\sem{\qT[n]{\mathbf{G}}} = \sem{\mathbf{G}}$.
\end{lemma}
\begin{proof}
    By induction on the structure of $\mathbf{G}$ and direct computation.
\end{proof}

\lemwfaith*
\begin{proof}
    We prove it by induction on the structure of $c$.
    \begin{itemize}
        \item $c = (-1)$. $\sem{\wsem{(-1)}} = \sem{\CZ{1}} = -1 =\sem{(-1)}$, where $\wsem{(-1)}\in \cG_1^*$.
        \item $c = \hadamard$. $\sem{\wsem{\hadamard}} = \sem{\CH{1,2}} = H = \sem{\hadamard}$, where $\wsem{\hadamard}\in \cG_2^*$.
        \item $c = \pid$, $\assoclp$, $\assocrp$, $\unitep$, $\unitip$, $\assoclt$, $\assocrt$, $\unitet$, $\unitit$, $\dist$, $\factor$, $\absorb$ or $ \factorz$. In the category of \Unitary{}, we have $\sem{c} = I_{\hdim(c)}$. With $\wsem{c} = \varepsilon \in \cG_{\hdim(c)}^*$, we have $\sem{\wsem{c}} = I_{\hdim(c)} = \sem{c}$.
        \item $c = \swapp$. Let $\pi:j\in[n_1+n_2] \mapsto \mathsf{if}\, j \leq n_1 \,\mathsf{then}\, j+n_2 \,\mathsf{else}\, j-n_1$, and $X_\pi$ be the matrix representation of $\pi$. Then, by the definition of $\hpermute$ and correctness of \cref{alg:exact_synthesis}, $\sem{\swapp} = X_{\pi} = \sem{\hpermute(\pi)} = \sem{\wsem{\swapp}}$.
        \item $c = \swapt$. Let $\pi:\pi:j\in[n_1\times n_2] \mapsto \mathsf{let}\, (k,l) = (j\ \mathrm{div}\ n_2, j\bmod n_2) \,\mathsf{in}\, l\times n_1+k+1$, and $X_\pi$ be the matrix representation of $\pi$. Then, by the definition of $\hpermute$ and correctness of \cref{alg:exact_synthesis}, $\sem{\swapt} = X_{\pi} = \sem{\hpermute(\pi)} = \sem{\wsem{\swapt}}$.
        \item $c = c_1\seqq c_2$. 
        \begin{align*}
            \sem{\wsem{c_1\seqq c_2}} &= \sem{\wsem{c_2}\wsem{c_1}} \\
            &= \sem{\wsem{c_2}}\cdot\sem{\wsem{c_1}} \\
            &= \sem{c_2}\cdot \sem{c_1} \tag{by inductive hypothesis} \\
            &= \sem{c_1\seqq c_2}.
        \end{align*}
        \item $c = c_1 + c_2$. \begin{align*}
            \sem{\wsem{c_1+c_2}} ={}& \sem{\wsem{c_1}\hshift(\wsem{c_2},\hdim(c_1))} \\
            ={}& \sem{\wsem{c_1}}\cdot\sem{\hshift(\wsem{c_2},\hdim(c_1))} \tag{$\wsem{c_1}, \wsem{c_2} \in \cG_{\hdim(c_1+c_2)}^*$} \\
            ={}& (\sem{\wsem{c_1}}\oplus I_{\hdim(c_2)})\cdot \sem{\hshift(\wsem{c_2},\hdim(c_1))} \tag{by \cref{lem:n_embedded_n+m} and $\wsem{c_1} \in \cG_{\hdim(c_1)}^*, \hdim(c_1+c_2) = \hdim(c_1)+\hdim(c_2)$} \\
            ={}& (\sem{\wsem{c_1}\in \cG_{\hdim(c_1)}^*}\oplus I_{\hdim(c_2)})\cdot (I_{\hdim(c_1)}\oplus \sem{\wsem{c_2}\in\cG_{\hdim(c_2)}^*}) \tag{by \cref{lem:n_shift}} \\
            ={}& \sem{\wsem{c_1}\in \cG_{\hdim(c_1)}^*}\oplus \sem{\wsem{c_2}\in\cG_{\hdim(c_2)}^*} \\
            ={}& \sem{c_1}\oplus \sem{c_2}. \tag{by inductive hypothesis}
        \end{align*}
        \item $c = c_1 \times c_2$. \begin{itemize}[leftmargin=10pt]
            \item For the special case $c = \pid_b \times c$, \begin{align*}
                & \sem{\wsem{\pid_{b}\times c}} \\
                ={}& \sem{\wsem{c}\shift\rbra[\big]{\wsem{c},\hdim(c)}\shift\rbra[\big]{\wsem{c},2\hdim(c)}\cdots\shift\rbra[\big]{\wsem{c},(\hdim(b)-1)\times \hdim(c)}} \\
                ={}& \sem{\wsem{c}}\cdot\sem{\shift\rbra[\big]{\wsem{c},\hdim(c)}}\cdot \\
                & \quad \sem{\shift\rbra[\big]{\wsem{c},2\hdim(c)}}\cdots \sem{\shift\rbra[\big]{\wsem{c},(\hdim(b)-1)\times \hdim(c)}} \\
                ={}& \rbra*{\sem{\wsem{c}\in \cG_{\hdim(c)}^*}\oplus I_{\hdim(b-1)\times \hdim(c)}}\cdot \sem{\shift\rbra[\big]{\wsem{c},\hdim(c)}}\cdot \\
                & \quad \sem{\shift\rbra[\big]{\wsem{c},2\hdim(c)}}\cdots\sem{\shift\rbra[\big]{\wsem{c},(\hdim(b)-1)\times \hdim(c)}} \tag{by \cref{lem:n_embedded_n+m}} \\
                ={}& \rbra*{\sem{\wsem{c}\in \cG_{\hdim(c)}^*}\oplus I_{\hdim(b-1)\times \hdim(c)}}\cdot \rbra*{I_{\hdim(c)}\oplus\sem{\wsem{c}\in \cG_{\hdim(b-1)\times \hdim(c)}^*}}\cdot \\
                & \quad \sem{\shift\rbra[\big]{\wsem{c},2\hdim(c)}}\cdots\sem{\shift\rbra[\big]{\wsem{c},(\hdim(b)-1)\times \hdim(c)}} \tag{by \cref{lem:n_shift}} \\
                ={}& \rbra*{\sem{\wsem{c}\in \cG_{\hdim(c)}^*}\oplus I_{\hdim(b-1)\times \hdim(c)}}\cdot \rbra*{I_{\hdim(c)}\oplus\sem{\wsem{c}\in \cG_{\hdim(c)}^*}\oplus I_{\hdim(b-2)\times \hdim(c)}}\cdot \\
                & \quad \sem{\shift\rbra[\big]{\wsem{c},2\hdim(c)}}\cdots\sem{\shift\rbra[\big]{\wsem{c},(\hdim(b)-1)\times \hdim(c)}} \tag{by \cref{lem:n_embedded_n+m}} \\
                ={}& \rbra*{\sem{\wsem{c}\in \cG_{\hdim(c)}^*}\oplus\sem{\wsem{c}\in \cG_{\hdim(c)}^*}\oplus I_{\hdim(b-2)\times \hdim(c)}}\cdot \\
                & \quad \sem{\shift\rbra[\big]{\wsem{c},2\hdim(c)}}\cdots\sem{\shift\rbra[\big]{\wsem{c},(\hdim(b)-1)\times \hdim(c)}} \\
                ={}& \cdots \\
                ={}& \underbrace{\sem{\wsem{c}\in \cG_{\hdim(c)}^*}\oplus\sem{\wsem{c}\in \cG_{\hdim(c)}^*}\oplus \cdots \oplus \sem{\wsem{c}\in \cG_{\hdim(c)}^*}}_{\hdim(b)} \\
                ={}& I_{\hdim(b)}\otimes \sem{\wsem{c}\in \cG_{\hdim(c)}^*} \\
                ={}& I_{\hdim(b)}\otimes \sem{c} \tag{by inductive hypothesis} \\
                ={}& \sem{\pid_{b}}\otimes \sem{c} = \sem{\pid_{b}\times c}.
            \end{align*}
            \item For general cases of $c = c_1\times c_2$ with $c_1:b_1\fromto b_3$ and $c_2:b_2\fromto b_4$. We have $\hdim(c_1) = \hdim(b_1)=\hdim(b_3), \hdim(c_2) = \hdim(b_2) = \hdim(b_4)$. \begin{align*}
                \sem{\wsem{c_1\times c_2}} ={}& \sem{\wsem{\swapt}\wsem{\pid_{b_4}\times c_1}\wsem{\swapt}\wsem{\pid_{b_1}\times c_2}} \\
                ={}& \sem{\wsem{\swapt}}\cdot\sem{\wsem{\pid_{b_4}\times c_1}}\cdot\sem{\wsem{\swapt}}\cdot\sem{\wsem{\pid_{b_1}\times c_2}} \\
                ={}& \sem{{\swapt}}\cdot\sem{{\pid_{b_4}\times c_1}}\cdot\sem{{\swapt}}\cdot\sem{{\pid_{b_1}\times c_2}} \tag{by inductive hypothesis} \\
                ={}& \sem{\swapt}\cdot (I_{\hdim(c_2)}\otimes \sem{c_1}) \cdot \sem{\swapt} \cdot (I_{\hdim(c_1)}\otimes \sem{c_2}) \\
                ={}& (\sem{c_1}\otimes I_{\hdim(c_2)})\cdot (I_{\hdim(c_1)}\otimes \sem{c_2}) \tag{by semantic of $\swapt$} \\
                ={}& \sem{c_1}\otimes \sem{c_2} = \sem{c_1\times c_2}.
            \end{align*}
            \qedhere
        \end{itemize}
    \end{itemize}
\end{proof}

\lemqceqs*
\begin{proof}
    \begin{align*}
        \scalebox{\scaleratio}{\begin{quantikz}[row sep=2mm,column sep=2mm]
            &\ctrl{1} &\\
            &\gate{\p{z}} &
        \end{quantikz}} & \fromto_2 \pid_{\bb{2}}+\pid_{\bbo}+(-1) \\
        &\fromto_2 (\pid_\bbo+\swapp+\pid_\bbo)\seqq (\pid_{\bb{2}}+\pid_{\bbo}+(-1))\seqq(\pid_\bbo+\swapp+\pid_\bbo) \tag{by bifunctoriality of $+$ and naturality of $\swapp$} \\
        &\fromto_2 \swapt_{\bb{2}\times \bb{2}}\seqq (\pid_{\bb{2}}+\pid_{\bbo}+(-1))\seqq\swapt_{\bb{2}\times\bb{2}} \tag{by completeness of $\Pi$} \\
        &\fromto_2 \scalebox{\scaleratio}{\begin{quantikz}[row sep=2mm,column sep=2mm]
            & \swap{1} & \ctrl{1} & \swap{1} &\\
            & \targX{} & \gate{\p{z}} & \targX{} &
        \end{quantikz}} \fromto_2 \scalebox{\scaleratio}{\begin{quantikz}[row sep=2mm,column sep=2mm]
            & \gate{\p{z}} &\\
            & \ctrl{-1} &
        \end{quantikz}} \tag{by \cref{eq:qt5}} \\
        \\
        \scalebox{\scaleratio}{\begin{quantikz}[row sep=2mm,column sep=2mm]
            &&\gate{\p{x}}&\ctrl{1} & \gate{\p{x}} &\\
            &\qwbundle{}&&\gate{c} & &
        \end{quantikz}}
        &\fromto_2{} (\swapp_{\bb{2}}\times \pid) \seqq (\pid+c) \seqq (\swapp_{\bb{2}}\times \pid) \\
        &\fromto_2{} \swapp \seqq (\pid+c) \seqq \swapp \tag{by completeness of $\Pi$} \\
        &\fromto_2{} c+\pid \tag{by naturality of $\swapp$}
        \fromto_2{} \scalebox{\scaleratio}{\begin{quantikz}[row sep=2mm,column sep=2mm]
            &&\octrl{1} &\\
            &\qwbundle{}&\gate{c} &
        \end{quantikz}} \\
        \\
        \scalebox{\scaleratio}{\begin{quantikz}[row sep=2mm,column sep=2mm]
            &&\ctrl{1} &\\
            &\qwbundle{}&\gate{d\seqq c\seqq d^{-1}} &
        \end{quantikz}} &\fromto_2 \scalebox{\scaleratio}{\begin{quantikz}[row sep=2mm,column sep=2mm]
            && \ctrl{1} & \ctrl{1} & \ctrl{1} & \\
            &\qwbundle{}&\gate{d} & \gate{c}& \gate{d^{-1}} &
        \end{quantikz}} \tag{by \cref{eq:qt0}} \\
        &\fromto_2 \scalebox{\scaleratio}{\begin{quantikz}[row sep=2mm,column sep=2mm]
            && & \octrl{1} & \ctrl{1} & \octrl{1} & &\\
            &\qwbundle{}&\gate{d} &\gate{d} & \gate{c}& \gate{d^{-1}} & \gate{d^{-1}}
        \end{quantikz}} \tag{by \cref{eq:qt2,eq:qt3}} \\
        &\fromto_2 \scalebox{\scaleratio}{\begin{quantikz}[row sep=2mm,column sep=2mm]
            && & \ctrl{1} & \octrl{1} & \octrl{1} & &\\
            &\qwbundle{}&\gate{d} &\gate{c} & \gate{d}& \gate{d^{-1}} & \gate{d^{-1}}
        \end{quantikz}} \tag{by \cref{eq:qt2}} \\
        &\fromto_2 \scalebox{\scaleratio}{\begin{quantikz}[row sep=2mm,column sep=2mm]
            && & \ctrl{1} & & \ctrl{1} & & &\\
            &\qwbundle{}&\gate{d} &\gate{c} & \gate{d}& \gate{d\seqq d^{-1}} & \gate{d^{-1}} & \gate{d^{-1}}
        \end{quantikz}} \tag{by \cref{eq:qt3,eq:qt0}} \\
        &\fromto_2 \scalebox{\scaleratio}{\begin{quantikz}[row sep=2mm,column sep=2mm]
            && & \ctrl{1} & & \ctrl{1} & & &\\
            &\qwbundle{}&\gate{d} &\gate{c} & \gate{d}& \gate{\pid} & \gate{d^{-1}} & \gate{d^{-1}}
        \end{quantikz}} \tag{$d^{-1}$ is the inverse of $d$}\\
        &\fromto_2 \scalebox{\scaleratio}{\begin{quantikz}[row sep=2mm,column sep=2mm]
            && & \ctrl{1} &  &\\
            &\qwbundle{}&\gate{d} &\gate{c} & \gate{d^{-1}}
        \end{quantikz}} \tag{by completeness of $\Pi$ and $d^{-1}$ is the inverse of $d$}
    \end{align*}
\end{proof}

\lemweq*
\begin{proof}
    \ifSHOWPROOF%
    In addition to \cref{rel:d3,rel:d4} proved in the main text, we get graphical proofs for the remaining cases of \cref{rel:a1,rel:a2,rel:d2} with fixed indices here.
    \begin{itemize}
        \item \cref{rel:a1}: $\CZ{1}^2\approx \varepsilon$. \begin{align*}
            \qT[1]{\CZ{1}^2} ={}& 
            \scalebox{\scaleratio}{\begin{quantikz}[row sep=3mm,column sep=2mm]
                \push{\bbo} & \gate{\qT[1]{\CZ{1}}} & \gate{\qT[1]{\CZ{1}}} & 
            \end{quantikz}} \\
            ={}& \scalebox{\scaleratio}{\begin{quantikz}[row sep=3mm,column sep=2mm]
                \push{\bbo} & \gate{(-1)} & \gate{(-1)} & 
            \end{quantikz}} \\
            \fromto_2{}& \scalebox{\scaleratio}{\begin{quantikz}[row sep=3mm,column sep=2mm]
                \push{\bbo} & & & &
            \end{quantikz}} \tag{by \cref{eq:e1}} \\
            \fromto_2{}& \scalebox{\scaleratio}{\begin{quantikz}[row sep=3mm,column sep=2mm]
                \push{\bbo} & \gate{\qT{\varepsilon}} & 
            \end{quantikz}} = \qT[1]{\varepsilon}.
        \end{align*}
        \item \cref{rel:a2}: $\CX{1,2}^2 \approx \varepsilon$. \begin{align*}
            \qT[2]{\CX{1,2}^2} ={}& \scalebox{\scaleratio}{\begin{quantikz}[row sep=3mm,column sep=2mm]
                \lstick[2,brackets=none]{$+$\hspace{-.65cm}}\push{\bbo} & \gate[2]{\qT[2]{\CX{1,2}}} & \gate[2]{\qT[2]{\CX{1,2}}} & \\
                \push{\bbo}& & &
            \end{quantikz}} \\
            =& \scalebox{\scaleratio}{\begin{quantikz}[row sep=3mm,column sep=2mm]
                \lstick[2,brackets=none]{$+$\hspace{-.65cm}}\push{\bbo} & \permute{2,1} & \permute{2,1} & \\
                \push{\bbo}& & &
            \end{quantikz}} \\
            \fromto_2{}& \scalebox{\scaleratio}{\begin{quantikz}[row sep=3mm,column sep=2mm]
                \lstick[2,brackets=none]{$+$\hspace{-.65cm}}\push{\bbo} &  & & \\
                \push{\bbo}& & &
            \end{quantikz}} \tag{by completeness of $\Pi$} \\
            \fromto_2{}& \scalebox{\scaleratio}{\begin{quantikz}[row sep=3mm,column sep=2mm]
                \lstick[2,brackets=none]{$+$\hspace{-.65cm}}\push{\bbo} & \gate[2]{\qT{\varepsilon}} & \\
                \push{\bbo}& &
            \end{quantikz}} = \qT[2]{\varepsilon}.
        \end{align*}
        \item \cref{rel:a3}: $\CH{1,2}^2 \approx \varepsilon$. \begin{align*}
            \qT[2]{\CH{1,2}^2} ={}& \scalebox{\scaleratio}{\begin{quantikz}[row sep=3mm,column sep=2mm]
                \lstick[2,brackets=none]{$+$\hspace{-.65cm}}\push{\bbo} & \gate[2]{\qT[2]{\CH{1,2}}} & \gate[2]{\qT[2]{\CH{1,2}}} & \\
                \push{\bbo}& & &
            \end{quantikz}} \\
            ={}& \scalebox{\scaleratio}{\begin{quantikz}[row sep=3mm,column sep=2mm]
                \lstick[2,brackets=none]{$+$\hspace{-.65cm}}\push{\bbo} & \gate[2]{\hadamard} & \gate[2]{\hadamard} & \\
                \push{\bbo}& & &
            \end{quantikz}} \\
            \fromto_2{}& \scalebox{\scaleratio}{\begin{quantikz}[row sep=3mm,column sep=2mm]
                \lstick[2,brackets=none]{$+$\hspace{-.65cm}}\push{\bbo} &  & & \\
                \push{\bbo}& & &
            \end{quantikz}} \tag{by \cref{eq:e2}} \\
            \fromto_2{}& \scalebox{\scaleratio}{\begin{quantikz}[row sep=3mm,column sep=2mm]
                \lstick[2,brackets=none]{$+$\hspace{-.65cm}}\push{\bbo} & \gate[2]{\qT{\varepsilon}} & \\
                \push{\bbo}& &
            \end{quantikz}} = \qT[2]{\varepsilon}.
        \end{align*}
        \item \cref{rel:b1}: $\CZ{1}\CZ{2} \approx \CZ{2}\CZ{1}$. \begin{align*}
            \qT[2]{\CZ{1}\CZ{2}} ={}& \scalebox{\scaleratio}{\begin{quantikz}[row sep=3mm,column sep=2mm]
                \lstick[2,brackets=none]{$+$\hspace{-.65cm}}\push{\bbo} & \gate[2]{\qT[2]{\CZ{2}}} & \gate[2]{\qT[2]{\CZ{1}}} & \\
                \push{\bbo}& & &
            \end{quantikz}} \\
            \fromto_2{}& \scalebox{\scaleratio}{\begin{quantikz}[row sep=3mm,column sep=2mm]
                \lstick[2,brackets=none]{$+$\hspace{-.65cm}}\push{\bbo} & & \permute{2,1} & & \permute{2,1} &\\
                \push{\bbo}& \gate{(-1)} & & \gate{(-1)} & &
            \end{quantikz}} \\
            \fromto_2{}& \scalebox{\scaleratio}{\begin{quantikz}[row sep=3mm,column sep=2mm]
                \lstick[2,brackets=none]{$+$\hspace{-.65cm}}\push{\bbo} & & \gate{(-1)} &\\
                \push{\bbo}& \gate{(-1)} &  &
            \end{quantikz}} \tag{by naturality of $\swapp$}\\
            \fromto_2{}& \scalebox{\scaleratio}{\begin{quantikz}[row sep=3mm,column sep=2mm]
                \lstick[2,brackets=none]{$+$\hspace{-.65cm}}\push{\bbo} & \gate{(-1)} & & \\
                \push{\bbo}& & \gate{(-1)} &
            \end{quantikz}} \tag{by bifunctoriality of $+$}\\
            \fromto_2{}& \scalebox{\scaleratio}{\begin{quantikz}[row sep=3mm,column sep=2mm]
                \lstick[2,brackets=none]{$+$\hspace{-.65cm}}\push{\bbo} & \permute{2,1} & & \permute{2,1} & &\\
                \push{\bbo} & & \gate{(-1)} & & \gate{(-1)} &
            \end{quantikz}} \tag{by naturality of $\swapp$} \\
            \fromto_2{}& \scalebox{\scaleratio}{\begin{quantikz}[row sep=3mm,column sep=2mm]
                \lstick[2,brackets=none]{$+$\hspace{-.65cm}}\push{\bbo} & \gate[2]{\qT[2]{\CZ{1}}} & \gate[2]{\qT[2]{\CZ{2}}} & \\
                \push{\bbo}& & &
            \end{quantikz}} = \qT[2]{\CZ{2}\CZ{1}}.
        \end{align*}
        \item \cref{rel:b2}: $\CZ{1}\CX{2,3} \approx \CX{2,3}\CZ{1}$. \begin{align*}
            \qT[3]{\CZ{1}\CX{2,3}} ={}& \scalebox{\scaleratio}{\begin{quantikz}[row sep=3mm,column sep=2mm]
                \lstick[2,brackets=none]{$+$\hspace{-.65cm}}\push{\bbo} & \gate[3]{\qT[3]{\CX{2,3}}} & \gate[3]{\qT[3]{\CZ{1}}} & \\
                \lstick[2,brackets=none]{$+$\hspace{-.65cm}}\push{\bbo}& & & \\
                \push{\bbo}& & &
            \end{quantikz}} \\
            \fromto_2{}& \scalebox{\scaleratio}{\begin{quantikz}[row sep=3mm,column sep=2mm]
                \lstick[2,brackets=none]{$+$\hspace{-.65cm}}\push{\bbo} & & \permute{3,2,1} && \permute{3,2,1} & \\
                \lstick[2,brackets=none]{$+$\hspace{-.65cm}}\push{\bbo}& \permute{2,1} & & & & \\
                \push{\bbo}& & & \gate{(-1)} & &
            \end{quantikz}} \\
            \fromto_2{}& \scalebox{\scaleratio}{\begin{quantikz}[row sep=3mm,column sep=2mm]
                \lstick[2,brackets=none]{$+$\hspace{-.65cm}}\push{\bbo} & & \gate{(-1)} & \\
                \lstick[2,brackets=none]{$+$\hspace{-.65cm}}\push{\bbo}& \permute{2,1} & & \\
                \push{\bbo}& & &
            \end{quantikz}} \tag{by naturality of $\swapp$} \\
            \fromto_2{}& \scalebox{\scaleratio}{\begin{quantikz}[row sep=3mm,column sep=2mm]
                \lstick[2,brackets=none]{$+$\hspace{-.65cm}}\push{\bbo} & \gate{(-1)} & & \\
                \lstick[2,brackets=none]{$+$\hspace{-.65cm}}\push{\bbo}& & \permute{2,1} & \\
                \push{\bbo}& & &
            \end{quantikz}} \tag{by bifunctoriality of $+$}\\
            \fromto_2{}& \scalebox{\scaleratio}{\begin{quantikz}[row sep=3mm,column sep=2mm]
                \lstick[2,brackets=none]{$+$\hspace{-.65cm}}\push{\bbo} & \permute{3,2,1} && \permute{3,2,1} & &\\
                \lstick[2,brackets=none]{$+$\hspace{-.65cm}}\push{\bbo}& & & & \permute{2,1} & \\
                \push{\bbo}& & \gate{(-1)} & & &
            \end{quantikz}} \tag{by naturality of $\swapp$}\\
            \fromto_2{}& \scalebox{\scaleratio}{\begin{quantikz}[row sep=3mm,column sep=2mm]
                \lstick[2,brackets=none]{$+$\hspace{-.65cm}}\push{\bbo} & \gate[3]{\qT[3]{\CZ{1}}} & \gate[3]{\qT[3]{\CX{2,3}}} & \\
                \lstick[2,brackets=none]{$+$\hspace{-.65cm}}\push{\bbo}& & & \\
                \push{\bbo}& & &
            \end{quantikz}} = \qT[3]{\CX{2,3}\CZ{1}}.
        \end{align*}
        \item \cref{rel:b3}: $\CX{1,2}\CX{3,4} \approx \CX{3,4}\CX{1,2}$. \begin{align*}
            \qT[4]{\CX{1,2}\CX{3,4}} ={}& \scalebox{\scaleratio}{\begin{quantikz}[row sep=3mm,column sep=2mm]
                \lstick[2,brackets=none]{$+$\hspace{-.65cm}}\push{\bbo} & \gate[4]{\qT[4]{\CX{3,4}}} & \gate[4]{\qT[4]{\CX{1,2}}} & \\
                \lstick[2,brackets=none]{$+$\hspace{-.65cm}}\push{\bbo}& & & \\
                \lstick[2,brackets=none]{$+$\hspace{-.65cm}}\push{\bbo}& & & \\
                \push{\bbo}& & &
            \end{quantikz}} \\
            \fromto_2{}& \scalebox{\scaleratio}{\begin{quantikz}[row sep=3mm,column sep=2mm]
                \lstick[2,brackets=none]{$+$\hspace{-.65cm}}\push{\bbo} & \permute{1,2,4,3} & \permute{2,1} & \\
                \lstick[2,brackets=none]{$+$\hspace{-.65cm}}\push{\bbo}& & & \\
                \lstick[2,brackets=none]{$+$\hspace{-.65cm}}\push{\bbo}& & & \\
                \push{\bbo}& & &
            \end{quantikz}} \\
            \fromto_2{}& \scalebox{\scaleratio}{\begin{quantikz}[row sep=3mm,column sep=2mm]
                \lstick[2,brackets=none]{$+$\hspace{-.65cm}}\push{\bbo} & \permute{2,1} & \permute{1,2,4,3} & \\
                \lstick[2,brackets=none]{$+$\hspace{-.65cm}}\push{\bbo}& & & \\
                \lstick[2,brackets=none]{$+$\hspace{-.65cm}}\push{\bbo}& & & \\
                \push{\bbo}& & &
            \end{quantikz}} \\
            \fromto_2{}& \scalebox{\scaleratio}{\begin{quantikz}[row sep=3mm,column sep=2mm]
                \lstick[2,brackets=none]{$+$\hspace{-.65cm}}\push{\bbo} & \gate[4]{\qT[4]{\CX{1,2}}} & \gate[4]{\qT[4]{\CX{3,4}}} & \\
                \lstick[2,brackets=none]{$+$\hspace{-.65cm}}\push{\bbo}& & & \\
                \lstick[2,brackets=none]{$+$\hspace{-.65cm}}\push{\bbo}& & & \\
                \push{\bbo}& & &
            \end{quantikz}} = \qT[4]{\CX{3,4}\CX{1,2}}.
        \end{align*}
        \item \cref{rel:b4}: $\CZ{1}\CH{2,3} \approx \CH{2,3}\CZ{1}$. \begin{align*}
            \qT[3]{\CZ{1}\CH{2,3}} ={}& \scalebox{\scaleratio}{\begin{quantikz}[row sep=3mm,column sep=2mm]
                \lstick[2,brackets=none]{$+$\hspace{-.65cm}}\push{\bbo} & \gate[3]{\qT[3]{\CH{2,3}}} & \gate[3]{\qT[3]{\CZ{1}}} & \\
                \lstick[2,brackets=none]{$+$\hspace{-.65cm}}\push{\bbo}& & & \\
                \push{\bbo}& & & \\
            \end{quantikz}} \\
            \fromto_2{}& \scalebox{\scaleratio}{\begin{quantikz}[row sep=3mm,column sep=2mm]
                \lstick[2,brackets=none]{$+$\hspace{-.65cm}}\push{\bbo} &  & \permute{3,2,1} & & \permute{3,2,1} & \\
                \lstick[2,brackets=none]{$+$\hspace{-.65cm}}\push{\bbo}& \gate[2]{\hadamard} & & & & \\
                \push{\bbo}& & & \gate{(-1)} & & \\
            \end{quantikz}} \\
            \fromto_2{}& \scalebox{\scaleratio}{\begin{quantikz}[row sep=3mm,column sep=2mm]
                \lstick[2,brackets=none]{$+$\hspace{-.65cm}}\push{\bbo} &  & \gate{(-1)} & \\
                \lstick[2,brackets=none]{$+$\hspace{-.65cm}}\push{\bbo}& \gate[2]{\hadamard} & & \\
                \push{\bbo}& & & \\
            \end{quantikz}} \tag{by naturality of $\swapp$}\\
            \fromto_2{}& \scalebox{\scaleratio}{\begin{quantikz}[row sep=3mm,column sep=2mm]
                \lstick[2,brackets=none]{$+$\hspace{-.65cm}}\push{\bbo} & \gate{(-1)} & & \\
                \lstick[2,brackets=none]{$+$\hspace{-.65cm}}\push{\bbo}& & \gate[2]{\hadamard} & \\
                \push{\bbo}& & &
            \end{quantikz}} \tag{by bifunctoriality of $+$}\\
            \fromto_2{}& \scalebox{\scaleratio}{\begin{quantikz}[row sep=3mm,column sep=2mm]
                \lstick[2,brackets=none]{$+$\hspace{-.65cm}}\push{\bbo} & \permute{3,2,1} & & \permute{3,2,1} & & \\
                \lstick[2,brackets=none]{$+$\hspace{-.65cm}}\push{\bbo}&&& & \gate[2]{\hadamard} & \\
                \push{\bbo}& & \gate{(-1)} & & &
            \end{quantikz}} \tag{by naturality of $\swapp$}\\
            \fromto_2{}& \scalebox{\scaleratio}{\begin{quantikz}[row sep=3mm,column sep=2mm]
                \lstick[2,brackets=none]{$+$\hspace{-.65cm}}\push{\bbo} & \gate[3]{\qT[3]{\CZ{1}}} & \gate[3]{\qT[3]{\CH{2,3}}} & \\
                \lstick[2,brackets=none]{$+$\hspace{-.65cm}}\push{\bbo}& & & \\
                \push{\bbo}& & & \\
            \end{quantikz}} = \qT[3]{\CH{2,3}\CZ{1}}.
        \end{align*}
        \item \cref{rel:b5}: $\CX{1,2}\CH{3,4} \approx \CH{3,4}\CX{1,2}$. \begin{align*}
            \qT[4]{\CX{1,2}\CH{3,4}} ={}& \scalebox{\scaleratio}{\begin{quantikz}[row sep=3mm,column sep=2mm]
                \lstick[2,brackets=none]{$+$\hspace{-.65cm}}\push{\bbo} & \gate[4]{\qT[4]{\CX{3,4}}} & \gate[4]{\qT[4]{\CH{1,2}}} & \\
                \lstick[2,brackets=none]{$+$\hspace{-.65cm}}\push{\bbo}& & & \\
                \lstick[2,brackets=none]{$+$\hspace{-.65cm}}\push{\bbo}& & & \\
                \push{\bbo}& & &
            \end{quantikz}} \\
            \fromto_2{}& \scalebox{\scaleratio}{\begin{quantikz}[row sep=3mm,column sep=2mm]
                \lstick[2,brackets=none]{$+$\hspace{-.65cm}}\push{\bbo} & \permute{1,2,4,3} & \permute{3,4,1,2}& & \permute{3,4,1,2} & \\
                \lstick[2,brackets=none]{$+$\hspace{-.65cm}}\push{\bbo}& & && & \\
                \lstick[2,brackets=none]{$+$\hspace{-.65cm}}\push{\bbo}& & &\gate[2]{\hadamard}& & \\
                \push{\bbo}& & && &
            \end{quantikz}} \\
            \fromto_2{}& \scalebox{\scaleratio}{\begin{quantikz}[row sep=3mm,column sep=2mm]
                \lstick[2,brackets=none]{$+$\hspace{-.65cm}}\push{\bbo} & \permute{1,2,4,3} & \gate[2]{\hadamard} & \\
                \lstick[2,brackets=none]{$+$\hspace{-.65cm}}\push{\bbo}& & & \\
                \lstick[2,brackets=none]{$+$\hspace{-.65cm}}\push{\bbo}& & & \\
                \push{\bbo}& & &
            \end{quantikz}} \tag{by naturality of $\swapp$}\\
            \fromto_2{}& \scalebox{\scaleratio}{\begin{quantikz}[row sep=3mm,column sep=2mm]
                \lstick[2,brackets=none]{$+$\hspace{-.65cm}}\push{\bbo} & \gate[2]{\hadamard} & \permute{1,2,4,3} & \\
                \lstick[2,brackets=none]{$+$\hspace{-.65cm}}\push{\bbo}& & & \\
                \lstick[2,brackets=none]{$+$\hspace{-.65cm}}\push{\bbo}& & & \\
                \push{\bbo}& & &
            \end{quantikz}} \tag{by bifunctoriality of $+$}\\
            \fromto_2{}& \scalebox{\scaleratio}{\begin{quantikz}[row sep=3mm,column sep=2mm]
                \lstick[2,brackets=none]{$+$\hspace{-.65cm}}\push{\bbo} & \permute{3,4,1,2}& & \permute{3,4,1,2} & \permute{1,2,4,3} & \\
                \lstick[2,brackets=none]{$+$\hspace{-.65cm}}\push{\bbo}& & && & \\
                \lstick[2,brackets=none]{$+$\hspace{-.65cm}}\push{\bbo}& &\gate[2]{\hadamard}& & & \\
                \push{\bbo}& & && &
            \end{quantikz}} \tag{by naturality of $\swapp$}\\
            \fromto_2{}& \scalebox{\scaleratio}{\begin{quantikz}[row sep=3mm,column sep=2mm]
                \lstick[2,brackets=none]{$+$\hspace{-.65cm}}\push{\bbo} & \gate[4]{\qT[4]{\CH{1,2}}} & \gate[4]{\qT[4]{\CX{3,4}}} & \\
                \lstick[2,brackets=none]{$+$\hspace{-.65cm}}\push{\bbo}& & & \\
                \lstick[2,brackets=none]{$+$\hspace{-.65cm}}\push{\bbo}& & & \\
                \push{\bbo}& & &
            \end{quantikz}} = \qT[4]{\CH{3,4}\CX{1,2}}.
        \end{align*}
        \item \cref{rel:b6}: $\CH{1,2}\CH{3,4} \approx \CH{3,4}\CH{1,2}$. \begin{align*}
            \qT[4]{\CH{1,2}\CH{3,4}} ={}& \scalebox{\scaleratio}{\begin{quantikz}[row sep=3mm,column sep=2mm]
                \lstick[2,brackets=none]{$+$\hspace{-.65cm}}\push{\bbo} & \gate[4]{\qT[4]{\CH{3,4}}} & \gate[4]{\qT[4]{\CH{1,2}}} & \\
                \lstick[2,brackets=none]{$+$\hspace{-.65cm}}\push{\bbo}& & & \\
                \lstick[2,brackets=none]{$+$\hspace{-.65cm}}\push{\bbo}& & & \\
                \push{\bbo}& & &
            \end{quantikz}} \\
            \fromto_2{}& \scalebox{\scaleratio}{\begin{quantikz}[row sep=3mm,column sep=2mm]
                \lstick[2,brackets=none]{$+$\hspace{-.65cm}}\push{\bbo} & & \permute{3,4,1,2} & & \permute{3,4,1,2} & \\
                \lstick[2,brackets=none]{$+$\hspace{-.65cm}}\push{\bbo}& & & &&\\
                \lstick[2,brackets=none]{$+$\hspace{-.65cm}}\push{\bbo}& \gate[2]{\hadamard} & & \gate[2]{\hadamard} &&\\
                \push{\bbo}& & & &&
            \end{quantikz}} \\
            \fromto_2{}& \scalebox{\scaleratio}{\begin{quantikz}[row sep=3mm,column sep=2mm]
                \lstick[2,brackets=none]{$+$\hspace{-.65cm}}\push{\bbo} & & \gate[2]{\hadamard} & \\
                \lstick[2,brackets=none]{$+$\hspace{-.65cm}}\push{\bbo}& & & \\
                \lstick[2,brackets=none]{$+$\hspace{-.65cm}}\push{\bbo}& \gate[2]{\hadamard} & & \\
                \push{\bbo}& & &
            \end{quantikz}} \tag{by naturality of $\swapp$}\\
            \fromto_2{}& \scalebox{\scaleratio}{\begin{quantikz}[row sep=3mm,column sep=2mm]
                \lstick[2,brackets=none]{$+$\hspace{-.65cm}}\push{\bbo} & \gate[2]{\hadamard} & & \\
                \lstick[2,brackets=none]{$+$\hspace{-.65cm}}\push{\bbo}& & & \\
                \lstick[2,brackets=none]{$+$\hspace{-.65cm}}\push{\bbo}& & \gate[2]{\hadamard} & \\
                \push{\bbo}& & &
            \end{quantikz}} \tag{by bifunctoriality of $+$}\\
            \fromto_2{}& \scalebox{\scaleratio}{\begin{quantikz}[row sep=3mm,column sep=2mm]
                \lstick[2,brackets=none]{$+$\hspace{-.65cm}}\push{\bbo} & \permute{3,4,1,2} & & \permute{3,4,1,2} & & \\
                \lstick[2,brackets=none]{$+$\hspace{-.65cm}}\push{\bbo}& & & &&\\
                \lstick[2,brackets=none]{$+$\hspace{-.65cm}}\push{\bbo}& & \gate[2]{\hadamard} & & \gate[2]{\hadamard} &\\
                \push{\bbo}& & & &&
            \end{quantikz}} \tag{by naturality of $\swapp$}\\
            \fromto_2{}& \scalebox{\scaleratio}{\begin{quantikz}[row sep=3mm,column sep=2mm]
                \lstick[2,brackets=none]{$+$\hspace{-.65cm}}\push{\bbo} & \gate[4]{\qT[4]{\CH{1,2}}} & \gate[4]{\qT[4]{\CH{3,4}}} & \\
                \lstick[2,brackets=none]{$+$\hspace{-.65cm}}\push{\bbo}& & & \\
                \lstick[2,brackets=none]{$+$\hspace{-.65cm}}\push{\bbo}& & & \\
                \push{\bbo}& & &
            \end{quantikz}} = \qT[4]{\CH{3,4}\CH{1,2}}.
        \end{align*}
        \item \cref{rel:c1}: $\CZ{1}\CX{1,2} \approx \CX{1,2}\CZ{2}$. \begin{align*}
            \qT[2]{\CZ{1}\CX{1,2}} ={}& \scalebox{\scaleratio}{\begin{quantikz}[row sep=3mm,column sep=2mm]
                \lstick[2,brackets=none]{$+$\hspace{-.65cm}}\push{\bbo} & \gate[2]{\qT[2]{\CX{1,2}}} & \gate[2]{\qT[2]{\CZ{1}}} & \\
                \push{\bbo}& & &
            \end{quantikz}} \\
            \fromto_2{}& \scalebox{\scaleratio}{\begin{quantikz}[row sep=3mm,column sep=2mm]
                \lstick[2,brackets=none]{$+$\hspace{-.65cm}}\push{\bbo} & \permute{2,1} & \permute{2,1} & & \permute{2,1} & \\
                \push{\bbo}& & & \gate{(-1)} & &
            \end{quantikz}} \\
            \fromto_2{}& \scalebox{\scaleratio}{\begin{quantikz}[row sep=3mm,column sep=2mm]
                \lstick[2,brackets=none]{$+$\hspace{-.65cm}}\push{\bbo} & & \permute{2,1} & \\
                \push{\bbo}& \gate{(-1)} & &
            \end{quantikz}} \tag{by completeness of $\Pi$} \\
            \fromto_2{}& \scalebox{\scaleratio}{\begin{quantikz}[row sep=3mm,column sep=2mm]
                \lstick[2,brackets=none]{$+$\hspace{-.65cm}}\push{\bbo}  & \gate[2]{\qT[2]{\CZ{2}}} & \gate[2]{\qT[2]{\CX{1,2}}} & \\
                \push{\bbo}& & &
            \end{quantikz}} = \qT[2]{\CX{1,2}\CZ{2}}.
        \end{align*}
        \item \cref{rel:c2}: $\CX{2,3}\CX{1,2} \approx \CX{1,2}\CX{1,3}$. \begin{align*}
            \qT[3]{\CX{2,3}\CX{1,2}} ={}& \scalebox{\scaleratio}{\begin{quantikz}[row sep=3mm,column sep=2mm]
                \lstick[2,brackets=none]{$+$\hspace{-.65cm}}\push{\bbo}  & \gate[3]{\qT[3]{\CX{1,2}}} & \gate[3]{\qT[3]{\CX{2,3}}} & \\
                \lstick[2,brackets=none]{$+$\hspace{-.65cm}}\push{\bbo}  & & & \\
                \push{\bbo}& & &
            \end{quantikz}} \\
            \fromto_2{}& \scalebox{\scaleratio}{\begin{quantikz}[row sep=3mm,column sep=2mm]
                \lstick[2,brackets=none]{$+$\hspace{-.65cm}}\push{\bbo}  & \permute{2,1} & \permute{1,3,2} & \\
                \lstick[2,brackets=none]{$+$\hspace{-.65cm}}\push{\bbo}  & & & \\
                \push{\bbo}& & &
            \end{quantikz}} \\
            \fromto_2{}& \scalebox{\scaleratio}{\begin{quantikz}[row sep=3mm,column sep=2mm]
                \lstick[2,brackets=none]{$+$\hspace{-.65cm}}\push{\bbo}  & \permute{3,2,1} & \permute{2,1} & \\
                \lstick[2,brackets=none]{$+$\hspace{-.65cm}}\push{\bbo}  & & & \\
                \push{\bbo}& & &
            \end{quantikz}} \tag{by completeness of $\Pi$} \\
            \fromto_2{}& \scalebox{\scaleratio}{\begin{quantikz}[row sep=3mm,column sep=2mm]
                \lstick[2,brackets=none]{$+$\hspace{-.65cm}}\push{\bbo}  & \gate[3]{\qT[3]{\CX{1,3}}} & \gate[3]{\qT[3]{\CX{1,2}}} & \\
                \lstick[2,brackets=none]{$+$\hspace{-.65cm}}\push{\bbo}  & & & \\
                \push{\bbo}& & &
            \end{quantikz}} = \qT[3]{\CX{1,2}\CX{2,3}}.
        \end{align*}
        \item \cref{rel:c3}: $\CX{1,3}\CX{2,3} \approx \CX{2,3}\CX{1,2}$. \begin{align*}
            \qT[3]{\CX{1,3}\CX{2,3}} ={}& \scalebox{\scaleratio}{\begin{quantikz}[row sep=3mm,column sep=2mm]
                \lstick[2,brackets=none]{$+$\hspace{-.65cm}}\push{\bbo}  & \gate[3]{\qT[3]{\CX{2,3}}} & \gate[3]{\qT[3]{\CX{1,3}}} & \\
                \lstick[2,brackets=none]{$+$\hspace{-.65cm}}\push{\bbo}  & & & \\
                \push{\bbo}& & &
            \end{quantikz}} \\
            \fromto_2{}& \scalebox{\scaleratio}{\begin{quantikz}[row sep=3mm,column sep=2mm]
                \lstick[2,brackets=none]{$+$\hspace{-.65cm}}\push{\bbo}  & \permute{1,3,2} & \permute{3,2,1} & \\
                \lstick[2,brackets=none]{$+$\hspace{-.65cm}}\push{\bbo}  & & & \\
                \push{\bbo}& & &
            \end{quantikz}} \\
            \fromto_2{}& \scalebox{\scaleratio}{\begin{quantikz}[row sep=3mm,column sep=2mm]
                \lstick[2,brackets=none]{$+$\hspace{-.65cm}}\push{\bbo}  & \permute{2,1} & \permute{1,3,2} & \\
                \lstick[2,brackets=none]{$+$\hspace{-.65cm}}\push{\bbo}  & & & \\
                \push{\bbo}& & &
            \end{quantikz}} \tag{by completeness of $\Pi$} \\
            \fromto_2{}& \scalebox{\scaleratio}{\begin{quantikz}[row sep=3mm,column sep=2mm]
                \lstick[2,brackets=none]{$+$\hspace{-.65cm}}\push{\bbo}  & \gate[3]{\qT[3]{\CX{1,2}}} & \gate[3]{\qT[3]{\CX{2,3}}} & \\
                \lstick[2,brackets=none]{$+$\hspace{-.65cm}}\push{\bbo}  & & & \\
                \push{\bbo}& & &
            \end{quantikz}} = \qT[3]{\CX{2,3}\CX{1,2}}.
        \end{align*}
        \item \cref{rel:c4}: $\CH{2,3}\CX{1,2} \approx \CX{1,2}\CH{1,3}$. \begin{align*}
            \qT[3]{\CH{2,3}\CX{1,2}} ={}& \scalebox{\scaleratio}{\begin{quantikz}[row sep=3mm,column sep=2mm]
                \lstick[2,brackets=none]{$+$\hspace{-.65cm}}\push{\bbo}  & \gate[3]{\qT[3]{\CX{1,2}}} & \gate[3]{\qT[3]{\CH{2,3}}} & \\
                \lstick[2,brackets=none]{$+$\hspace{-.65cm}}\push{\bbo}  & & & \\
                \push{\bbo}& & &
            \end{quantikz}} \\
            \fromto_2{}& \scalebox{\scaleratio}{\begin{quantikz}[row sep=3mm,column sep=2mm]
                \lstick[2,brackets=none]{$+$\hspace{-.65cm}}\push{\bbo}  & \permute{2,1} & & \\
                \lstick[2,brackets=none]{$+$\hspace{-.65cm}}\push{\bbo}  & & \gate[2]{\hadamard} & \\
                \push{\bbo}& & &
            \end{quantikz}} \\
            \fromto_2{}& \scalebox{\scaleratio}{\begin{quantikz}[row sep=3mm,column sep=2mm]
                \lstick[2,brackets=none]{$+$\hspace{-.65cm}}\push{\bbo}  & \permute{2,1} & & \permute{2,1} & \permute{2,1} & \\
                \lstick[2,brackets=none]{$+$\hspace{-.65cm}}\push{\bbo}  & & \gate[2]{\hadamard}& & & \\
                \push{\bbo}& & & & &
            \end{quantikz}} \tag{by completeness of $\Pi$} \\
            \fromto_2{}& \scalebox{\scaleratio}{\begin{quantikz}[row sep=3mm,column sep=2mm]
                \lstick[2,brackets=none]{$+$\hspace{-.65cm}}\push{\bbo}  & \gate[3]{\qT[3]{\CX{1,3}}} & \gate[3]{\qT[3]{\CX{1,2}}} & \\
                \lstick[2,brackets=none]{$+$\hspace{-.65cm}}\push{\bbo}  & & & \\
                \push{\bbo}& & &
            \end{quantikz}} = \qT[3]{\CX{1,2}\CX{2,3}}.
        \end{align*}
        \item \cref{rel:c5}: $\CH{1,3}\CX{2,3} \approx \CX{2,3}\CH{1,2}$. \begin{align*}
            \qT[3]{\CH{1,3}\CX{2,3}} ={}& \scalebox{\scaleratio}{\begin{quantikz}[row sep=3mm,column sep=2mm]
                \lstick[2,brackets=none]{$+$\hspace{-.65cm}}\push{\bbo}  & \gate[3]{\qT[3]{\CX{2,3}}} & \gate[3]{\qT[3]{\CH{1,3}}} & \\
                \lstick[2,brackets=none]{$+$\hspace{-.65cm}}\push{\bbo}  & & & \\
                \push{\bbo}& & &
            \end{quantikz}} \\
            \fromto_2{}& \scalebox{\scaleratio}{\begin{quantikz}[row sep=3mm,column sep=2mm]
                \lstick[2,brackets=none]{$+$\hspace{-.65cm}}\push{\bbo}  & \permute{1,3,2} & \permute{2,1} & & \permute{2,1} & \\
                \lstick[2,brackets=none]{$+$\hspace{-.65cm}}\push{\bbo}  & & & \gate[2]{\hadamard} & & \\
                \push{\bbo}& & & & &
            \end{quantikz}} \\
            \fromto_2{}& \scalebox{\scaleratio}{\begin{quantikz}[row sep=3mm,column sep=2mm]
                \lstick[2,brackets=none]{$+$\hspace{-.65cm}}\push{\bbo}& \permute{2,3,1} & & \permute{3,1,2} & \permute{1,3,2} & \\
                \lstick[2,brackets=none]{$+$\hspace{-.65cm}}\push{\bbo}& & \gate[2]{\hadamard} & & & \\
                \push{\bbo}& & & & &
            \end{quantikz}} \tag{by completeness of $\Pi$} \\
            \fromto_2{}& \scalebox{\scaleratio}{\begin{quantikz}[row sep=3mm,column sep=2mm]
                \lstick[2,brackets=none]{$+$\hspace{-.65cm}}\push{\bbo}  & \gate[3]{\qT[3]{\CH{1,2}}} & \gate[3]{\qT[3]{\CX{2,3}}} & \\
                \lstick[2,brackets=none]{$+$\hspace{-.65cm}}\push{\bbo}  & & & \\
                \push{\bbo}& & &
            \end{quantikz}} = \qT[3]{\CX{2,3}\CH{1,2}}.
        \end{align*}
        \item \cref{rel:d1}: $\CZ{1}\CZ{2}\CH{1,2} \approx \CH{1,2}\CZ{1}\CZ{2}$. By \cref{rel:b1}, we only need to prove the case of $\CZ{1}\CZ{2}\CH{1,2} \approx \CH{1,2}\CZ{2}\CZ{1}$.
        \begin{align*}
            & \qT[2]{\CZ{1}\CZ{2}\CH{1,2}} \\
            ={}& \scalebox{\scaleratio}{\begin{quantikz}[row sep=3mm,column sep=2mm]
                \lstick[2,brackets=none]{$+$\hspace{-.65cm}}\push{\bbo} & \gate[2]{\qT[2]{\CH{1,2}}} & \gate[2]{\qT[2]{\CZ{2}}} & \gate[2]{\qT[2]{\CZ{1}}} & \\
                \push{\bbo}& & & &
            \end{quantikz}} \\
            \fromto_2{}& \scalebox{\scaleratio}{\begin{quantikz}[row sep=3mm,column sep=2mm]
                \lstick[2,brackets=none]{$+$\hspace{-.65cm}}\push{\bbo} & \gate[2]{\hadamard} & & \permute{2,1} && \permute{2,1} & \\
                \push{\bbo}& & \gate{(-1)} & &\gate{(-1)}& &
            \end{quantikz}} \\
            \fromto_2{}& \scalebox{\scaleratio}{\begin{quantikz}[row sep=3mm,column sep=2mm]
                \lstick[2,brackets=none]{$+$\hspace{-.65cm}}\push{\bbo} & \gate[2]{\hadamard} & \gate[2]{\mathit{hxh}} & \gate[2]{\swapp} &\gate[2]{\mathit{hxh}}& \gate[2]{\swapp} & \\
                \push{\bbo}& &  & && &
            \end{quantikz}} \tag{by \cref{eq:e3}} \\
            \fromto_2{}& \scalebox{\scaleratio}{\begin{quantikz}[row sep=3mm,column sep=2mm]
                \lstick[2,brackets=none]{$+$\hspace{-.65cm}}\push{\bbo} & \gate[2]{\swapp} & \gate[2]{\hadamard} & \gate[2]{\swapp} &\gate[2]{\mathit{hxh}}& \gate[2]{\swapp} & \\
                \push{\bbo}& &  & && &
            \end{quantikz}} \tag{by \cref{eq:e2}} \\
            \fromto_2{}& \scalebox{\scaleratio}{\begin{quantikz}[row sep=3mm,column sep=2mm]
                \lstick[2,brackets=none]{$+$\hspace{-.65cm}}\push{\bbo} & \gate[2]{\swapp} & \gate[2]{\mathit{hxh}} & \gate[2]{\swapp} &\gate[2]{\hadamard}& \gate[2]{\swapp} & \\
                \push{\bbo}& &  & && &
            \end{quantikz}} \\
            \fromto_2{}& \scalebox{\scaleratio}{\begin{quantikz}[row sep=3mm,column sep=2mm]
                \lstick[2,brackets=none]{$+$\hspace{-.65cm}}\push{\bbo} & \gate[2]{\swapp} & \gate[2]{\mathit{hxh}} & \gate[2]{\swapp} &\gate[2]{\mathit{hxh}}& \gate[2]{\hadamard} & \\
                \push{\bbo}& &  & && &
            \end{quantikz}} \tag{by \cref{eq:e2}}\\
            \fromto_2{}& \scalebox{\scaleratio}{\begin{quantikz}[row sep=3mm,column sep=2mm]
                \lstick[2,brackets=none]{$+$\hspace{-.65cm}}\push{\bbo} & \permute{2,1} & & \permute{2,1} & & \gate[2]{\hadamard} & \\
                \push{\bbo}& & \gate{(-1)} & & \gate{(-1)} & &
            \end{quantikz}} \tag{by \cref{eq:e3}}\\
            \fromto_2{}& \scalebox{\scaleratio}{\begin{quantikz}[row sep=3mm,column sep=2mm]
                \lstick[2,brackets=none]{$+$\hspace{-.65cm}}\push{\bbo} & \gate[2]{\qT[2]{\CZ{1}}} & \gate[2]{\qT[2]{\CZ{2}}} & \gate[2]{\qT[2]{\CH{1,2}}} & \\
                \push{\bbo}& & & &
            \end{quantikz}} \\
            ={}& \qT[2]{\CH{1,2}\CZ{2}\CZ{1}}.
        \end{align*}
        \item \cref{rel:d2}: $\CZ{2}\CH{1,2} \approx \CH{1,2}\CX{1,2}$. \begin{align*}
            \qT[2]{\CZ{2}\CH{1,2}} ={}& \scalebox{\scaleratio}{\begin{quantikz}[row sep=3mm,column sep=2mm]
                \lstick[2,brackets=none]{$+$\hspace{-.65cm}}\push{\bbo} & \gate[2]{\qT[2]{\CH{1,2}}} & \gate[2]{\qT[2]{\CZ{2}}} & \\
                \push{\bbo}& & & 
            \end{quantikz}} \\
            \fromto_2{}& \scalebox{\scaleratio}{\begin{quantikz}[row sep=3mm,column sep=2mm]
                \lstick[2,brackets=none]{$+$\hspace{-.65cm}}\push{\bbo} & \gate[2]{\hadamard} &  & \\
                \push{\bbo}& & \gate{(-1)} & 
            \end{quantikz}} \\
            \fromto_2{}& \scalebox{\scaleratio}{\begin{quantikz}[row sep=3mm,column sep=2mm]
                \lstick[2,brackets=none]{$+$\hspace{-.65cm}}\push{\bbo} & \gate[2]{\hadamard} & \gate[2]{\mathit{hxh}} & \\
                \push{\bbo}& &  & 
            \end{quantikz}} \tag{by \cref{eq:e3}} \\
            \fromto_2{}& \scalebox{\scaleratio}{\begin{quantikz}[row sep=3mm,column sep=2mm]
                \lstick[2,brackets=none]{$+$\hspace{-.65cm}}\push{\bbo} & \gate[2]{\swapp} & \gate[2]{\hadamard} & \\
                \push{\bbo}& &  & 
            \end{quantikz}} \tag{by \cref{eq:e2}}\\
            \fromto_2{}& \scalebox{\scaleratio}{\begin{quantikz}[row sep=3mm,column sep=2mm]
                \lstick[2,brackets=none]{$+$\hspace{-.65cm}}\push{\bbo} & \gate[2]{\qT[2]{\CX{1,2}}} & \gate[2]{\qT[2]{\CH{1,2}}} & \\
                \push{\bbo}& &  & 
            \end{quantikz}} = \qT[2]{\CH{1,2}\CX{1,2}}.
        \end{align*}
    \end{itemize}
    \else%
    Set \verb|\SHOWPROOFtrue| to compile the proof.
    \fi
\end{proof}

\lemwqinv*
\begin{proof}
    We prove it by induction on the structure of $c$. For convenience, we omit the isomorphisms $\assocrp$, $\assoclp$, $\unitep$, $\unitip$, $\assocrt$, $\assoclt$, $\unitet$, $\unitit$, $\dist$, $\factor$, $\absorb$, and $\factorz$ used for type conversion.
    \begin{itemize}
        \item $c = (-1)$. \[\qT[\hdim((-1))]{\wsem{(-1)}} = \qT[1]{\CZ{1}} = \pid_{\bbo}\seqq(\pid_{\bbz}+(-1))\seqq\pid_{\bbo} \fromto_2{} (-1).\]
        \item $c = \hadamard$.
        \begin{align*}
            \qT[\hdim(\hadamard)]{\wsem{\hadamard}} ={}& \qT[2]{\CH{1,2}} \\
            ={}& \pid_{\bb{2}}\seqq\pid_{\bb{2}}\seqq(\pid_{\bbz}+\hadamard)\seqq\pid_{\bb{2}}\seqq\pid_{\bb{2}} \\
            \fromto_2{}& \hadamard.
        \end{align*}
        \item $c = \pid$, $\assocrp$, $\assoclp$, $\unitep$, $\unitip$, $\assocrt$, $\assoclt$, $\unitet$, $\unitit$, $\dist$, $\factor$, $\absorb$, or $\factorz$.
        \[ \qT[\hdim(c)]{\wsem{c}} = \qT[\hdim(c)]{\varepsilon} = \pid_{\hdim(c)\bbo} \fromto_2 c. \tag{by strictness}\]
        \item $c = \swapp$. By \cref{lem:w_faith,lem:q_faith}, we have \[\sem{\swapp} = \sem{\wsem{\swapp}} = \sem{\qT[\hdim(\swapp)]{\wsem{\swapp}}}.\] Then, we can get $\swapp \fromto_2 \qT[\hdim(\swapp)]{\wsem{\swapp}}$ through the completeness of $\Pi$, since they only contain primitives of classical $\Pi$.
        \item $c = \swapt$. By \cref{lem:w_faith,lem:q_faith}, we have \[\sem{\swapt} = \sem{\wsem{\swapt}} = \sem{\qT[\hdim(\swapt)]{\wsem{\swapt}}}.\] Then, we can get $\swapt \fromto_2 \qT[\hdim(\swapt)]{\wsem{\swapt}}$ through the completeness of $\Pi$, since they only contain primitives of classical $\Pi$.
        \item $c = c_1\seqq c_2$. It is easy to prove that for any $c: b_1\fromto b_2$, $\hdim(c) = \hdim(b_1) = \hdim(b_2)$, thus we can get $\hdim(c_1\seqq c_2) = \hdim(c_1)=\hdim(c_2)$. Then, \begin{align*}
            \qT[\hdim(c_1\seqq c_2)]{\wsem{c_1\seqq c_2}} ={}& \qT[\hdim(c_1\seqq c_2)]{\wsem{c_2}\wsem{c_1}} \\
            ={}& \qT[\hdim(c_1\seqq c_2)]{\wsem{c_1}}\seqq\qT[\hdim(c_1\seqq c_2)]{\wsem{c_2}} \\
            ={}& \qT[\hdim(c_1)]{\wsem{c_1}}\seqq\qT[\hdim(c_2)]{\wsem{c_2}} \\
            \fromto_2{}& c_1\seqq c_2. \tag{by inductive hypothesis}
        \end{align*}
        \item $c = c_1+c_2$. We already prove this case in the main text.
        \item $c = c_1\times c_2$. \begin{itemize}[leftmargin=10pt]
            \item For the special case $c = \pid_{b} \times c$. \begin{align*}
                & \qT[\hdim(c)]{\pid_b\times c_1} \\
                ={}& \cT_Q^{\hdim(c)}\big[\wsem{c_1} \\
                & \quad \shift(\wsem{c_1},\hdim(c_1))\shift(\wsem{c_1},2\hdim(c_1))\cdots\shift(\wsem{c_1},(\hdim(b)-1)\times \hdim(c_1))\big] \\
                ={}& \scalebox{.96}{$\qT[\hdim(c)]{\shift(\wsem{c_1},\hdim(c_1))\shift(\wsem{c_1},2\hdim(c_1))\cdots\shift(\wsem{c_1},(\hdim(b)-1)\times \hdim(c_1))}$}\seqq \\
                &\quad \qT[\hdim(c)]{\wsem{c_1}} \\
                \fromto_2{}& \scalebox{.96}{$\qT[\hdim(c)]{\shift(\wsem{c_1},\hdim(c_1))\shift(\wsem{c_1},2\hdim(c_1))\cdots\shift(\wsem{c_1},(\hdim(b)-1)\times \hdim(c_1))}$}\seqq \\
                &\quad \rbra*{\qT[\hdim(c_1)]{\wsem{c_1}}+\pid_{(\hdim(b)-1)\times\hdim(c_1)\bbo}} \tag{by \cref{lem:qT_nm} and $\hdim(c) = \hdim(b)\times \hdim(c_1)$} \\
                \fromto_2{}& \scalebox{.96}{$\qT[\hdim(c)]{\shift(\wsem{c_1},\hdim(c_1))\shift(\wsem{c_1},2\hdim(c_1))\cdots\shift(\wsem{c_1},(\hdim(b)-1)\times \hdim(c_1))}$}\seqq \\
                &\quad \rbra*{c_1+\pid_{(\hdim(b)-1)\times\hdim(c_1)\bbo}} \tag{by inductive hypothesis} \\
                ={}& \qT[\hdim(c)]{\shift(\wsem{c_1},2\hdim(c_1))\cdots\shift(\wsem{c_1},(\hdim(b)-1)\times \hdim(c_1))}\seqq \\
                &\quad \qT[\hdim(c)]{\shift(\wsem{c_1},\hdim(c_1))}\seqq\rbra*{c_1+\pid_{(\hdim(b)-1)\times\hdim(c_1)\bbo}} \\
                \fromto_2{}& \qT[\hdim(c)]{\shift(\wsem{c_1},2\hdim(c_1))\cdots\shift(\wsem{c_1},(\hdim(b)-1)\times \hdim(c_1))}\seqq \\
                &\quad \rbra*{\pid_{\hdim{c_1}\bbo}+\qT[(\hdim(b)-1)\times\hdim{c_1}]{\wsem{c_1}}}\seqq\rbra*{c_1+\pid_{(\hdim(b)-1)\times\hdim(c_1)\bbo}} \tag{by \cref{lem:qT_shift}} \\
                \fromto_2{}& \qT[\hdim(c)]{\shift(\wsem{c_1},2\hdim(c_1))\cdots\shift(\wsem{c_1},(\hdim(b)-1)\times \hdim(c_1))}\seqq \\
                &\quad \rbra*{\pid_{\hdim{c_1}\bbo}+\qT[\hdim{c_1}]{\wsem{c_1}}+\pid_{(\hdim(b)-2)\times \hdim(c_1)\bbo}}\seqq\rbra*{c_1+\pid_{(\hdim(b)-1)\times\hdim(c_1)\bbo}} \tag{by \cref{lem:qT_nm}} \\
                \fromto_2{}& \qT[\hdim(c)]{\shift(\wsem{c_1},2\hdim(c_1))\cdots\shift(\wsem{c_1},(\hdim(b)-1)\times \hdim(c_1))}\seqq \\
                &\quad \rbra*{\pid_{\hdim{c_1}\bbo}+c_1+\pid_{(\hdim(b)-2)\times \hdim(c_1)\bbo}}\seqq\rbra*{c_1+\pid_{(\hdim(b)-1)\times\hdim(c_1)\bbo}} \tag{by inductive hypothesis} \\
                \fromto_2{}& \qT[\hdim(c)]{\shift(\wsem{c_1},2\hdim(c_1))\cdots\shift(\wsem{c_1},(\hdim(b)-1)\times \hdim(c_1))}\seqq \\
                &\quad \rbra*{c_1+c_1+\pid_{(\hdim(b)-2)\times \hdim(c_1)\bbo}} \tag{by bifunctoriality of $+$} \\
                \fromto_2{}& \cdots \\
                \fromto_2{}& (\underbrace{c_1+c_1+\cdots+c_1}_{\hdim(b)}) \\
                \fromto_2{}& \pid_{\hdim(b)\bbo}\times c_1 \tag{by distributivity}\\
                \fromto_2{}& \pid_{b}\times c_1. \tag{by completeness of $\Pi$}
            \end{align*}
            \item For general cases of $c = c_1\times c_2$ with $c_1:b_1\fromto b_3$ and $c_2:b_2\fromto b_4$. We have $\hdim(c_1) = \hdim(b_1)=\hdim(b_3), \hdim(c_2) = \hdim(b_2) = \hdim(b_4)$. \begin{align*}
                & \qT[\hdim(c)]{\wsem{c_1\times c_2}} \\
                ={}& \qT[\hdim(c)]{\wsem{\swapt}\wsem{\pid_{b_4}\times c_1}\wsem{\swapt}\wsem{\pid_{b_1}\times c_2}} \\
                ={}&  \qT[\hdim(c)]{\wsem{\pid_{b_1}\times c_2}}\seqq \qT[\hdim(c)]{\wsem{\swapt}} \seqq \qT[\hdim(c)]{\wsem{\pid_{b_4}\times c_1}} \seqq \qT[\hdim(c)]{\wsem{\swapt}} \\
                \fromto_2{}& (\pid_{b_1}\times c_2)\seqq \swapt \seqq (\pid_{b_4}\times c_1) \seqq \swapt \tag{by inductive hypothesis} \\
                \fromto_2{}& (\pid_{b_1}\times c_2)\seqq (c_1\times\pid_{b_4}) \tag{by naturality of $\swapt$} \\
                \fromto_2{}& c_1\times c_2. \tag{by bifunctoriality of $\times$}
            \end{align*}
        \end{itemize}
    \end{itemize}
\end{proof}

\lemqtshift*
\begin{proof}
    By induction on $\mathbf{G}$.
    \begin{itemize}
        \item $\mathbf{G} = \CZ{a}$. \begin{align*}
            \qT[n+k]{\hshift(\CZ{a},k)}
            ={}& \qT[n+k]{\CZ{a+k}} \\
            ={}& \swapp(a+k,n+k)\seqq(\pid_{(n+k-1)\bbo}+(-1))\seqq\swapp(a+k,n+m) \\
            \fromto_2{}& (\pid_{k\bbo}+\swapp(a,n))\seqq(\pid_{k\bbo} + \pid_{(n-1)\bbo}+(-1))\seqq(\pid_{k\bbo}+\swapp(a,n)) \tag{by completeness of $\Pi$} \\
            \fromto_2{}& \pid_{k\bbo} + (\swapp(a,n)\seqq(\pid_{(n-1)\bbo}+(-1))\seqq\swapp(a,n)) \tag{by bifunctoriality of $+$} \\
            ={}& \pid_{k\bbo} + \qT[n]{\CZ{a}}.
        \end{align*}
        \item $\mathbf{G} = \CX{b,c}$. \begin{align*}
            \qT[n+k]{\hshift(\CX{b,c},k)}
            ={}& \qT[n+k]{\CX{b+k,c+k}} \\
            ={}& \swapp(b+k,c+k) \\
            \fromto_2{}& \pid_{k\bbo}+\swapp(b,c) \tag{by completeness of $\Pi$} \\
            =& \pid_{k\bbo}+\qT[n]{\CX{b,c}}.
        \end{align*}
        \item $\mathbf{G} = \CH{b,c}$. \begin{align*}
            & \qT[n+k]{\hshift(\CH{b,c},k)} \\
            ={}& \qT[n+k]{\CH{b+k,c+k}} \\
            ={}& \swapp(b+k,n+k-1)\seqq\swapp(c+k,n+k)\seqq(\pid_{(n+k-2)\bbo}+\hadamard)\seqq\\
            & \quad \swapp(c+k,n+k)\seqq\swapp(b+k,n+k-1) \\
            \fromto_2{}& (\pid_{k\bbo}+\swapp(b,n-1))\seqq(\pid_{k\bbo}+\swapp(c,n))\seqq(\pid_{k\bbo}+\pid_{(n-2)\bbo}+\hadamard)\seqq \\
            & \quad (\pid_{k\bbo}+\swapp(c,n))\seqq(\pid_{k\bbo}+\swapp(b,n-1)) \tag{by completeness of $\Pi$} \\
            \fromto_2{}& \pid_{k\bbo}+(\swapp(b,n-1)\seqq\swapp(c,n)\seqq(\pid_{(n-2)\bbo}+\hadamard)\seqq\swapp(c,n)\seqq\swapp(b,n-1)) \tag{by bifunctoriality of $+$} \\
            ={}& \pid_{k\bbo}+\qT[n]{\CH{b,c}}.
        \end{align*}
        \item $\mathbf{G} = \mathbf{G}_1\mathbf{G}_2$. \begin{align*}
            \qT[n+k]{\mathbf{G}_1\mathbf{G}_2}
            ={}& \qT[n+k]{\mathbf{G}_2}\seqq \qT[n+k]{\mathbf{G}_1} \\
            \fromto_2{}& (\pid_{k\bbo}+\qT[n]{\mathbf{G}_2})\seqq (\pid_{k\bbo}+\qT[n]{\mathbf{G}_1}) \tag{by inductive hypothesis} \\
            \fromto_2{}& \pid_{k\bbo}+(\qT[n]{\mathbf{G}_2}\seqq\qT[n]{\mathbf{G}_1}) \tag{by bifunctoriality of $+$} \\
            ={}& \pid_{k\bbo}+\qT[n]{\mathbf{G}_1\mathbf{G}_2}.
        \end{align*}
    \end{itemize}
\end{proof}

\endgroup
\section{Completeness of \texorpdfstring{\HPiLang{}}{Hadamard-Pi} \texorpdfstring{(\NoCaseChange{\cref{thm:hpi}})}{}}\label{app:hpi_proof}
\begingroup
\allowdisplaybreaks

For the convenience of the proof, we introduce the following lemmas and a new translation $\cT_{\mathit{HQ}}$.

\begin{lemma}
    \label{lem:type_conversion}
    Let $b$ be a type of $\Pi$, then we can construct a term $\mathit{tconv}_b:b \fromto \hdim(b)\bbo$ of $\Pi$ such that $\sem{\mathit{tconv}_b} = I_{\hdim(b)}$.
\end{lemma}
\begin{proof}
    By induction on the structure of $b$.
    \begin{itemize}
        \item $b = \bbz$. Let $\tconv_\bbz = \pid_{\bbz}$, we have $\sem{\tconv_\bbz} = \sem{\pid_{\bbz}} = I_{0} = I_{\hdim(\bbz)}$.
        \item $b = \bbo$. Let $\tconv_\bbo = \pid_{\bbo}$, we have $\sem{\tconv_\bbo} = \sem{\pid_{\bbo}} = I_{1} = I_{\hdim(\bbo)}$.
        \item $b = b_1+b_2$. Let \[\tconv_{b_1+b_2} = (\tconv_{b_1}+\tconv_{b_2})\seqq\assoclp\seqq(\assoclp+\pid_{\bbo})\seqq\cdots\seqq(\assoclp+\underbrace{\pid_{\bbo}+\pid_{\bbo}+\cdots+\pid_{\bbo}}_{\hdim(b_2)-1 \text{ times of } \pid_{\bbo}}).\] 
        Then,
        \begin{align*}
            & \sem{\tconv_{b_1+b_2}} \\
            ={}& \sem{\assoclp\seqq(\assoclp+\pid_{\bbo})\seqq\cdots\seqq(\assoclp+{\pid_{\bbo}+\pid_{\bbo}+\cdots+\pid_{\bbo}})} \cdot \sem{\tconv_{b_1}+\tconv_{b_2}} \\
            ={}& I_{\hdim(b_1+b_2)}\cdot (I_{\hdim(b_1)}\oplus I_{\hdim(b_2)}) \tag{by inductive hypothesis} \\
            ={}& I_{\hdim(b_1+b_2)}. \tag{by $\hdim(b_1+b_2) = \hdim(b_1) + \hdim(b_2)$}
        \end{align*}
        \item $b = b_1\times b_2$. Let \begin{align*}
            & \tconv_{b_1\times b_2} \\
            ={}& (\tconv_{b_1}\times \tconv_{b_2})\seqq\dist\seqq(\dist+\pid_{\hdim(b_2)\bbo})\seqq\cdots\seqq(\dist+\underbrace{\pid_{\hdim(b_2)\bbo}+\cdots+\pid_{\hdim(b_2)\bbo}}_{\hdim(b_1)-1 \text{ times } \pid_{\hdim(b_2)\bbo}})\seqq\\
            &\quad  \assoclp\seqq(\assoclp+\pid_{\bbo})\seqq\cdots\seqq(\assoclp+\underbrace{\pid_{\bbo}+\cdots+\pid_{\bbo}}_{\mathclap{\hdim(b_1)\hdim(b_2)-1 \text{ times of } \pid_{\bbo}}}).
        \end{align*}
        Then,
        \begin{align*}
            \sem{ \tconv_{b_1\times b_2}} = \cdots ={}& \sem{\tconv_{b_1}}\otimes \sem{\tconv_{b_2}} \\
            ={}& I_{\hdim(b_1)}\otimes I_{\hdim(b_2)} \tag{by inductive hypothesis} \\
            ={}& I_{\hdim(b_1\times b_2)}. \tag{by $\hdim(b_1\times b_2) = \hdim(b_1) \times \hdim(b_2)$}
        \end{align*}
    \end{itemize}
\end{proof}

\begin{definition}
    [Translation $\cT_{\mathit{HQ}}$]
    The translation $\cT_{\mathit{HQ}}$ that maps a term $c:b_1 \fromto b_2$ of \QPiLang{} to a term $\hqT{c}: \hdim(c)\bbo+ b_1 \fromto \hdim(c)\bbo+ b_2$ of \HPiLang{} is defined inductively on the structure of combinators as follows.
    \begin{gather*}
        \hqT{(-1)} = \hadamard\seqq\swapp\seqq\hadamard \\
        \hqT{\mathit{iso}} = \pid_{\hdim(\mathit{iso})}+\mathit{iso} \\
        \hqT{c_1\seqq c_2} = \hqT{c_1}\seqq\hqT{c_2} \\
        \begin{prooftree}
            \hypo{c_1:b_1\fromto b_3} \hypo{c_2:b_2\fromto b_3}
            \infer2{\hqT{c_1 + c_2} = (\pid_{b_1}+\swapp_{b_2+ \hdim(c_1)\bbo}+\pid_{b_2})\seqq(\hqT{c_1}+\hqT{c_2})\seqq(\pid_{b_1}+\swapp_{b_3+ \hdim(c_2)\bbo}+\pid_{b_4})}
        \end{prooftree} \\
        \begin{prooftree}
            \hypo{c_1:b_1\fromto b_3} \hypo{c_2:b_2\fromto b_3}
            \infer2{\mbox{$
            \begin{aligned}
                \hqT{c_1 \times c_2} ={}& \factor\seqq(\hqT{c_1}\times \pid_{b_2})\seqq\dist\seqq(((\tconv_{b_3}^{-1}\times \pid_{b_2})\seqq\swapt_{b_3\times b_2})+\swapt_{b_3\times b_2})\seqq \\
                &\quad \factor\seqq (\hqT{c_2}\times \pid_{b_3})\seqq\dist\seqq \\
                &\quad (((\pid_{\hdim(c_2)\bbo}\times\tconv_{b_3})\seqq\swapt_{(\hdim(c_2)\bbo)\times (\hdim(c_1)\bbo)})+\swapt_{b_4\times b_3})
            \end{aligned}
            $}}
        \end{prooftree}
    \end{gather*}
    where $\tconv_{b}^{-1}$ is the inverse of $\tconv_{b}$. When omitting the isomorphisms $\assocrp$, $\assoclp$, $\unitep$, $\unitip$, $\assocrt$, $\assoclt$, $\unitet$, $\unitit$, $\dist$, $\factor$, $\absorb$, \emph{$\hqT{c}$ can also be treated as the type of $\bb{2}\times b_1 \fromto \bb{2}\times b_2$}.
\end{definition}

The difference between $\cT_{\mathit{HQ}}$ and $\cT_{H}$ is that we double the dimension of the space, corresponding to taking the tensor product with a qubit, rather than merely adding one auxiliary dimension.
Then, the following lemmas can be established.

\begin{lemma}\label{lem:hqt0}
    Let $c:b_1\fromto b_3$ be a term of \QPiLang{} and $b$ be a type. When omitting the isomorphisms $\assocrp$, $\assoclp$, $\unitep$, $\unitip$, $\assocrt$, $\assoclt$, $\unitet$, $\unitit$, $\dist$, $\factor$, $\absorb$, and $\factorz$ used for type conversion, we have that
    \[\hqT{\pid_{b}\times c} \fromto_2{} (\swapt\times\pid_{b_1})\seqq(\pid_{b}\times \hqT{c})\seqq(\swapt\times \pid_{b_3})\]
    is a level-2 combinator of \HPiLang{}.
\end{lemma}
\begin{proof}
    \begin{align*}
        & \hqT{\pid_{b}\times c} \\
        \fromto_2{}& (\hqT{\pid_{b}}\times \pid_{b_1})\seqq (\swapt_{b\times b_1}+\swapt_{b\times b_1})\seqq (\hqT{c}\times \pid_{b})\seqq (\swapt_{b_3\times b}+\swapt_{b_3\times b}) \\
        \fromto_2{}&  (\pid_{\bb{2}\times b}\times \pid_{b_1})\seqq (\pid_{\bb{2}}\times \swapt_{b\times b_1})\seqq (\hqT{c}\times \pid_{b})\seqq (\pid_{\bb{2}}\times \swapt_{b_3\times b}) \tag{by completeness of $\Pi$} \\
        \fromto_2{}& (\pid_{\bb{2}\times b}\times \pid_{b_1})\seqq (\pid_{\bb{2}}\times \swapt_{b\times b_1})\seqq \swapt_{(\bb{2}\times b_1)\times b}\seqq(\pid_{b}\times \hqT{c})\seqq\swapt_{b\times (\bb{2}\times b_3)}\seqq \\
        & \quad (\pid_{\bb{2}}\times \swapt_{b_3\times b}) \tag{by naturality of $\swapt$} \\
        \fromto_2{}& (\swapt_{\bb{2}\times b}\times \pid_{b_1})\seqq(\pid_{b}\times \hqT{c})\seqq(\swapt_{b\times \bb{2}}\times \pid_{b_3}). \tag{by completeness of $\Pi$}
    \end{align*}
\end{proof}

\begin{lemma}\label{lem:ht_assoc_p}
    Let $c_1,c_2,c_3$ be terms of \QPiLang{}. When omitting the isomorphisms $\assocrp$, $\assoclp$, $\unitep$, $\unitip$, $\assocrt$, $\assoclt$, $\unitet$, $\unitit$, $\dist$, $\factor$, $\absorb$, and $\factorz$ used for type conversion, we have that
    \[\hT{(c_1+c_2)+c_3} \fromto_2{} \hT{c_1+(c_2+c_3)}\]
    is a level-2 combinator of \HPiLang{}.
\end{lemma}
\begin{proof}
    \ifSHOWPROOF%
    \begin{align*}
        & \hT{(c_1+c_2)+c_3} \\
        \fromto_2{}& \scalebox{\scaleratio}{\begin{quantikz}[row sep=3mm,column sep=2mm]
            \lstick[2,brackets=none]{$+$\hspace{-.65cm}}\push{\bbo}&\gate[3]{\hT{c_1+c_2}} & \permute{3,1,2} & & \permute{2,3,1} &\\
            \lstick[2,brackets=none]{$+$\hspace{-.65cm}}\push{\phantom{\bbo}}& & & & &\\
            \lstick[2,brackets=none]{$+$\hspace{-.65cm}}\push{\phantom{\bbo}}& & &\gate[2]{\hT{c_3}} & &\\
            \push{\phantom{\bbo}}&&&& &
        \end{quantikz}} \\
        \fromto_2{}& \scalebox{\scaleratio}{\begin{quantikz}[row sep=3mm,column sep=2mm]
            \lstick[2,brackets=none]{$+$\hspace{-.65cm}}\push{\bbo}&\gate[2]{\hT{c_1}} &\permute{2,1}&&\permute{2,1}& \permute{3,1,2} & & \permute{2,3,1} &\\
            \lstick[2,brackets=none]{$+$\hspace{-.65cm}}\push{\phantom{\bbo}}& & &\gate[2]{\hT{c_2}} & & & & &\\
            \lstick[2,brackets=none]{$+$\hspace{-.65cm}}\push{\phantom{\bbo}}& & &&& &\gate[2]{\hT{c_3}} & &\\
            \push{\phantom{\bbo}}&&&&&& &&
        \end{quantikz}} \\
        \fromto_2{}& \scalebox{\scaleratio}{\begin{quantikz}[row sep=3mm,column sep=2mm]
            \lstick[2,brackets=none]{$+$\hspace{-.65cm}}\push{\bbo}&\gate[2]{\hT{c_1}} &\permute{2,1}&&& \permute{1,3,2} & && \permute{2,1} &\\
            \lstick[2,brackets=none]{$+$\hspace{-.65cm}}\push{\phantom{\bbo}}& & &\gate[2]{\hT{c_2}} & & & & \permute{2,1} &&\\
            \lstick[2,brackets=none]{$+$\hspace{-.65cm}}\push{\phantom{\bbo}}& & &&& &\gate[2]{\hT{c_3}} & &&\\
            \push{\phantom{\bbo}}&&&&&& &&&
        \end{quantikz}} \tag{by completeness of $\Pi$} \\
        \fromto_2{}& \scalebox{\scaleratio}{\begin{quantikz}[row sep=3mm,column sep=2mm]
            \lstick[2,brackets=none]{$+$\hspace{-.65cm}}\push{\bbo}&\gate[2]{\hT{c_1}} &\permute{2,1}&&& \permute{2,1} &\\
            \lstick[2,brackets=none]{$+$\hspace{-.65cm}}\push{\phantom{\bbo}}& & &\gate[3]{\hT{c_2+c_3}} & &&\\
            \lstick[2,brackets=none]{$+$\hspace{-.65cm}}\push{\phantom{\bbo}}& &&& &&\\
            \push{\phantom{\bbo}}&&&&&& 
        \end{quantikz}} \\
        \fromto_2{}& \hT{c_1+(c_2+c_3)}.
    \end{align*}
    \else%
    Set \verb|\SHOWPROOFtrue| to compile the proof.
    \fi
\end{proof}

\begin{lemma}\label{lem:hqt1}
    Let $c:b_1\fromto b_3$ be a term of \QPiLang{} and $b$ be a type. When omitting the isomorphisms $\assocrp$, $\assoclp$, $\unitep$, $\unitip$, $\assocrt$, $\assoclt$, $\unitet$, $\unitit$, $\dist$, $\factor$, $\absorb$, and $\factorz$ used for type conversion, we have that
    \begin{align*}
        \hT{\pid_{b}\times c} \fromto_2{} \cT_{H}[\underbrace{c+c+\cdots+c}_{\hdim(b)}] \enspace \text{and}\enspace
        \hqT{\pid_{b}\times c} \fromto_2{} \cT_{HQ}[\underbrace{c+c+\cdots+c}_{\hdim(b)}]
    \end{align*}
    are level-2 combinators of \HPiLang{}.
\end{lemma}
\begin{proof}
    We prove the first one by induction on the value of $R(b)$, where $R$ has been defined in \cref{fig:def_ht}.
    \begin{itemize}
        \item $R(b)=1$. In this case, we have $b = \bbz$, which is the trivial case.
        \item $R(b) = 2$. In this case, $b = \bbo$. \begin{align*}
            \hT{\pid_{\bbo}\times c} ={}& (\pid_\bbo+\unitet)\seqq\hT{c}\seqq(\pid_\bbo+\unitit)
            \fromto_2{} \hT{c}. \tag{by strictness}
        \end{align*}
        \item Suppose $R(b) = k\geq 3$ and for any type $b_2$ with $R(b_2) < k$, we have 
        \[\hT{\pid_{b_2}\times c} \fromto_2{} \cT_{H}[\underbrace{c+c+\cdots+c}_{\hdim(b_2)}].\]
        \begin{itemize}[leftmargin=10pt]
            \item For the cases where $b$ takes the form $\bbz\times b_4, \bbo\times b_4, (b_5+b_6)\times b_4, (b_5\times b_6)\times b_4$, by the definition of $\cT_H$ in \cref{fig:def_ht}, we can easily reduce to smaller cases of $b'$ with $R(b')< R(b) = k$ such that $\hdim(b') = \hdim(b)$ and $\hT{\pid_{b}\times c} \fromto_2\hT{\pid_{b'}\times c}$. Thus the statement is easy to see.
            \item The only non-trivial case is $b = b_4+b_5$.  \begin{align*}
                &\hT{\pid_{b_4+b_5}\times c} \\
                \fromto_2{}& \hT{\pid_{b_4}\times c+\pid_{b_5}\times c} \\
                \fromto_2{}& (\hT{\pid_{b_4}\times c}+\pid_{b_5\times b_1})\seqq(\swapp_{\bbo+(b_4\times b_3)})\seqq(\pid_{b_4\times b_3}+\hT{\pid_{b_5}\times c})\seqq(\swapp_{(b_4\times b_3)+\bbo}) \\
                \fromto_2{}& (\cT_{H}[\underbrace{c+c+\cdots+c}_{\hdim(b_4)}]+\pid_{b_5\times b_1})\seqq(\swapp_{\bbo+(b_4\times b_3)})\seqq(\pid_{b_4\times b_3}+\cT_{H}[\underbrace{c+c+\cdots+c}_{\hdim(b_5)}])\seqq \\
                &\quad (\swapp_{(b_4\times b_3)+\bbo}) \tag{by inductive hypothesis} \\
                \fromto_2{}& (\cT_{H}[\underbrace{c+c+\cdots+c}_{\hdim(b_4)}]+\pid_{\scalebox{.8}{$\underbrace{b_1+b_1+\cdots+b_1}_{\hdim(b_5)}$}})\seqq(\swapp_{\bbo+(\scalebox{.8}{$\underbrace{b_3+b_3+\cdots+b_3}_{\hdim(b_4)}$})})\seqq \\
                & \quad (\pid_{\scalebox{.8}{$\underbrace{b_3+b_3+\cdots+b_3}_{\hdim(b_4)}$}}+\cT_{H}[\underbrace{c+c+\cdots+c}_{\hdim(b_5)}])\seqq(\swapp_{(\scalebox{.8}{$\underbrace{b_3+b_3+\cdots+b_3}_{\hdim(b_4)}$})+\bbo}) \tag{by strictness}\\
                \fromto_2{}& \cT_{H}[(\underbrace{c+c+\cdots+c}_{\hdim(b_4)})+(\underbrace{c+c+\cdots+c}_{\hdim(b_5)})] \\
                \fromto_2{}& \cT_{H}[(\underbrace{c+c+\cdots+c}_{\hdim(b_4+b_5)})].  \tag{by \cref{lem:ht_assoc_p} and $\hdim(b_4+b_5) = \hdim(b_4)+\hdim(b_5)$}
            \end{align*}
        \end{itemize}
    \end{itemize}
    For the second one, we prove it by induction on the value of $\hdim(b)$. \begin{itemize}
        \item $\hdim(b) = 0, 1$. By \cref{lem:hqt0}, the basic case is trivial.
        \item Suppose $\hdim(b) = k \geq 2$ and for any type $b_2$ with $\hdim(b_2) < k$, we have
        \[\hqT{\pid_{b_2}\times c} \fromto_2{} \cT_{HQ}[\underbrace{c+c+\cdots+c}_{\hdim(b_2)}].\]
        Since $\hdim(b)\geq 2$, there must have $b_4,b_5=\bbo$ such that $\hdim(b) = \hdim(b_4)+\hdim(b_5) = \hdim(b_4)+1$. Then, \begin{align*}
            & \hqT{\pid_{b}\times c} \\
            \fromto_2{}& (\swapt_{\bb{2}\times b}\times\pid_{b_1})\seqq(\pid_{b}\times \hqT{c})\seqq(\swapt_{\bb{2}\times b}\times \pid_{b_3}) \tag{by \cref{lem:hqt0}} \\
            \fromto_2{}& (\swapt_{\bb{2}\times (b_4+\bbo)}\times\pid_{b_1})\seqq((\pid_{b_4}+\pid_{\bbo})\times \hqT{c})\seqq(\swapt_{(b_4+\bbo)\times \bb{2}}\times \pid_{b_3}) \tag{by strictness} \\
            \fromto_2{}& (\swapt_{\bb{2}\times (b_4+\bbo)}\times\pid_{b_1})\seqq((\pid_{b_4}\times \hqT{c}+\hqT{c}))\seqq(\swapt_{(b_4+\bbo)\times \bb{2}}\times \pid_{b_3})  \tag{by distributivity and unitality} \\
            \fromto_2{}& (\swapt_{\bb{2}\times (b_4+\bbo)}\times\pid_{b_1})\seqq(((\swapt_{b_4\times \bb{2}}\times \pid_{b_1})\seqq\hqT{\pid_{b_4}\times c}\seqq(\swapt_{\bb{2}\times b_4}\times \pid_{b_3}))+\hqT{c})\seqq \\
            & \quad (\swapt_{(b_4+\bbo)\times \bb{2}}\times \pid_{b_3}) \tag{by \cref{lem:hqt0}} \\
            \fromto_2{}& (\pid_{b_4\times b_1}+\swapp_{b_1+(b_4\times b_1)}+\pid_{b_1})\seqq(\hqT{\pid_{b_4}\times c}+\hqT{c})\seqq (\pid_{b_4\times b_1}+\swapp_{(b_4\times b_1)+b_1}+\pid_{b_1}) \tag{by completeness of $\Pi$} \\
            \fromto_2{}& (\pid_{b_4\times b_1}+\swapp_{b_1+(b_4\times b_1)}+\pid_{b_1})\seqq(\cT_{HQ}[\underbrace{c+c+\cdots+c}_{\hdim(b_4)}]+\hqT{c})\seqq \\
            & \quad (\pid_{b_4\times b_1}+\swapp_{(b_4\times b_1)+b_1}+\pid_{b_1}) \tag{by inductive hypothesis} \\
            \fromto_2{}& \cT_{HQ}[\underbrace{c+c+\cdots+c}_{\hdim(b_4)+1}] = \cT_{HQ}[\underbrace{c+c+\cdots+c}_{\hdim(b)}].
        \end{align*}
    \end{itemize}
\end{proof}

%\subsubsection{A Stronger Lemma}
The following lemma is a stronger version of \cref{lem:ht2}. While we can prove the final conclusion using \cref{lem:ht2}, establishing this lemma will greatly simplify our presentation.

\begin{lemma}\label{lem:hqt2}
    Let $c$ be a term of \QPiLang{}. When omitting the isomorphisms $\assocrp$, $\assoclp$, $\unitep$, $\unitip$, $\assocrt$, $\assoclt$, $\unitet$, $\unitit$, $\dist$, $\factor$, $\absorb$, and $\factorz$ used for type conversion, we have that
    \[\pid_{\hdim(c)\bbo}+\hT{c} \fromto_2{} (\swapp_{\hdim(c)\bbo+\bbo}+\pid)\seqq(\pid_{\bbo}+\hqT{c})\seqq(\swapp_{\bbo+\hdim(c)\bbo}+\pid)\]
    is a level-2 combinator of \HPiLang{}. 
\end{lemma}
\begin{proof}
    \ifSHOWPROOF%
    We prove it by induction on the structure of $c$.
    \begin{itemize}
        \item $c = (-1)$. \begin{align*}
            & (\swapp_{\hdim((-1))\bbo+\bbo}+\pid)\seqq(\pid_{\bbo}+\hqT{(-1)})\seqq(\swapp_{\bbo+\hdim((-1))\bbo}+\pid) \\
            \fromto_2{}& \scalebox{\scaleratio}{\begin{quantikz}[row sep=3mm,column sep=2mm]
                \lstick[2,brackets=none]{$+$\hspace{-.65cm}}\push{\bbo}& \permute{2,1} & & \permute{2,1} &\\
                \lstick[2,brackets=none]{$+$\hspace{-.65cm}}\push{\bbo}& & \gate[2]{\hadamard\seqq\swapp\seqq\hadamard}& &\\
                & & & &
            \end{quantikz}} \\
            \fromto_2{}& \scalebox{\scaleratio}{\begin{quantikz}[row sep=3mm,column sep=2mm]
                \lstick[2,brackets=none]{$+$\hspace{-.65cm}}\push{\bbo}& &\\
                \lstick[2,brackets=none]{$+$\hspace{-.65cm}}\push{\bbo}& \gate[2]{\hadamard\seqq\swapp\seqq\hadamard} &\\
                & &
            \end{quantikz}} \tag{by \cref{eq:h3}} \\
            \fromto_2{}& \pid_{\hdim((-1))\bbo}+\hT{(-1)}.
        \end{align*}
        \item $c$ is another isomorphism. 
        \begin{align*}
            & (\swapp_{\hdim(c)\bbo+\bbo}+\pid)\seqq(\pid_{\bbo}+\hqT{c})\seqq(\swapp_{\bbo+\hdim(c)\bbo}+\pid) \\
            \fromto_2{}&\qquad \scalebox{\scaleratio}{\begin{quantikz}[row sep=3mm,column sep=2mm]
                \lstick[2,brackets=none]{$+$\hspace{-.65cm}}\push{\mathllap{\hdim(c)}\bbo}& \permute{2,1} & & \permute{2,1} &\\
                \lstick[2,brackets=none]{$+$\hspace{-.65cm}}\push{\bbo}& & & &\\
                & & \gate{c} & &
            \end{quantikz}} \\
            \fromto_2{}& \qquad  \scalebox{\scaleratio}{\begin{quantikz}[row sep=3mm,column sep=2mm]
                \lstick[2,brackets=none]{$+$\hspace{-.65cm}}\push{\mathllap{\hdim(c)}\bbo}& &\\
                \lstick[2,brackets=none]{$+$\hspace{-.65cm}}\push{\bbo}& &\\
                &\gate{c} &
            \end{quantikz}} \tag{by naturality of $\swapp$} \\
            \fromto_2{}& \pid_{\hdim(c)\bbo} + \hT{c}.
        \end{align*}
        \item $c = c_1\seqq c_2$. It is easy to see.
        \item $c = c_1+c_2$ and $c_1:b_1\fromto b_3, c_2:b_2\fromto b_4$. \begin{align*}
            &(\swapp_{\hdim(c_1+c_2)\bbo+\bbo}+\pid)\seqq(\pid_{\bbo}+\hqT{c_1+c_2})\seqq(\swapp_{\bbo+\hdim(c_1+c_2)\bbo}+\pid) \\
            \fromto_2{}& \qquad \scalebox{\scaleratio}{\begin{quantikz}[row sep=3mm,column sep=2mm]
                \lstick[2,brackets=none]{$+$\hspace{-.65cm}}\push{\mathllap{\hdim(b_1)}\bbo}& \permute{2,3,1} & & & & \permute{3,1,2}&\\
                \lstick[2,brackets=none]{$+$\hspace{-.65cm}}\push{\mathllap{\hdim(b_2)}\bbo}&& & \gate[2]{\hqT{c_1}} & &&\\
                \lstick[2,brackets=none]{$+$\hspace{-.65cm}}\push{\bbo}&& \permute{2,1} & & \permute{2,1} &&\\
                \lstick[2,brackets=none]{$+$\hspace{-.65cm}}\push{b_1}&& & \gate[2]{\hqT{c_2}}& &&\\
                \push{b_2}&& & & && 
            \end{quantikz}} \\
            \fromto_2{}& \qquad \scalebox{\scaleratio}{\begin{quantikz}[row sep=3mm,column sep=2mm]
                \lstick[2,brackets=none]{$+$\hspace{-.65cm}}\push{\mathllap{\hdim(b_1)}\bbo}& \permute{2,3,1} & \permute{2,1} & & \permute{2,1} & \permute{3,1,2}&\\
                \lstick[2,brackets=none]{$+$\hspace{-.65cm}}\push{\mathllap{\hdim(b_2)}\bbo}&& & \gate[2]{\hT{c_1}} & &&\\
                \lstick[2,brackets=none]{$+$\hspace{-.65cm}}\push{\bbo}&& \permute{2,1} & & \permute{2,1} &&\\
                \lstick[2,brackets=none]{$+$\hspace{-.65cm}}\push{b_1}&& & \gate[2]{\hqT{c_2}}& &&\\
                \push{b_2}&& & & && 
            \end{quantikz}} \tag{by inductive hypothesis of $c_1$} \\
            \fromto_2{}& \qquad \scalebox{\scaleratio}{\begin{quantikz}[row sep=3mm,column sep=2mm]
                \lstick[2,brackets=none]{$+$\hspace{-.65cm}}\push{\mathllap{\hdim(b_1)}\bbo}& \permute{1,3,2} & & & & \permute{1,3,2}&\\
                \lstick[2,brackets=none]{$+$\hspace{-.65cm}}\push{\mathllap{\hdim(b_2)}\bbo}&& & \gate[2]{\hT{c_1}} & &&\\
                \lstick[2,brackets=none]{$+$\hspace{-.65cm}}\push{\bbo}&& \permute{2,1} & & \permute{2,1} &&\\
                \lstick[2,brackets=none]{$+$\hspace{-.65cm}}\push{b_1}&& & \gate[2]{\hqT{c_2}}& &&\\
                \push{b_2}&& & & && 
            \end{quantikz}} \tag{by completeness of $\Pi$} \\
            \fromto_2{}& \qquad \scalebox{\scaleratio}{\begin{quantikz}[row sep=3mm,column sep=2mm]
                \lstick[2,brackets=none]{$+$\hspace{-.65cm}}\push{\mathllap{\hdim(b_1)}\bbo}&& \permute{1,3,2} & & & & \permute{1,3,2}&\\
                \lstick[2,brackets=none]{$+$\hspace{-.65cm}}\push{\mathllap{\hdim(b_2)}\bbo}&&& & & &&\\
                \lstick[2,brackets=none]{$+$\hspace{-.65cm}}\push{\bbo}&\gate[2]{\hT{c_1}} && \permute{2,1} & & \permute{2,1} &&\\
                \lstick[2,brackets=none]{$+$\hspace{-.65cm}}\push{b_1}&&& & \gate[2]{\hqT{c_2}}& &&\\
                \push{b_2}&&& & & && 
            \end{quantikz}} \tag{by bifunctoriality of $+$} \\
            \fromto_2{}& \qquad \scalebox{\scaleratio}{\begin{quantikz}[row sep=3mm,column sep=2mm]
                \lstick[2,brackets=none]{$+$\hspace{-.65cm}}\push{\mathllap{\hdim(b_1)}\bbo}&&  & & & &&&&&\\
                \lstick[2,brackets=none]{$+$\hspace{-.65cm}}\push{\mathllap{\hdim(b_2)}\bbo}& & & & \permute{2,1} & & & & \permute{2,1}& &\\
                \lstick[2,brackets=none]{$+$\hspace{-.65cm}}\push{\bbo}&\gate[2]{\hT{c_1}} && \permute{2,1} & & \permute{2,1} & & \permute{2,1} & & \permute{2,1}&\\
                \lstick[2,brackets=none]{$+$\hspace{-.65cm}}\push{b_1}& && & & & \gate[2]{\hqT{c_2}}& &&&\\
                \push{b_2}&&& & & && &&&
            \end{quantikz}} \tag{by completeness of $\Pi$ and naturality of $\swapp$} \\
            \fromto_2{}& \qquad \scalebox{\scaleratio}{\begin{quantikz}[row sep=3mm,column sep=2mm]
                \lstick[2,brackets=none]{$+$\hspace{-.65cm}}\push{\mathllap{\hdim(b_1)}\bbo}&&  & & & &&&\\
                \lstick[2,brackets=none]{$+$\hspace{-.65cm}}\push{\mathllap{\hdim(b_2)}\bbo}& & & & \permute{2,1} & & \permute{2,1}& &\\
                \lstick[2,brackets=none]{$+$\hspace{-.65cm}}\push{\bbo}&\gate[2]{\hT{c_1}} && \permute{2,1} & & & & \permute{2,1}&\\
                \lstick[2,brackets=none]{$+$\hspace{-.65cm}}\push{b_1}& && & & \gate[2]{\hT{c_2}} &&&\\
                \push{b_2}&&& & & && &&
            \end{quantikz}} \tag{by inductive hypothesis of $c_2$} \\
            \fromto_2{}& \qquad \scalebox{\scaleratio}{\begin{quantikz}[row sep=3mm,column sep=2mm]
                \lstick[2,brackets=none]{$+$\hspace{-.65cm}}\push{\mathllap{\hdim(b_1)}\bbo}&&  & & &&\\
                \lstick[2,brackets=none]{$+$\hspace{-.65cm}}\push{\mathllap{\hdim(b_2)}\bbo}& & & & & &\\
                \lstick[2,brackets=none]{$+$\hspace{-.65cm}}\push{\bbo}&\gate[2]{\hT{c_1}} && \permute{2,1} & & \permute{2,1}&\\
                \lstick[2,brackets=none]{$+$\hspace{-.65cm}}\push{b_1}& && & \gate[2]{\hT{c_2}} &&\\
                \push{b_2}&&& & & &&
            \end{quantikz}} \tag{by naturality of $\swapp$} \\
            \fromto_2{}& \pid_{\hdim(c_1+c_2)\bbo}+\hT{c_1+c_2}. \tag{by strictness and $\hdim(c_1+c_2) = \hdim(b_1)+\hdim(b_2)$}
        \end{align*}
        \item $c = c_1\times c_2$ and $c_1:b_1\fromto b_3, c_2:b_2\fromto b_4$. \begin{itemize}[leftmargin=10pt]
            \item For the basic case of $c = \pid_{b_1}\times c_2$, we have \begin{align*}
                &\pid_{\hdim(\pid_{b_1}\times c_2)\bbo}+\hT{\pid_{b_1}\times c_2} \\
                \fromto_2{}&  \pid_{\hdim(\pid_{b_1}\times c_2)\bbo}+\cT_{H}[\underbrace{c_2+\cdots+c_2}_{\hdim(b_1)}] \tag{by \cref{lem:hqt1}} \\
                \fromto_2{}& \cdots \\
                \fromto_2{}& (\swapp_{\hdim(b_1\times b_2)\bbo+\bbo}+\pid_{b_1\times b_2})\seqq(\pid_{\bbo}+\cT_{HQ}[\underbrace{c_2+\cdots+c_2}_{\hdim(b_1)}])\seqq(\swapp_{\bbo+\hdim(b_1\times b_4)\bbo}+\pid_{b_1\times b_4})  \tag{by inductive hypothesis} \\
                \fromto_2{}& (\swapp_{\hdim(b_1\times b_2)\bbo+\bbo}+\pid_{b_1\times b_2})\seqq(\pid_{\bbo}+\hqT{\pid_{b_1}\times c_2})\seqq(\swapp_{\bbo+\hdim(b_1\times b_4)\bbo}+\pid_{b_1\times b_4}). \tag{by \cref{lem:hqt1}}
            \end{align*}
            \item For the general case of $c = c_1\times c_2$.
            {\begin{align*}
                & \pid_{\dim(c_1\times c_2)\bbo}+\hT{c_1\times c_2} \\
                \fromto_2{}&  \pid_{\hdim(c_1\times c_2)\bbo}+ ((\pid_{\bbo}+\swapt_{b_1\times b_2})\seqq\hT{\pid_{b_2}\times c_1}\seqq(\pid_{\bbo}+\swapt_{b_2\times b_3})\seqq\hT{\pid_{b_3}\times c_2}) \\
                \fromto_2{}&  (\swapt_{b_1\times b_2}+ (\pid_{\bbo}+\swapt_{b_1\times b_2}))\seqq(\pid_{\hdim(\pid_{b_2}\times c_1)\bbo}+ \hT{\pid_{b_2}\times c_1})\seqq \\
                & \quad (\swapt_{b_2\times b_3}+ (\pid_{\bbo}+\swapt_{b_2\times b_3}))\seqq(\pid_{\hdim(\pid_{b_3}\times c_2)\bbo}+ \hT{\pid_{b_3}\times c_2}) \tag{by bifunctoriality of $+$, strictness and naturality of $\swapt$}\\
                \fromto_2{}& (\swapt_{b_1\times b_2}+ (\pid_{\bbo}+\swapt_{b_1\times b_2}))\seqq \\
                & \quad (\swapp_{\hdim(b_2\times b_1)\bbo+\bbo}+\pid_{b_2\times b_1})\seqq(\pid_{\bbo}+\hqT{\pid_{b_2}\times c_1})\seqq(\swapp_{\bbo+\hdim(b_2\times b_3)\bbo}+\pid_{b_2\times b_3})\seqq \\
                & \quad (\swapt_{b_2\times b_3}+ (\pid_{\bbo}+\swapt_{b_2\times b_3}))\seqq \\
                & \quad (\swapp_{\hdim(b_3\times b_2)\bbo+\bbo}+\pid_{b_3\times b_2})\seqq(\pid_{\bbo}+\hqT{\pid_{b_3}\times c_2})\seqq(\swapp_{\bbo+\hdim(b_3\times b_4)\bbo}+\pid_{b_3\times b_4}) \tag{by inductive hypothesis on the case of $\pid_{b}\times c$} \\
                \fromto_2{}&  (\swapt_{b_1\times b_2}+ (\pid_{\bbo}+\swapt_{b_1\times b_2}))\seqq(\swapp_{\hdim(b_2\times b_1)\bbo+\bbo}+\pid_{b_2\times b_1})\seqq \\
                & \quad (\pid_{\bbo}+((\swapt_{\bb{2}\times b_2}\times \pid_{b_1})\seqq(\pid_{b_2}\times \hqT{c_1})\seqq(\swapt_{b_2\times \bb{2}}\times \pid_{b_3})))\seqq\\
                & \quad (\swapp_{\bbo+\hdim(b_2\times b_3)\bbo}+\pid_{b_2\times b_3})\seqq(\swapt_{b_2\times b_3}+ (\pid_{\bbo}+\swapt_{b_2\times b_3}))\seqq\\
                & \quad (\swapp_{\hdim(b_3\times b_2)\bbo+\bbo}+\pid_{b_3\times b_2})\seqq \\
                & \quad (\pid_{\bbo}+((\swapt_{\bb{2}\times b_3}\times \pid_{b_2})\seqq(\pid_{b_3}\times \hqT{c_2})\seqq(\swapt_{b_3\times \bb{2}}\times \pid_{b_4})))\seqq\\
                & \quad (\swapp_{\bbo+\hdim(b_3\times b_4)\bbo}+\pid_{b_3\times b_4}) \tag{by \cref{lem:hqt0}} \\
                \fromto_2{}&  (\swapt_{b_1\times b_2}+ (\pid_{\bbo}+\swapt_{b_1\times b_2}))\seqq(\swapp_{\hdim(b_2\times b_1)\bbo+\bbo}+\pid_{b_2\times b_1})\seqq \\
                & \quad (\pid_{\bbo}+((\swapt_{\bb{2}\times b_2}\times \pid_{b_1})\seqq(\swapt_{b_2\times (\bb{2}\times b_1)})\seqq(\hqT{c_1}\times \pid_{b_2})\seqq \\
                & \qquad (\swapt_{(\bb{2}\times b_3)\times b_2})\seqq(\swapt_{b_2\times \bb{2}}\times \pid_{b_3})))\seqq\\
                & \quad (\swapp_{\bbo+\hdim(b_2\times b_3)\bbo}+\pid_{b_2\times b_3})\seqq(\swapt_{b_2\times b_3}+ (\pid_{\bbo}+\swapt_{b_2\times b_3}))\seqq\\
                & \quad (\swapp_{\hdim(b_3\times b_2)\bbo+\bbo}+\pid_{b_3\times b_2})\seqq \\
                & \quad (\pid_{\bbo}+((\swapt_{\bb{2}\times b_3}\times \pid_{b_2})\seqq(\swapt_{b_3\times (\bb{2}\times b_2)})\seqq(\hqT{c_2}\times \pid_{b_3})\seqq \\
                & \qquad (\swapt_{(\bb{2}\times b_4)\times b_3})\seqq(\swapt_{b_3\times \bb{2}}\times \pid_{b_4})))\seqq\\
                & \quad (\swapp_{\bbo+\hdim(b_3\times b_4)\bbo}+\pid_{b_3\times b_4}) \tag{by naturality of $\swapt$} \\
                \fromto_2{}&  (\swapt_{b_1\times b_2}+ (\pid_{\bbo}+\swapt_{b_1\times b_2}))\seqq(\swapp_{\hdim(b_2\times b_1)\bbo+\bbo}+\pid_{b_2\times b_1})\seqq \\
                & \quad (\pid_{\bbo}+((\swapt_{\bb{2}\times b_2}\times \pid_{b_1})\seqq(\swapt_{b_2\times (\bb{2}\times b_1)})\seqq(\hqT{c_1}\times \pid_{b_2})))\seqq\\
                & \quad (\pid_{\bbo}+ (\swapt_{b_3\times b_2} + \swapt_{b_3\times b_2}))\seqq \\
                & \quad (\pid_{\bbo}+((\hqT{c_2}\times \pid_{b_3})\seqq(\swapt_{(\bb{2}\times b_4)\times b_3})\seqq(\swapt_{b_3\times \bb{2}}\times \pid_{b_4})))\seqq\\
                & \quad (\swapp_{\bbo+\hdim(b_3\times b_4)\bbo}+\pid_{b_3\times b_4}) \tag{by completeness of $\Pi$ and bifunctoriality of $+$} \\
                \fromto_2{}&  (\swapt_{b_1\times b_2}+ (\pid_{\bbo}+\swapt_{b_1\times b_2}))\seqq(\swapp_{\hdim(b_2\times b_1)\bbo+\bbo}+\pid_{b_2\times b_1})\seqq \\
                & \quad (\pid_{\bbo}+((\swapt_{\bb{2}\times b_2}\times \pid_{b_1})\seqq(\swapt_{b_2\times (\bb{2}\times b_1)})))\seqq\\
                & \quad (\pid_{\bbo}+ ((\hqT{c_1}\times \pid_{b_2})\seqq(\swapt_{b_3\times b_2} + \swapt_{b_3\times b_2})\seqq(\hqT{c_2}\times \pid_{b_3})\seqq\\
                & \qquad (\swapt_{b_4\times b_3} + \swapt_{b_4\times b_3})))\seqq \\
                & \quad (\pid_{\bbo}+((\swapt_{b_3\times b_4} + \swapt_{b_3\times b_4})\seqq(\swapt_{(\bb{2}\times b_4)\times b_3})\seqq(\swapt_{b_3\times \bb{2}}\times \pid_{b_4})))\seqq\\
                & \quad (\swapp_{\bbo+\hdim(b_3\times b_4)\bbo}+\pid_{b_3\times b_4}) \tag{by bifunctoriality of $+$ and naturality of $\swapt$} \\
                \fromto_2{}&  (\swapt_{b_1\times b_2}+ (\pid_{\bbo}+\swapt_{b_1\times b_2}))\seqq(\swapp_{\hdim(b_2\times b_1)\bbo+\bbo}+\pid_{b_2\times b_1})\seqq \\
                & \quad (\pid_{\bbo}+((\swapt_{\bb{2}\times b_2}\times \pid_{b_1})\seqq(\swapt_{b_2\times (\bb{2}\times b_1)})))\seqq\\
                & \quad (\pid_{\bbo}+ \hqT{c_1\times c_2})\seqq \\
                & \quad (\pid_{\bbo}+((\swapt_{b_3\times b_4} + \swapt_{b_3\times b_4})\seqq(\swapt_{(\bb{2}\times b_4)\times b_3})\seqq(\swapt_{b_3\times \bb{2}}\times \pid_{b_4})))\seqq\\
                & \quad (\swapp_{\bbo+\hdim(b_3\times b_4)\bbo}+\pid_{b_3\times b_4}) \tag{by definition of $\cT_{HQ}$} \\
                \fromto_2{}&  (\swapp_{(b_1\times b_2)+\bbo} +\pid_{b_1\times b_2})\seqq (\pid_{\bbo}+ \hqT{c_1\times c_2})\seqq (\swapp_{\bbo+\hdim(b_3\times b_4)\bbo}+\pid_{b_3\times b_4}) \tag{by completeness of $\Pi$} \\
                \fromto_2{}&  (\swapp_{\hdim(c_1\times c_2)\bbo+\bbo} +\pid_{b_1\times b_2})\seqq (\pid_{\bbo}+ \hqT{c_1\times c_2})\seqq (\swapp_{\bbo+\hdim(c_1\times c_2)\bbo}+\pid_{b_3\times b_4}). \tag{by strictness}
            \end{align*}}
        \end{itemize} 
    \end{itemize}
    \else%
    Set \verb|\SHOWPROOFtrue| to compile the proof.
    \fi
\end{proof}

The following lemma is extracted from \[\hT{c_1}\times \hT{c_2}\fromto_2{} (\hT{c_1}\times \pid_{\bbo+b_2})\seqq(\pid_{\bbo+b_3}\times \hT{c_2}) \fromto_2{} (\pid_{\bbo+b_1}\times \hT{c_2})\seqq(\hT{c_1}\times \pid_{\bbo+b_4}).\]

\begin{lemma}\label{lem:ht_times4}
    Let $c_1:b_1\fromto b_3, c_2:b_2\fromto b_4$ be two terms of \QPiLang{}. When omitting the isomorphisms $\assocrp$, $\assoclp$, $\unitep$, $\unitip$, $\assocrt$, $\assoclt$, $\unitet$, $\unitit$, $\dist$, $\factor$, $\absorb$, and $\factorz$ used for type conversion, we have
    \begin{equation*}
        \adjustbox{max width=\linewidth}{$\begin{aligned}
            &(\swapp_{b_2+b_1}+\swapt_{b_2\times b_1})\seqq(\pid_{b_1}+\hT{c_1}\times \pid_{b_2})\seqq (\swapp_{b_3+b_2}+\swapt_{b_3\times b_2})\seqq(\pid_{b_2}+\hT{c_2}\times \pid_{b_3}) \\
            \fromto_2{}& (\pid_{b_2}+\hT{c_2}\times \pid_{b_1})\seqq (\swapp_{b_2+b_1}+\swapt_{b_2\times b_1})\seqq (\pid_{b_1}+\hT{c_1}\times \pid_{b_4})\seqq(\swapp_{b_3+b_4}+\swapt_{b_3\times b_4}).
        \end{aligned}$}
    \end{equation*}
\end{lemma}
\begin{proof}
    \ifSHOWPROOF%
    We first prove some equations required.
    Consider the following graphical proof.
        \begin{align*}
            &\pid_{\bbo}+((\hT{c_1}+\pid_{b_2})\seqq(\pid_{\bbo}+\swapp_{b_3+b_2})\seqq(\hT{c_2}+\pid_{b_3})\seqq(\pid_{\bbo}+\swapp_{b_4+b_3})) \\
            \fromto_2{}& \scalebox{\scaleratio}{\begin{quantikz}[row sep=3mm,column sep=2mm]
                \lstick[2,brackets=none]{$+$\hspace{-.65cm}}\push{\bbo}& & & & &\\
                \lstick[2,brackets=none]{$+$\hspace{-.65cm}}\push{\bbo}& \gate[2]{\hT{c_1}} & & \gate[2]{\hT{c_2}} & &\\
                \lstick[2,brackets=none]{$+$\hspace{-.65cm}}\push{b_1}& & \permute{2,1} & & \permute{2,1} &\\
                \push{b_2}& & & & &
            \end{quantikz}} \\
            \fromto_2{}& \scalebox{\scaleratio}{\begin{quantikz}[row sep=3mm,column sep=2mm]
                \lstick[2,brackets=none]{$+$\hspace{-.65cm}}\push{\bbo}& & \permute{2,1} & & \permute{2,1} &\\
                \lstick[2,brackets=none]{$+$\hspace{-.65cm}}\push{\bbo}& \gate[2]{\hT{c_1}} & & \gate[2]{\hT{c_2}} & &\\
                \lstick[2,brackets=none]{$+$\hspace{-.65cm}}\push{b_1}& & \permute{2,1} & & \permute{2,1} &\\
                \push{b_2}& & & & &
            \end{quantikz}} \tag{by \cref{lem:ht2}} \\
            \fromto_2{}& \scalebox{\scaleratio}{\begin{quantikz}[row sep=3mm,column sep=2mm]
                \lstick[2,brackets=none]{$+$\hspace{-.65cm}}\push{\bbo}& \permute{2,1} & & \permute{2,1} & &\\
                \lstick[2,brackets=none]{$+$\hspace{-.65cm}}\push{\bbo}& & \gate[2]{\hT{c_2}} & & \gate[2]{\hT{c_1}} &\\
                \lstick[2,brackets=none]{$+$\hspace{-.65cm}}\push{b_1} & \permute{2,1} & & \permute{2,1} & &\\
                \push{b_2}& & & & &
            \end{quantikz}} \tag{by bifunctoriality of $+$} \\
            \fromto_2{}& \scalebox{\scaleratio}{\begin{quantikz}[row sep=3mm,column sep=2mm]
                \lstick[2,brackets=none]{$+$\hspace{-.65cm}}\push{\bbo}& & & & &\\
                \lstick[2,brackets=none]{$+$\hspace{-.65cm}}\push{\bbo}& & \gate[2]{\hT{c_2}} & & \gate[2]{\hT{c_1}} &\\
                \lstick[2,brackets=none]{$+$\hspace{-.65cm}}\push{b_1} & \permute{2,1} & & \permute{2,1} & &\\
                \push{b_2}& & & & &
            \end{quantikz}} \tag{by \cref{lem:ht2}} \\
            \fromto_2{}& \pid_{\bbo}+((\pid_{\bbo}+\swapp_{b_1+b_2})\seqq(\hT{c_2}+\pid_{b_1})\seqq(\pid_{\bbo}+\swapp_{b_4+b_3}) \seqq (\hT{c_1}+\pid_{b_4})),
        \end{align*}
        then, by \cref{eq:h3}, we have
        \begin{equation}\label{eq:ht_times1}
            \begin{aligned}
                &(\hT{c_1}+\pid_{b_2})\seqq(\pid_{\bbo}+\swapp_{b_3+b_2})\seqq(\hT{c_2}+\pid_{b_3})\seqq(\pid_{\bbo}+\swapp_{b_4+b_3}) \\
                \fromto_2{}& (\pid_{\bbo}+\swapp_{b_1+b_2})\seqq(\hT{c_2}+\pid_{b_1})\seqq(\pid_{\bbo}+\swapp_{b_4+b_1}) \seqq (\hT{c_1}+\pid_{b_4}).
            \end{aligned} 
        \end{equation}
        By \cref{lem:ht2}, we also have
        \begin{align*}
            &\rbra*{(\swapp_{\bbo+\bbo}+\pid_{b_1})\seqq(\pid_{\bbo}+\hT{c_1})\seqq(\swapp_{\bbo+\bbo}+\pid_{b_3})} \times \pid_{b_2}
            \fromto_2{} \rbra*{\pid_{\bbo}+\hT{c_1}} \times \pid_{b_2} \\
            \Rightarrow{}& (\swapp_{\bbo+\bbo}\times \pid_{b_2}+\pid_{b_1}\times \pid_{b_2})\seqq (\pid_{\bbo}\times \pid_{b_2}+\hT{c_1}\times \pid_{b_2})\seqq(\swapp_{\bbo+\bbo}\times \pid_{b_2}+\pid_{b_3}\times \pid_{b_2}) \\
            & \quad \fromto_2{} \pid_{\bbo}\times \pid_{b_2}+\hT{c_1}\times \pid_{b_2}, \tag{by bifunctoriality of $\times$ and distributivity}
        \end{align*}
        then, by completeness of $\Pi$, we obtain \begin{equation}\label{eq:ht_times2}
            (\swapp_{b_2+b_2}+\pid_{b_1\times b_2})\seqq(\pid_{b_2}+\hT{c_1}\times \pid_{b_2}) \seqq (\swapp_{b_2+b_2}+\pid_{b_3\times b_2})  \fromto_2{} \pid_{b_2}+\hT{c_1}\times \pid_{b_2}.
        \end{equation}
        Consider the following graphical proof.
        \begin{align*}
            &\pid_{b_2}+\big((\pid_{\bbo+b_1}+\hT{c_1}\times \pid_{b_2})\seqq(\pid_{\bbo}+\swapp_{b_1+b_2}+\pid_{b_3\times b_2})\seqq \\
            & \quad (\hT{c_2}+\pid_{b_1}+\pid_{b_3\times b_2})\seqq(\pid_{\bbo}+\swapp_{b_4+b_1}+\pid_{b_3\times b_2})\big) \\
            \fromto_2{}& \quad \scalebox{\scaleratio}{\begin{quantikz}[row sep=3mm,column sep=2mm]
                \lstick[2,brackets=none]{$+$\hspace{-.65cm}}\push{b_2}& & & & &\\
                \lstick[2,brackets=none]{$+$\hspace{-.65cm}}\push{\bbo}& & & \gate[2]{\hT{c_2}} & &\\
                \lstick[2,brackets=none]{$+$\hspace{-.65cm}}\push{b_1}& & \permute{2,1} & & \permute{2,1} &\\
                \lstick[2,brackets=none]{$+$\hspace{-.65cm}}\push{\mathllap{\bbo\times{}}b_2}& \gate[2]{\hT{c_1}\times \pid_{b_2}} & & & &\\
                \push{\mathllap{b_1\times{}} b_2}& & & & &
            \end{quantikz}} \\
            \fromto_2{}&\quad \scalebox{\scaleratio}{\begin{quantikz}[row sep=3mm,column sep=2mm]
                \lstick[2,brackets=none]{$+$\hspace{-.65cm}}\push{b_2}&\permute{3,2,1} & & \permute{3,2,1}& & & &\\
                \lstick[2,brackets=none]{$+$\hspace{-.65cm}}\push{\bbo}&&& & & \gate[2]{\hT{c_2}} & &\\
                \lstick[2,brackets=none]{$+$\hspace{-.65cm}}\push{b_1}&&& & \permute{2,1} & & \permute{2,1} &\\
                \lstick[2,brackets=none]{$+$\hspace{-.65cm}}\push{\mathllap{\bbo\times{}}b_2}&& \gate[2]{\hT{c_1}\times \pid_{b_2}} && & & &\\
                \push{\mathllap{b_1\times{}} b_2}&&& & & & &
            \end{quantikz}} \tag{by naturality of $\swapp$}\\
            \fromto_2{}&\quad \scalebox{\scaleratio}{\begin{quantikz}[row sep=3mm,column sep=2mm]
                \lstick[2,brackets=none]{$+$\hspace{-.65cm}}\push{b_2}&\permute{3,2,1} & & & & \permute{3,2,1}& & & &\\
                \lstick[2,brackets=none]{$+$\hspace{-.65cm}}\push{\bbo}&&&&& & & \gate[2]{\hT{c_2}} & &\\
                \lstick[2,brackets=none]{$+$\hspace{-.65cm}}\push{b_1}&&\permute{2,1}&&\permute{2,1} & & \permute{2,1} & & \permute{2,1} &\\
                \lstick[2,brackets=none]{$+$\hspace{-.65cm}}\push{\mathllap{\bbo\times{}}b_2}&&& \gate[2]{\hT{c_1}\times \pid_{b_2}} &&& & & &\\
                \push{\mathllap{b_1\times{}} b_2}&&&&& & & & &
            \end{quantikz}} \tag{by \cref{eq:ht_times2}} \\
            \fromto_2{}&\quad \scalebox{\scaleratio}{\begin{quantikz}[row sep=3mm,column sep=2mm]
                \lstick[2,brackets=none]{$+$\hspace{-.65cm}}\push{b_2}&&&&\permute{3,2,1} & & & & \permute{3,2,1}&\\
                \lstick[2,brackets=none]{$+$\hspace{-.65cm}}\push{\bbo}& & \gate[2]{\hT{c_2}} & &&&&&&\\
                \lstick[2,brackets=none]{$+$\hspace{-.65cm}}\push{b_1}& \permute{2,1} & & \permute{2,1} & &\permute{2,1}& &\permute{2,1} &&\\
                \lstick[2,brackets=none]{$+$\hspace{-.65cm}}\push{\mathllap{\bbo\times{}}b_2}& & & &&& \gate[2]{\hT{c_1}\times \pid_{b_2}} &&&\\
                \push{\mathllap{b_1\times{}} b_2}&&&&& & & & &
            \end{quantikz}} \tag{by bifunctoriality of $+$} \\
            \fromto_2{}&\quad \scalebox{\scaleratio}{\begin{quantikz}[row sep=3mm,column sep=2mm]
                \lstick[2,brackets=none]{$+$\hspace{-.65cm}}\push{b_2}&&&&\permute{3,2,1} & & \permute{3,2,1}&\\
                \lstick[2,brackets=none]{$+$\hspace{-.65cm}}\push{\bbo}& & \gate[2]{\hT{c_2}} & &&&&\\
                \lstick[2,brackets=none]{$+$\hspace{-.65cm}}\push{b_1}& \permute{2,1} & & \permute{2,1} & & &&\\
                \lstick[2,brackets=none]{$+$\hspace{-.65cm}}\push{\mathllap{\bbo\times{}}b_2}& & & && \gate[2]{\hT{c_1}\times \pid_{b_2}} &&\\
                \push{\mathllap{b_1\times{}} b_2}&&&&& & &
            \end{quantikz}} \tag{by \cref{eq:ht_times2}} \\
            \fromto_2{}&\quad \scalebox{\scaleratio}{\begin{quantikz}[row sep=3mm,column sep=2mm]
                \lstick[2,brackets=none]{$+$\hspace{-.65cm}}\push{b_2}&&&&&\\
                \lstick[2,brackets=none]{$+$\hspace{-.65cm}}\push{\bbo}& & \gate[2]{\hT{c_2}} & &&\\
                \lstick[2,brackets=none]{$+$\hspace{-.65cm}}\push{b_1}& \permute{2,1} & & \permute{2,1} & &\\
                \lstick[2,brackets=none]{$+$\hspace{-.65cm}}\push{\mathllap{\bbo\times{}}b_2}& & & & \gate[2]{\hT{c_1}\times \pid_{b_2}} &\\
                \push{\mathllap{b_1\times{}} b_2}&&&&&
            \end{quantikz}} \tag{by naturality of $\swapp$} \\
            \fromto_2{}& \pid_{b_2}+\big((\pid_{\bbo}+\swapp_{b_1+b_2}+\pid_{b_1\times b_2})\seqq(\hT{c_2}+\pid_{b_2}+\pid_{b_1\times b_2})\seqq\\
            & \quad (\pid_{\bbo}+\swapp_{b_4+b_1}+\pid_{b_1\times b_2})\seqq(\pid_{\bbo+b_1}+\hT{c_1}\times \pid_{b_2})\big).
        \end{align*}
        By \cref{eq:h3}, we have 
        \begin{equation}\label{eq:ht_times3}
            \begin{aligned}
                &(\pid_{\bbo+b_1}+\hT{c_1}\times \pid_{b_2})\seqq(\pid_{\bbo}+\swapp_{b_1+b_2}+\pid_{b_3\times b_2})\seqq(\hT{c_2}+\pid_{b_1}+\pid_{b_3\times b_2})\seqq \\
                & \quad (\pid_{\bbo}+\swapp_{b_4+b_1}+\pid_{b_3\times b_2}) \\
                \fromto_2{}& (\pid_{\bbo}+\swapp_{b_1+b_2}+\pid_{b_1\times b_2})\seqq(\hT{c_2}+\pid_{b_1}+\pid_{b_1\times b_2})\seqq (\pid_{\bbo}+\swapp_{b_4+b_1}+\pid_{b_1\times b_2})\seqq\\
                & \quad (\pid_{\bbo+b_1}+\hT{c_1}\times \pid_{b_2})
            \end{aligned}
        \end{equation}
        Thus, \begin{align*}
            &(\hT{c_1}\times \pid_{\bbo+b_2})\seqq(\pid_{\bbo+b_3}\times \hT{c_2}) \\
            \fromto_2{}& \blue{\swapt_{(\bbo+b_1)\times (\bbo+b_2)}\seqq(\pid_{\bbo+b_2}\times \hT{c_1})\seqq\swapt_{(\bbo+b_2)\times (\bbo+b_3)}}\seqq(\pid_{\bbo+b_3}\times \hT{c_2}) \tag{by naturality of $\swapt$} \\
            \fromto_2{}& \swapt_{(\bbo+b_1)\times (\bbo+b_2)}\seqq(\blue{\hT{c_1}+\pid_{b_2}\times \hT{c_1}})\seqq\swapt_{(\bbo+b_2)\times (\bbo+b_3)}\seqq(\blue{\hT{c_2}+\pid_{b_3}\times \hT{c_2}}) \tag{by distributivity} \\
            \fromto_2{}& \swapt_{(\bbo+b_1)\times (\bbo+b_2)}\seqq  (\hT{c_1}+(\blue{\swapt_{b_2\times (\bbo+b_1)}\seqq(\hT{c_1}\times \pid_{b_2})\seqq(\swapt_{(\bbo+b_3)\times b_2})}))\seqq \\
            & \quad \swapt_{(\bbo+b_2)\times (\bbo+b_3)}\seqq(\hT{c_2}+\pid_{b_3}\times \hT{c_2}) \tag{by naturality of $\swapt$} \\
            \fromto_2{}& \swapt_{(\bbo+b_1)\times (\bbo+b_2)}\seqq \blue{(\pid_{\bbo+b_1}+\swapt_{b_2\times (\bbo+b_1)})\seqq (\hT{c_1}+(\hT{c_1}\times \pid_{b_2}))\seqq} \\
            & \blue{(\pid_{\bbo+b_3}+\swapt_{(\bbo+b_3)\times b_2})}\seqq \swapt_{(\bbo+b_2)\times (\bbo+b_3)}\seqq(\hT{c_2}+\pid_{b_3\times(\bbo+b_2)})\seqq(\pid_{\bbo+b_4}+\pid_{b_3}\times \hT{c_2}) \tag{by bifunctoriality of $+$} \\
            \fromto_2{}& \swapt_{(\bbo+b_1)\times (\bbo+b_2)}\seqq (\pid_{\bbo+b_1}+\swapt_{b_2\times (\bbo+b_1)})\seqq (\hT{c_1}+(\hT{c_1}\times \pid_{b_2}))\seqq \\
            & \quad (\pid_{\bbo+b_3}+\swapt_{(\bbo+b_3)\times b_2})\seqq \swapt_{(\bbo+b_2)\times (\bbo+b_3)}\seqq \\
            & \quad \blue{(\pid_{\bbo+b_2}+\swapt_{b_3\times (\bbo+b_2)})\seqq(\hT{c_2}+\swapt_{(\bbo+b_2)\times b_3})}\seqq (\pid_{\bbo+b_4}+\pid_{b_3}\times \hT{c_2}) \tag{by naturality of $\swapt$ and bifunctoriality of $+$} \\
            \fromto_2{}& \swapt_{(\bbo+b_1)\times (\bbo+b_2)}\seqq (\pid_{\bbo+b_1}+\swapt_{b_2\times (\bbo+b_1)})\seqq (\hT{c_1}+(\hT{c_1}\times \pid_{b_2}))\seqq \\
            & \quad (\pid_{\bbo+b_3}+\swapt_{(\bbo+b_3)\times b_2})\seqq\swapt_{(\bbo+b_2)\times (\bbo+b_3)}\seqq(\pid_{\bbo+b_2}+\swapt_{b_3\times (\bbo+b_2)})\seqq\\
            & \quad \blue{(\pid_{\bbo+b_2+b_3}+\swapt_{b_2\times b_3})\seqq} \blue{(\hT{c_2}+((\pid_{b_3}+\swapt_{b_3\times b_2})\seqq\swapt_{(\bbo+b_2)\times b_3}))}\seqq \\
            & \quad (\pid_{\bbo+b_4}+\pid_{b_3}\times \hT{c_2}) \tag{by naturality of $\swapt$ and bifunctoriality of $+$} \\
            \fromto_2{}& \swapt_{(\bbo+b_1)\times (\bbo+b_2)}\seqq (\pid_{\bbo+b_1}+\swapt_{b_2\times (\bbo+b_1)})\seqq (\hT{c_1}+(\hT{c_1}\times \pid_{b_2}))\seqq \\
            &\quad \blue{(\pid_{\bbo}+\swapp_{b_3+b_2}+\pid_{b_3\times b_2})} \seqq
            (\hT{c_2}+((\pid_{b_3}+\swapt_{b_3\times b_2})\seqq\swapt_{(\bbo+b_2)\times b_3})) \\
            &\quad \seqq (\pid_{\bbo+b_4}+\pid_{b_3}\times \hT{c_2}) \tag{by completeness of $\Pi$} \\
            \fromto_2{}& \swapt_{(\bbo+b_1)\times (\bbo+b_2)}\seqq (\pid_{\bbo+b_1}+\swapt_{b_2\times (\bbo+b_1)})\seqq (\hT{c_1}+\pid_{(\bbo+b_1)\times b_2})\seqq \\
            &\quad \blue{(\pid_{\bbo+b_3}+(\hT{c_1}\times \pid_{b_2}))}\seqq(\pid_{\bbo}+\swapp_{b_3+b_2}+\pid_{b_3\times b_2}) \seqq \blue{(\hT{c_2}+\pid_{b_3}+\pid_{b_3\times b_2})\seqq}\\
            &\quad \blue{(\pid_{\bbo+b_4}+((\pid_{b_3}+\swapt_{b_3\times b_2})\seqq\swapt_{(\bbo+b_2)\times b_3}))}\seqq (\pid_{\bbo+b_4}+\pid_{b_3}\times \hT{c_2}) \tag{by bifunctoriality of $+$} \\
            \fromto_2{}& \swapt_{(\bbo+b_1)\times (\bbo+b_2)}\seqq (\pid_{\bbo+b_1}+\swapt_{b_2\times (\bbo+b_1)})\seqq (\hT{c_1}+\pid_{(\bbo+b_1)\times b_2})\seqq \\
            &\quad \blue{(\pid_{\bbo}+\swapp_{b_3+b_2}+\pid_{b_3\times b_2}) \seqq (\hT{c_2}+\pid_{b_3}+\pid_{b_3\times b_2})\seqq (\pid_{\bbo}+\swapp_{b_4+b_3}+\pid_{b_3\times b_2})\seqq}\\
            &\quad \blue{(\pid_{\bbo+b_3}+(\hT{c_1}\times \pid_{b_2}))\seqq(\pid_{\bbo}+\swapp_{b_3+b_4}+\pid_{b_3\times b_2})}\seqq \\
            & \quad (\pid_{\bbo+b_4}+((\pid_{b_3}+\swapt_{b_3\times b_2})\seqq\swapt_{(\bbo+b_2)\times b_3}))\seqq (\pid_{\bbo+b_4}+\pid_{b_3}\times \hT{c_2}) \tag{by \cref{eq:ht_times3} and strictness}
        \end{align*}
        In the same way, by the naturality of $\swapt$, we can get \begin{align*}
            &(\pid_{\bbo+b_1}\times \hT{c_2})\seqq(\hT{c_1}\times \pid_{\bbo+b_4}) \\
            \fromto_2{}& \swapt_{(\bbo+b_1)\times (\bbo+b_2)}\seqq (\hT{c_2}\times \pid_{\bbo+b_1})\seqq (\pid_{\bbo+b_4}\times \hT{c_1}) \seqq \swapt_{(\bbo+b_4)\times (\bbo+b_3)} \\
            \fromto_2{}& \swapt_{(\bbo+b_1)\times (\bbo+b_2)}\seqq \\
            & \quad \swapt_{(\bbo+b_2)\times (\bbo+b_1)}\seqq (\pid_{\bbo+b_2}+\swapt_{b_1\times (\bbo+b_2)})\seqq (\hT{c_2}+\pid_{(\bbo+b_2)\times b_1})\seqq \\
            &\quad (\pid_{\bbo}+\swapp_{b_4+b_1}+\pid_{b_4\times b_1}) \seqq (\hT{c_1}+\pid_{b_4}+\pid_{b_4\times b_1})\seqq (\pid_{\bbo}+\swapp_{b_3+b_4}+\pid_{b_4\times b_1})\seqq\\
            &\quad (\pid_{\bbo+b_4}+(\hT{c_2}\times \pid_{b_1}))\seqq(\pid_{\bbo}+\swapp_{b_4+b_3}+\pid_{b_4\times b_1})\seqq \\
            & \quad (\pid_{\bbo+b_3}+((\pid_{b_4}+\swapt_{b_4\times b_1})\seqq\swapt_{(\bbo+b_1)\times b_4}))\seqq (\pid_{\bbo+b_3}+\pid_{b_4}\times \hT{c_1})\seqq \\
            &\quad \swapt_{(\bbo+b_4)\times (\bbo+b_3)} \\
            \fromto_2{}& (\pid_{\bbo+b_2}+\swapt_{b_1\times (\bbo+b_2)})\seqq (\hT{c_2}+\pid_{(\bbo+b_2)\times b_1})\seqq \\
            &\quad (\pid_{\bbo}+\swapp_{b_4+b_1}+\pid_{b_4\times b_1}) \seqq (\hT{c_1}+\pid_{b_4}+\pid_{b_4\times b_1})\seqq (\pid_{\bbo}+\swapp_{b_3+b_4}+\pid_{b_4\times b_1})\seqq\\
            &\quad (\pid_{\bbo+b_4}+(\hT{c_2}\times \pid_{b_1}))\seqq(\pid_{\bbo}+\swapp_{b_4+b_3}+\pid_{b_4\times b_1})\seqq \\
            & \quad (\pid_{\bbo+b_3}+((\pid_{b_4}+\swapt_{b_4\times b_1})\seqq\swapt_{(\bbo+b_1)\times b_4}))\seqq (\pid_{\bbo+b_3}+\pid_{b_4}\times \hT{c_1})\seqq \swapt_{(\bbo+b_4)\times (\bbo+b_3)}.
        \end{align*}
        Combined with 
        \[\hT{c_1}\times \hT{c_2}\fromto_2{} (\hT{c_1}\times \pid_{\bbo+b_2})\seqq(\pid_{\bbo+b_3}\times \hT{c_2}) \fromto_2{} (\pid_{\bbo+b_1}\times \hT{c_2})\seqq(\hT{c_1}\times \pid_{\bbo+b_4}),\]
        we have 
        \begin{align*}
            & \swapt_{(\bbo+b_1)\times (\bbo+b_2)}\seqq (\pid_{\bbo+b_1}+\swapt_{b_2\times (\bbo+b_1)})\seqq (\hT{c_1}+\pid_{(\bbo+b_1)\times b_2})\seqq \\
            &\quad (\pid_{\bbo}+\swapp_{b_3+b_2}+\pid_{b_3\times b_2}) \seqq (\hT{c_2}+\pid_{b_3}+\pid_{b_3\times b_2})\seqq (\pid_{\bbo}+\swapp_{b_4+b_3}+\pid_{b_3\times b_2})\seqq\\
            &\quad (\pid_{\bbo+b_3}+(\hT{c_1}\times \pid_{b_2}))\seqq(\pid_{\bbo}+\swapp_{b_3+b_4}+\pid_{b_3\times b_2})\seqq \\
            & \quad (\pid_{\bbo+b_4}+((\pid_{b_3}+\swapt_{b_3\times b_2})\seqq\swapt_{(\bbo+b_2)\times b_3}))\seqq (\pid_{\bbo+b_4}+\pid_{b_3}\times \hT{c_2}) \\
            &\fromto_2{} (\pid_{\bbo+b_2}+\swapt_{b_1\times (\bbo+b_2)})\seqq (\hT{c_2}+\pid_{(\bbo+b_2)\times b_1})\seqq \\
            &\quad (\pid_{\bbo}+\swapp_{b_4+b_1}+\pid_{b_4\times b_1}) \seqq (\hT{c_1}+\pid_{b_4}+\pid_{b_4\times b_1})\seqq (\pid_{\bbo}+\swapp_{b_3+b_4}+\pid_{b_4\times b_1})\seqq\\
            &\quad (\pid_{\bbo+b_4}+(\hT{c_2}\times \pid_{b_1}))\seqq(\pid_{\bbo}+\swapp_{b_4+b_3}+\pid_{b_4\times b_1})\seqq \\
            & \quad (\pid_{\bbo+b_3}+((\pid_{b_4}+\swapt_{b_4\times b_1})\seqq\swapt_{(\bbo+b_1)\times b_4}))\seqq (\pid_{\bbo+b_3}+\pid_{b_4}\times \hT{c_1})\seqq \swapt_{(\bbo+b_4)\times (\bbo+b_3)} \\
            \Rightarrow{}& \swapt_{(\bbo+b_1)\times (\bbo+b_2)}\seqq (\pid_{\bbo+b_1}+\swapt_{b_2\times (\bbo+b_1)})\seqq\\
            &\quad \blue{(((\hT{c_1}+\pid_{b_2})\seqq (\pid_{\bbo}+\swapp_{b_3+b_2}) \seqq (\hT{c_2}+\pid_{b_3})\seqq (\pid_{\bbo}+\swapp_{b_4+b_3}))+\pid_{b_3\times b_2})}\seqq\\
            &\quad (\pid_{\bbo+b_3}+(\hT{c_1}\times \pid_{b_2}))\seqq(\pid_{\bbo}+\swapp_{b_3+b_4}+\pid_{b_3\times b_2})\seqq \\
            & \quad (\pid_{\bbo+b_4}+((\pid_{b_3}+\swapt_{b_3\times b_2})\seqq\swapt_{(\bbo+b_2)\times b_3}))\seqq (\pid_{\bbo+b_4}+\pid_{b_3}\times \hT{c_2}) \\
            &\fromto_2{} (\pid_{\bbo+b_2}+\swapt_{b_1\times (\bbo+b_2)})\seqq (\hT{c_2}+\pid_{(\bbo+b_2)\times b_1})\seqq \\
            &\quad (\pid_{\bbo}+\swapp_{b_4+b_1}+\pid_{b_4\times b_1}) \seqq (\hT{c_1}+\pid_{b_4}+\pid_{b_4\times b_1})\seqq (\pid_{\bbo}+\swapp_{b_3+b_4}+\pid_{b_4\times b_1})\seqq\\
            &\quad (\pid_{\bbo+b_4}+(\hT{c_2}\times \pid_{b_1}))\seqq(\pid_{\bbo}+\swapp_{b_4+b_3}+\pid_{b_4\times b_1})\seqq \\
            & \quad (\pid_{\bbo+b_3}+((\pid_{b_4}+\swapt_{b_4\times b_1})\seqq\swapt_{(\bbo+b_1)\times b_4}))\seqq (\pid_{\bbo+b_3}+\pid_{b_4}\times \hT{c_1})\seqq \swapt_{(\bbo+b_4)\times (\bbo+b_3)} \tag{by bifunctoriality of $+$ and strictness} \\
            \Rightarrow{}& \swapt_{(\bbo+b_1)\times (\bbo+b_2)}\seqq (\pid_{\bbo+b_1}+\swapt_{b_2\times (\bbo+b_1)})\seqq\\
            &\quad \blue{(((\pid_{\bbo}+\swapp_{b_1+b_2}) \seqq (\hT{c_2}+\pid_{b_1})\seqq (\pid_{\bbo}+\swapp_{b_4+b_1})\seqq(\hT{c_1}+\pid_{b_4}))+\pid_{b_3\times b_2})}\seqq\\
            &\quad (\pid_{\bbo+b_3}+(\hT{c_1}\times \pid_{b_2}))\seqq(\pid_{\bbo}+\swapp_{b_3+b_4}+\pid_{b_3\times b_2})\seqq \\
            & \quad (\pid_{\bbo+b_4}+((\pid_{b_3}+\swapt_{b_3\times b_2})\seqq\swapt_{(\bbo+b_2)\times b_3}))\seqq (\pid_{\bbo+b_4}+\pid_{b_3}\times \hT{c_2}) \\
            &\fromto_2{} (\pid_{\bbo+b_2}+\swapt_{b_1\times (\bbo+b_2)})\seqq (\hT{c_2}+\pid_{(\bbo+b_2)\times b_1})\seqq \\
            &\quad (\pid_{\bbo}+\swapp_{b_4+b_1}+\pid_{b_4\times b_1}) \seqq (\hT{c_1}+\pid_{b_4}+\pid_{b_4\times b_1})\seqq (\pid_{\bbo}+\swapp_{b_3+b_4}+\pid_{b_4\times b_1})\seqq\\
            &\quad (\pid_{\bbo+b_4}+(\hT{c_2}\times \pid_{b_1}))\seqq(\pid_{\bbo}+\swapp_{b_4+b_3}+\pid_{b_4\times b_1})\seqq \\
            & \quad (\pid_{\bbo+b_3}+((\pid_{b_4}+\swapt_{b_4\times b_1})\seqq\swapt_{(\bbo+b_1)\times b_4}))\seqq (\pid_{\bbo+b_3}+\pid_{b_4}\times \hT{c_1})\seqq \swapt_{(\bbo+b_4)\times (\bbo+b_3)} \tag{by \cref{eq:ht_times2}} \\
            \Rightarrow{}& \blue{(\pid_{\bbo+b_2}+\swapt_{(\bbo+b_2)\times b_1})\seqq}\swapt_{(\bbo+b_1)\times (\bbo+b_2)}\seqq (\pid_{\bbo+b_1}+\swapt_{b_2\times (\bbo+b_1)})\seqq\\
            & \quad \blue{(\pid_{\bbo}+\swapp_{b_1+b_2}+\pid_{b_3\times b_2})} \seqq\\
            &\quad (((\hT{c_2}+\pid_{b_1})\seqq (\pid_{\bbo}+\swapp_{b_4+b_1})\seqq(\hT{c_1}+\pid_{b_4}))+\pid_{b_3\times b_2})\seqq\\
            &\quad (\pid_{\bbo+b_3}+(\hT{c_1}\times \pid_{b_2}))\seqq(\pid_{\bbo}+\swapp_{b_3+b_4}+\pid_{b_3\times b_2})\seqq \\
            & \quad (\pid_{\bbo+b_4}+((\pid_{b_3}+\swapt_{b_3\times b_2})\seqq\swapt_{(\bbo+b_2)\times b_3}))\seqq (\pid_{\bbo+b_4}+\pid_{b_3}\times \hT{c_2}) \\
            &\fromto_2{} (\hT{c_2}+\pid_{(\bbo+b_2)\times b_1})\seqq \\
            &\quad (\pid_{\bbo}+\swapp_{b_4+b_1}+\pid_{b_4\times b_1}) \seqq (\hT{c_1}+\pid_{b_4}+\pid_{b_4\times b_1})\seqq (\pid_{\bbo}+\swapp_{b_3+b_4}+\pid_{b_4\times b_1})\seqq\\
            &\quad (\pid_{\bbo+b_4}+(\hT{c_2}\times \pid_{b_1}))\seqq(\pid_{\bbo}+\swapp_{b_4+b_3}+\pid_{b_4\times b_1})\seqq \\
            & \quad (\pid_{\bbo+b_3}+((\pid_{b_4}+\swapt_{b_4\times b_1})\seqq\swapt_{(\bbo+b_1)\times b_4}))\seqq (\pid_{\bbo+b_3}+\pid_{b_4}\times \hT{c_1})\seqq \swapt_{(\bbo+b_4)\times (\bbo+b_3)} \tag{by completeness of $\Pi$ and bifunctoriality of $+$} \\
            \Rightarrow{}& \blue{(\pid_{\bbo+b_2+b_1}+\swapt_{b_2\times b_1})} \seqq\\
            &\quad (((\hT{c_2}+\pid_{b_1})\seqq (\pid_{\bbo}+\swapp_{b_4+b_1})\seqq(\hT{c_1}+\pid_{b_4}))+\pid_{b_3\times b_2})\seqq\\
            &\quad (\pid_{\bbo+b_3}+(\hT{c_1}\times \pid_{b_2}))\seqq(\pid_{\bbo}+\swapp_{b_3+b_4}+\pid_{b_3\times b_2})\seqq \\
            & \quad (\pid_{\bbo+b_4}+((\pid_{b_3}+\swapt_{b_3\times b_2})\seqq\swapt_{(\bbo+b_2)\times b_3}))\seqq (\pid_{\bbo+b_4}+\pid_{b_3}\times \hT{c_2}) \\
            &\fromto_2{} (\hT{c_2}+\pid_{(\bbo+b_2)\times b_1})\seqq \\
            &\quad (\pid_{\bbo}+\swapp_{b_4+b_1}+\pid_{b_4\times b_1}) \seqq (\hT{c_1}+\pid_{b_4}+\pid_{b_4\times b_1})\seqq (\pid_{\bbo}+\swapp_{b_3+b_4}+\pid_{b_4\times b_1})\seqq\\
            &\quad (\pid_{\bbo+b_4}+(\hT{c_2}\times \pid_{b_1}))\seqq(\pid_{\bbo}+\swapp_{b_4+b_3}+\pid_{b_4\times b_1})\seqq \\
            & \quad (\pid_{\bbo+b_3}+((\pid_{b_4}+\swapt_{b_4\times b_1})\seqq\swapt_{(\bbo+b_1)\times b_4}))\seqq (\pid_{\bbo+b_3}+\pid_{b_4}\times \hT{c_1})\seqq \swapt_{(\bbo+b_4)\times (\bbo+b_3)} \tag{by completeness of $\Pi$} \\
            \Rightarrow{}& \blue{(\pid_{\bbo+b_2}+\pid_{b_1}+\swapt_{b_2\times b_1}) \seqq (\hT{c_2}+\pid_{b_1}+\pid_{b_1\times b_2})\seqq}\\
            &\quad (((\pid_{\bbo}+\swapp_{b_4+b_1})\seqq(\hT{c_1}+\pid_{b_4}))+\pid_{b_3\times b_2})\seqq\\
            &\quad (\pid_{\bbo+b_3}+(\hT{c_1}\times \pid_{b_2}))\seqq(\pid_{\bbo}+\swapp_{b_3+b_4}+\pid_{b_3\times b_2})\seqq \\
            & \quad (\pid_{\bbo+b_4}+((\pid_{b_3}+\swapt_{b_3\times b_2})\seqq\swapt_{(\bbo+b_2)\times b_3}))\seqq (\pid_{\bbo+b_4}+\pid_{b_3}\times \hT{c_2}) \\
            &\fromto_2{} (\blue{\hT{c_2}+\pid_{b_1}+\pid_{b_2\times b_1}})\seqq \\
            &\quad (\pid_{\bbo}+\swapp_{b_4+b_1}+\pid_{b_4\times b_1}) \seqq (\hT{c_1}+\pid_{b_4}+\pid_{b_4\times b_1})\seqq (\pid_{\bbo}+\swapp_{b_3+b_4}+\pid_{b_4\times b_1})\seqq\\
            &\quad (\pid_{\bbo+b_4}+(\hT{c_2}\times \pid_{b_1}))\seqq(\pid_{\bbo}+\swapp_{b_4+b_3}+\pid_{b_4\times b_1})\seqq \\
            & \quad (\pid_{\bbo+b_3}+((\pid_{b_4}+\swapt_{b_4\times b_1})\seqq\swapt_{(\bbo+b_1)\times b_4}))\seqq (\pid_{\bbo+b_3}+\pid_{b_4}\times \hT{c_1})\seqq \swapt_{(\bbo+b_4)\times (\bbo+b_3)} \tag{by bifunctoriality of $+$ and strictness} \\
            \Rightarrow{}& \blue{(\hT{c_2}+\pid_{b_1}+\pid_{b_2\times b_1})\seqq(\pid_{\bbo+b_4}+\pid_{b_1}+\swapt_{b_2\times b_1})}\seqq\\
            &\quad (((\pid_{\bbo}+\swapp_{b_4+b_1})\seqq(\hT{c_1}+\pid_{b_4}))+\pid_{b_3\times b_2})\seqq\\
            &\quad (\pid_{\bbo+b_3}+(\hT{c_1}\times \pid_{b_2}))\seqq(\pid_{\bbo}+\swapp_{b_3+b_4}+\pid_{b_3\times b_2})\seqq \\
            & \quad (\pid_{\bbo+b_4}+((\pid_{b_3}+\swapt_{b_3\times b_2})\seqq\swapt_{(\bbo+b_2)\times b_3}))\seqq (\pid_{\bbo+b_4}+\pid_{b_3}\times \hT{c_2}) \\
            &\fromto_2{} (\hT{c_2}+\pid_{b_1}+\pid_{b_2\times b_1})\seqq \\
            &\quad (\pid_{\bbo}+\swapp_{b_4+b_1}+\pid_{b_4\times b_1}) \seqq (\hT{c_1}+\pid_{b_4}+\pid_{b_4\times b_1})\seqq (\pid_{\bbo}+\swapp_{b_3+b_4}+\pid_{b_4\times b_1})\seqq\\
            &\quad (\pid_{\bbo+b_4}+(\hT{c_2}\times \pid_{b_1}))\seqq(\pid_{\bbo}+\swapp_{b_4+b_3}+\pid_{b_4\times b_1})\seqq \\
            & \quad (\pid_{\bbo+b_3}+((\pid_{b_4}+\swapt_{b_4\times b_1})\seqq\swapt_{(\bbo+b_1)\times b_4}))\seqq (\pid_{\bbo+b_3}+\pid_{b_4}\times \hT{c_1})\seqq \swapt_{(\bbo+b_4)\times (\bbo+b_3)} \tag{by bifunctoriality of $+$ and naturality of $\swapt$} \\
            \Rightarrow{}& (\pid_{\bbo+b_4}+\pid_{b_1}+\swapt_{b_2\times b_1})\seqq (((\pid_{\bbo}+\swapp_{b_4+b_1})\seqq(\hT{c_1}+\pid_{b_4}))+\pid_{b_3\times b_2})\seqq\\
            &\quad (\pid_{\bbo+b_3}+(\hT{c_1}\times \pid_{b_2}))\seqq(\pid_{\bbo}+\swapp_{b_3+b_4}+\pid_{b_3\times b_2})\seqq \\
            & \quad (\pid_{\bbo+b_4}+((\pid_{b_3}+\swapt_{b_3\times b_2})\seqq\swapt_{(\bbo+b_2)\times b_3}))\seqq (\pid_{\bbo+b_4}+\pid_{b_3}\times \hT{c_2}) \\
            &\fromto_2{} (\pid_{\bbo}+\swapp_{b_4+b_1}+\pid_{b_4\times b_1}) \seqq (\hT{c_1}+\pid_{b_4}+\pid_{b_4\times b_1})\seqq (\pid_{\bbo}+\swapp_{b_3+b_4}+\pid_{b_4\times b_1})\seqq\\
            &\quad (\pid_{\bbo+b_4}+(\hT{c_2}\times \pid_{b_1}))\seqq(\pid_{\bbo}+\swapp_{b_4+b_3}+\pid_{b_4\times b_1})\seqq \\
            & \quad (\pid_{\bbo+b_3}+((\pid_{b_4}+\swapt_{b_4\times b_1})\seqq\swapt_{(\bbo+b_1)\times b_4}))\seqq (\pid_{\bbo+b_3}+\pid_{b_4}\times \hT{c_1})\seqq \swapt_{(\bbo+b_4)\times (\bbo+b_3)} \tag{by reversibility} \\
            \Rightarrow{}& \blue{(((\pid_{\bbo}+\swapp_{b_4+b_1})\seqq(\hT{c_1}+\pid_{b_4}))+\pid_{b_2\times b_1})\seqq(\pid_{\bbo+b_4}+\pid_{b_1}+\swapt_{b_2\times b_1})}\seqq\\
            &\quad (\pid_{\bbo+b_3}+(\hT{c_1}\times \pid_{b_2}))\seqq(\pid_{\bbo}+\swapp_{b_3+b_4}+\pid_{b_3\times b_2})\seqq \\
            & \quad (\pid_{\bbo+b_4}+((\pid_{b_3}+\swapt_{b_3\times b_2})\seqq\swapt_{(\bbo+b_2)\times b_3}))\seqq (\pid_{\bbo+b_4}+\pid_{b_3}\times \hT{c_2}) \\
            &\fromto_2{} (\pid_{\bbo}+\swapp_{b_4+b_1}+\pid_{b_4\times b_1}) \seqq (\hT{c_1}+\pid_{b_4}+\pid_{b_4\times b_1})\seqq (\pid_{\bbo}+\swapp_{b_3+b_4}+\pid_{b_4\times b_1})\seqq\\
            &\quad (\pid_{\bbo+b_4}+(\hT{c_2}\times \pid_{b_1}))\seqq(\pid_{\bbo}+\swapp_{b_4+b_3}+\pid_{b_4\times b_1})\seqq \\
            & \quad (\pid_{\bbo+b_3}+((\pid_{b_4}+\swapt_{b_4\times b_1})\seqq\swapt_{(\bbo+b_1)\times b_4}))\seqq (\pid_{\bbo+b_3}+\pid_{b_4}\times \hT{c_1})\seqq \swapt_{(\bbo+b_4)\times (\bbo+b_3)} \tag{by bifunctoriality of $+$, naturality of $\swapt$ and strictness} \\
            \Rightarrow{}& (\pid_{\bbo+b_4}+\pid_{b_1}+\swapt_{b_2\times b_1})\seqq (\pid_{\bbo+b_3}+(\hT{c_1}\times \pid_{b_2}))\seqq(\pid_{\bbo}+\swapp_{b_3+b_4}+\pid_{b_3\times b_2})\seqq \\
            & \quad (\pid_{\bbo+b_4}+((\pid_{b_3}+\swapt_{b_3\times b_2})\seqq\swapt_{(\bbo+b_2)\times b_3}))\seqq (\pid_{\bbo+b_4}+\pid_{b_3}\times \hT{c_2}) \\
            &\fromto_2{} (\pid_{\bbo}+\swapp_{b_3+b_4}+\pid_{b_4\times b_1})\seqq (\pid_{\bbo+b_4}+(\hT{c_2}\times \pid_{b_1}))\seqq(\pid_{\bbo}+\swapp_{b_4+b_3}+\pid_{b_4\times b_1})\seqq \\
            & \quad (\pid_{\bbo+b_3}+((\pid_{b_4}+\swapt_{b_4\times b_1})\seqq\swapt_{(\bbo+b_1)\times b_4}))\seqq (\pid_{\bbo+b_3}+\pid_{b_4}\times \hT{c_1})\seqq \swapt_{(\bbo+b_4)\times (\bbo+b_3)} \tag{by bifunctoriality of $+$ and reversibility} \\
            \Rightarrow{}& (\pid_{\bbo+b_4}+\pid_{b_1}+\swapt_{b_2\times b_1})\seqq (\pid_{\bbo+b_3}+(\hT{c_1}\times \pid_{b_2}))\seqq(\pid_{\bbo}+\swapp_{b_3+b_4}+\pid_{b_3\times b_2})\seqq \\
            & \quad (\pid_{\bbo+b_4}+((\pid_{b_3}+\swapt_{b_3\times b_2})\seqq\swapt_{(\bbo+b_2)\times b_3}))\seqq (\pid_{\bbo+b_4}+\pid_{b_3}\times \hT{c_2}) \\
            &\fromto_2{} (\pid_{\bbo}+\swapp_{b_3+b_4}+\pid_{b_4\times b_1})\seqq (\pid_{\bbo+b_4}+(\hT{c_2}\times \pid_{b_1}))\seqq(\pid_{\bbo}+\swapp_{b_4+b_3}+\pid_{b_4\times b_1})\seqq \\
            & \quad (\pid_{\bbo+b_3}+((\pid_{b_4}+\swapt_{b_4\times b_1})\seqq\swapt_{(\bbo+b_1)\times b_4}))\seqq (\pid_{\bbo+b_3}+\pid_{b_4}\times \hT{c_1})\seqq \\
            & \quad (\pid_{\bbo+b_3}+\swapt_{b_4\times(\bbo+b_3)})\seqq(\pid_{\bbo}+\swapp_{b_3+b_4}+\swapt_{b_3\times b_4})\seqq(\pid_{\bbo+b_4}+\swapt_{(\bbo+b_4)\times b_3}) \tag{by completeness of $\Pi$} \\
            \Rightarrow{}& (\pid_{\bbo+b_4}+\pid_{b_1}+\swapt_{b_2\times b_1})\seqq (\pid_{\bbo+b_3}+(\hT{c_1}\times \pid_{b_2}))\seqq(\pid_{\bbo}+\swapp_{b_3+b_4}+\pid_{b_3\times b_2})\seqq \\
            & \quad (\pid_{\bbo+b_4}+((\pid_{b_3}+\swapt_{b_3\times b_2})\seqq\swapt_{(\bbo+b_2)\times b_3}))\seqq (\pid_{\bbo+b_4}+\pid_{b_3}\times \hT{c_2}) \\
            &\fromto_2{} (\pid_{\bbo}+\swapp_{b_3+b_4}+\pid_{b_4\times b_1})\seqq (\pid_{\bbo+b_4}+(\hT{c_2}\times \pid_{b_1}))\seqq(\pid_{\bbo}+\swapp_{b_4+b_3}+\pid_{b_4\times b_1})\seqq \\
            & \quad (\pid_{\bbo+b_3}+((\pid_{b_4}+\swapt_{b_4\times b_1})))\seqq\blue{(\pid_{\bbo+b_3}+\swapt_{(\bbo+b_1)\times b_4})}\seqq (\pid_{\bbo+b_3}+\pid_{b_4}\times \hT{c_1})\seqq \\
            & \quad (\pid_{\bbo+b_3}+\swapt_{b_4\times(\bbo+b_3)})\seqq(\pid_{\bbo}+\swapp_{b_3+b_4}+\swapt_{b_3\times b_4})\seqq(\pid_{\bbo+b_4}+\swapt_{(\bbo+b_4)\times b_3}) \tag{by bifunctoriality of $+$} \\
            \Rightarrow{}& (\pid_{\bbo+b_4}+\pid_{b_1}+\swapt_{b_2\times b_1})\seqq (\pid_{\bbo+b_3}+(\hT{c_1}\times \pid_{b_2}))\seqq(\pid_{\bbo}+\swapp_{b_3+b_4}+\pid_{b_3\times b_2})\seqq \\
            & \quad (\pid_{\bbo+b_4}+((\pid_{b_3}+\swapt_{b_3\times b_2})\seqq\swapt_{(\bbo+b_2)\times b_3}))\seqq (\pid_{\bbo+b_4}+\pid_{b_3}\times \hT{c_2}) \\
            &\fromto_2{} (\pid_{\bbo}+\swapp_{b_3+b_4}+\pid_{b_4\times b_1})\seqq (\pid_{\bbo+b_4}+(\hT{c_2}\times \pid_{b_1}))\seqq(\pid_{\bbo}+\swapp_{b_4+b_3}+\pid_{b_4\times b_1})\seqq \\
            & \quad (\pid_{\bbo+b_3}+((\pid_{b_4}+\swapt_{b_4\times b_1})))\seqq\blue{(\pid_{\bbo+b_3}+\hT{c_1}\times \pid_{b_4})}\seqq \\
            & \quad (\pid_{\bbo}+\swapp_{b_3+b_4}+\swapt_{b_3\times b_4})\seqq(\pid_{\bbo+b_4}+\swapt_{(\bbo+b_4)\times b_3}) \tag{by naturality of $\swapt$} \\
            \Rightarrow{}& (\pid_{\bbo+b_4}+\pid_{b_1}+\swapt_{b_2\times b_1})\seqq (\pid_{\bbo+b_3}+(\hT{c_1}\times \pid_{b_2}))\seqq(\pid_{\bbo}+\swapp_{b_3+b_4}+\pid_{b_3\times b_2})\seqq \\
            & \quad (\pid_{\bbo+b_4}+((\pid_{b_3}+\swapt_{b_3\times b_2})))\seqq\\
            & \quad \blue{(\pid_{\bbo+b_4}+\swapt_{(\bbo+b_2)\times b_3})}\seqq (\pid_{\bbo+b_4}+\pid_{b_3}\times \hT{c_2})\blue{\seqq(\pid_{\bbo+b_4}+\swapt_{b_3\times (\bbo+b_4)})}\seqq \\
            &\fromto_2{} (\pid_{\bbo}+\swapp_{b_3+b_4}+\pid_{b_4\times b_1})\seqq (\pid_{\bbo+b_4}+(\hT{c_2}\times \pid_{b_1}))\seqq(\pid_{\bbo}+\swapp_{b_4+b_3}+\pid_{b_4\times b_1})\seqq \\
            & \quad (\pid_{\bbo+b_3}+((\pid_{b_4}+\swapt_{b_4\times b_1})))\seqq(\pid_{\bbo+b_3}+\hT{c_1}\times \pid_{b_4})\seqq (\pid_{\bbo}+\swapp_{b_3+b_4}+\swapt_{b_3\times b_4}) \tag{by completeness of $\Pi$ and bifunctoriality of $+$} \\
            \Rightarrow{}& (\pid_{\bbo+b_4}+\pid_{b_1}+\swapt_{b_2\times b_1})\seqq (\pid_{\bbo+b_3}+(\hT{c_1}\times \pid_{b_2}))\seqq(\pid_{\bbo}+\swapp_{b_3+b_4}+\pid_{b_3\times b_2})\seqq \\
            & \quad (\pid_{\bbo+b_4}+((\pid_{b_3}+\swapt_{b_3\times b_2})))\seqq \blue{(\pid_{\bbo+b_4}+\hT{c_2}\times\pid_{b_3})}\seqq \\
            &\fromto_2{} (\pid_{\bbo}+\swapp_{b_3+b_4}+\pid_{b_4\times b_1})\seqq (\pid_{\bbo+b_4}+(\hT{c_2}\times \pid_{b_1}))\seqq(\pid_{\bbo}+\swapp_{b_4+b_3}+\pid_{b_4\times b_1})\seqq \\
            & \quad (\pid_{\bbo+b_3}+((\pid_{b_4}+\swapt_{b_4\times b_1})))\seqq(\pid_{\bbo+b_3}+\hT{c_1}\times \pid_{b_4})\seqq (\pid_{\bbo}+\swapp_{b_3+b_4}+\swapt_{b_3\times b_4}) \tag{by naturality of $\swapt$} \\
            \Rightarrow{}& \blue{(\pid_{\bbo}+\swapp_{b_4+b_3}+\pid_{b_4\times b_1})}\seqq(\pid_{\bbo+b_4}+\pid_{b_1}+\swapt_{b_2\times b_1})\seqq (\pid_{\bbo+b_3}+(\hT{c_1}\times \pid_{b_2}))\seqq\\
            & \quad (\pid_{\bbo}+\swapp_{b_3+b_4}+\pid_{b_3\times b_2})\seqq (\pid_{\bbo+b_4}+((\pid_{b_3}+\swapt_{b_3\times b_2})))\seqq (\pid_{\bbo+b_4}+\hT{c_2}\times\pid_{b_3})\seqq \\
            &\fromto_2{} (\pid_{\bbo+b_4}+(\hT{c_2}\times \pid_{b_1}))\seqq(\pid_{\bbo}+\swapp_{b_4+b_3}+\pid_{b_4\times b_1})\seqq \\
            & \quad (\pid_{\bbo+b_3}+((\pid_{b_4}+\swapt_{b_4\times b_1})))\seqq(\pid_{\bbo+b_3}+\hT{c_1}\times \pid_{b_4})\seqq (\pid_{\bbo}+\swapp_{b_3+b_4}+\swapt_{b_3\times b_4}) \tag{by reversibility} \\
            \Rightarrow{}& \pid_{\bbo}+ \big((\swapp_{b_2+b_1}+\swapt_{b_2\times b_1})\seqq(\pid_{b_1}+\hT{c_1}\times \pid_{b_2})\seqq\\
            & \quad (\swapp_{b_3+b_2}+\swapt_{b_3\times b_2})\seqq(\pid_{b_2}+\hT{c_2}\times \pid_{b_3})\big) \\
            & \fromto{}_2{} \pid_{\bbo}+\big((\pid_{b_2}+\hT{c_2}\times \pid_{b_1})\seqq (\swapp_{b_2+b_1}+\swapt_{b_2\times b_1})\seqq \\
            & \quad (\pid_{b_1}+\hT{c_1}\times \pid_{b_4})\seqq(\swapp_{b_3+b_4}+\swapt_{b_3\times b_4})\big) \tag{by bifunctoriality of $+$ and strictness}
        \end{align*}
        Therefore, by \cref{eq:h3}, we have 
        \begin{equation*}
            \adjustbox{max width=\linewidth}{$\begin{aligned}
                &(\swapp_{b_2+b_1}+\swapt_{b_2\times b_1})\seqq(\pid_{b_1}+\hT{c_1}\times \pid_{b_2})\seqq (\swapp_{b_3+b_2}+\swapt_{b_3\times b_2})\seqq(\pid_{b_2}+\hT{c_2}\times \pid_{b_3}) \\
                \fromto_2{}& (\pid_{b_2}+\hT{c_2}\times \pid_{b_1})\seqq (\swapp_{b_2+b_1}+\swapt_{b_2\times b_1})\seqq (\pid_{b_1}+\hT{c_1}\times \pid_{b_4})\seqq(\swapp_{b_3+b_4}+\swapt_{b_3\times b_4}).
            \end{aligned}$}
        \end{equation*}
    \else%
    Set \verb|\SHOWPROOFtrue| to compile the proof.
    \fi
\end{proof}

\begin{lemma}\label{lem:ht_times5}
    Let $c_1:b_1\fromto b_3, c_2:b_2\fromto b_4$ be two terms of \QPiLang{}. When omitting the isomorphisms $\assocrp$, $\assoclp$, $\unitep$, $\unitip$, $\assocrt$, $\assoclt$, $\unitet$, $\unitit$, $\dist$, $\factor$, $\absorb$, and $\factorz$ used for type conversion, we have
    \begin{equation*}
        \begin{aligned}
            &(\hqT{c_1}\times \pid_{b_2})\seqq(\swapt_{b_3\times b_2}+\swapt_{b_3\times b_2})\seqq(\hqT{c_2}\times \pid_{b_3})\seqq(\swapt_{b_4\times b_3}+\swapt_{b_4\times b_3}) \\
            \fromto_2{}& (\swapt_{b_1\times b_2}+\swapt_{b_1\times b_2})\seqq(\hqT{c_2}\times \pid_{b_1})\seqq(\swapt_{b_4\times b_1}+\swapt_{b_4\times b_1})\seqq(\hqT{c_1}\times \pid_{b_4}).
        \end{aligned}
    \end{equation*}
\end{lemma}
\begin{proof}
    Consider \begin{align*}
        &\pid_{b_2\times b_1}+ ((\swapp_{b_2+b_1}+\swapt_{b_2\times b_1})\seqq(\pid_{b_1}+\hT{c_1}\times \pid_{b_2})\seqq \\
        & \quad (\swapp_{b_3+b_2}+\swapt_{b_3\times b_2})\seqq (\pid_{b_2}+\hT{c_2}\times \pid_{b_3})) \\
        \fromto_2{}& (\swapt_{b_2\times b_1} + (\swapp_{b_2+b_1}+\swapt_{b_2\times b_1}))\seqq(\pid_{b_1\times b_2}+(\pid_{b_1}+\hT{c_1}\times \pid_{b_2}))\seqq \\
        & \quad (\swapt_{b_1\times b_2} + ((\swapp_{b_3+b_2}+\swapt_{b_3\times b_2})\seqq(\pid_{b_2}+\hT{c_2}\times \pid_{b_3}))) \tag{by bifunctoriality of $+$ and naturality of $\swapt$} \\
        \fromto_2{}& (\swapt_{b_2\times b_1} + (\swapp_{b_2+b_1}+\swapt_{b_2\times b_1}))\seqq (\swapp_{(b_1\times b_2)+b_1}+\pid_{(\bbo+b_1)\times b_2})\seqq \\
        & \quad (\pid_{b_1}+\pid_{b_1\times b_2}+\hT{c_1}\times \pid_{b_2})\seqq(\swapp_{b_1+(b_1\times b_2)}+\pid_{(\bbo+b_3)\times b_2})\seqq \\
        & \quad (\swapt_{b_1\times b_2} + ((\swapp_{b_3+b_2}+\swapt_{b_3\times b_2})\seqq(\pid_{b_2}+\hT{c_2}\times \pid_{b_3}))) \tag{by naturality of $\swapp$} \\
        \fromto_2{}& (\swapt_{b_2\times b_1} + (\swapp_{b_2+b_1}+\swapt_{b_2\times b_1}))\seqq(\swapp_{(b_1\times b_2)+b_1}+\pid_{(\bbo+b_1)\times b_2})\seqq \\
        & \quad (\pid_{b_1}+(\pid_{b_1}+\hT{c_1})\times \pid_{b_2})\seqq(\swapp_{b_1+(b_1\times b_2)}+\pid_{(\bbo+b_3)\times b_2})\seqq \\
        & \quad (\swapt_{b_1\times b_2} + ((\swapp_{b_3+b_2}+\swapt_{b_3\times b_2})\seqq(\pid_{b_2}+\hT{c_2}\times \pid_{b_3}))) \tag{by distributivity} \\
        \fromto_2{}& (\swapt_{b_2\times b_1} + (\swapp_{b_2+b_1}+\swapt_{b_2\times b_1}))\seqq(\swapp_{(b_1\times b_2)+b_1}+\pid_{(\bbo+b_1)\times b_2})\seqq \\
        & \quad ((\pid_{b_1}+((\swapp_{b_1+\bbo}+\pid_{b_1})\seqq(\pid_{\bbo}+\hqT{c_1})\seqq(\swapp_{\bbo+b_3}+\pid_{b_3})))\times \pid_{b_2})\seqq \\
        & \quad (\swapp_{b_1+(b_1\times b_2)}+\pid_{(\bbo+b_3)\times b_2})\seqq \\
        & \quad (\swapt_{b_1\times b_2} + ((\swapp_{b_3+b_2}+\swapt_{b_3\times b_2})\seqq(\pid_{b_2}+\hT{c_2}\times \pid_{b_3}))) \tag{by \cref{lem:hqt2}} \\
        \fromto_2{}& (\swapt_{b_2\times b_1} + (\swapp_{b_2+b_1}+\swapt_{b_2\times b_1}))\seqq(\swapp_{(b_1\times b_2)+b_1}+\pid_{(\bbo+b_1)\times b_2})\seqq \\
        & \quad (\pid_{b_1}+(\swapp_{b_1+\bbo}+\pid_{b_1})\times \pid_{b_2})\seqq(\pid_{b_1}+(\pid_{\bbo}+\hqT{c_1})\times \pid_{b_2})\seqq\\
        & \quad (\pid_{b_1}+(\swapp_{\bbo+b_3}+\pid_{b_3})\times \pid_{b_2})\seqq (\swapp_{b_1+(b_1\times b_2)}+\pid_{(\bbo+b_3)\times b_2})\seqq \\
        & \quad (\swapt_{b_1\times b_2} + ((\swapp_{b_3+b_2}+\swapt_{b_3\times b_2})\seqq(\pid_{b_2}+\hT{c_2}\times \pid_{b_3}))) \tag{by bifunctoriality of $+$ and $\times$} \\
        \fromto_2{}& (\swapt_{b_2\times b_1} + (\swapp_{b_2+b_1}+\swapt_{b_2\times b_1}))\seqq(\swapp_{(b_1\times b_2)+b_1}+\pid_{(\bbo+b_1)\times b_2})\seqq \\
        & \quad (\pid_{b_1}+(\swapp_{b_1+\bbo}+\pid_{b_1})\times \pid_{b_2})\seqq(\pid_{b_1}+(\pid_{\bbo}+\hqT{c_1})\times \pid_{b_2})\seqq \\
        & \quad (\pid_{b_1}+(\swapp_{\bbo+b_3}+\pid_{b_3})\times \pid_{b_2})\seqq \\
        & \quad (\swapp_{b_1+(b_1\times b_2)}+\pid_{(\bbo+b_3)\times b_2})\seqq (\swapt_{b_1\times b_2} + (\swapp_{b_3+b_2}+\swapt_{b_3\times b_2}))\seqq\\
        & \quad (\pid_{b_2\times b_3}+\pid_{b_2}+\hT{c_2}\times \pid_{b_3}) \tag{by bifunctoriality of $+$} \\
        \fromto_2{}& (\swapt_{b_2\times b_1} + (\swapp_{b_2+b_1}+\swapt_{b_2\times b_1}))\seqq(\swapp_{(b_1\times b_2)+b_1}+\pid_{(\bbo+b_1)\times b_2})\seqq \\
        & \quad (\pid_{b_1}+(\swapp_{b_1+\bbo}+\pid_{b_1})\times \pid_{b_2})\seqq(\pid_{b_1}+(\pid_{\bbo}+\hqT{c_1})\times \pid_{b_2})\seqq\\
        & \quad (\pid_{b_1}+(\swapp_{\bbo+b_3}+\pid_{b_3})\times \pid_{b_2})\seqq \\
        & \quad (\swapp_{b_1+(b_1\times b_2)}+\pid_{(\bbo+b_3)\times b_2})\seqq (\swapt_{b_1\times b_2} + (\swapp_{b_3+b_2}+\swapt_{b_3\times b_2}))\seqq \\
        & \quad (\swapp_{(b_2\times b_3)+b_2}+\pid_{b_2\times b_3})\seqq(\pid_{b_2}+(\pid_{b_2}+\hT{c_2})\times \pid_{b_3})\seqq (\swapp_{(b_3\times b_2)+b_2}+\pid_{b_4\times b_3}) \tag{by naturality of $\swapt$ and distributivity} \\
        \fromto_2{}& (\swapt_{b_2\times b_1} + (\swapp_{b_2+b_1}+\swapt_{b_2\times b_1}))\seqq(\swapp_{(b_1\times b_2)+b_1}+\pid_{(\bbo+b_1)\times b_2})\seqq \\
        & \quad (\pid_{b_1}+((\swapp_{b_1+\bbo}+\pid_{b_1}))\times \pid_{b_2})\seqq(\pid_{b_1}+(\pid_{\bbo}+\hqT{c_1})\times \pid_{b_2})\seqq \\
        & \quad (\pid_{b_1}+(\swapp_{\bbo+b_3}+\pid_{b_3})\times \pid_{b_2})\seqq \\
        & \quad (\swapp_{b_1+(b_1\times b_2)}+\pid_{(\bbo+b_3)\times b_2})\seqq (\swapt_{b_1\times b_2} + (\swapp_{b_3+b_2}+\swapt_{b_3\times b_2}))\seqq \\
        & \quad (\swapp_{(b_2\times b_3)+b_2}+\pid_{b_2\times b_3})\seqq \\
        & \quad (\pid_{b_2}+((\swapp_{b_2+\bbo}+\pid_{b_2})\seqq(\pid_{\bbo}+\hqT{c_2})\seqq (\swapp_{\bbo+b_4}+\pid_{b_4}))\times \pid_{b_3})\seqq \\
        & \quad (\swapp_{(b_3\times b_2)+b_2}+\pid_{b_4\times b_3}) \tag{by \cref{lem:hqt2}} \\
        \fromto_{}& (\swapp_{(b_2\times b_1)+(b_2+b_1)}+\pid_{b_1\times b_2})\seqq(\swapp_{b_2+b_1}+\swapt_{b_2\times b_1} +\swapt_{b_2\times b_1})\seqq \\
        & \quad (\pid_{b_1}+\pid_{b_2}+\hqT{c_1}\times \pid_{b_2})\seqq  (\swapp_{b_1+b_2}+ \swapt_{b_1\times b_2}+ \swapt_{b_1\times b_2})\seqq \\
        & \quad (\pid_{b_2}+\pid_{b_1}+\hqT{c_2}\times \pid_{b_3})\seqq (\swapp_{(b_2+b_1)+(b_2\times b_1)}+\pid_{b_4\times b_3}).\tag{by completeness of $\Pi$ and strictness}
    \end{align*}
    Similarly, we also have \begin{align*}
        &\pid_{b_2\times b_1}+((\pid_{b_2}+\hT{c_2}\times \pid_{b_1})\seqq (\swapp_{b_2+b_1}+\swapt_{b_2\times b_1})\seqq \\
        & \quad (\pid_{b_1}+\hT{c_1}\times \pid_{b_4})\seqq(\swapp_{b_3+b_4}+\swapt_{b_3\times b_4})) \\
        \fromto_2{}& (\swapp_{(b_2\times b_1)+(b_2+b_1)}+\pid_{b_2\times b_1})\seqq(\pid_{b_2}+\pid_{b_1}+\hqT{c_2}\times \pid_{b_1})\seqq \\
        & \quad (\swapp_{b_2+b_1}+\swapt_{b_2\times b_1}+\swapt_{b_2\times b_1})\seqq (\pid_{b_1}+\pid_{b_2}+\hqT{c_1}\times \pid_{b_4})\seqq \\
        & \quad (\swapp_{b_1+b_2}+\swapt_{b_3\times b_4}+\swapt_{b_3\times b_4}) \seqq(\swapp_{(b_2+b_1)+(b_4\times b_3)}+\pid_{b_4\times b_3}).
    \end{align*}
    By \cref{lem:ht_times4}, we have \begin{align*}
        & (\swapp_{(b_2\times b_1)+(b_2+b_1)}+\pid_{b_1\times b_2})\seqq(\swapp_{b_2+b_1}+\swapt_{b_2\times b_1} +\swapt_{b_2\times b_1})\seqq \\
        & \quad (\pid_{b_1}+\pid_{b_2}+\hqT{c_1}\times \pid_{b_2})\seqq (\swapp_{b_1+b_2}+ \swapt_{b_1\times b_2}+ \swapt_{b_1\times b_2})\seqq \\
        & \quad (\pid_{b_2}+\pid_{b_1}+\hqT{c_2}\times \pid_{b_3})\seqq (\swapp_{(b_2+b_1)+(b_2\times b_1)}+\pid_{b_4\times b_3}) \\
        &\fromto_2{} (\swapp_{(b_2\times b_1)+(b_2+b_1)}+\pid_{b_2\times b_1})\seqq(\pid_{b_2}+\pid_{b_1}+\hqT{c_2}\times \pid_{b_1})\seqq \\
        & \quad (\swapp_{b_2+b_1}+\swapt_{b_2\times b_1}+\swapt_{b_2\times b_1})\seqq (\pid_{b_1}+\pid_{b_2}+\hqT{c_1}\times \pid_{b_4})\seqq \\
        & \quad (\swapp_{b_1+b_2}+\swapt_{b_3\times b_4}+\swapt_{b_3\times b_4}) \seqq(\swapp_{(b_2+b_1)+(b_4\times b_3)}+\pid_{b_4\times b_3}) \\
        \Rightarrow{}& (\swapp_{b_2+b_1}+\swapt_{b_2\times b_1} +\swapt_{b_2\times b_1})\seqq (\pid_{b_1}+\pid_{b_2}+\hqT{c_1}\times \pid_{b_2})\seqq \\
        &\quad (\swapp_{b_1+b_2}+ \swapt_{b_1\times b_2}+ \swapt_{b_1\times b_2})\seqq (\pid_{b_2}+\pid_{b_1}+\hqT{c_2}\times \pid_{b_3}) \\
        &\fromto_2{} (\pid_{b_2}+\pid_{b_1}+\hqT{c_2}\times \pid_{b_1})\seqq(\swapp_{b_2+b_1}+\swapt_{b_2\times b_1}+\swapt_{b_2\times b_1})\seqq \\
        & \quad (\pid_{b_1}+\pid_{b_2}+\hqT{c_1}\times \pid_{b_4})\seqq (\swapp_{b_1+b_2}+\swapt_{b_3\times b_4}+\swapt_{b_3\times b_4}) \tag{by reversibility and strictness} \\
        \Rightarrow{}& \pid_{b_1}+\pid_{b_2}+ \\
        & \quad ((\swapt_{b_2\times b_1} +\swapt_{b_2\times b_1})\seqq(\hqT{c_1}\times \pid_{b_2})\seqq(\swapt_{b_1\times b_2}+ \swapt_{b_1\times b_2})\seqq(\hqT{c_2}\times \pid_{b_3})) \\
        &\fromto_2{} \pid_{b_2}+\pid_{b_1}+\\
        & \quad ((\hqT{c_2}\times \pid_{b_1})\seqq(\swapt_{b_2\times b_1}+\swapt_{b_2\times b_1})\seqq(\hqT{c_1}\times \pid_{b_4})\seqq(\swapt_{b_3\times b_4}+\swapt_{b_3\times b_4})) \tag{by bifunctoriality of $+$ and naturality of $\swapp$} \\
        \Rightarrow{}& (\swapt_{b_2\times b_1} +\swapt_{b_2\times b_1})\seqq(\hqT{c_1}\times \pid_{b_2})\seqq(\swapt_{b_1\times b_2}+ \swapt_{b_1\times b_2})\seqq(\hqT{c_2}\times \pid_{b_3}) \\
        &\fromto_2{} (\hqT{c_2}\times \pid_{b_1})\seqq(\swapt_{b_2\times b_1}+\swapt_{b_2\times b_1})\seqq(\hqT{c_1}\times \pid_{b_4})\seqq(\swapt_{b_3\times b_4}+\swapt_{b_3\times b_4}) \tag{by \cref{eq:h3}} \\
        \Rightarrow{}& (\hqT{c_1}\times \pid_{b_2})\seqq(\swapt_{b_3\times b_2}+\swapt_{b_3\times b_2})\seqq(\hqT{c_2}\times \pid_{b_3})\seqq(\swapt_{b_4\times b_3}+\swapt_{b_4\times b_3}) \\
        &\fromto_2{} (\swapt_{b_1\times b_2}+\swapt_{b_1\times b_2})\seqq(\hqT{c_2}\times \pid_{b_1})\seqq(\swapt_{b_4\times b_1}+\swapt_{b_4\times b_1})\seqq(\hqT{c_1}\times \pid_{b_4}). \tag{by naturality of $\swapt$ and bifunctoriality of $+$}
    \end{align*}
\end{proof}

\lemht*
\begin{proof}
    For convenience, we omit the isomorphisms $\assocrp$, $\assoclp$, $\unitep$, $\unitip$, $\assocrt$, $\assoclt$, $\unitet$, $\unitit$, $\dist$, $\factor$, $\absorb$, and $\factorz$ used for type conversion.
    \begin{itemize}
        \item For \cref{eq:ht_plus}, \begin{align*}
            &\pid_{\hdim(c_1+c_2)\bbo}+(\hT{c_1+\pid_{b_2}}\seqq\hT{\pid_{b_3}+c_2}) \\
            ={}& \pid_{\hdim(c_1+c_2)\bbo}+\hT{(c_1+\pid_{b_2})\seqq(\pid_{b_3}+c_2)} \tag{by definition of $\cT_H$} \\
            \fromto_2{}& (\swapp_{\hdim(c_1+c_2)\bbo +\bbo}+\pid_{b_1+b_2})\seqq(\pid_{\bbo}+\hqT{(c_1+\pid_{b_2})\seqq(\pid_{b_3}+c_2)})\seqq \\
            & \quad (\swapp_{\hdim(c_1+c_2)\bbo +\bbo}+\pid_{b_3+b_4}) \tag{by \cref{lem:hqt2} and strictness} \\
            ={}& (\swapp_{\hdim(c_1+c_2)\bbo +\bbo}+\pid_{b_1+b_2})\seqq(\pid_{\bbo}+\hqT{c_1+\pid_{b_2}}\seqq\hqT{\pid_{b_3}+c_2})\seqq \\
            & \quad (\swapp_{\hdim(c_1+c_2)\bbo +\bbo}+\pid_{b_3+b_4}) \tag{by definition of $\cT_{HQ}$} \\
            \fromto_2{}& (\swapp_{\hdim(c_1+c_2)\bbo +\bbo}+\pid_{b_1+b_2})\seqq \\
            & \quad (\pid_{\bbo}+((\pid_{b_1}+\swapp_{b_2+\hdim(c_1)\bbo}+\pid_{b_2})\seqq(\hqT{c_1}+\pid_{\bb{2}}\times \pid_{b_2})\seqq \\
            & \qquad (\pid_{b_1}+\swapp_{b_3+\hdim(b_2)\bbo}+\pid_{b_2})\seqq (\pid_{b_3}+\swapp_{b_2+\hdim(b_3)\bbo}+\pid_{b_2})\seqq \\
            & \qquad (\pid_{\bb{2}}\times \pid_{b_3}+\hqT{c_2})\seqq(\pid_{b_3}+\swapp_{b_3+\hdim(c_2)\bbo}+\pid_{b_4})))\seqq \\
            & \quad (\swapp_{\hdim(c_1+c_2)\bbo +\bbo}+\pid_{b_3+b_4}) \tag{by definition of $\cT_{HQ}$} \\
            \fromto_2{}& (\swapp_{\hdim(c_1+c_2)\bbo +\bbo}+\pid_{b_1+b_2})\seqq \\
            & \quad (\pid_{\bbo}+((\pid_{b_1}+\swapp_{b_2+\hdim(c_1)\bbo}+\pid_{b_2})\seqq(\hqT{c_1}+\hqT{c_2})\seqq \\
            & \qquad (\pid_{b_3}+\swapp_{b_3+\hdim(c_2)\bbo}+\pid_{b_4})))\seqq \\
            & \quad (\swapp_{\hdim(c_1+c_2)\bbo +\bbo}+\pid_{b_3+b_4}) \tag{by naturality of $\swapp$ and bifunctoriality of $+$} \\
            \fromto_2{}& (\swapp_{\hdim(c_1+c_2)\bbo +\bbo}+\pid_{b_1+b_2})\seqq  (\pid_{\bbo}+\hqT{c_1+c_2})\seqq  (\swapp_{\hdim(c_1+c_2)\bbo +\bbo}+\pid_{b_3+b_4}) \tag{by definition of $\cT_{HQ}$} \\
            \fromto_2{}& \pid_{\hdim(c_1+c_2)\bbo}+\hT{c_1+c_2}. \tag{by \cref{lem:hqt2}}
        \end{align*}
        Then, by \cref{eq:h3}, we have $\hT{c_1+\pid_{b_2}}\seqq\hT{\pid_{b_3}+c_2} \fromto_2{} \hT{c_1+c_2}$. Similarly, we obtain $\hT{\pid_{b_1}+c_2}\seqq\hT{c_1+\pid_{b_4}} \fromto_2{} \hT{c_1+c_2}$.
        \item For \cref{eq:ht_times}, \begin{align*}
            &\pid_{\hdim(c_1\times c_2)\bbo}+(\hT{c_1\times \pid_{b_2}}\seqq\hT{\pid_{b_3}\times c_2}) \\
            ={}& \pid_{\hdim(c_1\times c_2)\bbo}+\hT{(c_1\times \pid_{b_2})\seqq(\pid_{b_3}\times c_2)} \tag{by definition of $\cT_H$} \\
            \fromto_2{}& (\swapp_{\hdim(c_1\times c_2)\bbo +\bbo}+\pid_{b_1\times b_2})\seqq(\pid_{\bbo}+\hqT{(c_1\times \pid_{b_2})\seqq(\pid_{b_3}\times c_2)})\seqq \\
            & \quad (\swapp_{\hdim(c_1\times c_2)\bbo +\bbo}+\pid_{b_3\times b_4}) \tag{by \cref{lem:hqt2} and strictness} \\
            ={}& (\swapp_{\hdim(c_1\times c_2)\bbo +\bbo}+\pid_{b_1\times b_2})\seqq(\pid_{\bbo}+\hqT{c_1\times \pid_{b_2}}\seqq\hqT{\pid_{b_3}\times c_2})\seqq \\
            & \quad (\swapp_{\hdim(c_1\times c_2)\bbo +\bbo}+\pid_{b_3\times b_4}) \tag{by definition of $\cT_{HQ}$} \\
            \fromto_2{}& (\swapp_{\hdim(c_1\times c_2)\bbo +\bbo}+\pid_{b_1\times b_2})\seqq \\
            & \quad (\pid_{\bbo}+((\hqT{c_1}\times \pid_{b_2})\seqq(\swapt_{b_3\times b_2}+\swapt_{b_3\times b_2})\seqq(\pid_\bb{2}\times \pid_{b_2}\times \pid_{b_3})\seqq \\
            & \qquad (\swapt_{b_2\times b_3}+\swapt_{b_2\times b_3})\seqq (\pid_{\bb{2}}\times \pid_{b_3}\times \pid_{b_2})\seqq(\swapt_{b_3\times b_2}+\swapt_{b_3\times b_2})\seqq \\
            & \qquad (\hqT{c_2}\times \pid_{b_3})\seqq(\swapt_{b_4\times b_3}+\swapt_{b_4\times b_3})))\seqq \\
            & \quad (\swapp_{\hdim(c_1\times c_2)\bbo +\bbo}+\pid_{b_3\times b_4}) \tag{by definition of $\cT_{HQ}$ and strictness} \\
            \fromto_2{}& (\swapp_{\hdim(c_1\times c_2)\bbo +\bbo}+\pid_{b_1\times b_2})\seqq \\
            & \quad (\pid_{\bbo}+((\hqT{c_1}\times \pid_{b_2})\seqq(\swapt_{b_3\times b_2}+\swapt_{b_3\times b_2})\seqq \\
            & \qquad (\hqT{c_2}\times \pid_{b_3})\seqq(\swapt_{b_4\times b_3}+\swapt_{b_4\times b_3})))\seqq \\
            & \quad (\swapp_{\hdim(c_1\times c_2)\bbo +\bbo}+\pid_{b_3\times b_4}) \tag{by naturality of $\swapt$ and bifunctoriality of $+$} \\
            \fromto_2{}& (\swapp_{\hdim(c_1\times c_2)\bbo +\bbo}+\pid_{b_1\times b_2})\seqq (\pid_{\bbo}+\hqT{c_1\times c_2})\seqq (\swapp_{\hdim(c_1\times c_2)\bbo +\bbo}+\pid_{b_3\times b_4}) \tag{by definition of $\cT_{HQ}$ and strictness} \\
            \fromto_2{}& \pid_{\hdim(c_1\times c_2)\bbo}+\hT{c_1\times c_2}. \tag{by \cref{lem:hqt2}}
        \end{align*}
        Then, by \cref{eq:h3}, we have $\hT{c_1\times \pid_{b_2}}\seqq\hT{\pid_{b_3}\times c_2} \fromto_2{} \hT{c_1\times c_2}$. Similarly, we have \begin{align*}
            & \pid_{\hdim(c_1\times c_2)\bbo}+(\hT{\pid_{b_1}\times c_2}\seqq\hT{c_1\times \pid_{b_4}}) \\
            \fromto_2{}& (\swapp_{\hdim(c_1\times c_2)\bbo +\bbo}+\pid_{b_1\times b_2})\seqq \\
            & \quad (\pid_{\bbo}+((\swapt_{b_1\times b_2}+\swapt_{b_1\times b_2})\seqq(\hqT{c_2}\times \pid_{b_1})\seqq \\
            & \qquad (\swapt_{b_4\times b_1}+\swapt_{b_4\times b_1})\seqq(\hqT{c_1}\times \pid_{b_4})))\seqq \\
            & \quad (\swapp_{\hdim(c_1\times c_2)\bbo +\bbo}+\pid_{b_3\times b_4}) \\
            \fromto{}& (\swapp_{\hdim(c_1\times c_2)\bbo +\bbo}+\pid_{b_1\times b_2})\seqq \\
            & \quad (\pid_{\bbo}+((\hqT{c_1}\times \pid_{b_2})\seqq(\swapt_{b_3\times b_2}+\swapt_{b_3\times b_2})\seqq \\
            & \qquad (\hqT{c_2}\times \pid_{b_3})\seqq(\swapt_{b_4\times b_3}+\swapt_{b_4\times b_3})))\seqq \\
            & \quad (\swapp_{\hdim(c_1\times c_2)\bbo +\bbo}+\pid_{b_3\times b_4}) \tag{by \cref{lem:ht_times5}} \\
            \fromto_2{}& (\swapp_{\hdim(c_1\times c_2)\bbo +\bbo}+\pid_{b_1\times b_2})\seqq (\pid_{\bbo}+\hqT{c_1\times c_2})\seqq (\swapp_{\hdim(c_1\times c_2)\bbo +\bbo}+\pid_{b_3\times b_4}) \tag{by definition of $\cT_{HQ}$ and strictness} \\
            \fromto_2{}& \pid_{\hdim(c_1\times c_2)\bbo}+\hT{c_1\times c_2}. \tag{by \cref{lem:hqt2}}
        \end{align*}
        Thus, by \cref{eq:h3}, we also have $\hT{\pid_{b_1}\times c_2}\seqq\hT{c_1\times \pid_{b_4}} \fromto_2{} \hT{c_1\times c_2}$.
        \qedhere
    \end{itemize}
\end{proof}

\lemqeq*
\begin{proof}
    We prove it by induction on the level-2 combinators of \QPiLang{} deriving $c_1\fromto_2{} c_2$. For convenience, we suppress the isomorphisms $\assocrp$, $\assoclp$, $\unitep$, $\unitip$, $\assocrt$, $\assoclt$, $\unitet$, $\unitit$, $\dist$, $\factor$, $\absorb$, and $\factorz$ used for type conversion.
    \begin{itemize}
        \item Level-2 combinators of $\Pi$. \begin{itemize}[leftmargin=10pt]
            \item Reflexivity, symmetry and transitivity. Since $\fromto_2{}$ in \HPiLang{} is also an equivalence relation, this case is easily proved by the induction hypothesis.
            \item Unitality. Suppose $c:b_1\fromto b_2$, we have \begin{align*}
                \hT{\pid_{b_1}\seqq c} ={}& \hT{\pid_{b_1}}\seqq\hT{c} \\
                ={}& (\pid_{\bbo}+\pid_{b_1})\seqq \hT{c} \\
                \fromto_2{}& \pid_{\bbo+b_1} \seqq \hT{c} \tag{by bifunctoriality}\\
                \fromto_2{}& \hT{c}, \tag{by unitality} \\ \\
                \hT{\unitip\seqq(\pid_{\bbz}+c)\seqq\unitep} ={}& \hT{\unitip}\seqq\hT{\pid_{\bbz}+c}\seqq\hT{\unitep} \\
                \fromto_2{}& (\pid_{\bbo}+\pid_{b_1})\seqq((\pid_{\bbo}+\pid_{\bbz})+\pid_{b_1})\seqq \\
                & \quad (\swapp_{\bbo+\bbz}+\pid_{b_1})\seqq(\pid_{\bbz}+\hT{c})\seqq(\swapp_{\bbz+\bbo}+\pid_{b_2})\seqq \\
                & \quad (\pid_{\bbo}+\pid_{b_2}) \\
                \fromto_2{}& \pid_{\bbo+b_1}\seqq(\pid_{\bbz}+\hT{c})\seqq \pid_{\bbo+b_2} \tag{by completeness of $\Pi$} \\
                \fromto_2{}& \hT{c}, \tag{by unitality} \\ \\
                \hT{\unitit\seqq(\pid_{\bbo}\times c)\seqq\unitet} ={}& \hT{\unitit}\seqq\hT{\pid_{\bbo}\times c}\seqq\hT{\unitet} \\
                \fromto_2{}& (\pid_{\bbo}+\pid_{b_1})\seqq (\pid_{\bbo}+\pid_{b_1})\seqq \hT{c} \seqq \\
                & \quad (\pid_{\bbo}+\pid_{b_1})\seqq (\pid_{\bbo}+\pid_{b_1})\\
                \fromto_2{}& \hT{c}. \tag{by unitality and bifunctoriality}
            \end{align*}
            The remaining three are also proven in the same way.
            \item Associativity. \begin{align*}
                \hT{(c_1\seqq c_2)\seqq c_3} ={}& \hT{c_1\seqq c_2}\seqq \hT{c_3} \\
                ={}& (\hT{c_1}\seqq\hT{c_2})\seqq\hT{c_3} \\
                \fromto_2{}& \hT{c_1}\seqq(\hT{c_2}\seqq\hT{c_3}) \tag{by associativity} \\
                ={}& \hT{c_1}\seqq\hT{c_2\seqq c_3} \\
                ={}& \hT{c_1\seqq(c_2\seqq c_3)}, \\ \\
                \hT{c_1+(c_2+c_3)} \fromto_2{}& \hT{(c_1+c_2)+c_3}. \tag{by \cref{lem:ht_assoc_p}}
            \end{align*}
            For the associativity of $\times$, we can prove $\hqT{c_1\times (c_2\times c_3)} \fromto_2 \hqT{(c_1\times c_2)\times c_3}$ using a similar graphical proof as in \cref{lem:ht_assoc_p}, but replacing $+$ between wires with $\times$ between wires. Then applying \cref{lem:hqt2} and \cref{eq:h3}, we have \[ \hT{c_1\times (c_2\times c_3)} \fromto_2{} \hT{(c_1\times c_2)\times c_3}.\]
            \begin{align*}
                & \hT{((\assocrp+\pid)\seqq\assocrp)\seqq(\pid+\assocrp)} \\
                \fromto_2{}& \hT{\assocrp+\pid}\seqq\hT{\assocrp}\seqq\hT{\pid+\assocrp} \\
                \fromto_2{}& (\pid_{\bbo}+\assocrp+\pid)\seqq(\swapp+\pid)\seqq (\pid+\pid_{\bbo}+\pid)\seqq(\swapp+\pid)\seqq \\
                & \quad \hT{\assocrp}\seqq\hT{\pid+\assocrp} \\
                \fromto_2{}&(\pid_{\bbo}+(\assocrp+\pid))\seqq\hT{\assocrp}\seqq\\
                & \quad \hT{\pid+\assocrp} \tag{by bifunctoriality and naturality} \\
                \fromto_2{}& \cdots \\
                \fromto_2{}& \pid_{\bbo}+ (((\assocrp+\pid)\seqq\assocrp)\seqq(\pid+\assocrp)) \\
                \fromto_2{}& \pid_{\bbo}+(\assocrp\seqq\assocrp) \tag{by associativity}\\
                \fromto_2{}& (\pid_{\bbo}+\assocrp)\seqq(\pid_{\bbo}+\assocrp) \tag{by bifunctoriality}\\
                ={}& \hT{\assocrp\seqq\assocrp}.
            \end{align*}
            The remaining three are also proven in the same way.
            \item Annihilativity. Suppose $c:b_1\fromto b_2$, we have \begin{align*}
                & \hT{\factorz\seqq(c\times \pid_{\bbz})\seqq\absorb} \\
                ={}& \hT{\factorz}\seqq \hT{c\times \pid_{\bbz}} \seqq\hT{\absorb} \\
                \fromto_2{}& (\pid_{\bbo}+\pid_{\bbz})\seqq (\pid_{\bbo}+\swapt_{b_1\times \bbz})\seqq\hT{\pid_{\bbz}\times c}\seqq (\pid_{\bbo}+\swapt_{\bbz\times b_2})\seqq \\
                & \quad \hT{\pid_{b_2}\times \pid_{\bbz}}\seqq (\pid_{\bbo}+\pid_{\bbz}) \\
                \fromto_2{}& \hT{\pid_{\bbz}\times c}\seqq\hT{\pid_{b_2}\times \pid_{\bbz}} \tag{by completeness of $\Pi$ and unitality} \\
                ={}& (\pid_{\bbo}+\swapt_{\bbz\times b_1}\seqq\absorb\seqq\pid_{\bbz}\seqq\factorz\seqq\swapt_{b_2\times \bbz})\seqq \hT{\pid_{b_2}\times \pid_{\bbz}} \\
                \fromto_2{}& \hT{\pid_{b_2}\times \pid_{\bbz}} \tag{by completeness of $\Pi$ and unitality} \\
                \fromto_2{}& \cT_{H}[\underbrace{\pid_\bbz+\cdots+\pid_\bbz}_{\hdim(b_2)}] \tag{by \cref{lem:hqt1}} \\
                \fromto_2{}& \pid_{\bbz}+ \cT_{H}[\underbrace{\pid_\bbz+\cdots+\pid_\bbz}_{\hdim(b_2)-1}] \tag{by the proof for unitality} \\
                \fromto_2{}& \cT_{H}[\underbrace{\pid_\bbz+\cdots+\pid_\bbz}_{\hdim(b_2)-1}] \tag{by unitality} \\
                \fromto_2{}& \cdots \\
                \fromto_2{}& \hT{\pid_\bbz}.
            \end{align*}
            The remaining one can be proven in the same way.
            \item Bifunctoriality. \begin{itemize}[leftmargin=10pt]
                \item For $ \begin{prooftree}
                    \hypo{c_1\fromto_2 c_3} \hypo{c_2\fromto_2 c_4}
                    \infer2{c_1+c_2 \fromto_2 c_3+c_4}
                \end{prooftree} $, suppose $c_1\fromto_2 c_3$ and $c_2\fromto_2 c_4$, then by inductive hypothesis, we have $\hT{c_1}\fromto \hT{c_3}, \hT{c_2}\fromto_2 \hT{c_4}$. Since $\hT{c_1+c_2}, \hT{c_3+c_4}$ are defined inductively on $\hT{c_1},\hT{c_2}$ and $\hT{c_3},\hT{c_4}$ respectively, we can easily see that $\hT{c_1+c_2}\fromto_2 \hT{c_3+c_4}$.
                \item For $\begin{prooftree}
                    \hypo{c_1\fromto_2 c_3} \hypo{c_2\fromto_2 c_4}
                    \infer2{c_1\times c_2 \fromto_2 c_3\times c_4}
                \end{prooftree}$, suppose $c_1\fromto_2 c_3$ and $c_2\fromto_2 c_4$, then by inductive hypothesis, we have $\hT{c_1}\fromto \hT{c_3}, \hT{c_2}\fromto_2 \hT{c_4}$. With \cref{lem:hqt2} and \cref{eq:h3}, we obtain $\hqT{c_1}\fromto_2{} \hqT{c_3}, \hqT{c_2}\fromto \hqT{c_4}$.
                Since $\hqT{c_1\times c_2}, \hqT{c_3\times c_4}$ are defined inductively on $\hqT{c_1},\hqT{c_2}$ and $\hqT{c_3},\hqT{c_4}$ respectively, we can easily see that $\hqT{c_1\times c_2}\fromto_2 \hqT{c_3\times c_4}$.
                Applying \cref{lem:hqt2} and \cref{eq:h3} again, we have $\hT{c_1\times c_2}\fromto_2 \hT{c_3\times c_4}$.
                \item For $\pid + \pid \fromto_2 \pid$ and $\pid \times  \pid \fromto_2 \pid$,  they are trivial.
                \item For $(c_1+c_2)\seqq(c_3+c_4) \fromto_2 (c_1\seqq c_3) + (c_2\seqq c_4)$, we have \begin{align*}
                    & \hT{(c_1+c_2)\seqq(c_3+c_4)} \\
                    ={}& \hT{c_1+c_2}\seqq\hT{c_3+c_4} \\
                    \fromto_2{}& \hT{c_1+\pid}\seqq\hT{\pid+c_2}\seqq\hT{c_3+\pid}\seqq\hT{\pid+c_4} \tag{by \cref{lem:ht3}} \\
                    \fromto_2{}& \hT{c_1+\pid}\seqq\hT{c_3+\pid}\seqq\hT{\pid+c_2}\seqq\hT{\pid+c_4} \tag{by \cref{lem:ht3}} \\
                    \fromto_2{}& (\hT{c_1}+\pid)\seqq(\swapp+\pid)\seqq(\pid+\pid_\bbo +\pid)\seqq(\swapp+\pid)\seqq \\
                    & \quad (\hT{c_3}+\pid)\seqq(\swapp+\pid)\seqq(\pid+\pid_\bbo +\pid)\seqq(\swapp+\pid)\seqq \\
                    & \quad \hT{\pid+c_2}\seqq\hT{\pid+c_4} \tag{by definition of $\cT_H$}\\
                    \fromto_2{}& (\hT{c_1}\seqq\hT{c_3}+\pid)\seqq(\swapp+\pid)\seqq(\pid+\pid_\bbo +\pid)\seqq(\swapp+\pid)\seqq \\
                    & \quad \hT{\pid+c_2}\seqq\hT{\pid+c_4} \tag{by naturality of $\swapp$} \\
                    ={}& (\hT{c_1\seqq c_3}+\pid)\seqq(\swapp+\pid)\seqq(\pid+\pid_\bbo +\pid)\seqq(\swapp+\pid)\seqq \\
                    & \quad \hT{\pid+c_2}\seqq\hT{\pid+c_4} \tag{by definition of $\cT_H$} \\
                    \fromto_2{}& \hT{(c_1\seqq c_3)+\pid}\seqq \hT{\pid+c_2}\seqq\hT{\pid+c_4} \tag{by definition of $\cT_H$} \\
                    \fromto_2{}& \cdots \\
                    \fromto_2{}& \hT{(c_1\seqq c_3)+\pid}\seqq \hT{\pid+(c_2\seqq c_4)} \\
                    \fromto_2{}& \hT{(c_1\seqq c_3)+ (c_2\seqq c_4)}. \tag{by \cref{lem:ht3}}
                \end{align*}
            \item For $(c_1\times c_2)\seqq(c_3\times c_4) \fromto_2 (c_1\seqq c_3) \times  (c_2\seqq c_4)$, we have \begin{align*}
                & \hT{(c_1\times c_2)\seqq(c_3\times c_4)} \\
                ={}& \hT{c_1\times c_2}\seqq\hT{c_3\times c_4} \\
                \fromto_2{}& \hT{c_1\times \pid}\seqq\hT{\pid\times c_2}\seqq\hT{c_3\times \pid}\seqq\hT{\pid\times c_4} \tag{by \cref{lem:ht3}} \\
                \fromto_2{}& \hT{c_1\times \pid}\seqq\hT{c_3\times \pid}\seqq\hT{\pid\times c_2}\seqq\hT{\pid\times c_4} \tag{by \cref{lem:ht3}} \\
                \fromto_2{}& (\pid_{\bbo}+\swapt)\seqq\hT{\pid\times c_1}\seqq(\pid_{\bbo}+\swapt)\seqq\hT{\pid\times \pid}\seqq \\
                & \quad (\pid_{\bbo}+\swapt)\seqq\hT{\pid\times c_3}\seqq(\pid_{\bbo}+\swapt)\seqq\hT{\pid\times \pid}\seqq \\
                & \quad \hT{\pid\times c_2}\seqq\hT{\pid\times c_4} \tag{by definition of $\cT_H$} \\
                \fromto_2{}& (\pid_{\bbo}+\swapt)\seqq\hT{\pid\times c_1}\seqq(\pid_{\bbo}+\swapt)\seqq(\pid_{\bbo}+\pid)\seqq \\
                & \quad (\pid_{\bbo}+\swapt)\seqq\hT{\pid\times c_3}\seqq(\pid_{\bbo}+\swapt)\seqq\hT{\pid\times \pid}\seqq \\
                & \quad \hT{\pid\times c_2}\seqq\hT{\pid\times c_4} \\
                \fromto_2{}& (\pid_{\bbo}+\swapt)\seqq\hT{\pid\times c_1}\seqq\hT{\pid\times c_3}\seqq(\pid_{\bbo}+\swapt)\seqq\hT{\pid\times \pid}\seqq \\
                & \quad \hT{\pid\times c_2}\seqq\hT{\pid\times c_4} \tag{by naturality of $\swapt$} \\
                \fromto_2{}& (\pid_{\bbo}+\swapt)\seqq\hT{c_1+\cdots+c_1}\seqq\hT{c_3+\cdots+c_3}\seqq \\
                & \quad (\pid_{\bbo}+\swapt)\seqq\hT{\pid\times \pid}\seqq  \hT{c_2+\cdots+c_2}\seqq\hT{c_4+\cdots+c_4} \tag{by \cref{lem:hqt1}}\\
                \fromto_2{}& (\pid_{\bbo}+\swapt)\seqq\hT{(c_1\seqq c_3)+\cdots+(c_1\seqq c_3)}\seqq \\
                & \quad (\pid_{\bbo}+\swapt)\seqq\hT{\pid\times \pid}\seqq \hT{(c_2\seqq c_4)+\cdots+(c_2\seqq c_4)} \tag{by samilar proofs for $\hT{(c_1+c_2)\seqq(c_3+c_4)} \fromto_2 \hT{(c_1\seqq c_3)+ (c_2\seqq c_4)}$} \\
                \fromto_2{}& (\pid_{\bbo}+\swapt)\seqq\hT{\pid\times (c_1\seqq c_3)}\seqq  (\pid_{\bbo}+\swapt)\seqq\hT{\pid\times \pid}\seqq \hT{\pid\times (c_2\seqq c_4)} \tag{by \cref{lem:hqt1}} \\
                \fromto_2{}& \hT{(c_1\seqq c_3)\times \pid}\seqq \hT{\pid\times (c_2\seqq c_4)} \tag{by definition of $\cT_H$} \\
                \fromto_2{}& \hT{(c_1\seqq c_3)\times (c_2\seqq c_4)}. \tag{by \cref{lem:ht3}}
            \end{align*}
            \end{itemize}
            \item Distributivity. Suppose $c_1:b_1\fromto b_4, c_2:b_2 \fromto b_5, c_3:b_3\fromto b_6$. Consider \begin{align*}
                &\hqT{\factor\seqq((c_1+c_2)\times c_3)\seqq\dist} \\
                \fromto_2{}& (\hqT{c_1+c_2}\times \pid_{b_3})\seqq(\swapt_{(b_4+b_5)\times b_3}+\swapt_{(b_4+b_5)\times b_3})\seqq \\
                &\quad (\hqT{c_3}\times \pid_{b_4+b_5})\seqq(\swapt_{b_6\times (b_4+b_5)}+\swapt_{b_6\times (b_4+b_5)}) \tag{by definition of $\cT_{HQ}$} \\
                \fromto_2{}& (((\pid_{b_1}+\swapp_{b_2+b_1}+\pid_{b_2})\seqq(\hqT{c_1}+\hqT{c_2})\seqq(\pid_{b_4}+\swapp_{b_4+b_5}+\pid_{b_5}))\times \pid_{b_3})\seqq\\
                & \quad (\swapt_{(b_4+b_5)\times b_3}+\swapt_{(b_4+b_5)\times b_3})\seqq \\
                &\quad (\hqT{c_3}\times \pid_{b_4+b_5})\seqq(\swapt_{b_6\times (b_4+b_5)}+\swapt_{b_6\times (b_4+b_5)}) \tag{by definition of $\cT_{HQ}$} \\
                \fromto_2{}& ((\pid_{b_1}+\swapp_{b_2+b_1}+\pid_{b_2})\times \pid_{b_3})\seqq((\hqT{c_1}+\hqT{c_2})\times \pid_{b_3})\seqq \\
                & \quad ((\pid_{b_4}+\swapp_{b_4+b_5}+\pid_{b_5})\times \pid_{b_3})\seqq (\swapt_{(b_4+b_5)\times b_3}+\swapt_{(b_4+b_5)\times b_3})\seqq \\
                &\quad (\hqT{c_3}\times \pid_{b_4+b_5})\seqq(\swapt_{b_6\times (b_4+b_5)}+\swapt_{b_6\times (b_4+b_5)}) \tag{by bifunctoriality} \\
                \fromto_2{}& ((\pid_{b_1}+\swapp_{b_2+b_1}+\pid_{b_2})\times \pid_{b_3})\seqq(\hqT{c_1}\times \pid_{b_3}+\hqT{c_2}\times \pid_{b_3})\seqq \\
                & \quad ((\pid_{b_4}+\swapp_{b_4+b_5}+\pid_{b_5})\times \pid_{b_3})\seqq (\swapt_{(b_4+b_5)\times b_3}+\swapt_{(b_4+b_5)\times b_3})\seqq \\
                &\quad (\hqT{c_3}\times \pid_{b_4+b_5})\seqq(\swapt_{b_6\times (b_4+b_5)}+\swapt_{b_6\times (b_4+b_5)}) \tag{by distributivity} \\
                \fromto_2{}& ((\pid_{b_1}+\swapp_{b_2+b_1}+\pid_{b_2})\times \pid_{b_3})\seqq(\hqT{c_1}\times \pid_{b_3}+\hqT{c_2}\times \pid_{b_3})\seqq \\
                & \quad ((\pid_{b_4}+\swapp_{b_4+b_5}+\pid_{b_5})\times \pid_{b_3})\seqq (\swapt_{(b_4+b_5)\times b_3}+\swapt_{(b_4+b_5)\times b_3})\seqq \\
                &\quad (\hqT{c_3}\times \pid_{b_4+b_5})\seqq(\swapt_{b_6\times (b_4+b_5)}+\swapt_{b_6\times (b_4+b_5)}) \tag{by naturality} \\
                \fromto_2{}& ((\pid_{b_1}+\swapp_{b_2+b_1}+\pid_{b_2})\times \pid_{b_3})\seqq(\hqT{c_1}\times \pid_{b_3}+\hqT{c_2}\times \pid_{b_3})\seqq \\
                & \quad ((\pid_{b_4}+\swapp_{b_4+b_5}+\pid_{b_5})\times \pid_{b_3})\seqq (\swapt_{(b_4+b_5)\times b_3}+\swapt_{(b_4+b_5)\times b_3})\seqq \\
                &\quad \swapt_{(\bb{2}\times b_3)\times (b_4+b_5)}\seqq(\pid_{b_4}\times \hqT{c_3}+\pid_{b_5}\times \hqT{c_3})\seqq\swapt_{(b_4+b_5)\times (\bb{2}\times b_6)}\seqq\\
                & \quad (\swapt_{b_6\times (b_4+b_5)}+\swapt_{b_6\times (b_4+b_5)}) \tag{by naturality and distributivity} \\
                \fromto_2{}& ((\pid_{b_1}+\swapp_{b_2+b_1}+\pid_{b_2})\times \pid_{b_3})\seqq(\hqT{c_1}\times \pid_{b_3}+\hqT{c_2}\times \pid_{b_3})\seqq \\
                &\quad (\swapt_{\bb{2}\times b_4}\times \pid_{b_3}+\swapt_{\bb{2}\times b_5}\times \pid_{b_3})\seqq(\pid_{b_4}\times \hqT{c_3}+\pid_{b_5}\times \hqT{c_3})\seqq\\
                & \quad (\swapt_{(b_4+b_5)\times \bb{2}}\times \pid_{b_6}) \tag{by completeness of $\Pi$ and bifunctoriality} \\
                \fromto_2{}& ((\pid_{b_1}+\swapp_{b_2+b_1}+\pid_{b_2})\times \pid_{b_3})\seqq \\
                & (((\hqT{c_1}\times \pid_{b_3})\seqq(\swapt_{\bb{2}\times b_4}\times \pid_{b_3})\seqq(\pid_{b_4}\times \hqT{c_3}))+ \\
                &\qquad ((\hqT{c_2}\times \pid_{b_3})\seqq(\swapt_{\bb{2}\times b_5}\times \pid_{b_3})\seqq(\pid_{b_5}\times \hqT{c_3})))\seqq \\
                &\quad (\swapt_{(b_4+b_5)\times \bb{2}}\times \pid_{b_6}) \tag{by bifunctoriality} \\
                \fromto_2{}& ((\pid_{b_1}+\swapp_{b_2+b_1}+\pid_{b_2})\times \pid_{b_3})\seqq \\
                & (((\hqT{c_1}\times \pid_{b_3})\seqq(\swapt_{\bb{2}\times b_4}\times \pid_{b_3})\seqq(\swapt_{b_4\times(\bb{2}\times b_3)})\seqq\\
                & \qquad\quad (\hqT{c_3}\times \pid_{b_4})\seqq(\swapt_{(\bb{2}\times b_6)\times b_4}))+ \\
                &\qquad ((\hqT{c_2}\times \pid_{b_3})\seqq(\swapt_{\bb{2}\times b_5}\times \pid_{b_3})\seqq(\swapt_{b_5\times (\bb{2}\times b_3)})\seqq \\
                & \qquad\quad (\hqT{c_3}\times \pid_{b_5})\seqq(\swapt_{(\bb{2}\times b_6)\times b_5})))\seqq \\
                &\quad (\swapt_{(b_4+b_5)\times \bb{2}}\times \pid_{b_6}) \tag{by naturality}\\
                \fromto_2{}& ((\pid_{b_1}+\swapp_{b_2+b_1}+\pid_{b_2})\times \pid_{b_3})\seqq \\
                & (((\hqT{c_1}\times \pid_{b_3})\seqq(\swapt_{b_4\times b_3}+\swapt_{b_4\times b_3})\seqq(\hqT{c_3}\times \pid_{b_4})\seqq \\
                & \qquad \quad (\swapt_{b_6\times b_4}+\swapt_{b_6\times b_4})\seqq (\swapt_{\bb{2}\times b_4}\times \pid_{b_6}))+ \\
                &\qquad ((\hqT{c_2}\times \pid_{b_3})\seqq(\swapt_{b_5\times b_3}+\swapt_{b_5\times b_3})\seqq(\hqT{c_3}\times \pid_{b_5})\seqq \\
                & \qquad \quad (\swapt_{b_6\times b_5}+\swapt_{b_6\times b_5})\seqq (\swapt_{\bb{2}\times b_5}\times \pid_{b_6})))\seqq \\
                &\quad (\swapt_{(b_4+b_5)\times \bb{2}}\times \pid_{b_6}) \tag{by completeness of $\Pi$} \\
                \fromto_2{}& ((\pid_{b_1}+\swapp_{b_2+b_1}+\pid_{b_2})\times \pid_{b_3})\seqq \\
                & ((\hqT{c_1\times c_3}\seqq(\swapt_{\bb{2}\times b_4}\times \pid_{b_6}))+ (\hqT{c_2\times c_3}\seqq(\swapt_{\bb{2}\times b_5}\times \pid_{b_6})))\seqq \\
                &\quad (\swapt_{(b_4+b_5)\times \bb{2}}\times \pid_{b_6}) \tag{by definition of $\cT_{HQ}$} \\
                \fromto_2{}& (\pid_{b_1\times b_3}+\swapp_{(b_2\times b_3)+(b_1\times b_3)}+\pid_{b_2\times b_3})\seqq(\hqT{c_1\times c_3} + \hqT{c_2\times c_3})\seqq \\
                & \quad (\pid_{b_3\times b_6}+\swapp_{(b_4\times b_6)+(b_5\times b_6)}+\pid_{b_5\times b_6}) \tag{by completeness of $\Pi$ and bifunctoriality} \\
                \fromto_2{}& \hT{(c_1\times c_3)+(c_2\times c_3)}, \tag{by definition of $\cT_{HQ}$}
            \end{align*}
            with the help of \cref{lem:hqt2,eq:h3}, we have \[\hT{\factor\seqq((c_1+c_2)\times c_3)\seqq\dist} \fromto_2{} \hT{(c_1\times c_3)+(c_2\times c_3)}.\]
            The remaining one can be proven in the same way.
            \item Coherence. \begin{itemize}[leftmargin=10pt]
                \item For $+$, we can easily verify that  \begin{align*}
                    \hT{\assocrp} ={}& \pid_{\bbo}+\assocrp, \\
                    \hT{\assoclp} ={}& \pid_{\bbo}+\assoclp, \\
                    \hT{\swapp} ={}& \pid_{\bbo}+\swapp, \\
                    \hT{\pid+\unitep} \fromto_2{}& \pid_{\bbo}+(\pid+\unitep), \\
                    \hT{\unitep+\pid} \fromto_2{}& \pid_{\bbo}+(\unitep+\pid), \\
                    \hT{\swapp+\pid} \fromto_2{}& \pid_{\bbo}+(\swapp+\pid), \\
                    \hT{\pid+\swapp} \fromto_2{}& \pid_{\bbo}+(\pid+\swapp).
                \end{align*}
                Thus, it is straightforward to get the equivalence for coherence of $+$ by bifunctoriality.
                \item For $\times$, we can easily verify that \begin{align*}
                    \hqT{\assocrt} ={}& \pid_{\hdim(\assocrt)\bbo}+\assocrt, \\
                    \hqT{\assoclt} ={}& \pid_{\hdim(\assoclt)\bbo}+\assoclt, \\
                    \hqT{\swapt} ={}& \pid_{\hdim(\swapt)\bbo}+\swapt, \\
                    \hqT{\pid\times \unitet} \fromto_2{}& \pid_{\hdim(\pid\times \unitet)\bbo}+(\pid\times \unitet), \\
                    \hqT{\unitet\times \pid} \fromto_2{}& \pid_{\hdim(\unitet\times \pid)\bbo}+(\unitet\times \pid), \\
                    \hqT{\swapt\times \pid} \fromto_2{}& \pid_{\hdim(\swapt\times \pid)\bbo}+(\swapt\times \pid), \\
                    \hqT{\pid\times \swapt} \fromto_2{}& \pid_{\hdim(\pid\times \swapt)\bbo}+(\pid\times \swapt).
                \end{align*}
                Then, it is straightforward to obtain the equivalence for coherence of $\times$ by bifunctoriality when applied with the translation $\cT_{HQ}$. By \cref{lem:hqt2,eq:h3}, we derive the equivalence for coherency of $\times$ by bifunctoriality when applied with the translation $\cT_{H}$.
            \end{itemize}
            \item Naturality. Suppose $c_1:b_1\fromto b_3, c_2:b_2\fromto b_4$. 
            \begin{itemize}[leftmargin=10pt]
                \item For $\swapp$, we have \begin{align*}
                    & \hT{\swapp\seqq(c_1+c_2)\seqq\swapp} \\
                    ={}&\hT{\swapp_{b_2+b_1}}\seqq\hT{c_1+c_2}\seqq\hT{\swapp_{b_3+b_4}} \\
                    \fromto_2{}& \hT{\swapp_{b_2+b_1}}\seqq\hT{\pid_{b_1}+c_2}\seqq\hT{c_1+\pid_{b_4}}\seqq\hT{\swapp_{b_3+b_4}} \tag{by \cref{lem:ht3}} \\
                    \fromto_2{}& (\pid_{\bbo}+\swapp_{b_2+b_1})\seqq \\
                    & \quad (\pid_{\bbo}+\pid_{b_1}+\pid_{b_2})\seqq(\swapp_{\bbo+b_1}+\pid_{b_2})\seqq(\pid_{b_1}+\hT{c_2})\seqq(\swapp_{b_1+\bbo}+\pid_{b_4})\seqq \\
                    & \quad (\hT{c_1}+\pid_{b_4})\seqq(\swapp_{\bbo+b_3}+\pid_{b_4})\seqq(\pid_{b_3}+\pid_{\bbo}+\pid_{b_4})\seqq(\swapp_{b_3+\bbo}+\pid_{b_4})\seqq \\
                    & \quad (\pid_{\bbo}+\swapp_{b_3+b_4}) \tag{by definition of $\cT_H$} \\
                    \fromto_2{}& \swapp_{(\bbo+b_2)+b_1}\seqq(\pid_{b_1}+\hT{c_2})\seqq(\swapp_{b_1+\bbo}+\pid_{b_4})\seqq \\
                    &\quad (\hT{c_1}+\pid_{b_4})\seqq (\pid_{\bbo}+\swapp_{b_3+b_4}) \tag{by completeness of $\Pi$} \\
                    \fromto_2{}& (\hT{c_2}+\pid_{b_1})\seqq \swapp_{(\bbo+b_4)+b_1}\seqq(\swapp_{b_1+\bbo}+\pid_{b_4})\seqq\\
                    &\quad (\hT{c_1}+\pid_{b_4})\seqq (\pid_{\bbo}+\swapp_{b_3+b_4}) \tag{by naturality} \\
                    \fromto_2{}& (\hT{c_2}+\pid_{b_1})\seqq (\swapp_{\bbo+b_4}+\pid_{b_1})\seqq(\swapp_{b_4+(\bbo+b_1)})\seqq \\
                    &\quad  (\hT{c_1}+\pid_{b_4})\seqq (\pid_{\bbo}+\swapp_{b_3+b_4}) \tag{by completeness of $\Pi$} \\
                    \fromto_2{}& (\hT{c_2}+\pid_{b_1})\seqq (\swapp_{\bbo+b_4}+\pid_{b_1})\seqq  (\pid_{b_4}+\hT{c_1})\seqq \\
                    & \quad (\swapp_{b_4+(\bbo+b_3)})\seqq (\pid_{\bbo}+\swapp_{b_3+b_4}) \tag{by naturality} \\
                    \fromto_2{}& (\hT{c_2}+\pid_{b_1})\seqq (\swapp_{\bbo+b_4}+\pid_{b_1})\seqq  (\pid_{b_4}+\hT{c_1})\seqq (\swapp_{b_4+\bbo}+\pid_{b_3}) \tag{by completeness of $\Pi$} \\
                    ={}& \hT{c_2+c_1}. \tag{by definition of $\cT_H$}
                \end{align*}
                \item For $\swapt$, we have \begin{align*}
                    & \hqT{\swapt\seqq(c_1\times c_2)\seqq\swapt} \\
                    ={}& \hqT{\swapt}\seqq\hqT{c_1\times c_2}\seqq\hqT{\swapt} \\
                    \fromto_2{}& \hqT{\swapt_{b_2\times b_1}}\seqq\hqT{\pid_{b_1}\times c_2}\seqq\hqT{c_1\times \pid_{b_4}}\seqq \hqT{\swapt_{b_3\times b_4}} \tag{by \cref{lem:ht3,lem:hqt2,eq:h3}} \\
                    \fromto_2{}& (\pid_{b_2\times b_1}+\swapt_{b_2\times b_1})\seqq \\
                    & \quad (\swapt_{b_1\times b_2}+\swapt_{b_1\times b_2})\seqq(\hqT{c_2}\times \pid_{b_1})\seqq(\swapt_{b_4\times b_1}+\swapt_{b_4\times b_1})\seqq\\
                    & \quad (\hqT{c_1}\times \pid_{b_4})\seqq  (\pid_{b_3\times b_4}+\swapt_{b_3\times b_4}) \tag{by definition of $\cT_{HQ}$ and completeness of $\Pi$} \\
                    \fromto_2{}& (\swapt_{b_1\times b_2}+\pid_{b_2\times b_1})\seqq(\hqT{c_2}\times \pid_{b_1})\seqq(\swapt_{b_4\times b_1}+\swapt_{b_4\times b_1})\seqq\\
                    & \quad (\hqT{c_1}\times \pid_{b_4})\seqq  (\swapt_{b_3\times b_4}+\swapt_{b_3\times b_4})\seqq(\swapt_{b_4\times b_3}+\pid_{b_4\times b_3}) \tag{by completeness of $\Pi$} \\
                    \fromto_2{}& (\swapt_{b_1\times b_2}+\pid_{b_2\times b_1})\seqq\hqT{c_2\times c_1}\seqq(\swapt_{b_4\times b_3}+\pid_{b_4\times b_3}). \tag{by definition of $\cT_{HQ}$}
                \end{align*}
                By \cref{lem:hqt2}, we obtain \begin{align*}
                    & \pid_{\hdim(c_2\times c_1)\bbo}+\hT{\swapt\seqq(c_1\times c_2)\seqq\swapt} \\
                    \fromto_2{}& (\swapp_{\hdim(c_2\times c_1)\bbo+\bbo}+\pid_{b_2\times b_1}) \seqq \\
                    & \quad (\pid_{\bbo}+ \hqT{\swapt\seqq(c_1\times c_2)\seqq\swapt})\seqq \\
                    &\quad (\pid_\bbo+\swapp_{\hdim(c_2\times c_1)\bbo}+\pid_{b_4\times b_3}) \\
                    \fromto_2{}& (\swapp_{\hdim(c_2\times c_1)\bbo+\bbo}+\pid_{b_2\times b_1}) \seqq \\
                    & \quad (\pid_{\bbo}+ ((\swapt_{b_1\times b_2}+\pid_{b_2\times b_1})\seqq\hqT{c_2\times c_1}\seqq(\swapt_{b_4\times b_3}+\pid_{b_4\times b_3})))\seqq \\
                    &\quad (\pid_\bbo+\swapp_{\hdim(c_2\times c_1)\bbo}+\pid_{b_4\times b_3}) \\
                    \fromto_2{}& (\swapp_{\hdim(c_2\times c_1)\bbo+\bbo}+\pid_{b_2\times b_1}) \seqq (\pid_{\bbo}+ (\swapt_{b_1\times b_2}+\pid_{b_2\times b_1}))\seqq\\
                    & \quad (\pid_{\bbo}+\hqT{c_2\times c_1})\seqq\\
                    & \quad (\pid_{\bbo}+(\swapt_{b_4\times b_3}+\pid_{b_4\times b_3}))\seqq (\pid_\bbo+\swapp_{\hdim(c_2\times c_1)\bbo}+\pid_{b_4\times b_3}) \tag{by bifunctoriality}\\
                    \fromto_2{}& (\swapp_{\hdim(c_2\times c_1)\bbo+\bbo}+\pid_{b_2\times b_1}) \seqq (\pid_{\bbo}+ (\swapt_{b_1\times b_2}+\pid_{b_2\times b_1}))\seqq\\
                    & \quad (\swapp_{\bbo+\hdim(c_2\times c_1)\bbo}+\pid_{b_2\times b_1})\seqq(\pid_{\hdim(c_2\times c_1)\bbo}+\hT{c_2\times c_1})\seqq \\
                    & \quad (\swapp_{\hdim(c_2\times c_1)\bbo+\bbo}+\pid_{b_4\times b_3})\seqq\\
                    & \quad (\pid_{\bbo}+(\swapt_{b_4\times b_3}+\pid_{b_4\times b_3}))\seqq (\bbo+\swapp_{\hdim(c_2\times c_1)\bbo}+\pid_{b_4\times b_3}) \tag{by \cref{lem:hqt2}} \\
                    \fromto_2{}& (\swapt_{b_1\times b_2}+\pid_{\bbo}+\pid_{b_2\times b_1})\seqq(\pid_{\hdim(c_2\times c_1)\bbo}+\hT{c_2\times c_1})\seqq \\
                    & \quad (\swapt_{b_4\times b_3}+\pid_{\bbo}+\pid_{b_4\times b_3}) \tag{by completeness of $\Pi$} \\
                    \fromto_2{}& (\swapt_{(\hdim(b_1)\bbo)\times(\hdim(b_2)\bbo)}+\pid_{\bbo}+\pid_{b_2\times b_1})\seqq(\pid_{\hdim(c_2\times c_1)\bbo}+\hT{c_2\times c_1})\seqq \\
                    & (\swapt_{(\hdim(b_4)\bbo)\times(\hdim(b_3)\bbo)}+\pid_{\bbo}+\pid_{b_4\times b_3}) \tag{by completeness of $\Pi$} \\
                    \fromto_2{}& \pid_{\hdim(c_2\times c_1)\bbo}+\hT{c_2\times c_1}. \tag{by bifunctoriality and naturality}
                \end{align*}
                Then, applying \cref{eq:h3}, we have have $\hT{\swapt\seqq(c_1\times c_2)\seqq\swapt} \fromto_2 \hT{c_2\times c_1}$.
            \end{itemize}    
        \end{itemize}
        \item Additional level-2 combinators of \QPiLang{}. \begin{itemize}[leftmargin=10pt]
            \item \cref{eq:e1}. \begin{align*}
                \hT{(-1)^2} ={}& \hT{(-1)}\seqq\hT{(-1)} \\
                \fromto_2{}& \hadamard\seqq\swapp\seqq\hadamard\seqq\\
                & \quad \hadamard\seqq\swapp\seqq\hadamard \\
                \fromto_2{}& \pid_{\bbo+\bbo} \tag{by \cref{eq:h1}, naturality and unitality} \\
                \fromto_2{}& \pid_{\bbo}+\pid_{\bbo} \\
                ={}& \hT{\pid_{\bbo}}.
            \end{align*}
            \item \cref{eq:e2}. \begin{align*}
                \hT{\hadamard^2} ={}& (\pid_{\bbo}+\hadamard)^2 \\
                \fromto_2{}& \pid_{\bbo}+\hadamard^2 \tag{by bifunctoriality} \\
                \fromto_2{}& \pid_{\bbo}+\pid_{\bbo+\bbo} \tag{by \cref{eq:h1}} \\
                ={}& \hT{\pid_{\bbo+\bbo}}.
            \end{align*}
            \item \cref{eq:e3}. \begin{align*}
                & \hT{\pid_{\bbo}+(-1)} \\
                \fromto_2{}& (\pid_{\bbo}+\pid_{\bbo}+\pid_{\bbo})\seqq(\swapp_{\bbo+\bbo}+\pid_{\bbo})\seqq\\
                & \quad (\pid_{\bbo}+\hadamard\seqq\swapp_{\bbo+\bbo}\seqq\hadamard)\seqq(\swapp_{\bbo+\bbo}+\pid_{\bbo}) \\
                \fromto_2{}& (\swapp_{\bbo+\bbo}+\pid_{\bbo})\seqq\\
                & \quad (\pid_{\bbo}+\hadamard\seqq\swapp_{\bbo+\bbo}\seqq\hadamard)\seqq(\swapp_{\bbo+\bbo}+\pid_{\bbo}) \\
                \fromto_2{}& (\pid_{\bbo}+\hadamard\seqq\swapp_{\bbo+\bbo}\seqq\hadamard) \tag{by \cref{eq:h2}} \\
                \fromto_2{}& (\pid_{\bbo}+\hadamard)\seqq(\pid_\bbo+\swapp)\seqq(\pid_{\bbo}+\hadamard) \tag{by bifunctoriality}\\
                ={}& \hT{\hadamard}\seqq\hT{\swapp}\seqq\hT{\hadamard} \\
                \fromto_2{}& \hT{\hadamard\seqq\swapp\seqq\hadamard}.
            \end{align*}
            \item \cref{eq:e4}. Suppose $c_1+c_2\fromto_2{} \pid$ and $c_1:b_1\fromto b_3, c_2:b_2\fromto b_4$, then
            \begin{align*}
                & \pid_{\bbo}+\pid_{\bbo+b_1+b_2} \\
                \fromto_2{}&
                \pid_{\bbo}+\hT{\pid_{b_1+b_2}} \tag{by definition of $\cT_H$} \\
                \fromto_2{}& \pid_{\bbo}+\hT{c_1+c_2} \tag{by inductive hypothesis}\\
                \fromto_2{}& \pid_{\bbo}+((\hT{c_1}+\pid_{b_2})\seqq \\
                & \qquad (\swapp_{\bbo+b_3}+\pid_{b_2})\seqq (\pid_{b_3}+\hT{c_2})\seqq(\swapp_{b_3+\bbo}+\pid_{b_4})) \tag{by definition of $\cT_H$}\\
                \fromto_2{}& (\pid_{\bbo}+(\hT{c_1}+\pid_{b_2}))\seqq \\
                & \quad (\pid_{\bbo}+((\swapp_{\bbo+b_3}+\pid_{b_2})\seqq(\pid_{b_3}+\hT{c_2})\seqq(\swapp_{b_3+\bbo}+\pid_{b_4}))) \tag{by bifunctoriality of $+$} \\
                \fromto_2{}& (((\swapp_{\bbo+\bbo}+\pid_{b_1})\seqq(\pid_{\bbo}+\hT{c_1})\seqq(\swapp_{\bbo+\bbo}+\pid_{b_3}))+\pid_{b_2})\seqq \\
                & \quad (\pid_{\bbo}+((\swapp_{\bbo+b_3}+\pid_{b_2})\seqq(\pid_{b_3}+\hT{c_2})\seqq(\swapp_{b_3+\bbo}+\pid_{b_4}))) \tag{by \cref{lem:ht2}} \\
                \fromto_2{}& (\swapp_{\bbo+\bbo}+\pid_{b_1}+\pid_{b_2})\seqq(\pid_{\bbo}+\hT{c_1}+\pid_{b_2})\seqq(\swapp_{\bbo+\bbo}+\pid_{b_3}+\pid_{b_2})\seqq \\
                & \quad (\pid_{\bbo}+\swapp_{\bbo+b_3}+\pid_{b_2})\seqq(\pid_{\bbo}+\pid_{b_3}+\hT{c_2})\seqq(\pid_{\bbo}+\swapp_{b_3+\bbo}+\pid_{b_4}) \tag{by bifunctoriality of $+$} \\
                \fromto_2{}& (\swapp_{\bbo+\bbo}+\pid_{b_1}+\pid_{b_2})\seqq(\pid_{\bbo}+\hT{c_1}+\pid_{b_2})\seqq(\swapp_{\bbo+(\bbo+b_3)}+\pid_{b_2})\seqq \\
                & \quad (\pid_{\bbo}+\pid_{b_3}+\hT{c_2})\seqq(\pid_{\bbo}+\swapp_{b_3+\bbo}+\pid_{b_4}) \tag{by completeness of $\Pi$} \\
                \fromto_2{}& (\swapp_{\bbo+\bbo}+\pid_{b_1}+\pid_{b_2})\seqq(\swapp_{\bbo+(\bbo+b_1)}+\pid_{b_2})\seqq(\hT{c_1}+\pid_{\bbo}+\pid_{b_2})\seqq \\
                & \quad (\pid_{\bbo}+\pid_{b_3}+\hT{c_2})\seqq(\pid_{\bbo}+\swapp_{b_3+\bbo}+\pid_{b_4}) \tag{by naturality of $\swapp$} \\
                \fromto_2{}& (\swapp_{\bbo+\bbo}+\pid_{b_1}+\pid_{b_2})\seqq(\swapp_{\bbo+(\bbo+b_1)}+\pid_{b_2})\seqq (\hT{c_1}+\hT{c_2})\seqq(\pid_{\bbo}+\swapp_{b_3+\bbo}+\pid_{b_4}) \tag{by bifunctoriality of $+$} \\
                \fromto_2{}& (\pid_{\bbo}+\swapp_{\bbo+b_1}+\pid_{b_2})\seqq (\hT{c_1}+\hT{c_2})\seqq(\pid_{\bbo}+\swapp_{b_3+\bbo}+\pid_{b_4}). \tag{by completeness of $\Pi$}
            \end{align*}
            Using the naturality of $\swapp$, and bifunctoriality of $+$, we obtain
            \begin{align*}
                \hT{c_1}+\hT{c_2} \fromto_2{} (\pid_{\bbo}+\pid_{b_1})+(\pid_{\bbo}+\pid_{b_2}).
            \end{align*}
            Then, from \cref{eq:h3} and the definition of $\cT_H$, it follows that 
            \[\hT{c_1} \fromto_2{} \hT{\pid_{b_1}}. \qedhere\]
        \end{itemize}
    \end{itemize}
\end{proof}

\endgroup

\section{Misc}\label{app:misc}

\subsection{Level-2 Combinators of \texorpdfstring{$\Pi$}{Pi}}\label{app:pi_equations}

\begin{itemize}
    \item Reflexivity, symmetry and transitivity: \begin{gather*}
        c \fromto_2 c \\[1mm]
        \begin{prooftree}
            \hypo{c_1 \fromto_2 c_2}
            \infer1{c_2 \fromto_2 c_1}
        \end{prooftree} \\[1mm]
        \begin{prooftree}
            \hypo{c_1 \fromto_2 c_2} \hypo{c_2 \fromto_2 c_3}
            \infer2{c_1 \fromto_2 c_3}
        \end{prooftree}
    \end{gather*}
    \item Unitality: \begin{gather*}
        \pid\seqq c \fromto_2 c \\
        c\seqq \pid \fromto_2 c \\
        \unitip\seqq(\pid_{\bbz}+c)\seqq\unitep \fromto_2 c\\
        \unitep\seqq c\seqq \unitip \fromto_2  \pid_{\bbz}+c\\
        \unitit\seqq(\pid_{\bbo}\times c)\seqq\unitet \fromto_2 c\\
        \unitet\seqq c\seqq \unitit \fromto_2  \pid_{\bbo}\times c
    \end{gather*}
    \item Associativity: \begin{gather*}
        c_1\seqq(c_2\seqq c_3) \fromto_2{} (c_1\seqq c_2)\seqq c_3 \\
        \assocrp\seqq(c_1+(c_2+c_3))\seqq \assoclp \fromto_2{} (c_1+c_2)+c_3 \\
        \assoclp\seqq((c_1+c_2)+c_3)\seqq \assocrp \fromto_2{} c_1+(c_2+c_3) \\
        \assocrt\seqq(c_1\times (c_2\times c_3))\seqq \assoclt \fromto_2{} (c_1\times c_2)\times c_3 \\
        \assoclt\seqq((c_1\times c_2)\times c_3)\seqq \assocrt \fromto_2{} c_1\times (c_2\times c_3) \\
        \assocrp\seqq\assocrp \fromto_2 ((\assocrp+\pid)\seqq\assocrp)\seqq(\pid+\assocrp) \\
        \assocrt\seqq\assocrt \fromto_2 ((\assocrt+\pid)\seqq\assocrt)\seqq(\pid+\assocrt)\\
    \end{gather*}
    \item Annihilativity: \begin{gather*}
        \factorz\seqq(c\times \pid_{\bbz})\seqq\absorb \fromto_2 \pid_{\bbz} \\
        \absorb\seqq\factorz \fromto_2 c\times \pid_{\bbz}
    \end{gather*}
    \item Bifunctoriality: \begin{gather*}
        \begin{prooftree}
            \hypo{c_1\fromto_2 c_3} \hypo{c_2\fromto_2 c_4}
            \infer2{c_1+c_2 \fromto_2 c_3+c_4}
        \end{prooftree} \\[1mm]
        \begin{prooftree}
            \hypo{c_1\fromto_2 c_3} \hypo{c_2\fromto_2 c_4}
            \infer2{c_1\times c_2 \fromto_2 c_3\times c_4}
        \end{prooftree} \\[1mm]
        \pid + \pid \fromto_2 \pid \\
        \pid \times  \pid \fromto_2 \pid \\
        (c_1+c_2)\seqq(c_3+c_4) \fromto_2 (c_1\seqq c_3) + (c_2\seqq c_4) \\
        (c_1\times c_2)\seqq(c_3\times c_4) \fromto_2 (c_1\seqq c_3) \times  (c_2\seqq c_4)
    \end{gather*}
    \item Distributivity: \begin{gather*}
        \factor\seqq((c_1+c_2)\times c_3)\seqq \dist \fromto_2 (c_1\times c_3)+(c_2\times c_3) \\
        \dist\seqq((c_1\times c_3)+(c_2\times c_3))\seqq\factor \fromto_2{} (c_1+c_2)\times c_3
    \end{gather*}
    \item Coherence: \begin{gather*}
        \assocrp\seqq(\pid+\unitep)\fromto_2 (\swapp+\pid)\seqq(\unitep+\pid) \\
        \assocrt\seqq(\pid\times \unitet)\fromto_2 (\swapt\times \pid)\seqq(\unitet\times \pid) \\
        (\assocrp\seqq\swapp)\seqq \assocrp \fromto_2 ((\swapp+\pid)\seqq\assocrp)\seqq(\pid+\swapp) \\
        (\assoclp\seqq\swapp)\seqq \assoclp \fromto_2 ((\pid+\swapp)\seqq\assoclp)\seqq(\swapp+\pid) \\
        (\assocrt\seqq\swapt)\seqq \assocrt \fromto_2 ((\swapt\times \pid)\seqq\assocrt)\seqq(\pid\times \swapt) \\
        (\assoclt\seqq\swapt)\seqq \assoclt \fromto_2 ((\pid\times \swapt)\seqq\assoclt)\seqq(\swapt\times \pid) \\
    \end{gather*}
    \item Naturality:  \begin{gather*}
        \swapp\seqq (c_1+c_2)\seqq\swapp \fromto_2 c_2 + c_1 \\
        \swapt\seqq (c_1\times c_2)\seqq\swapt \fromto_2 c_2 \times  c_1 \\
    \end{gather*}
\end{itemize}

Due to the strictness, we will omit the isomorphisms $\assocrp$, $\assoclp$, $\unitep$, $\unitip$, $\assocrt$, $\assoclt$, $\unitet$, $\unitit$, $\dist$, $\factor$, $\absorb$, and $\factorz$ used for type conversion in proofs for convenience.
For example, we rewrite \[c_1 + (c_2 + c_3) \fromto_2{} (c_1 + c_2) + c_3\] for simplicity.

\subsection{Additional Figure}

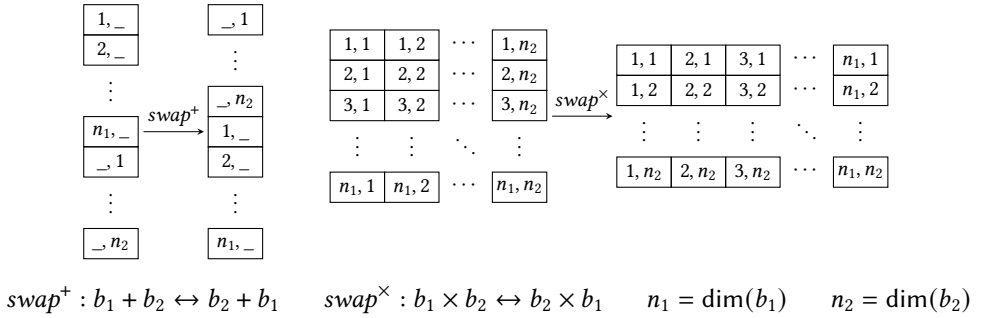
\begin{figure}[ht]
    \centering
    \scalebox{.8}{\begin{tikzpicture}
    [every node/.style={inner sep=2,minimum width=9mm,minimum height=5mm}]
    \matrix[
        nodes={
            rectangle,draw,
            align=center,
        },
        column sep=-.5pt,row sep=-.5pt,line width=.5pt,
    ] (swapp0){
        & \node{$\mathclap{1,\_}$}; \\
        & \node{$\mathclap{2,\_}$}; \\
        & \node[draw=none]{$\vdots$}; \\[2mm]
        & \node{$\mathclap{n_1,\_}$}; \\
        & \node{$\mathclap{\_,1}$}; \\
        & \node[draw=none]{$\vdots$}; \\[2mm]
        & \node{$\mathclap{\_,n_2}$}; \\
    };

    \matrix[
        right=of swapp0,
        nodes={
            rectangle,draw,
            align=center,
        },
        column sep=-.5pt,row sep=-.5pt,line width=.5pt,
    ] (swapp1){
        & \node{$\mathclap{\_,1}$}; \\
        & \node[draw=none]{$\vdots$}; \\[2mm]
        & \node{$\mathclap{\_,n_2}$}; \\
        & \node{$\mathclap{1,\_}$}; \\
        & \node{$\mathclap{2,\_}$}; \\
        & \node[draw=none]{$\vdots$}; \\[2mm]
        & \node{$\mathclap{n_1,\_}$}; \\
    };
    \draw[-stealth,] (swapp0) --node[above](swapp){$\swapp$} (swapp1);

    \matrix[
        right=2cm of swapp,
        nodes={
            rectangle,draw,
            align=center,
        },
        column sep=-.5pt,row sep=-.5pt,line width=.5pt,
    ] (swapt0){
        & \node{$\mathclap{1,1}$}; & \node{$\mathclap{1,2}$}; & \node[draw=none]{$\cdots$}; & \node{$\mathclap{1,n_2}$}; \\
        & \node{$\mathclap{2,1}$}; & \node{$\mathclap{2,2}$}; & \node[draw=none]{$\cdots$}; & \node{$\mathclap{2,n_2}$}; \\
        & \node{$\mathclap{3,1}$}; & \node{$\mathclap{3,2}$}; & \node[draw=none]{$\cdots$}; & \node{$\mathclap{3,n_2}$}; \\
        & \node[draw=none]{$\vdots$}; & \node[draw=none]{$\vdots$}; & \node[draw=none]{$\ddots$}; & \node[draw=none]{$\vdots$}; \\[2mm]
        & \node{$\mathclap{n_1,1}$}; & \node{$\mathclap{n_1,2}$}; & \node[draw=none]{$\cdots$}; & \node{$\mathclap{n_1,n_2}$}; \\
    };

    \matrix[
        right= of swapt0,
        nodes={
            rectangle,draw,
            align=center,
        },
        column sep=-.5pt,row sep=-.5pt,line width=.5pt,
    ] (swapt1){
        & \node{$\mathclap{1,1}$}; & \node{$\mathclap{2,1}$}; & \node{$\mathclap{3,1}$}; & \node[draw=none]{$\cdots$}; & \node{$\mathclap{n_1,1}$}; \\
        & \node{$\mathclap{1,2}$}; & \node{$\mathclap{2,2}$}; & \node{$\mathclap{3,2}$}; & \node[draw=none]{$\cdots$}; & \node{$\mathclap{n_1,2}$}; \\
        & \node[draw=none]{$\vdots$}; & \node[draw=none]{$\vdots$}; & \node[draw=none]{$\vdots$}; & \node[draw=none]{$\ddots$}; & \node[draw=none]{$\vdots$}; \\[2mm]
        & \node{$\mathclap{1,n_2}$}; & \node{$\mathclap{2,n_2}$}; & \node{$\mathclap{3,n_2}$}; & \node[draw=none]{$\cdots$}; & \node{$\mathclap{n_1,n_2}$}; \\
    };
    \draw[-stealth,] (swapt0) --node[above]{$\swapt$} (swapt1);

\end{tikzpicture}}
    \begin{align*}
        \swapp&: b_1+b_2 \fromto b_2+b_1 & \swapt &: b_1\times b_2 \fromto b_2\times b_1 & n_1 &= \hdim(b_1) & n_2 &= \hdim(b_2) 
    \end{align*}
    \caption{Visualisation of permutations in the definitions of $\wsem{\swapp}, \wsem{\swapt}$. The squares are numbered by row-major order.}
    \label{fig:visual_permutation}
\end{figure}
\fi

\end{document}